\pdfoutput=1
\documentclass{article}
\usepackage{etoolbox}
\newtoggle{draft}
\toggletrue{draft}
\usepackage[utf8]{inputenc}
\usepackage{amssymb,amsfonts}
\usepackage{amsthm}
\usepackage{mathtools}
\usepackage{hyperref}
\usepackage{natbib}
\hypersetup{
  colorlinks=true,
  linkcolor=green!30!black,
  linktoc=all,
  citecolor=blue,
  bookmarksnumbered=true,
  pdfproducer={},
  pdfauthor={},
  pdfkeywords={},
  pdfsubject={},
  pdfcreator={},
  pdflang={}
}
\usepackage{cleveref}

\newtheorem{theorem}{Theorem}[section]
\newtheorem{observation}[theorem]{Observation}
\newtheorem{claim}[theorem]{Claim}

\newtheorem{lemma}[theorem]{Lemma}
\newtheorem{notation}[theorem]{Notation}
\newtheorem{corollary}[theorem]{Corollary}
\newtheorem{proposition}[theorem]{Proposition}
\newtheorem{property}[theorem]{Property}

\newtheorem{problem}{Problem}
\theoremstyle{definition}
\newtheorem{definition}[theorem]{Definition}
\newtheorem{remark}[theorem]{Remark}

\crefname{theorem}{Theorem}{Theorems}
\crefname{observation}{Observation}{Observations}
\crefname{claim}{Claim}{Claims}
\crefname{condition}{Condition}{Conditions}
\crefname{example}{Example}{Examples}
\crefname{fact}{Fact}{Facts}
\crefname{lemma}{Lemma}{Lemmas}
\crefname{corollary}{Corollary}{Corollaries}
\crefname{definition}{Definition}{Definitions}
\crefname{remark}{Remark}{Remarks}
\crefname{proposition}{Proposition}{Propositions}
\crefname{property}{Property}{Properties}
\crefname{problem}{Problem}{Problems}
\crefname{section}{Section}{Sections}

\usepackage{mathrsfs}
\usepackage{amsmath}
\usepackage{tikz}
\usetikzlibrary{shapes,arrows,positioning}
\usetikzlibrary{er,positioning,bayesnet}
\usepackage{fullpage}
\usepackage[tracking]{microtype}
\usepackage{float}

\usepackage[linesnumbered,ruled]{algorithm2e}
\crefname{algocf}{alg.}{algs.}

\title{A Fixed-Parameter Tractable Algorithm for Counting Markov Equivalence Classes with the same Skeleton \footnote{Accepted to the Proceedings of the 38th Annual AAAI Conference on Artificial Intelligence (AAAI 2024)}} 
\author{}
 \author{Vidya Sagar Sharma\thanks{Tata Institute of Fundamental Research,
     Mumbai. Email: \texttt{vidyasagartifr@gmail.com}.} 
}
\date{}

\newcommand{\canonical}{canonical source vertex}

\newcommand{\cp}[0]{chordless path}
\newcommand{\cps}[0]{chordless paths}

\newcommand{\cc}[0]{chordless cycle}
\newcommand{\ucc}[0]{undirected connected component}
\newcommand{\uccs}[0]{undirected connected components}
\newcommand{\uccc}[0]{undirected connected chordal component}

\newcommand{\skel}[1]{\textup{skeleton}({#1})}
\newcommand{\skeleton}[1]{\text{skeleton}({#1})}

\newcommand{\setofpartialMECs}[1]{\textup{PMEC}(#1)}
\newcommand{\BaseCount}{\textsc{BruteForceCount}}
\newcommand{\setofMECs}[1]{\textup{MEC}(#1)}
\newcommand{\tfp}[0]{\text{triangle-free path}}
\newcommand{\tfps}[0]{\text{triangle-free paths}}
\newcommand{\spe}[0]{\text{strongly protected}}
\newcommand{\shadow}[0]{\textup{shadow}}

\newcommand{\shadowofMEC}[1]{\textup{shadow}(#1)}
\newcommand{\shadowofudgraph}[1]{\textup{shadow}(#1)}
\newcommand{\epfs}[0]{\text{derived path function}}
\newcommand{\EPF}[1]{\text{DPF}(#1)}

\newcounter{casenum}

\makeatletter
\newcommand*{\defeq}{\mathrel{\rlap{\raisebox{0.3ex}{$\m@th\cdot$}}\raisebox{-0.3ex}{$\m@th\cdot$}}=}
\newcommand*{\eqdef}{=
  \mathrel{\rlap{\raisebox{0.3ex}{$\m@th\cdot$}}\raisebox{-0.3ex}{$\m@th\cdot$}}}
\makeatother

\renewcommand{\emptyset}[0]{\varnothing}

\iftoggle{draft}{\usepackage{marginnote}
      }{}

\newcommand{\undir}{\ensuremath{-}}

\begin{document}

\maketitle
Causal DAGs (also known as Bayesian networks) are a popular tool for encoding
conditional dependencies between random variables.  In a causal DAG, the random
variables are modeled as vertices in the DAG, and it is stipulated that every
random variable is independent of its ancestors conditioned on its parents.  It
is possible, however, for two different causal DAGs on the same set of random
variables to encode exactly the same set of conditional dependencies.  Such
causal DAGs are said to be \emph{Markov equivalent}, and equivalence classes of
Markov equivalent DAGs are known as \emph{Markov Equivalent Classes} (MECs).
Beautiful combinatorial characterizations of MECs have been developed in the
past few decades, and it is known, in particular that all DAGs in the same MEC
must have the same ``skeleton'' (underlying undirected graph) and v-structures (induced subgraph of the form $a\rightarrow b \leftarrow c$).

These combinatorial characterizations also suggest several natural algorithmic
questions.  One of these is: given an undirected graph $G$ as input, how many
distinct Markov equivalence classes have the skeleton $G$?  Much work has been
devoted in the last few years to this and other closely related problems.
However, to the best of our knowledge, a polynomial time algorithm for the
problem remains unknown.

In this paper, we make progress towards this goal by giving a fixed parameter
tractable algorithm for the above problem, with the parameters being the
treewidth and the maximum degree of the input graph $G$.  The main technical
ingredient in our work is a construction we refer to as \emph{shadow},
which lets us create a ``local description'' of long-range constraints imposed
by the combinatorial characterizations of MECs.
 \section{Introduction}
\label{sec:introduction}
A graphical model is used to graphically represent a set of conditional independence relations between random variables. In the literature, both directed and undirected variations have been utilized to model various types of dependency structures. In this paper, our focus is on studying graphical models represented by directed acyclic graphs (DAGs), also called the Bayesian graphical model, which effectively conveys conditional independence relations and causal influences between random variables through the use of directed acyclic graphs~\citep{Pearl2009}. This model is well-studied and has been found extensive applications in many fields, such as material science \citep{ren_embedding_2020}, game theory \citep{kearns2013graphical}, and biology \citep{friedman2004inferring,finegold2011robust}.

In the Bayesian graphical model, a probability distribution over a set of random variables $V$ is said to satisfy a DAG $G$ with vertices $V$ if and only if for every $w \in V$, $w$ is independent of the set of all its non-descendants conditional on the set of its parents (the definitions of ``parent'' and ``non-descendant'' are provided in \Cref{sec:preliminary}). The set of probability distributions that satisfy $G$ is denoted by Markov$(G)$. Further, for disjoint $A, B, S \subseteq V$, $G$ is said to \emph{entail} that $A$ is independent of $B$ given $S$ (written $G\models A\perp B | S$) if and only if $A$ is independent of $B$ given $S$ in every probability distribution in Markov$(G)$.  Two DAGs are said to be \emph{Markov equivalent} if both entail the same set of such conditional independence relations.  \citet{verma1990equivalence} gave an elegant graphical characterization of this equivalence: two DAGs are Markov equivalent if and only if they have the same \emph{skeleton} (underlying undirected graph) and the same set of v-structures (induced subgraphs of the form $a\rightarrow b \leftarrow c$).  DAGs that are Markov equivalent to each other are said to belong to the same \emph{Markov equivalence class} (MEC) (see \Cref{fig:MEC-examples}).

\begin{figure}[ht]
    \centering
    \begin{tikzpicture}
    \node[](a1){$A$};
    \node[](b1)[below left=0.5 and 0.2 of a1] {$B$};
    \node[](c1)[below right=0.5 and 0.2 of a1] {$C$};    
    \draw[->](a1)--(b1);
    \draw[->](a1)--(c1);
    \node[][below = 0.8 of a1]{$D_2$};

    \node[](a2)[left = 1.5 of a1]{$A$};
    \node[](b2)[below left=0.5 and 0.2 of a2] {$B$};
    \node[](c2)[below right=0.5 and 0.2 of a2] {$C$};    
    \draw[<-](a2)--(b2);
    \draw[->](a2)--(c2);
    \node[][below = 0.8 of a2]{$D_1$};

    \node[](a3)[right = 1.5 of a1]{$A$};
    \node[](b3)[below left=0.5 and 0.2 of a3] {$B$};
    \node[](c3)[below right=0.5 and 0.2 of a3] {$C$};    
    \draw[->](a3)--(b3);
    \draw[<-](a3)--(c3);
    \node[][below = 0.8 of a3]{$D_3$};

    \node[](a4)[right = 1.5 of a3]{$A$};
    \node[](b4)[below left=0.5 and 0.2 of a4] {$B$};
    \node[](c4)[below right=0.5 and 0.2 of a4] {$C$};    
    \draw[<-](a4)--(b4);
    \draw[<-](a4)--(c4);
    \node[][below = 0.8 of a4]{$D_4$};

    \node[](a0)[below = 1.5 of a2]{$A$};
    \node[](b0)[below left=0.5 and 0.2 of a0] {$B$};
    \node[](c0)[below right=0.5 and 0.2 of a0] {$C$};    
    \draw[-](a0)--(b0);
    \draw[-](a0)--(c0);
    \node[][below = 0.9 of a0]{undirected graph $G$};

    \node[](a5)[below right = 1.5 and 0.7 of a1]{$A$};
    \node[](b5)[below left=0.5 and 0.2 of a5] {$B$};
    \node[](c5)[below right=0.5 and 0.2 of a5] {$C$};    
    \draw[-](a5)--(b5);
    \draw[-](a5)--(c5);
    \node[][below = 0.9 of a5]{$M_1 = \{D_1, D_2, D_3\}$};

    \node[](a6)[below = 1.5 of a4]{$A$};
    \node[](b6)[below left=0.5 and 0.2 of a6] {$B$};
    \node[](c6)[below right=0.5 and 0.2 of a6] {$C$};    
    \draw[<-](a6)--(b6);
    \draw[<-](a6)--(c6);
    \node[][below = 0.9of a6]{$M_2 = \{D_4\}$};

    \end{tikzpicture}
    \caption{Markov equivalent DAGs and MECs with the same skeleton. $D_1, D_2$, and $D_3$ are Markov equivalent, while $D_4$ is not equivalent to $D_1, D_2$, and $D_3$.  The MEC $M_1 = \{D_1, D_2, D_3\}$ contains $D_1, D_2$ and $D_3$, and its graphical representation is the union of $D_1, D_2$, and $D_3$. The MEC $M_2$ contains only $D_4$, and its graphical representation matches $D_4$.  The MECs $M_1$ and $M_2$ share $G$ as their skeleton, and in fact are the only MECs with skeleton $G$.  Both $M_1$ and $M_2$ entail a conditional independence relation of the form $B \perp C \mid S$, where in $M_1$, $S = \{C\}$, and in $M_2$, $S = \emptyset$.}
\label{fig:MEC-examples}
\end{figure}
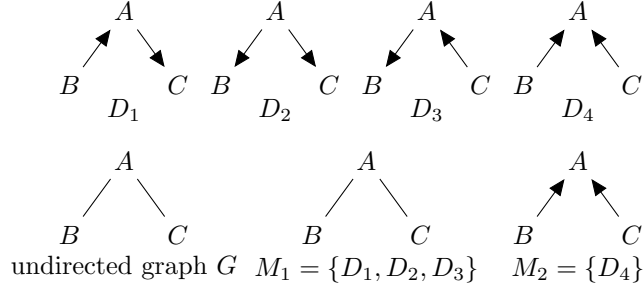

As an MEC consists of Markov equivalent DAGs, it uniquely represents the set of conditional independence relations represented by the DAGs it contains and is graphically represented by a partially directed graph which is the graphical union of the DAGs it contains.
We treat an MEC and its graphical representation as the same.
Since all the DAGs that belong to an MEC have the same skeleton and the set of v-structures, an MEC is uniquely determined by its skeleton and v-structures.
\citet{andersson1997characterization} gave a necessary and sufficient condition for a partially directed graph to be an MEC (see Theorem \ref{thm:nes-and-suf-cond-for-chordal-graph-to-be-an-MEC} below). For convenience, the graph representing an MEC is often considered synonymous with the MEC itself, and both are referred to as an ``MEC''.

\citet{meek1995causal} formulated rules to construct an MEC based on knowledge about the conditional independence relations among random variables. For a set of conditional independence relations involving a set of random variables $V$ and represented by an MEC $M$, \citet[p.~3]{meek1995causal} also showed that two random variables $A$ and $B$ are non-adjacent in the skeleton of $M$ if and only if there exists $S\subseteq V\setminus{\{A, B\}}$ such that $M$ entails that $A \perp B \mid S$. 
This implies that two MECs $M_1$ and $M_2$ sharing the same skeleton entail sets of conditional independence relations $\mathcal{M}_1$ and $\mathcal{M}_2$, respectively, which have the following important relation: For two random variables $A$ and $B$, $\mathcal{M}_1$ contains a conditional independent relation indicating that $A$ is independent of $B$ given some $S_1\subseteq V\setminus{\{A, B\}}$ if and only if $\mathcal{M}_2$ also contains a conditional independence relation signifying that $A$ is independent of $B$ given some $S_2\subseteq V\setminus{\{A, B\}}$ ($S_1$ are $S_2$ may not be same).  In other words, MECs with the same skeleton can be related by an equivalence relation that has not just a graphical representation (i.e., that they have the same underlying undirected graph), but also a natural statistical one (the one given above).

This connection between MECs with the same skeleton motivates the problem of understanding this class.   In particular, a natural question to ask is: how many MECs are in this class?

\paragraph{Our Contributions}
We now formalize the problem outlined above. Our input consists of a connected undirected graph $G$ with $n$ nodes. Our objective is to determine the count of MECs that have $G$ as their skeleton.
The primary contribution of this paper is the introduction of a fixed-parameter tractable (FPT) algorithm. This algorithm, when given an undirected graph $G$, computes the number of MECs whose skeletons match $G$. The algorithm's parameters are the degree and the treewidth of the input undirected graph.  (An algorithm is said to be FPT with respect to a parameter if there is a constant $c$ and a computable function $f$ such that the algorithm's runtime on instances of size $n$ for which the value of the parameter is $k$ is bounded above by $f(k)\cdot n^c$: the crucial point here is that the degree of the polynomial in $n$ \emph{does not} depend upon the parameter $k$ \citep{cygan_parameterized_2015}.)
Our main result presents an algorithm capable of counting the MECs associated with an input undirected graph $G$ with $n$ nodes and having a degree of $\delta$ and a treewidth of $k$. This counting can be achieved in $O(n(2^{O(k^4 \delta^4)} + n^2))$ time.
Importantly, the runtime of our algorithm remains polynomially bounded when the parameters $\delta$ and $k$ are both bounded above by constants. As an illustrative example, our algorithm demonstrates polynomial runtime for tree graphs with bounded degrees.

As of now, we do not know the precise computational complexity of this problem, and we consider our result as a first step towards a complete resolution of this question.
This mirrors the situation of the problem of counting DAGs of an MEC, where initially an algorithm that was exponential in the degree of the graph was given by \citet{ghassami2019counting} and then improved by \citet{talvitie2019counting} who gave a fixed parameter tractable algorithm for the problem.  Finally, a polynomial algorithm was provided by \citet{wienobst2020polynomial}.  We hope that the techniques introduced in our paper will be useful in the further study for the problem.

\paragraph{Related Work}
The problem of counting MECs with a given number of nodes (instead of a skeleton) has received extensive attention in the literature.  \citet{gillispie2013enumerating} developed a computer program for computing the number of MECs with $n$ nodes. \citet{gillispie2002size} created a computer program to enumerate Markov equivalence classes, studying class size distributions and the number of edges for graphs up to 10 vertices. They also observed that the ratio of DAGs to the number of MECs seems to asymptotically converge to around 3.7. \citet{steinsky2003enumeration} presented a recursive formula for counting Markov equivalence classes of size 1. \citet{gillispie2006formulas} provided a recursive algorithm for counting MECs of any size. 
\citet{he2013reversible} analyze the set of MECs  with $n$ nodes by constructing a Markov chain on the space of MECs and show that most edges of an MECs are directed. 
More recently, \citet{schmid2022number} show that the expected ratio of the number of DAGs and the number of MECs approaches a positive constant when the number of nodes goes to infinity.  

 All of the previous results focused on the class of MECs with a given number of nodes, and not with a given skeleton.
\citet{radhakrishnan2016counting} focused on counting MECs with the same skeleton. They classified MECs based on the number of v-structures present and derived a generating function for counting MECs. They experimentally demonstrated that the generating function varies for graphs with the same number of vertices. In subsequent work, \citet{radhakrishnan2018counting} delved further into the problem of counting MECs with the same skeleton. They explored generating functions for specific graph structures (e.g., path graphs, cycle graphs, star graphs, and bi-star graphs) and provided tight lower and upper bounds for the number of MECs in any tree.

We are not aware of any progress on the problem of counting MECs for general graphs, and to the best of our knowledge, this paper is the first to introduce a fixed-parameter tractable algorithm for this problem.

\paragraph{Technical Overview:}
For an undirected graph \(G\), we address the problem of counting Markov Equivalence Classes (MECs) of \(G\), i.e., counting MECs that have the skeleton \(G\). A trivial brute force approach is to iterate through each possible directed acyclic graph (DAG) having the skeleton \(G\). Since each DAG with the skeleton \(G\) is a member of a unique MEC of \(G\), and each MEC of \(G\) is a nonempty set of DAGs with the skeleton \(G\), we count the different MECs to which these DAGs belong. This approach yields the count of MECs with the skeleton \(G\). However, due to the fact that the number of DAGs with the skeleton \(G\) is \(2^{|E_G|}\) ($E_G$ is the set of edges of $G$), the runtime of this approach becomes exponential in the size of \(G\).

In this paper, we provide a fixed-parameter tractable (FPT) algorithm to count MECs of \(G\) such that the run time of our algorithm is exponential in terms of the parameters, degree, and treewidth of the input graph, but is polynomial in the size of the input graph, where the degree of the polynomial does not depend on the parameters.

We coin a new term ``shadow" (\cref{def:shadow}). A shadow of an MEC $M$ on $X \subseteq V_M$ is a triple $(O, P_1, P_2)$, where $O$ is the induced subgraph of $M$ on $X$, $P_1: E_O \times E_O \rightarrow \{0, 1\}$ is a function that for any $(u,v), (x,y) \in E_O$,  answers whether there exists a \tfp{} (\cref{def:tfp-path}) of length three or more in $M$ from $(u,v)$ to $(x,y)$, and $P_2: E_O \times V_O \rightarrow \{0,1\}$ ($V_O$ is the vertex set of $O$) is a function that for any  $((u,v), w) \in E_O \times V_O$,  answers whether there exists a \tfp{} of length three or more in $M$ from $(u,v)$ to $w$. A \tfp{} of a graph is a path of the graph such that for any three consecutive nodes $a, b$ and $c$ of the path, there is no edge in the graph with endpoints $a$ and $c$. A \tfp{} starts from an edge $(u,v)$ if the first two nodes of the path are $u$ and $v$. A \tfp{} ends on an edge $(x,y)$ if the last two nodes of the path are $x$ and $y$. A \tfp{} ends on a vertex $w$ if the last node of the path is $w$.

With a slight abuse of notation, for an undirected graph $G$, we define a triple $(O, P_1, P_2)$ as a shadow of $G$ if $O$ is a partial MEC of $G$ (a graph with skeleton $G$ and obeys \cref{item-1-theorem-nec-suf-cond-for-MEC,item-2-theorem-nec-suf-cond-for-MEC,item-3-theorem-nec-suf-cond-for-MEC} of \cref{thm:nes-and-suf-cond-for-chordal-graph-to-be-an-MEC}, \cref{def:partial-MEC}), and $P_1: E_O \times E_O \rightarrow \{0,1\}$ and $P_2: E_O \times V_O \rightarrow \{0,1\}$ are two functions. We show that for any undirected graph $G$, for any $X\subseteq V_G$, summation of $|\setofMECs{G, O, P_1, P_2}|$, the size of the set of MECs of $G$ having shadow $(O, P_1, P_2)$, over all shadows of $G[X]$ equals to the number of MECs of $G$ (\cref{lem:partition-of-MECs-of-H}). This reduces the problem of counting MECs of an undirected graph $G$ into counting MECs of $G$ with shadow $(O, P_1, P_2)$ for each possible shadow $(O, P_1, P_2)$ of $G[X]$ for some $X\subseteq V_G$.  It is the idea of ``shadow'', a nontrivial concise representation of an MEC, which is the key for the construction of the fixed-parameter tractable algorithm (\cref{alg:counting-MEC}).

We provide a recursive algorithm (\cref{alg:counting-MEC-of-general-graph}) to solve the reduced problem.
For this, we find a relation between the MECs of an undirected graph \(G\) and the MECs of an induced subgraph \(G'\) of \(G\). We note that for an MEC $M$ of $G$, $M[V_{G'}]$ may not be an MEC of $G'$. But, we found that there exists a unique MEC $M'$ of $G$ that has the same set of v-structure as $M[V_{G'}]$ (\Cref{lem:projection-of-an-MEC-is-unique}). We call the MEC $M'$ as the projection of $M$ on $V_{G'}$ (\cref{def:projection}).  We find a graphical resemblance between an MEC \(M\) of \(G\) and its projection (\cref{lem:directed-edge-is-same-in-projected-MEC}) that if an edge is directed in the projection then it is also directed in the MEC. 

For the recursion, we divide the graph using its tree decomposition.
Suppose $T$ is a tree decomposition of $G$ and $S_1$ is the root node of $T$.
We pick a neighbor $S_2$ of $T$.
We cut the edge $S_1-S_2$ of $T$ such that we get two induced subtrees $T_1$ and $T_2$ with roots $S_1$ and $S_2$, respectively, and representing the induced subgraphs $G_1$ and $G_2$ of $G$, respectively. 
From the tree decomposition properties, we have (a) $G$ is the union of $G_1$ and $G_2$, and (b) $I = V_{G_1} \cap V_{G_2}$ is a vertex separator of \(G\).
We start with finding a graphical relation between an MEC $M$ of $G$ and its projections $M_1$ and $M_2$ on $V_{G_1}$ and $V_{G_2}$ respectively.
Let $(O, P_1, P_2)$ be the shadow of $M$ on $S_1\cup S_2 \cup N(S_1\cup S_2, G)$ ($N(X, G)$ is the set of neighbors of the vertices in $X$ in $G$), $(O_1, P_{11}, P_{12})$ be the shadow of $M_1$ on $S_1 \cup N(S_1, G_1)$, and $(O_2, P_{21}, P_{22})$ be the shadow of $M_2$ on $S_2 \cup N(S_2, G_2)$. We find a graphical resemblance between the shadow of $M$, $M_1$ and $M_2$. We find necessary conditions that are obeyed by the shadows (\cref{obs1:O-structure-for-existence-of-MEC}). We later show that the conditions are also sufficient (\cref{obs2:O-structure-for-existence-of-MEC}). That means, if we have shadows $(O, P_1, P_2)$ of $G[S_1\cup S_2\cup N(S_1\cup S_2, G)]$, $(O_1, P_{11}, P_{12})$ of $G_1[S_1\cup N(S_1, G_1)]$, and $(O_2, P_{21}, P_{22})$ of $G_2[S_2\cup N(S_2, G_2)]$ that obeys the necessary and sufficient conditions then for each pair of MECs $(M_1, M_2)$ such that $M_1$ is an MEC of $G_1$ with shadow $(O_1, P_{11}, P_{12})$ and $M_2$ is an MEC of $G_2$ with shadow $(O_2, P_{21}, P_{22})$ there exists a unique MEC $M$ of $G$ with shadow $(O, P_1, P_2)$.
This relation shows that if we have the knowledge of $|\setofMECs{G_1, O_1, P_{11}, P_{12}}|$ for each shadow $(O_1, P_{11}, P_{12})$ of $G_1[S_1\cup N(S_1, G_1)]$, and $|\setofMECs{G_2, O_2, P_{21}, P_{22}}|$ for each shadow $(O_2, P_{21}, P_{22})$ of $G_2[S_2\cup N(S_2, G_2)]$ then we can compute $|\setofMECs{G, O, P_1, P_2}|$ for each shadow $(O, P_{1}, P_{2})$ of $G[S_1\cup S_2 \cup N(S_1\cup S_2, G)]$ (\cref{lem:equivalence-between-MECs}).

\Cref{lem:equivalence-between-MECs} provides us an approach to compute $|\setofMECs{G, O, P_1, P_2}|$ for each shadow $(O, P_1, P_2)$ of $G[S_1\cup S_2 \cup N(S_1\cup S_2, G)]$. But, for the recursion, we need to compute $|\setofMECs{G, O, P_1, P_2}|$ for each shadow $(O, P_1, P_2)$ of $G[S_1  \cup N(S_1, G)]$. Since a shadow of $G[S_1\cup S_2 \cup N(S_1\cup S_2, G)]$ contains more information than a shadow of $G[S_1  \cup N(S_1, G)]$, we can trim them to fulfill our purpose. We define the projection of shadow (\cref{def:proj-of-partial-MEC-and-functions}) for this.  A shadow $(O', P_1', P_2')$ of $G[S_1  \cup N(S_1, G)]$ is said to be a projection of a shadow $(O, P_1, P_2)$ of $G[S_1\cup S_2 \cup N(S_1\cup S_2, G)]$ if $O'$ is an induced subgraph of $O$, for all $((x,y), (u,v)) \in E_{O'}\times E_{O'}$, $P_1'((x,y), (u,v)) = P_1((x,y), (u,v))$, and for all $((u,v), w) \in E_{O'}\times V_{O'}$, $P_2'((u,v), w) = P_2((u,v), w)$. If a shadow $(O', P_1', P_2')$ of $G[S_1\cup N(S_1, G)]$ is a projection of a shadow $(O, P_1, P_2)$ of $G[S_1\cup S_2\cup N(S_1\cup S_2, G)]$ then if $M$ is an MEC of $G$ with shadow $(O, P_1, P_2)$ then $M$ has also $(O', P_1', P_2')$ as its shadow. For each shadow $(O', P_1', P_2')$ of $G[S_1\cup N(S_1, G)]$, this leads us to count MECs of $G$ with shadow $(O', P_1', P_2')$ when for each shadow $(O, P_1, P_2)$ of $G[S_1\cup S_2 \cup N(S_1\cup S_2, G)]$, we have the knowledge of $|\setofMECs{G, O, P_1, P_2}|$ (\cref{lem:counting-MECs-corresponding-to-projected-partial-MECs}).

For the recursion, we cut an edge of the tree decomposition of \(G\) to construct induced subgraphs \(G_1\) and \(G_2\). Since the number of edges of the tree decomposition of \(G\) is linear in the size of \(G\), the number of times we recursively call our algorithm is linear in the size of the input graph. The number of possible shadows is exponential in the treewidth and the degree of the input graph. The computation of the number of MECs having a specific shadow also takes time exponential in terms of the treewidth and the degree of the input graph. The overall time complexity (\Cref{thm:time-complexity-of-alg:counting-MEC}) of the algorithm is \(O(n(2^{O(k^4\delta^4}) + n^2))\), where the parameters \(k\) and \(\delta\) are the treewidth and the degree of the input graph, respectively, and \(n\) is the number of nodes of the input graph.

\paragraph{Paper Structure:} In \cref{sec:preliminary}, we define the terminologies used in this paper. In \Cref{subsection:old-results}, we go through basic graph theory terminologies (eg: path, cycle, chordal graphs, chain graphs, the union of graphs, tree decompositions, LBFS ordering, skeleton, v-structure), Markov equivalence class, and related results. In \cref{subsection:new-results}, we define the terminologies that are newly introduced in this paper, eg: Markov union of graphs, partial MECs, \tfp{}, a shadow of an MEC, projection of an MEC, and related results. In \cref{sec:counting-MEcs-of-a-graph}, we formally define the problem, and provide the necessary and sufficient conditions for the shadows (\cref{subsection:necessary-condition-of-shadow,subsection:sufficient-condition}). In \cref{subsection:LBFS}, we provide a modified version of LBFS algorithm. In \cref{subsection:dependence-of-P1-and-P2-on-O-and-others}, we define derived path functions and extension of shadows. In \cref{subsection:finding-size-of-MEC-H}, we provide a formula to compute $|\setofMECs{G, O, P_1, P_2}|$.  In \cref{sec:algorithm-for-counting-MECs}, we provide a fixed parameter tractable algorithm to count MECs of an undirected graph. \Cref{sec:time-complexity} analyzes the time complexity of the algorithms discussed in this paper.  In \Cref{sec:conclusion}, we discuss the open problems.

 \section{Preliminary}
\label{sec:preliminary}

\subsection{Old existing definitions and related results}
\label{subsection:old-results}
The following definitions and observations are well-studied terms and results in graph theory and graphical models.
\paragraph{Basic graph terminologies:}
\label{def:graph-terminologies}
A graph $G$ is a pair $(V, E)$, where $V$ is said to be the set of vertices of $G$, and $E \subseteq V \times V$ is said to be the set of edges of $G$. For $u, v \in V$, if $(u, v), (v, u) \in E$, then we say there is an undirected edge between $u$ and $v$, denoted as $u-v$. For $u, v \in V$, if $(u, v) \in E$ and $(v, u) \notin E$, then we say there is a directed edge from $u$ to $v$, denoted as $u \rightarrow v$. For $u, v \in V$, when we say that $(u, v)$ is an edge in $G$, it means that either there exists an undirected edge $u-v$ in $G$ or there exists a directed edge $u \rightarrow v \in G$.

For a graph $G$, we denote $V_G$ as the set of vertices of $G$, and $E_G$ as the set of edges of $G$. 
A partially directed graph is a graph in which some edges are directed and the remaining edges are undirected. 
If all the edges of a graph are undirected, then the graph is undirected. If all the edges of a graph are directed, then the graph is directed.
$u$ is said to be a neighbor of $v$ in a graph $G$ if either $(u,v) \in E_G$ or $(v,u) \in E_G$. $N(X, G)$ is the set of neighbors of $X$ in the graph $G$, i.e., $N(X, G) = \{v: \exists u \in X \text{ and either } (u, v) \in E_G \text{ or } (v,u) \in E_G$\}.
The degree of a node is the number of neighbors of the node.
The degree of a graph is the maximum degree of its nodes.

\begin{definition}[\textbf{Path, chord of a path, descendant, parent}]
\label{def:path}
For a graph $G$, a sequence $P = (u_1, u_2, \ldots, u_{l-1}, u_l)$ of distinct vertices is said to be a \textbf{\emph{path}} from $u_1$ to $u_{l}$ if for all $1 \leq i < l$, $(u_i, u_{i+1}) \in E_G$ (i.e., either $u_i-u_{i+1} \in E_G$ or $u_i\rightarrow u_{i+1} \in E_G$). The \textbf{\emph{length of the path}} is denoted by $l$.
$P$ is said to be an \textbf{\emph{undirected path}} of $G$ if for all $1 \leq i < l$, $u_i - u_{i+1} \in E_G$. Otherwise, if the path contains directed edges, it is called a \emph{directed path} of $G$.
For an undirected path $P = (u_1, u_2, \ldots, u_{l-1}, u_l)$ of $G$, we denote the \textbf{\emph{reverse of path}} $P$ as $\overline{P}$, i.e., $\overline{P} = (u_l, u_{l-1}, \ldots, u_2, u_1)$.
For a path $P$ of $G$, an edge in $G$ between two non-adjacent nodes of $P$ is said to be a \textbf{\emph{chord}} of $P$.
A path $P$ is said to be a \textbf{\emph{\cp{}}} of a graph $G$ if $G$ does not have any chord of $P$, i.e, between any two non-adjacent nodes of $P$ there is no edge in $G$. In other words, if $P = (u_1, u_2, \ldots, u_{l-1}, u_l)$ is a \cp{} of $G$ then for $1 \leq i < j \leq l$, if $j \neq i+1$, then neither $u_i - u_j \in E_G$, nor $u_i \rightarrow u_j \in E_G$, nor $u_j \rightarrow u_i \in E_G$.
In a directed graph $G$, $v$ is said to be a \textbf{\emph{descendant}} of $u$ if there exists a path from $u$ to $v$.
For a directed edge $y \rightarrow x$,  $y$ is said to be a \textbf{\emph{parent}} of $x$.

For a path $P= (u_1, u_2, \ldots , u_{l-1}, u_l)$, if the first two nodes of $P$ are $u$ and $v$ (i.e., $u_1 =u$ and $u_2 =v$), and the last
  two nodes of $P$ are $x$ and $y$ (i.e., $u_{l-1} = x$ and $u_l = y$), then we say $P$ is a \textbf{path from $(u,v)$ to
  $(x,y)$}
  (note that the boundary cases $(u, v) = (x, y)$ and $u \neq v = x \neq y$ are
  both allowed by this definition).
  Similarly, if the first two nodes of $P$ are $u$ and $v$, and the last node of
  $P$ is $w$ then we say $P$ is a \textbf{path from $(u,v)$ to $w$}
  (again note that the degenerate case $v = w$ is allowed by this definition).
\end{definition}

\begin{definition}[\textbf{Cycle, chord of a cycle, DAG}] 
\label{def:cycle}
For a graph $G$, a \emph{cycle} $C$ is a sequence $(u_1, u_2, \dots, u_{l}, u_{l+1} = u_1)$ of vertices in $G$ such that $l \geq 3$, and both $(u_1, u_2, \dots, u_{l})$ and $(u_l, u_1)$ are paths in $G$. The edge between two consecutive nodes of the sequence is said to be an edge of the cycle $C$.
An edge between two non-adjacent nodes of a cycle is said to be a \emph{chord of the cycle}.
$C = (u_1, u_2, \dots, u_{l}, u_{l+1} = u_1)$ is said to be a \textbf{\emph{chordless cycle}} if it does not have any chord, i.e.,  $l \geq 4$ and both $P_1 = (u_1, u_2, \ldots , u_{l-1})$ and $P_2 = (u_2, \ldots , u_{l-1}, u_{l})$ are chordless paths in $G$.
$C$ is said to be an \textbf{undirected cycle} of $G$ if all the edges of the cycle are undirected, i.e., for all $1 \leq i \leq l$, $u_i - u_{i+1} \in E_G$.
$C$ is said to be a \textbf{directed cycle} if at least one edge of $C$ is directed, i.e., there exists an $i$ such that $1 \leq i \leq l$ and $u_i \rightarrow u_{i+1} \in E_G$. A directed graph is said to be a \textbf{\emph{directed acyclic graph (DAG)}} if the graph does not have a cycle.
\end{definition}

\begin{definition}[\textbf{Chordal Graph}]
    An undirected graph $G$ is said to be a \emph{chordal graph} if there does not exist a chordless cycle in $G$.
\end{definition}

\begin{definition}[\textbf{Chain Graph}]
    A graph $G$ is said to be a \emph{chain graph} if there does not exist a directed cycle in $G$.
\end{definition}

\begin{proposition}
    \label{prop:every-chain-graph-has-a-topological-ordering}
    For every chain graph $G$, there always exists a topological ordering $\tau$ of the vertices of $G$ such that for $u\rightarrow v \in E_G$, $\tau(u) < \tau(v)$.
\end{proposition}
\begin{proof}
    We can construct a topological ordering \(\tau\) by first creating a topological ordering \(\tau_1\) of the vertices of the directed subgraph of \(G\) and then placing the nodes in \(\tau_1\) into \(\tau\) in the same order as they appear in \(\tau_1\). Next, we add the remaining nodes of \(G\) in any arbitrary order to complete \(\tau\). 
\end{proof}

\begin{proposition}
    \label{prop:edge-between-2-nodes-of-same-ucc-is-ud}
    Let $G$ be a chain graph, $\mathcal{C}$ be an \ucc{} of $G$, and $u,v \in \mathcal{C}$. Then $u\rightarrow v \notin G$.
\end{proposition}
\begin{proof}
    As both $u$ and $v$ belong to the same \ucc{} of $G$, there must exist an undirected path connecting them. If $u\rightarrow v \in G$, then combining this directed edge with the undirected path forms a directed cycle in $G$, which contradicts the definition of a chain graph. Thus, we conclude that $u\rightarrow v \notin G$. This completes the proof.
\end{proof}

\begin{proposition}
\label{prop:chordal-chain-graph-contains-directed-cycle-of-length-three}
Let $G$ be a graph with a chordal skeleton. If $C = (u_1, u_2, \ldots, u_l, u_1)$ is a directed cycle in $G$, then there exists a directed cycle $C' = (a, b, c, a)$ in $G$ of length three such that $a,b,c$ are in the cycle $C$.
\end{proposition}

\begin{proof}
Consider the smallest directed cycle $C'$ in $G$ that contains vertices of $C$. Such a cycle exists because $C$ exists. We aim to show that $C'$ must be of length three.

Suppose, for the sake of contradiction, that $C'$ has more than three vertices, so it can be expressed as $C' = (v_1, v_2, \ldots, v_m, v_1)$. Since the skeleton of $G$ is chordal, there must exist a chord in $C'$. Let's assume this chord is between vertices $v_i$ and $v_j$, where $i < j$, and $v_i$ and $v_j$ are not adjacent in $C$.

There are three possible cases: either $v_i - v_j$, $v_i \rightarrow v_j$, or $v_i \leftarrow v_j$. In each of these cases, either $(v_1, v_2, \ldots, v_i, v_j, v_{j+1}, v_m, v_1)$ or $(v_i, v_{i+1}, \ldots, v_j, v_i)$ forms a directed cycle in $G$ of length less than $m$, which contradicts the minimality of $C'$. Therefore, we conclude that $C'$ must be of length three.

This completes the proof of \cref{prop:chordal-chain-graph-contains-directed-cycle-of-length-three}.
\end{proof}

\begin{definition}[\textbf{Union of graphs}]
\label{def:union-of-graphs}
Let $G_1$ and $G_2$ be two graphs. $G$ is said to be the \emph{union} of $G_1$ and $G_2$ if $V_G = V_{G_1}\cup V_{G_2}$ and $E_G = E_{G_1} \cup E_{G_2}$.
\end{definition}

\begin{definition}[\textbf{Tree decomposition}]
\label{def:tree-decomposition}
A \emph{tree decomposition} of an undirected graph $G$ is a pair $(\mathcal{X}, T)$, where $\mathcal{X} = \{X_1, X_2, \ldots, X_l\}$ is a family of subsets of $V_G$, and $T$ is a tree whose nodes are subsets $X_i$, satisfying the following properties:
\begin{enumerate}
    \item $V_G = \bigcup_{i=1}^{l} X_i$.
    \item For every edge $(u,v)\in E_G$, there is a subset $X_i$ such that $u,v \in X_i$.
    \item For any node $v \in V_G$, the set of nodes $\{X_i \mid v \in X_i\}$ induces a connected subgraph in $T$. This is equivalent to if $X_i$ and $X_j$ both contain a node $v$, then all nodes in the path between $X_i$ and $X_j$ contain $v$.
\end{enumerate}
\end{definition}

\Cref{prop:tree-decomposition-property} is an intrinsic property of tree decomposition.
\begin{property}[\textbf{Intersection property of tree decomposition}, \cite{diestel2005graph}]
\label{prop:tree-decomposition-property}
Let $G$ be an undirected graph, and $(\mathcal{X} = \{X_1, X_2, \ldots, X_l\}, T)$ is a tree decomposition of $G$. For any edge $X_i - X_j \in E_T$, $X_i \cap X_j$ is a vertex separator of $G$.
\end{property}

\begin{definition}[\textbf{Perfect elimination ordering (PEO)}]
    Let $G$ be an undirected graph. Let $\tau$ be an ordering of its vertices. $\tau$ is said to be a \emph{perfect elimination ordering}  of $G$ if for all nodes $v$ of $G$, the subgraph induced in $G$ by the neighbors of $v$ that come before $v$ in $\tau$ forms a clique, i.e., for $u-v, w-v \in E_G$, if $\tau(u)<\tau(v)$ and $\tau(w)<\tau(v)$ then $u-w\in E_G$. 
\end{definition}
The following result is a classical result on the chordal graphs:
\begin{proposition}[\cite{dirac1961rigid}] 
    \label{prop:PEO-and-chordal-graph-relations}
    For a graph $G$, there exists a PEO if, and only if, $G$ is a chordal graph.
\end{proposition}
\begin{definition}[\textbf{LBFS Ordering}]
\label{def:LBFS}
    For a chordal graph $G$, \cite{rose1976algorithmic} construct a lexicographical breadth-first search (LBFS) algorithm that returns a PEO ordering of a chordal graph. An ordering of a graph $G$ that can be returned by the LBFS algorithm is said to be an \emph{LBFS ordering} of $G$. A modified LBFS algorithm is provided in \Cref{alg:LBFSwithO}.
\end{definition}

The following observations on chordal graphs come from \cite{rose1976algorithmic}.

\begin{observation}
\label{obs:LBGS-gives-PEO}
Let $\tau$ be an LBFS ordering of a chordal graph $G$. For $x,y,z\in V_G$, if both $x$ and $y$ are adjacent to $z$ (i.e., $x-z,y-z \in E_G$) and both are of lesser rank than $z$ in $\tau$ (i.e., $\tau(x)<\tau(z)$ and $\tau(y) < \tau(z)$) then $x$ and $y$ are also adjacent in $G$ (i.e., $x-y\in E_G$).
\end{observation}
\begin{proof}
An LBFS ordering of a chordal graph $G$ is known to be a perfect elimination ordering of $G$ \citep{rose1976algorithmic}. Given that $\tau$ represents an LBFS ordering of a chordal graph $G$, we can deduce that the subgraph induced by the neighbors of $z$ that appear before $z$ in $\tau$ forms a clique. Consequently, if $x$ and $y$ are neighbors of $z$ and satisfy $\tau(x) < \tau(z)$ and $\tau(y) < \tau(z)$ (i.e., both appear before $z$ in $\tau$), then $x$ and $y$ must belong to the same clique, implying that $(x,y) \in E_G$.
\end{proof}

\begin{observation}
\label{obs:chordless-path-and-LBFS-ordering}
Let $G$ be an undirected chordal graph, and $\tau$ be an LBFS ordering of the
vertices of $G$. For vertices $u, v \in V_G$ if
$P = (p_1 = u, p_2, \ldots, p_l =v)$ is a \cp{} from $u$ to $v$ in $G$ then
there exists $1 \leq i \leq l$ such that that
$\tau(p_1) > \tau(p_2) > \ldots > \tau(p_i)$, and
$\tau(p_i) < \tau(p_{i+1}) \dots < \tau(p_{l})$ (note that when $i = 1$ or
$i = l$ then one or the other of the above series of inequalities is empty).  In
other words, there are no ``internal local maxima'' of $\tau$ in $P$: $\tau$ is
either strictly increasing, or strictly decreasing, or has a unique global
minimum in the interior of $P$.
\end{observation}
\begin{proof}
  There cannot exist vertices $p_{j-1}, p_{j}, p_{j+1}$ such that
  $\tau(p_{j-1}) < \tau(p_j)$, and $\tau(p_{j+1}) < \tau(p_j)$. Otherwise, from
  \Cref{obs:LBGS-gives-PEO}, $p_{j-1}-p_{j+1} \in E_G$, which further implies
  that $P$ is not a \cp{}.
\end{proof}

\begin{definition}[\textbf{Skeleton} {\citep{verma1990equivalence}}]
\label{def:skeleton}
For a graph $G$, the \emph{skeleton} of $G$, denoted as $skel(G)$ or $G_U$, is the underlying undirected graph of $G$. More formally, the skeleton of a graph $G$ is an undirected graph obtained by replacing each directed edge $u\rightarrow v$ of $G$ with an undirected edge $u-v$.
\end{definition}

\begin{definition}[\textbf{v-structure} \citep{verma1990equivalence}]
\label{def:v-structure}
    For a graph $G$, an induced subgraph of $G$ of the form $a\rightarrow
    b \leftarrow c$ is said to be a v-structure of $G$. We denote the set of v-structures of $G$ by $\mathcal{V}(G)$. 
\end{definition}
\paragraph{\textbf{MEC:}}
 A directed acyclic graph (DAG) represents a set of conditional independence relations between random variables. Two DAGs are said to be Markov equivalent if both represent the same set of conditional independence relations between the random variables. A Markov equivalent class (MEC) is a set of DAGs that are Markov equivalent.  An MEC can be represented by a complete partial directed acyclic graph (CPDAG) which is a union of the DAGs contained by that MEC ~\citep{andersson1997characterization}. With a slight abuse of terminology, we equate the graph which represents an MEC with the MEC itself, and refer to both as an ``MEC''.  An MEC contains a unique set of v-structures. For an MEC, if we have given the skeleton of the MEC and its set of v-structures, we can construct the MEC  using Meek's rules \citep{meek1995causal}.  
A DAG and its corresponding MEC both have the same skeleton and the same set of v-structures. If we have a DAG then using Meek's rules, we can construct the MEC to which the DAG belongs.

\cite{andersson1997characterization} give the following necessary and sufficient conditions for a graph to be an MEC.
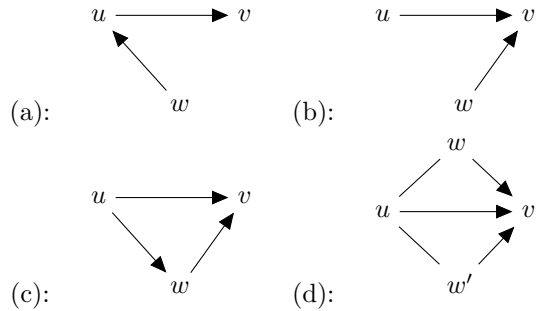
\begin{figure}[ht]
    \begin{center}
    \begin{tabular}{ c c c c }
 (a): &
 \begin{tikzpicture}
     \label{fig-a:strongly-protected-edge}
    \node[](u){$u$};
    \node[](v)[right= 1.5 of u]{$v$};
    \node[](w)[below right = 0.8 and 0.6 of u]{$w$};
    \draw[->](u)--(v);
    \draw[->](w)--(u);
    \end{tikzpicture}
    & (b): &
    \begin{tikzpicture}
    \node[](u){$u$};
    \node[](v)[right= 1.5 of u]{$v$};
    \node[](w)[below right = 0.8 and 0.6 of u]{$w$};
    \draw[->](u)--(v);
    \draw[->](w)--(v);
     \label{fig-b:strongly-protected-edge}
    \end{tikzpicture}
    \\ 
 (c): &
 \begin{tikzpicture}
    \node[](u){$u$};
    \node[](v)[right= 1.5 of u]{$v$};
    \node[](w)[below right = 0.8 and 0.6 of u]{$w$};
    \draw[->](u)--(v);
    \draw[->](u)--(w);
    \draw[->](w)--(v);
     \label{fig-c:strongly-protected-edge}
    \end{tikzpicture}
    & (d): &
    \begin{tikzpicture}
    \node[](u){$u$};
    \node[](v)[right= 1.5 of u]{$v$};
\node[](b)[below right=0.5 and 0.5 of u]{$w'$};
    \node[](a)[above right=0.5 and 0.5 of u]{$w$};
    \draw[->](u)--(v);
    \draw[-](u)--(b);
    \draw[->](b)--(v);
    \draw[-](a)--(u);
    \draw[->](a)--(v);
    \label{fig-d:strongly-protected-edge}
    \end{tikzpicture}
    \\     
\end{tabular}
\end{center}
\caption{Strongly protected $u \rightarrow v$.}
\label{fig:strongly-protected-edge}
\end{figure} \begin{theorem}[\cite{andersson1997characterization}]
\label{thm:nes-and-suf-cond-for-chordal-graph-to-be-an-MEC}
A graph $G$ is an MEC if, and only if, 
\begin{enumerate}
    \item
    \label{item-1-theorem-nec-suf-cond-for-MEC}
    $G$ is a chain graph.
    \item
    \label{item-2-theorem-nec-suf-cond-for-MEC}
    For every chain component $\tau$ of $G$, $G_{\tau}$ is chordal, i.e., every \ucc{} of $G$ is chordal.
    \item
    \label{item-3-theorem-nec-suf-cond-for-MEC}
    The configuration $a\rightarrow b - c$ does not occur as an induced subgraph of $G$.
    \item
    \label{item-4-theorem-nec-suf-cond-for-MEC}
    Every directed edge $u\rightarrow v \in G$ is strongly protected in $G$, i.e., $u\rightarrow v$ is a part of at least one of the subgraphs of $G$ as shown in \Cref{fig:strongly-protected-edge}.
\end{enumerate}
\end{theorem}

\subsection{New Terminologies and Related Results}
\label{subsection:new-results}
In this paper, we introduce the following definitions and observations, which play a crucial role in simplifying our main result. Firstly, we define synchronous graphs, the Markov union of synchronous graphs, and partial MECs. Additionally, we introduce new notations used in this paper to represent these concepts. In the latter part of this subsection, we conduct a detailed study of three important terms: (a) \tfp{}, (b) the shadow of an MEC and a graph, and (c) the projection of an MEC. We thoroughly examine these three terms and also explore the relevant results and implications of these newly introduced terminologies in the context of our paper.

\begin{definition}[\textbf{Synchronous Graphs}]
\label{def:synchronous-graphs} 
Two graphs $G$ and $H$ are said to be \emph{synchronous graphs} if there do not exist vertices $x, y \in V_{G}\cap V_{H}$ such that $x\rightarrow y \in E_G$ and $y\rightarrow x \in E_H$.
\end{definition}

\begin{definition}[\textbf{Markov Union of Synchronous Graphs}]
\label{def:Markov-union-of-graphs} 
    Let $G_1, G_2, \ldots, G_l$ be pairwise synchronous graphs. The \emph{Markov union} of $G_1, G_2, \ldots, G_l$ is a graph $G$, denoted by $G = U_M(G_1, G_2, \dots, G_l)$, such that $\skel{G}$ is the skeleton of the union of $G_1, G_2, \ldots, G_l$ (see \cref{def:union-of-graphs} for the definition of the graph union), and for any edge $u-v \in \skel{G}$, $u\rightarrow v \in G$ if $u\rightarrow v \in G_i$ for any $G_i$. More formally, $V_G = \bigcup_{i=1}^{l} V_{G_i}$, the set of directed edges of $G$ is $D_G = \{u\rightarrow v: u\rightarrow v \in E_{G_i}$ for some $1 \leq i \leq l\}$, and the set of undirected edges of $G$ is $U_G = \{u-v: u-v \in E_{G_i}$ for some $1 \leq i \leq l$, and there does not exist a $j$ such that $u\rightarrow v \in E_{G_j}$ or $v\rightarrow u \in E_{G_j}\}$. In other words, if $u\rightarrow v$ is a directed edge in any graph $G_i$, then $u\rightarrow v$ is also a directed edge in the Markov union graph, and an undirected edge $u-v$ of $G_i$ becomes an undirected edge of $G$ only if neither $u\rightarrow v$ nor $v\rightarrow u$ is part of any graph in the union. Since the graphs are pairwise synchronous, if one graph $G_i$ contains a directed edge $u\rightarrow v$, then no other graph $G_j$ can contain the directed edge $v\rightarrow u$.
\end{definition}

\begin{definition}[\textbf{Partial MEC}\footnote{The same notion appeared under the name `Restricted Chain Graph' in the earlier work of \citet{van2016separators,bang2023we}. We want to thank an anonymous reviewer for pointing out the references.}]
\label{def:partial-MEC}
A graph $M$ is said to be a \emph{partial MEC} if it satisfies the following conditions:
\begin{enumerate}
    \item
    \label{item-1-of-def:partial-MEC}
    $M$ is a chain graph.
    \item 
    \label{item-3-of-def:partial-MEC}
    Undirected connected components of $M$ are chordal.
    \item
    \label{item-2-of-def:partial-MEC}
    There does not exist a subgraph of the form $u\rightarrow v-w \in M$.
\end{enumerate}
i.e., a partial MEC obeys  \cref{item-1-theorem-nec-suf-cond-for-MEC,item-2-theorem-nec-suf-cond-for-MEC,item-3-theorem-nec-suf-cond-for-MEC} of \cref{thm:nes-and-suf-cond-for-chordal-graph-to-be-an-MEC}. 	
In other words, a partial MEC is a chain graph with chordal undirected connected components and no subgraph of the form $u\rightarrow v-w$.
\end{definition}
\Cref{def:partial-MEC} implies that every MEC is a partial MEC, and a partial MEC is an MEC if it also obeys \cref{item-4-theorem-nec-suf-cond-for-MEC} of \cref{thm:nes-and-suf-cond-for-chordal-graph-to-be-an-MEC}.

\begin{notation}
    \label{new-notation}
We say an MEC $M$ is an MEC of an undirected graph $G$ if $\skeleton{M} = G$. Similarly, we say a partial MEC $M$ is a partial MEC of an undirected graph $G$ if $\skeleton{M} = G$. For an undirected graph $G$, we denote MEC$(G)$ as the set of MECs of $G$, and $\setofpartialMECs{G}$ as the set of partial MECs of $G$.  
\end{notation}

\subsubsection{Triangle-Free Path}
\label{subsection:tfp}
\begin{definition}[\textbf{Triangle-free path}]
\label{def:tfp-path}
For a graph $G$, a path
  $P = (u_1, u_2, \ldots , u_{l-1}, u_l)$ is said to be a \emph{\tfp{}} of $G$ if for any node $u_i$ in $P$, there is no edge between the two adjacent nodes of $u_i$ in $P$. More formally, 
  for
  any $1\leq i \leq l-2$, neither $u_i - u_{i+2} \in E_G$, nor
  $u_i \rightarrow u_{i+2} \in E_G$, nor $u_{i+2} \rightarrow u_i \in E_G$,
  i.e., $u_i - u_{i+2} \notin \skel{G}$.

  Note that if $v\rightarrow u \in E_G$ then there cannot be a \tfp{} with the first two nodes $u$ and $v$ (if $v\rightarrow u \in E_G$ then from \Cref{def:path}, there cannot be a path with the first two nodes $u$ and $v$, as neither $u-v \in E_G$ nor $u\rightarrow v \in E_G$).
\end{definition}
\begin{observation}
\label{obs:cp-is-tfp}
For a graph $G$, every \cp{} $P = (u_1, u_2, \ldots, u_l)$ of $G$ is a \tfp{} of $G$.
\end{observation}
\begin{proof}
Since $P$ is a \cp{}, there cannot be an edge between two non-adjacent nodes of $P$. This implies there cannot be an edge between $u_i$ and $u_{i+2}$ for any $1\leq i \leq l-2$. This further implies $P$ is a \tfp{} of $G$.
\end{proof}
For our purposes, the utility of the notion of a \tfp{} comes from the following
partial converse.
\begin{proposition}
\label{obv:tfps-are-chordless-in-chordal-graphs}
  Let $G$ be an undirected chordal graph, and let $P$ be a \tfp{} in $G$.  Then,
  $P$ is also a \cp{}.
\end{proposition}
\begin{proof}
  Suppose, if possible, that $P = (u_{0}, u_1, \dots, u_{\ell})$ is a \tfp{} in
  $G$ that is not a \cp{}.  Choose a chord $u_{i} \undir{} u_{j}$ (assuming
  $i < j$ without loss of generality) of $P$ with the smallest possible value of
  $j - i$.  Since $P$ is a \tfp{}, we must have $j - i \geq 3$.  However, in
  that case $C = (u_i, u_{i+1}, \dots, u_j, u_i)$ is a simple cycle in $G$ of
  length at least four.  Since $G$ is chordal, this cycle must have a chord.  By
  the construction of $C$, this chord must be of the form $u_a \undir{} u_b$
  with $i \leq a < b \leq j$ and $1 < b - a < j - i$.  But then,
  $u_a \undir{} u_b$ is a chord of $P$ which contradicts the minimal choice of
  $j - i$.  We, therefore, conclude that $P$ must be a \cp{}.
\end{proof}

The following corollary of this partial converse is also useful.
\begin{corollary}\label{cor-adjacent-node-in-triangle-free-path}
  Let $G$ be an undirected chordal graph, and suppose that $u_1$ and $u_2$ are
  adjacent in $G$.  Suppose that $P = (v_1 = u_1, v_2, v_3, \dots, v_{l})$ is a
  \tfp{} in $G$ which contains $u_2$.  Then, it must be the case that $u_2$ is
  the second vertex in $P$ (i.e., $v_2 = u_2$).
\end{corollary}
\begin{proof}
  By \Cref{obv:tfps-are-chordless-in-chordal-graphs}, $P$ is a \cp{} in G, and
  hence no $v_i$ for $i > 2$ can be adjacent to $v_{1} = u_1$ (otherwise,
  $v_1 - v_i$ would be a chord of $P$).  Since $u_2$ is adjacent to $u_1$ and is
  contained in $P$, this only leaves open the possibility $v_2 = u_2$.
\end{proof}
The following observation extends the above to MECs.
\begin{observation}
\label{obs:tfp-in-M-is-a-cp}
Let $G$ be a chain graph with chordal undirected components, and let $P$ be an
undirected \tfp{} in $G$. Then, $P$ is an undirected \cp{} in $G$.
\end{observation}
\begin{proof}
  Since $P$ is undirected, it must be contained wholly in some undirected
  chordal component $H$ of $G$.  Thus,
  \Cref{obv:tfps-are-chordless-in-chordal-graphs} implies that $P$ must in fact
  be a \cp{} in $H$ (and therefore also in $G$).
\end{proof}

The following observation shows the transitive nature of \tfps{} in a chain graph.

\begin{observation}[\textbf{Concatenation of Triangle-Free Paths}]
\label{obs:concatenate-triangle-free-paths}
Consider a chain graph $G$ with chordal undirected components. Let $u, v, x, y, w\in V_G$ be (possibly non-distinct) vertices of $G$. Suppose $P_1 = (a_1 = u, a_2 = v, \dots, a_{l-1} = x, a_l = y)$ and $P_2 = (b_1 = x, b_2 = y, \dots, b_{m-1}, b_m = w)$ are \tfps{} in $G$ from $(u,v)$ to $(x,y)$ and from $(x, y)$ to $w$, respectively. Then, their concatenation $P \defeq (a_1 = u, a_2 = v, \ldots, a_{l-1} = x, a_l = y, b_3, b_4, \ldots, b_{m-1}, b_m = w)$ is a \tfp{} in $G$ from $(u,v)$ to $w$.
\end{observation}

\begin{proof}
From the definition of $P$, we observe that (i) for any two consecutive vertices $p$ and $q$ in $P$, either the edge $p \rightarrow q$ or the edge $p \undir{} q$ exists in $G$ because $p$ and $q$ must be consecutive in either $P_1$ or $P_2$, both of which are paths in $G$, and 
(ii) for any three consecutive vertices $p$, $q$, and $r$ in $P$, the vertices $p$, $q$, and $r$ are distinct, and $p$ and $r$ are not adjacent in $G$ because $p$, $q$, and $r$ must be consecutive in one of $P_1$ or $P_2$, both of which are \tfps{}.

  Thus, to show that $P$ is a \tfp{} in $G$ it only remains to show that all
  vertices in $P$ are distinct.
Suppose, if possible, that this is not the case.
Then, there must exist a segment of $P$ of the form
  $C \defeq (c_1 = p, c_2, \dots, c_{k-1}, c_k = p)$ in which the vertices
  $c_1, c_2, \dots, c_{k-1}$ are distinct.
From the observations above about $P$, we see that $C$ must be a cycle in $G$.
  Since $G$ is a chain graph, $C$ cannot have any directed edges. 
Therefore, $C$ must be entirely contained within an undirected component $H$ of $G$.
  Furthermore, from our second observation, $C$ must have at least four nodes. If $C$ does not have a chord in $H$ then $C$ is a chordless cycle in $G$,
which contradicts the assumption that the \uccs{} of $G$ are chordal. 
  Suppose $C$ has a chord $(c_i, c_j)$. We choose such $(c_i, c_j)$ with $i<j$ and $j-i$ being minimal. In this case, we have another cycle $C' \defeq (c_i, c_{i+1}, \ldots, c_j, c_{j+1} =c_i)$. We have selected $c_i$ and $c_j$ such that there is no edge between any two non-adjacent nodes of $C'$ (due to the minimality of $j-i$).
  According to our second observation, $C'$ must have at least four nodes. This makes $C'$ a chordless cycle, contradicting the assumption that the undirected components of $G$ are chordal.

  We conclude that all vertices in $P$ are distinct.  Combined with the above
  two observations about $P$, this shows that $P$ is a \tfp{} in $G$.
\end{proof}

The following observation is an easy consequence of the fact that the reverse of an undirected path is also a path. 

\begin{observation}
\label{obs:undirected-tfp-is-bidirectional}
For a graph $G$, if $P = (u_1 = u, u_2 =v, \ldots, u_{l-1} = x, u_l =y)$ is an
undirected triangle free path from $(u,v)$ to $(x,y)$ in $G$ then
$\bar{P} = (u_l = y, u_{l-1} = x, \ldots, u_2= v, u_1=u)$, the reverse of the
path $P$, is also a \tfp{} from $(y,x)$ to $(v,u)$.
\end{observation}
\begin{proof}
Suppose that $\bar{P}$ is not a \tfp{}, then there must exist nodes $u_i$ and $u_{i+2}$ such that there is an edge between $u_i$ and $u_{i+2}$. But, then, $P$ is also not a \tfp{}, which is a contradiction.
\end{proof}

The following observation comes from the \cref{item-3-theorem-nec-suf-cond-for-MEC} of \Cref{thm:nes-and-suf-cond-for-chordal-graph-to-be-an-MEC}.

\begin{observation}
\label{obs:undirected-tfp-with-ud-end-implies-source-is-ud}
\label{obs:undirected-tfp-in-M-is-a-cp}
Let $M$ be an MEC, and $P = (u_1=x, u_2=y, \ldots, u_{l-1}=u, u_l=v)$ be a \tfp{} from $(x,y)$ to $(u,v)$ in $M$. If $u-v \in M$ then  $P$ is an undirected \cp{} in $M$, more specifically $x-y \in M$.
\end{observation}
\begin{proof}
Since $P$ is a path, for each $1\leq i \leq l-1$, either $u_i-u_{i+1} \in M$ or $u_i\rightarrow u_{i+1} \in M$.
Suppose \Cref{obs:undirected-tfp-with-ud-end-implies-source-is-ud} is not true.
Then, pick the highest $i$ such that  $u_i\rightarrow u_{i+1} \in M$. $i$ must be less than $l-1$, as $u_{l-1}-u_l \in M$. But then $M$ has an induced subgraph $u_i \rightarrow u_{i+1} - u_{i+2}$, which is a contradiction from \Cref{item-3-theorem-nec-suf-cond-for-MEC} of \Cref{thm:nes-and-suf-cond-for-chordal-graph-to-be-an-MEC}. This implies all the edges of $P$ are undirected. This further implies $P$ is an undirected \tfp{} in $M$. From \Cref{obs:tfp-in-M-is-a-cp}, $P$ is a \cp{} in $M$ (from \Cref{thm:nes-and-suf-cond-for-chordal-graph-to-be-an-MEC}, $M$ is a chain graph with chordal undirected components).
\end{proof}

\Cref{obs:cp-is-tfp,obv:tfps-are-chordless-in-chordal-graphs,obs:concatenate-triangle-free-paths} imply the following:
\begin{lemma}
\label{lem:concatenation-of-chordal-paths}
    Let $G$ be an undirected chordal graph. Suppose $P_1 = (a_1 = u, a_2 = v, \ldots, a_{l-1} = x, a_l = y)$ and $P_2 = (b_1 = x, b_2 = y, \ldots, b_{m-1}, b_m = w)$ are \cp{} in $G$ from $(u,v)$ to $(x,y)$ and from $(x,y)$ to $w$, respectively. Then, their concatenation $P\defeq (a_1 = u, a_2 = v, \ldots, a_{l-1} = x, a_l = y, b_3, b_4, \ldots, b_{m-1}, b_m = w)$ is a \cp{} in $G$ from $(u,v)$ to $w$.
\end{lemma}

\begin{proof}
    From \cref{obs:cp-is-tfp}, $P_1$ and $P_2$ are \tfps{} in $G$. From \cref{obs:concatenate-triangle-free-paths}, the path $P$, concatenation of $P_1$ and $P_2$, is a \tfp{} in $G$. And, from \cref{obv:tfps-are-chordless-in-chordal-graphs}, $P$ is a \cp{} in $G$.
\end{proof}

We now introduce a \tfp{} variant of a source node. 

\begin{definition}[\textbf{Canonical Source Node}]
\label{def:canonical-source-node}
Let $G$ be a graph. A node $s\in V_G$ is said to be a \emph{canonical source node} of
$G$ if there does not exist a directed edge $u\rightarrow v \in G$ such that
there exists a \tfp{} from $(u,v)$ to $s$ in $G$.

In particular, note that the existence of an edge of the form $u \rightarrow s$
also prohibits $s$ from being a canonical source node, due to the presence of
the (degenerate) \tfp{} $(u, s)$ from $(u, s)$ to $s$.
\end{definition}

\begin{definition}[\textbf{Certificate for a non-canonical source node}]
\label{def:certificate}
    Let $G$ be a graph. For any non-canonical source node $z \in V_G$, we say $x$ is a \emph{certificate} of $z$ if there exists a $y\in V_G$ such that $x\rightarrow y \in G$ and there exists a \tfp{} from $(x,y)$ to $z$. In other words, $x$ certifies that $z$ is not a canonical source node in $G$.
\end{definition}
The following observation shows that self-certification is not possible.
\begin{observation}
    \label{obs:self-certification-not-possible}
    Let $G$ be a graph, and $u\in V_G$ such that $u$ is not a canonical source node. Then $u$ cannot be a certificate of itself.
\end{observation}
\begin{proof}
    Suppose that, if possible, $u$ is a certificate of itself. Then from \Cref{def:certificate}, there exists a $v\in V_G$ such that $u\rightarrow v \in G$ and there exists a \tfp{} $Q = (u_1=u, u_2 = v, \ldots, u_l =u)$ from $(u,v)$ to $u$. But,  then from the definition of a path (\Cref{def:path}), $Q$ is not a path. Hence, $Q$ is not a \tfp{}. Therefore, $u$ cannot be a certificate of itself. 
\end{proof}

We now provide an algorithm that for an input chain graph $G$ with chordal \uccs{}, outputs a pair of functions $(P_1, P_2)$ such that for two distinct edges $(u,v)$ and $(x,y)$ in $G$, $P_1((u,v),(x,y)) = 1$ if there exists a \tfp{} from $(u,v)$ to $(x,y)$ in $G$, else $P_1((u,v),(x,y)) = 0$. And, for an edge $(u,v)$ and a vertex $w$ of $G$, $P_2((u,v),w) =1$ if $v\neq w$, and there exists a \tfp{} from $(u,v)$ to $w$ in $G$, else $P_2((u,v),w) = 0$.

\begin{algorithm}[ht]
\caption{TFP($G$)}
\label{alg:tfp}
\SetAlgoLined
\SetKwInOut{KwIn}{Input}
\SetKwInOut{KwOut}{Output}
\SetKwFunction{MEC-Construction}{MEC-Construction}
\KwIn{A chain graph $G$ with chordal \uccs{}}
    \KwOut{$(P_1, P_2)$ such that \\
    $P_1:E_{G} \times E_{G} \rightarrow \{0,1\}$ and 
    $P_2: E_{G} \times V_{G} \rightarrow \{0,1\}$ such that \\
    (a) for $(u,v),(x,y) \in E_{G}$, $P_1((u,v),(x,y)) = 1$ if $(u,v)\neq (x,y)$ and\\
    there exists a \tfp{} from $(u,v)$ to $(x,y)$ in $G$ else $P_1((u,v),(x,y)) = 0$, and\\
    (b) for $((u,v), w) \in E_{G} \times V_{G}\rightarrow \{0,1\}$, $P_2((u,v),w) = 1$ if $v\neq w$ and\\ 
    there exists a \tfp{} from $(u,v)$ to $w$ in $G$ else $P_2((u,v),w) = 0$.}

    Construct a function $P_1: E_G\times E_G \rightarrow \{0,1\}$.

    Construct a function $P_2: E_G \times V_G \rightarrow \{0,1\}$.

\ForEach{$((u,v), (x,y)) \in E_{G} \times E_{G}$\label{alg:tfp:first-for-each-start}}
    {
        $P_1((u,v), (x,y)) \leftarrow 0$
    }\label{alg:tfp:first-for-each-end}

    \ForEach{$((u,v), w) \in E_{G} \times V_{G}$\label{alg:tfp:second-for-each-start}}
    {
        $P_2((u,v), w) \leftarrow 0$
    }\label{alg:tfp:second-for-each-end}

   \ForEach{triple $(u,v,w)$\label{alg:tfp:third-for-each-start}}
   {
        \If{$(u,v,w)$ is a \tfp{} in $G$}
        {
            $P_1((u,v),(v,w)) = 1$

            $P_2((u,v), w) = 1$
        }        
   }\label{alg:tfp:third-for-each-end}

\While{$\exists (u,v),(x,y),(z_1, z_2) \in E_{G}$ such that $P_1((u,v),(x,y)) =0$, $(u,v)\neq (x,y)$, and $P_1((u,v),(z_1,z_2)) = P_1((z_1,z_2),(x,y)) = 1$ \label{alg:tfp:first-while-loop-start}}
   {
        $P_1((u,v),(x,y)) = 1$
   }\label{alg:tfp:first-while-loop-end}

\While{$\exists (u,v), (z_1, z_2) \in E_{G}$ and $w\in V_{G}$ such that $P_2((u,v),w) =0$, $v\neq w$, and $P_1((u,v),(z_1,z_2)) = P_2((z_1,z_2),w) = 1$\label{alg:tfp:second-while-loop-start}}
   {
        $P_2((u,v),w) = 1$
   }\label{alg:tfp:second-while-loop-end}
   
   \KwRet $(P_1, P_2)$ \label{alg:tfp:return-statement}
\end{algorithm} \begin{lemma}
\label{lem:alg:tfp-is-valid}
Let $G$ be a chain graph with chordal \uccs{}. For input $G$, \Cref{alg:tfp} returns $(P_1, P_2)$, a tuple of functions, such that (a) $P_1:E_{G} \times E_{G} \rightarrow \{0,1\}$ and for $(u,v),(x,y) \in E_{G}$, $P_1((u,v),(x,y)) = 1$ if and only if $(u,v)\neq (x,y)$ and there exists a \tfp{} from $(u,v)$ to $(x,y)$ in $G$, and (b) $P_2:E_{G} \times V_{G} \rightarrow \{0,1\}$, and for $((u,v),w) \in E_{G} \times V_{G}$, $P_2((u,v),w) = 1$ if and only if $v\neq w$ and there exists a \tfp{} from $(u,v)$ to $w$ in $G$.
\end{lemma}
\begin{proof}
We will prove the correctness of $P_1$ and $P_2$ separately.

\textbf{Correctness of $P_1$:}
We first prove that for all $(u,v),(x,y) \in E_{G}$ such that $(u,v) \neq (x,y)$ and there exists a \tfp{} $Q = (u_1 = u, u_2 = v, \ldots, u_{l-1} = x, u_l = y)$ from $(u,v)$ to $(x,y)$ in $G$, we have $P_1((u,v),(x,y)) = 1$. We will use induction on the length $l$ of the \tfp{}. Since $(u,v)\neq (x,y)$, $l$ must be greater than one.

If $l = 2$, then $v=x$, and $Q = (u,v=x, y)$ is a \tfp{} from $(u,v)$ to $(x,y)$ in $G$. Therefore, while running lines \ref{alg:tfp:third-for-each-start}-\ref{alg:tfp:third-for-each-end} of \Cref{alg:tfp}, we have $P_1((u,v),(x,y)) = 1$. This implies that for all pairs of distinct edges $(u,v),(x,y) \in E_{G}$, if there exists a \tfp{} from $(u,v)$ to $(x,y)$ of length two then $P_1((u,v),(x,y)) = 1$.

Let us assume that for some $k \geq 2$, for all pairs of distinct edges $(u,v),(x,y)\in E_{G}$, if there exists a \tfp{} from $(u,v)$ to $(x,y)$ in $G$ of length at most $k$ then $P_1((u,v),(x,y)) = 1$. We now show that for $(u,v),(x,y) \in E_{G}$ such that $(u,v)\neq (x,y)$, if there exists a \tfp{} $Q =(u_1 =u, u_2 =v, \ldots, u_{l-1} =x, u_l =y)$ from $(u,v)$ to $(x,y)$ in $G$ of length $l=k+1$, then $P_1((u,v),(x,y)) = 1$.

Take two subpaths of $Q$: $Q_1 = (u_1 =u, u_2 =v, u_3)$ and $Q_2 = (u_2 =v, \ldots, u_{l-1} =x, u_l =y)$. Both paths must be \tfps{} of length at most $k$. Therefore, from our assumption, $P_1((u_1=u, u_2=v),(u_2,u_3)) = P_1((u_2,u_3),(u_{l-1}=x,u_l=y)) = 1$. But then, while running lines \ref{alg:tfp:first-while-loop-start}-\ref{alg:tfp:first-while-loop-end}, we get $P_1((u,v),(x,y)) = 1$. This completes the induction. Thus, we have proven that for all $(u,v),(x,y) \in E_{G}$ such that $(u,v)\neq (x,y)$ and there exists a \tfp{} from $(u,v)$ to $(x,y)$ in $G$, we have $P_1((u,v),(x,y)) = 1$.

We now prove that for all $(u,v),(x,y) \in E_{G}$, if $P_1((u,v),(x,y)) = 1$ then $(u,v)\neq (x,y)$ and there exists a \tfp{} from $(u,v)$ to $(x,y)$ in $G$. Suppose that this is not true. Then there must exist a pair of edges $(u,v),(x,y) \in E_{G}$ such that $P_1((u,v),(x,y)) =1$ and either $(u,v)\neq (x,y)$ or there does not exist a \tfp{} from $(u,v)$ to $(x,y)$ in $G$.

For all pairs of edges $(u,v),(x,y) \in E_{G}$, \Cref{alg:tfp} initializes $P_1((u,v),(x,y)) =0$ (lines \ref{alg:tfp:first-for-each-start}-\ref{alg:tfp:first-for-each-end}).
Pick the first such pair $(u,v),(x,y) \in E_{G}$ for which $P_1((u,v),(x,y)) =1$ and either $(u,v)\neq (x,y)$ or there does not exist a \tfp{} from $(u,v)$ to $(x,y)$, i.e., for all the pairs $(u',v'),(x',y') \in E_{G}$, for which the value of $P_1((u',v'),(x',y'))$ has been updated to 1 before the update of $P_1((u,v),(x,y))$, $(u',v')\neq (x',y')$ and there exists a \tfp{} from $(u',v')$ to $(x',y')$.
In the algorithm, for any pair of edges $(u,v)$ and $(x,y)$, \Cref{alg:tfp} updates $P_1$  either at lines \ref{alg:tfp:third-for-each-start}-\ref{alg:tfp:third-for-each-end}, or at lines \ref{alg:tfp:first-while-loop-start}-\ref{alg:tfp:first-while-loop-end}.
If the value of $P_1((u,v),(x,y))$ has been updated while running lines \ref{alg:tfp:third-for-each-start}-\ref{alg:tfp:third-for-each-end} of \Cref{alg:tfp} then we have $v=x$ and $(u,v=x, y)$ is a \tfp{}. This further implies $(u,v)\neq (x,y)$ and there exists a \tfp{} from $(u,v)$ to $(x,y)$. 

And, if $P_1((u,v),(x,y))$ has been updated while running lines \ref{alg:tfp:first-while-loop-start}-\ref{alg:tfp:first-while-loop-end} of \Cref{alg:tfp} then we have $(u,v)\neq (x,y)$ and there exists a $(z_1,z_2) \in E_{G}$ such that $P_1((u,v),(z_1, z_2)) = P_1((z_1, z_2),(x,y)) =1$. Since $P_1((u,v),(z_1, z_2))$ and $P_1((z_1, z_2),(x,y))$ have been updated before the update of $P_1((u,v),(x,y))$, therefore, from our assumption, (a) $(u,v)\neq (z_1,z_2)$ and there exists a \tfp{} $Q_1=(u_1 = u, u_2 =v, u_{l'-1} =z_1, u_{l'} =z_2)$ from $(u,v)$ to $(z_1,z_2)$, and (b) $(z_1, z_2)\neq (x,y)$ and $Q_2 = (v_1 = z_1, v_2 =z_2, \ldots, v_{m-1}=x, v_m =y)$ is a \tfp{} from $(z_1, z_2)$ to $(x,y)$, respectively. But, then from \Cref{obs:concatenate-triangle-free-paths}, concatenation of $Q_1$ and $Q_2$ gives us a \tfp{} from $(u,v)$ to $(x,y)$. Thus, $P_1((u,v),(x,y))=1$ correctly depicts that $(u,v)\neq (x,y)$ and there exists a \tfp{} from $(u,v)$ to $(x,y)$. This completes the proof of correctness of $P_1$. 

\textbf{Correctness of $P_2$:}
We first prove that for all $((u,v),w) \in E_{G}\times V_{G}$ such that $v\neq w$ and there exists a \tfp{} $Q =(u_1 =u, u_2 =v, \ldots,  u_l =w)$ from $(u,v)$ to $w$ in $G$ we have $P_2((u,v),w) = 1$. We prove this using induction on the length $l$ of the \tfp{}.  Since $v\neq w$, $l$ must be greater than one. 

If $l = 2$, then $Q = (u,v,w)$ is a \tfp{} from $(u,v)$ to $w$ in $G$. Therefore, while running lines \ref{alg:tfp:third-for-each-start}-\ref{alg:tfp:third-for-each-end} of \Cref{alg:tfp}, we get $P_2((u,v),w) = 1$. This implies that for all $((u,v),w) \in E_{G}\times V_{G}$, if there exists a \tfp{} from $(u,v)$ to $w$ in $G$ of length two then $P_2((u,v),w) = 1$.

Let us assume that for some $k \geq 2$, for all $((u,v),w) \in E_{G}\times V_{G}$, if $v\neq w$ and there exists a \tfp{} from $(u,v)$ to $w$ in $G$ of length at most $k$ then $P_2((u,v),w) = 1$. We now show that for all $((u,v),w) \in E_{G}\times V_{G}$ such that $v\neq w$, if there exists a \tfp{} $Q =(u_1 =u, u_2 =v, \ldots, u_l =w)$ from $(u,v)$ to $w$ in $G$ of length $l=k+1$, then $P_2((u,v),w) = 1$.
 
 Take two subpaths of $Q$: $Q_1 = (u_1 =u, u_2 =v, u_3)$ and $Q_2 = (u_2 =v, \ldots, u_{l-1}, u_l =w)$. Since $Q$ is a \tfp{}, therefore, the subpath $Q_1$ is a \tfp{} from $(u_1=u, u_2=v)$ to $(u_2,u_3)$.  Similarly, the subpath $Q_2$ is a \tfp{} from $(u_2, u_3)$ to $u_l=w$ of length $k$. Since the length of $Q_1$ is two, therefore $(u_1,u_2)\neq (u_2,u_3)$. Therefore, from the correctness of $P_1$, $P_1((u_1=u, u_2=v),(u_2,u_3)) =1$. Since the length of $Q_2$ is $k$, therefore, $u_3\neq u_l$ (otherwise $Q_2$ is of length $1 < k$, a contradiction).  Therefore, from our assumption, $P_2((u_2,u_3), u_l = w) =1$. But, then while running lines \ref{alg:tfp:second-while-loop-start}-\ref{alg:tfp:second-while-loop-end}, we get $P_2((u, v),w) =1$ (since $P_1((u_1=u, u_2=v), (u_2, u_3)) = P_2((u_2,u_3),u_l =w) =1$). This completes the induction. Thus, we prove that for all $((u,v),w) \in E_{G}\times V_{G}$ such that $v\neq w$ and there exists a \tfp{} from $(u,v)$ to $w$ in $G$, we have $P_2((u,v),w) = 1$.

We now prove that for all $((u,v),w) \in E_{G}\times V_{G}$, if $P_2((u,v),w) = 1$ then $v\neq w$ and there exists a \tfp{} from $(u,v)$ to $w$ in $G$. Suppose this is not true. Then there must exist $((u,v),w) \in E_{G}\times V_{G}$ such that $P_2((u,v),w) =1$ and either $v= w$ or there does not exist a \tfp{} from $(u,v)$ to $w$ in $G$.

For all the pair $((u,v),w) \in E_{G}\times V_{G}$, \Cref{alg:tfp} initializes  $P_2((u,v),w) =0$ (lines \ref{alg:tfp:first-for-each-start}-\ref{alg:tfp:first-for-each-end}). If possible, pick the first such $((u,v),w) \in E_{G}\times V_{G}$ for which $P_2((u,v),w) =1$ and either $v = w$ or there does not exist a \tfp{} from $(u,v)$ to $w$ in $G$, i.e., for all $((u',v'),w') \in E_{G}\times V_{G}$ for which the value of $P_2((u',v'),w')$ has been updated to 1 before the update of $P_2((u,v),w)$, we have $v'\neq w'$ and there exists a \tfp{} from $(u',v')$ to $w'$.
 In the algorithm, for any $((u,v),w) \in E_{G}\times V_{G}$, \Cref{alg:tfp} updates $P_2$ either at lines \ref{alg:tfp:third-for-each-start}-\ref{alg:tfp:third-for-each-end}, or at lines \ref{alg:tfp:second-while-loop-start}-\ref{alg:tfp:second-while-loop-end}. Suppose the value of $P_2((u,v),w)$ has been updated while running lines \ref{alg:tfp:third-for-each-start}-\ref{alg:tfp:third-for-each-end} of \Cref{alg:tfp} then we must have $(u,v,w)$ as a \tfp{} from $(u,v)$ to $w$. Also, $v\neq w$ (otherwise, $(u,v,w)$ is even not a path). 

 Now we move to the remaining possibility. Suppose $P_2((u,v),w)$ has been updated while running lines \ref{alg:tfp:second-while-loop-start}-\ref{alg:tfp:second-while-loop-end} of \Cref{alg:tfp} then we must have $v\neq w$ and there exists a $(z_1,z_2) \in E_{G}$ such that $P_1((u,v),(z_1, z_2)) = P_2((z_1, z_2),w) =1$. We have shown the correctness of $P_1$ earlier. Therefore, $P_1((u,v),(z_1, z_2)) = 1$ implies that $(u,v)\neq (z_1,z_2)$ and there exist \tfps{} $Q_1=(u_1 = u, u_2 =v, u_{l'-1} =z_1, u_{l'} =z_2)$ from $(u,v)$ to $(x,y)$. Also, $P_2((z_1, z_2),w)$ has been already updated. Therefore, from our assumption, there exists a \tfp{} $Q_2 = (v_1 = z_1, v_2 =z_2, \ldots, v_{m-1}, v_m =w)$ from $(z_1, z_2)$ to $w$. But, then from \Cref{obs:concatenate-triangle-free-paths}, concatenation of $Q_1$ and $Q_2$ gives us a \tfp{} from $(u,v)$ to $w$. Thus, $P_2((u,v),w)=1$ corrects depicts that $v\neq w$ and there exists a \tfp{} from $(u,v)$ to $w$ in $G$. This completes the proof of \Cref{lem:alg:tfp-is-valid}.
\end{proof}
 
\subsubsection{\textbf{Shadow of an MEC}}
\label{subsection:shadow}
\begin{definition}[\textbf{Shadow of an MEC}]
    \label{def:shadow}
    Let $M$ be an MEC of $G$, $Y\subseteq V_G$,  $O\in \setofpartialMECs{G[Y]}$ (i.e., $O$ be a partial MEC with skeleton $G[Y]$), and $P_1: E_{O}\times E_{O} \rightarrow \{0,1\}$ and $P_2: E_{O}\times V_O \rightarrow \{0,1\}$ be two functions. We define $(O, P_1, P_2)$ as a \shadow{} of the MEC $M$ on $Y$, denoted by $(O, P_1, P_2) = \shadowofMEC{M, Y}$, if the following occur:
    \begin{enumerate}
        \item
        \label{item-1-of-def:shadow}
        $M[Y] = O$, i.e., $O$ is an induced subgraph of $M$ on $Y$.
        \item
        \label{item-2-of-def:shadow}
        For $(u,v), (x,y) \in E_{O}$, $P_1((u,v),(x,y)) = 1$ if $(u, v) \neq (x, y)$ and there exists a \tfp{} from $(u,v)$ to $(x,y)$ in $M$, else $P_1((u,v),(x,y)) =0$.
        \item 
        \label{item-3-of-def:shadow}
        For $(u,v) \in E_{O}$ and $w\in V_O$, $P_2((u,v),w) =1$, if $v \neq w$
        and there exists a \tfp{} from $(u,v)$ to $w$ in $M$, else
        $P_2((u,v), w) =0$.
    \end{enumerate}
    With a slight abuse of notation, we say $(O, P_1, P_2)$ as a shadow of an undirected graph $G$, denoted by $(O, P_1, P_2) \in \shadowofudgraph{G}$,  if $O\in \setofpartialMECs{G}$, and $P_1: E_{O}\times E_{O} \rightarrow \{0,1\}$ and $P_2: E_{O}\times V_O \rightarrow \{0,1\}$ are two functions.
\end{definition}
\begin{remark}[\textbf{$P_1$ does not determine $P_2$}]\label{rem:p2-not-determined-by-p1}
  Note that $P_2$ is not necessarily determined by $P_1$.  This is because there
  can exist $(u,v) \in E_{O}$ and $w \in Y$ such that there exists a \tfp{}
  from $(u,v)$ to $w$ in $M$, i.e., $P_2((u,v),w) =1$, but for which there does
  not exist any edge $(x,w) \in E_{O}$ such that $P_1((u,v),(x,w)) =1$.
This can happen when every \tfp{}
  $P = (u_1 =u, u_2 =v, \ldots, u_{l-1}, u_l =w)$ in $M$ from $(u,v)$ to $w$
  has $u_{l-1} \notin Y$.
\end{remark}

\begin{definition}[$\setofMECs{G, O, P_1, P_2}$]
\label{def:subsets-of-MEC-based-on-O-P1-and-P2}
Let $G$ be an undirected graph, $Y\subseteq V_G$, and $(O, P_1, P_2) \in \shadowofudgraph{G[Y]}$.
We define $\setofMECs{G, O, P_1, P_2}$ to be the set of MECs $M$ of $G$ for which $(O, P_1, P_2)$ is a \shadow{} of $M$ on $Y$. More formally, $\setofMECs{G, O, P_1, P_2} \coloneqq \{M: M \in \setofMECs{G}$ \text{ and } $(O, P_1, P_2) = \shadowofMEC{M, V_O}\}$.
\end{definition}

The following observation is intuitive from \cref{def:shadow}.
\begin{observation}
    \label{obs:shadow-of-M-is-shadow-of-G}
    Let $G$ be an undirected graph, and let $M$ be an MEC of $G$. For $Y \subseteq V_G$, if $(O, P_1, P_2)$ is a shadow of the MEC $M$ on $Y$, then $(O, P_1, P_2)$ is a shadow of the undirected graph $G[Y]$.
\end{observation}
\begin{proof}
    Suppose $(O, P_1, P_2)$ is the shadow of $M$ on $Y$. According to \cref{def:shadow}, $O = M[Y]$, and $P_1: E_O \times E_O \rightarrow \{0,1\}$ and $P_2:E_O \times V_O \rightarrow \{0,1\}$ are two functions. Thus, to establish that $(O, P_1, P_2)$ is a shadow of $G[Y]$, we only need to demonstrate that $O$ is a partial MEC.
    
    Since $M$ is an MEC, as stated in \cref{thm:nes-and-suf-cond-for-chordal-graph-to-be-an-MEC}, $M$ is a chain graph with chordal \ucc{} and does not possess any induced subgraph of the form $u \rightarrow v - w$. It's worth noting that an induced subgraph of a chain graph with chordal \ucc{} is also a chain graph with chordal \ucc{}. Furthermore, if $M$ has no induced subgraph of the form $u \rightarrow v - w$, then its induced subgraphs must also lack such a structure.
    
    This implies that $O$ also adheres to the conditions specified in \cref{thm:nes-and-suf-cond-for-chordal-graph-to-be-an-MEC}. Thus, by \cref{def:partial-MEC}, $O$ qualifies as a partial MEC. This concludes the proof of \cref{obs:shadow-of-M-is-shadow-of-G}.
\end{proof}

\Cref{def:subsets-of-MEC-based-on-O-P1-and-P2} gives us a way to partition the set of MECs of any undirected graph $G$. \Cref{lem:partition-of-MECs-of-H} partitions the set of MECs of $G$ based on its \shadow{}.
\begin{lemma}
\label{lem:partition-of-MECs-of-H}
Let $G$ be an undirected graph. Then, for any $Y \subseteq V_G$,
\begin{equation}
    |\setofMECs{G}| = \sum_{(O, P_1, P_2) \in \shadowofudgraph{G[Y]}}{|\setofMECs{G, O, P_1, P_2}|}
\end{equation}
\end{lemma}
\begin{proof}
From \cref{def:subsets-of-MEC-based-on-O-P1-and-P2}, for each $M \in \setofMECs{G}$, there exists a unique \shadow{} $(O, P_1, P_2)$ of $M$ on $Y$. From \cref{obs:shadow-of-M-is-shadow-of-G}, $(O,P_1, P_2) \in \shadowofudgraph{G[Y]}$. This proves \cref{lem:partition-of-MECs-of-H}.
\end{proof}

\subsubsection{\textbf{Projection of an MEC}}

\begin{definition}[\textbf{Projection of an MEC}]
\label{def:projection}
Let $G$ be an undirected graph, $S\subseteq V_G$, $G'=G[S]$ be an induced subgraph of $G$, $M$ be an MEC of $G$, and $M'$ be an MEC of $G'$. We say $M'$ is a \emph{projection} of $M$ on $S$, denoted as $\mathcal{P}(M,S)=M'$, if $\mathcal{V}(M')$ = $\mathcal{V}(M[S])$. 

With slight overloading of the function $\mathcal{P}$, for two subsets $S_1,S_2\subseteq V_G$, we define $(M_1, M_2)$ as  a projection of $M$ on $(S_1, S_2)$, denoted  as: $\mathcal{P}(M,S_1,S_2) = (M_1, M_2)$, if $M_1 = \mathcal{P}(M, S_1)$, and $M_2 = \mathcal{P}(M, S_2)$.
\end{definition}

\begin{lemma}
\label{lem:projection-of-an-MEC-is-unique}
Let $G$ be an undirected graph, and $M$ be an MEC of $G$. For any $S\subseteq V_G$, there exists a unique MEC $M'$ of the graph $G[S]$ such that $\mathcal{P}(M,S)=M'$.
\end{lemma}
\begin{proof}
As a reminder from the introduction, it is well-known that every DAG corresponds to a unique MEC, and the set of v-structures in a DAG is identical to the set of v-structures in its corresponding MEC.
We first show the existence of such an MEC $M'$.
Let $D$ be a DAG such that $D\in M$. The set of v-structures contained by $D[S]$ equals $\mathcal{V}(M[S])$ (since both $D$ and $M$ have the same set of v-structures). 
Let $M'$ be an MEC of $G[S]$ such that $D[S]\in M'$. The set of v-structures of $M'$ must be the same as the set of v-structures of $D[S]$, i.e. $\mathcal{V}(M')$ = $\mathcal{V}(M[S])$. The uniqueness of $M'$ comes from the fact that if two MECs of a graph have the same set of v-structure then both are the same.
\end{proof}

\Cref{lem:projection-of-an-MEC-is-unique} further implies that the following corollaries:
\begin{corollary}
\label{corr:projection-of-an-MEC-is-unique}
Let $G$ be an undirected graph, and $M$ be an MEC of $G$. For $S_1,S_2 \subseteq V_G$, there exists a unique tuple $(M_1,M_2) \in $ MEC$(G[S_1])\times$ MEC$(G[S_2])$ such that $\mathcal{P}(M,S_1,S_2)=(M_1,M_2)$. 
\end{corollary}
\begin{proof}
From \cref{def:projection}, $M_1 = \mathcal{P}(M, S_1)$ and $M_2 = \mathcal{P}(M, S_2)$. From \cref{lem:projection-of-an-MEC-is-unique}, $\mathcal{P}(M, S_1)$ and $\mathcal{P}(M, S_2)$ both are unique. This proves the uniqueness of $(M_1, M_2)$.
\end{proof}

\begin{corollary}
\label{corr:projection-of-subgraph}
Let $G$ be an undirected graph, $S\subseteq V_G$, $M\in$ MEC$(G)$, and $M'\in $ MEC$(G[S])$. If $\mathcal{P}(M,S)=M'$ then for any $S'\subseteq S$,  $\mathcal{V}(M[S']) = \mathcal{V}(M'[S'])$, and $\mathcal{P}(M,S')=\mathcal{P}(M',S')$.
\end{corollary}
\begin{proof}
Since $S' \subseteq S$, for $u,v,w \in S'$, if a v-structure $u\rightarrow v \leftarrow w \in \mathcal{P}(M', S')$ then from \cref{def:projection}, the v-structures also belongs to $M'$ and $\mathcal{P}(M, S')$. Similarly, for $u,v,w \in S'$, if a v-structure $u\rightarrow v \leftarrow w \in \mathcal{P}(M, S')$ then from \cref{def:projection}, the v-structures also belongs to $M'$, and $\mathcal{P}(M', S)$. This implies $\mathcal{V}(M, S') = \mathcal{V}(M', S')$. From \cref{def:projection}, this further implies that $\mathcal{P}(M, S') = \mathcal{P}(M', S')$.
\end{proof}

The following lemma shows some structural similarity between an MEC $M$ and its projection $M'$.

\begin{lemma}
\label{lem:directed-edge-is-same-in-projected-MEC}
Let $G$ be a graph, $S \subseteq V_G$, $M \in \setofMECs{G}$, and $\mathcal{P}(M, S) = M'$. If $u \rightarrow v \in M'$, then $u \rightarrow v \in M$. 
\end{lemma}
\begin{proof}
\citet{meek1995causal} gives a set of rules using which, if we have knowledge about the v-structure of an MEC, then we can construct the MEC. According to \cref{def:projection}, $\mathcal{V}(M') \subseteq \mathcal{V}(M)$. Therefore, if an edge can be directed using Meek's rule in $M'$, then it must be directed in $M$ (using the same Meek's rule). This proves \cref{lem:directed-edge-is-same-in-projected-MEC}. However, we have another proof. This proof is more rigorous and has better clarity. We are presenting this proof because this approach of proving will be used multiple times in this paper.

Let us assume that \cref{lem:directed-edge-is-same-in-projected-MEC} is not true, i.e., there exists an edge $u\rightarrow v \in M'$ such that $u\rightarrow v \notin M$. Since $M'$ is a chain graph (from \cref{item-1-theorem-nec-suf-cond-for-MEC} of \cref{thm:nes-and-suf-cond-for-chordal-graph-to-be-an-MEC}), from \cref{prop:every-chain-graph-has-a-topological-ordering}, there exists a topological ordering $\tau$ of the vertices of $M'$ such that if $u\rightarrow v \in M'$ then $\tau(u) < \tau(v)$. 
Pick 2 vertices $u,v \in V_{M'}$ such that $u\rightarrow v \in M'$, $u\rightarrow v \notin M$, and for any edge $u'\rightarrow v' \in M'$, if $u'\rightarrow v' \notin M$ then either $\tau(v') > \tau(v)$, or $v = v'$ and $\tau(u') \leq \tau(u)$. The chosen vertices $u$ and $v$ imply that the following claim:
\begin{claim}
\label{obs:implicationn-of-vertices-chosen}
For any edge $u'\rightarrow v' \in M'$, if $\tau(v') <\tau(v)$, or  $v = v'$ and $\tau(u) < \tau(u')$  then $u'\rightarrow v' \in M$.
\end{claim}
Since $u\rightarrow v \in M'$, it must be strongly protected in $M'$ (\cref{item-4-theorem-nec-suf-cond-for-MEC} of \cref{thm:nes-and-suf-cond-for-chordal-graph-to-be-an-MEC}). Then, $u\rightarrow v$ is part of one of the subgraphs shown in \cref{fig:strongly-protected-edge}. We go through each possibility and show that in each case, $u\rightarrow v \in M$, contradicting our assumption.

\textbf{Case-1:} $u\rightarrow v$ is a part of item-(a) of \cref{fig:strongly-protected-edge}, i.e., $w\rightarrow u\rightarrow v \in M'$. Then, $\tau(w) < \tau(u) < \tau(v)$. From \cref{obs:implicationn-of-vertices-chosen}, $w\rightarrow u \in M$. $u-v \notin M$, as it violates \cref{item-3-theorem-nec-suf-cond-for-MEC} of \cref{thm:nes-and-suf-cond-for-chordal-graph-to-be-an-MEC}. Neither $v\rightarrow u \in M$, otherwise, $w\rightarrow u \leftarrow v$ is a v-structure in $M$. Since $M'$ is a projection of $M$, and $u,v,w\in V_{M'}$, from \cref{def:projection}, if the v-structure $w\rightarrow u\leftarrow v\in M$ then $w\rightarrow u\leftarrow v \in M'$. But, this is a contradiction as $u\rightarrow v \in M'$. Thus the only option that remains is $u\rightarrow v \in M$.

\textbf{Case-2:} $u\rightarrow v$ a is part of item-(b) of \cref{fig:strongly-protected-edge}, i.e., $u\rightarrow v\leftarrow w \in M'$. Since $u\rightarrow v \leftarrow w$ is a v-structure, from \cref{def:projection}, $u\rightarrow v \leftarrow w \in M$, as $M'$ is a projection of $M$. This implies that $u\rightarrow v \in M$.

\textbf{Case-3:} $u\rightarrow v$ is a part of item-(c) of \cref{fig:strongly-protected-edge}, i.e., $u\rightarrow w\rightarrow v \leftarrow u \in M'$. Then, $\tau(u) < \tau(w) < \tau(v)$. From \cref{obs:implicationn-of-vertices-chosen}, $u\rightarrow w \rightarrow v \in M$. If $u-v \in M$, or $v \rightarrow u \in M$ then $(u,v,w,u)$ is a directed cycle in $M$, contradicting \cref{item-1-theorem-nec-suf-cond-for-MEC} of \cref{thm:nes-and-suf-cond-for-chordal-graph-to-be-an-MEC}. This implies $u\rightarrow v \in M$.

\textbf{Case-4:}  $u\rightarrow v$ is a part of item-(d) of \cref{fig:strongly-protected-edge}. Since $w\rightarrow v \leftarrow w'$ is a v-structure, from \cref{def:projection}, $w\rightarrow v \leftarrow w' \in M$, as $M'$ is a projection of $M$. If $v-u \in M$, or $v\rightarrow u \in M$ then $w\rightarrow u\leftarrow w' \in M$, otherwise $M$ has a directed cycle $(w,v,u,w)$ (if $w\rightarrow u \notin M$), or $(w', v, u, w')$ (if $w'\rightarrow u \notin M)$, which violates \cref{item-1-theorem-nec-suf-cond-for-MEC} of \cref{thm:nes-and-suf-cond-for-chordal-graph-to-be-an-MEC}. Since, $w\rightarrow u \leftarrow w'$ is a v-structure of $M$, and $w, u, w' \in V_{M'}$, from  \cref{def:projection}, $w\rightarrow u \leftarrow w' \in M'$, which is a contradiction, as $w-u-w' \in M'$. Thus the only option that remains is $u\rightarrow v \in M$.

We have shown that $u\rightarrow v \in M$ in each case. This proves \cref{lem:directed-edge-is-same-in-projected-MEC}.
\end{proof}

\cref{lem:directed-edge-is-same-in-projected-MEC} further implies the following corollaries:
\begin{corollary}
\label{corr:directed-edge-of-main-graph-is-either-dir-or-ud-in-proj-graph}
Let $G$ be a graph, $S\subseteq V_G$, $M\in \setofMECs{G}$, and $\mathcal{P}(M,S)=M'$. If $u\rightarrow v \in M$, and $u,v\in S$ then either $u\rightarrow v \in M'$, or $u-v\in M'$.
\end{corollary}
\begin{proof}
Suppose $u,v\in S$ and $u\rightarrow v \in M$.
    From the construction of projection, $\skel{M'} = \skel{M[S]}$. This implies if $u\rightarrow v \in M$ then either  $u\rightarrow v \in M'$ or $v\rightarrow u \in M'$ or $u-v \in M'$. But, if $v\rightarrow u \in M'$ then from \cref{lem:directed-edge-is-same-in-projected-MEC}, $v\rightarrow u \in M$, a contradiction. This implies either $u\rightarrow v \in M'$, or $u-v\in M'$.
\end{proof}

\begin{corollary}
\label{corr:undirected-edge-in-an-MEC-implies-undirected-edge-in-projected-MEC}
Let $G$ be a graph, $S\subseteq V_G$, $M\in \setofMECs{G}$, and $\mathcal{P}(M,S)=M'$. If $u- v \in M$, and $u,v\in S$ then $u-v \in M'$.
\end{corollary}
\begin{proof}
    Suppose $u,v\in S$ and $u-v \in M$.
    From the construction of projection, $\skel{M'} = \skel{M[S]}$. This implies if $u-v \in M$ then either  $u\rightarrow v \in M'$ or $v\rightarrow u \in M'$ or $u-v \in M'$. From \cref{lem:directed-edge-is-same-in-projected-MEC}, if there is a directed edge between $u$ and $v$ in $M'$ then  $M$ also has a directed edge between $u$ and $v$, a contradiction. This implies $u-v \in M'$.
\end{proof}

\begin{corollary}
\label{corr:tfp-in-main-graph-implies-tfp-in-projected-graph}
    Let $M$ be an MEC and $M'$ be a projection of $M$.
    Let $P = (u_1,u_2,\ldots, u_l)$ be a \tfp{} in $M$ such that all the nodes of $P$ are in $M'$, i.e., for all $1\leq i \leq l$, $u_i \in V_{M'}$  then $P$ is a \tfp{} in $M'$.
\end{corollary}
\begin{proof}
    Suppose $P = (u_1,u_2,\ldots, u_l)$ is a \tfp{} in $M$ such that all the nodes of $P$ are in $M'$.
    From \cref{lem:directed-edge-is-same-in-projected-MEC,corr:undirected-edge-in-an-MEC-implies-undirected-edge-in-projected-MEC}, $P$ is a path in $M'$. $P$ must be a \tfp{} in $M'$, otherwise, if there exist $u_{i-1}, u_i, u_{i+1}$ such that there is an edge between $u_{i-1}$ and $u_{i+1}$ in $M'$ (of any form undirected or directed i.e., either $u_{i-1}-u_{i+1} \in M'$, or $u_{i-1}\rightarrow u_{i+1} \in M'$, or $u_{i+1}\rightarrow u_{i-1} \in M'$)  then even in $M$ there is an edge between $u_{i-1}$ and $u_{i+1}$ (since $M'$ is a projection of $M$, if $u-v \in \skel{M'}$ then $u-v \in \skel{M[V_M]}$). This implies $P$ is not a \tfp{} in $M$, a contradiction. This implies $P$ is a \tfp{} in $M'$.
\end{proof}

The following corollary is a partial converse of \cref{corr:tfp-in-main-graph-implies-tfp-in-projected-graph}.

\begin{corollary}
\label{obs:cond-for-ud-path-in-M_a-to-be-ud-path-in-M}
Let $M$ be an MEC and $M'$ be a projection of $M$.
Let $P = (u_1,u_2,\ldots, u_l)$ be a \tfp{} in $M'$. If $u_2\rightarrow u_1 \notin M$ then $P$ is also a \tfp{} in $M$.
\end{corollary}
\begin{proof}
Suppose that $u_2\rightarrow u_1 \notin M$.
Since $P$ is a \tfp{} in $M'$, if $P$ is a path in $M$ then it is a \tfp{} in $M$.
We show that $P$ is a path in $M$.
Suppose that $P$ is not a path in $M$. Then there exists an edge $u_i\leftarrow u_{i+1} \in M$. Pick the least $i$ such that $u_i\leftarrow u_{i+1} \in M$. Clearly, $i\geq 2$ (from our assumption $u_2\rightarrow u_1\notin M$). Since $P$ is a \tfp{} in $M'$, there does not exist an edge in between $u_{i-1}$ and $u_{i+1}$. This implies that either $u_{i-1}-u_i \leftarrow u_{i+1} \in M$ or $u_{i-1} \rightarrow u_i \leftarrow u_{i+1} \in M$. From \cref{item-2-theorem-nec-suf-cond-for-MEC} of \cref{thm:nes-and-suf-cond-for-chordal-graph-to-be-an-MEC}, $u_{i-1}-u_i \leftarrow u_{i+1} \notin M$. 
From \cref{def:projection}, $\mathcal{V}(M') = \mathcal{V}(M[V_{M'}])$. This implies if  the v-structure $u_{i-1} \rightarrow u_i \leftarrow u_{i+1} \in M$ then we have  $u_{i-1} \rightarrow u_i \leftarrow u_{i+1} \in M'$, as $u_{i-1}, u_i, u_{i+1} \in V_{M'}$. But, then $P$ is not a path in $M'$ (because $u_i \leftarrow u_{i+1} \in M'$), which is a contradiction. This implies that $P$ is a path in $M$. 
\end{proof}

The following corollary is an extension of \cref{obs:cond-for-ud-path-in-M_a-to-be-ud-path-in-M}.

\begin{corollary}
\label{obs:every-edge-of-triangle-free-path-is-directed}
Let $M$ be an MEC and $M'$ be a projection of $M$.
If $P = (u_1 = x, u_2 = y, u_3, \ldots, u_{l-1} = u, u_l =v)$ is a \tfp{} from $(x,y)$ to $(u,v)$ in $M'$, and $x\rightarrow y \in M$ then for all $1\leq i < l$, $u_i\rightarrow u_{i+1} \in M$.
\end{corollary}
\begin{proof}
For $i=1$, \cref{obs:every-edge-of-triangle-free-path-is-directed} is true as $u_1\rightarrow u_2\in M$. Pick the least $i>1$ such that $u_i \rightarrow u_{i+1}\notin M$. If $u_i-u_{i+1} \in M$ then $u_{i-1}\rightarrow u_i -u_{i+1} \in M$ which disobeys \cref{item-2-theorem-nec-suf-cond-for-MEC} of \cref{thm:nes-and-suf-cond-for-chordal-graph-to-be-an-MEC}. And, if $u_i\leftarrow u_{i+1}\in M$ then $u_{i-1}\rightarrow u_i \leftarrow u_{i+1}$ is a v-structure in $M$. This implies that $u_{i-1}\rightarrow u_i \leftarrow u_{i+1} \in M'$, as from from \cref{def:projection}, $\mathcal{V}(M[V_{M'}]) = \mathcal{V}(M')$. But, then, $P$ is not a path in $M'$, a contradiction. This implies that no such $i$ exists, i.e., for all $i$, $u_i\rightarrow u_{i+1} \in M$.
\end{proof}

The following proposition is due to the non-existence of a directed cycle in an MEC.

\begin{proposition}
\label{obs:edge-in--ucc-is-undirected}
Let $M$ be an MEC, and $\mathcal{C}$ be an \ucc{} of $M$. There cannot be a directed edge between two vertices of $\mathcal{C}$, i.e., if $u,v \in \mathcal{C}$ and there is an edge between $u$ and $v$ in $M$ then $u-v\in M$.
\end{proposition}
\begin{proof}
Suppose that $u,v \in \mathcal{C}$, and $u\rightarrow v \in M$. Since $u,v \in \mathcal{C}$, there exists an undirected path from $v$ to $u$. Combining the undirected path with $u\rightarrow v$ gives a directed cycle in $M$, which contradicts \cref{item-1-theorem-nec-suf-cond-for-MEC} of \cref{thm:nes-and-suf-cond-for-chordal-graph-to-be-an-MEC}. This implies if there is an edge between $u$ and $v$ then it must be undirected. This proves \cref{obs:edge-in--ucc-is-undirected}.
\end{proof}

The following lemma shows the unidirectional nature of \tfps{} in an MEC. We use this later for the verification of an MEC.

\begin{lemma}
\label{lem:nes-tfp-cond-for-an-MEC}
Let $M$ be an MEC of an undirected graph $G$. For $(u,v), (x,y) \in E_G$ such
that $(u,v)\neq (x,y)$, if there exists a \tfp{} from $(u,v)$ to $(x,y)$ in $M$
then there cannot exist a \tfp{} from $(x,y)$ to $(u,v)$.
\end{lemma}

\begin{proof}
  Suppose, if possible, that there exist \tfps{} $P_1$ from $(u, v)$ to $(x, y)$
  and $P_2$ from $(x, y)$ to $u$ ($P_2$ is a prefix of the hypothesised \tfp{}
  from $(x, y)$ to $(u, v)$).
Then by
  \cref{obs:concatenate-triangle-free-paths}, the concatenation $Q$ of $P_1$ and
  $P_2$ (in the sense of \cref{obs:concatenate-triangle-free-paths}) is a \tfp{}
  from $(u, v)$ to $u$.  But this is a contradiction, since by the definition of path, \cref{def:path}, no
  vertex can appear twice in a path. This implies there cannot exist a \tfp{} from $(x,y)$ to $u$.  Thus, there cannot exist a \tfp{} from $(x,y)$ to $(u,v)$.
\end{proof}

We now describe our approach to counting MECs of an undirected graph.  
 \section{Counting MECs of a Graph}
 \label{sec:counting-MEcs-of-a-graph}
 Let $G$ be an undirected graph. We want to count MECs that have the same skeleton as $G$, i.e., $|\setofMECs{G}|$. We provide a formal description of the problem.

\begin{problem}[\textbf{Counting MECs of an Undirected Graph}]
\label{prob:counting-MECs}
\textbf{Input:} An undirected graph $G$.

\noindent
\textbf{Output:}  $|\setofMECs{G}|$.
\end{problem}

We present a recursive algorithm to solve \cref{prob:counting-MECs}. For this recursion, we divide the input graph using its tree decomposition. We construct a tree decomposition $(X = \{X_1, X_2,\ldots, X_s\}, T)$ of $G$. We select node $X_1$ as the root node of $T$. For each node $X_i$ in $T$, we denote the subtree of $T$ rooted at $X_i$ as $T_i$. The induced subgraph of $G$ represented by $T_i$ is denoted as $G_i$. This implies that $T=T_1$, $G=G_1$, and our goal is to count the MECs of $G_1$.

In the subtree $T_i$, let the children of $X_i$ be denoted as $(X_{i1}, X_{i2}, \ldots, X_{il})$ (i.e., for $T_i$, we consider $X_{i1}$ as the first child of $X_i$, $X_{i2}$ as the second child of $X_i$, and so on). For any $0 \leq j \leq l$, we denote $T_i^j$ as the induced subtree of $T_i$ that contains node $X_i$ and the nodes of $T_{i1}, T_{i2}, \ldots, T_{ij}$. In other words, $T_i^0$ contains only one node, $X_i$, and $T_i^l = T_i$. The induced subgraph of $G$ represented by $T_i^j$ is denoted as $G_i^j$. If the number of children of $X_1$ is $p$, then $T_1^p = T_1 = T$, and $G_1^p = G_1 = G$, and \textbf{we want to count MECs of $G_1^p$}. Refer to \Cref{fig:graph-example} for further clarity.

\begin{figure*}[ht]
    \centering
    \begin{tikzpicture}
    \node[](u1){$1$};
    \node[](u2)[below left=0.5 and 1 of u1] {$2$};
    \node[](u3)[below right=0.5 and 1 of u1] {$3$};
    \node[](u4)[below left=0.5 and 0.5 of u2] {$4$};
    \node[](u5)[below =0.5 of u2] {$5$};
    \node[](u6)[below right =0.5 and 1 of u2] {$6$};
    \node[](u7)[below left=0.5 and 0.25 of u3] {$7$};
    \node[](u8)[below right=0.5 and 0.25 of u3] {$8$};
    \node[](u9)[below left=0.5 and 0.25 of u5] {$9$};
    \node[](u10)[below right=0.5 and 0.25 of u5] {$10$};
    \node[](u11)[below left=0.5 and 0.25 of u7] {$11$};
    \node[](u12)[below right=0.5 and 0.25 of u7] {$12$};
    \node[](u13)[below =0.5 of u8] {$13$};
    \node[](u14)[below right=0.5 and 0.25 of u8] {$14$};
    \node[](u15)[below left=0.5 and 0.25 of u12] {$15$};
    \node[](u16)[below right=0.5 and 0.25 of u12] {$16$};
    \node[](u17)[below = 0.5 of u14] {$17$};
    \node[ left=1.0 of u2](G){$G =$};
    \draw[-](u1)--(u2);
    \draw[-](u1)--(u3);
    \draw[-](u2)--(u4);
    \draw[-](u2)--(u5);
    \draw[-](u4)--(u5);
    \draw[-](u2)--(u6);
    \draw[-](u3)--(u6);
    \draw[-](u3)--(u7);
    \draw[-](u3)--(u8);
    \draw[-](u7)--(u8);
    \draw[-](u5)--(u9);
    \draw[-](u5)--(u10);
    \draw[-](u9)--(u10);
    \draw[-](u7)--(u11);
    \draw[-](u7)--(u12);
    \draw[-](u11)--(u12);
    \draw[-](u8)--(u13);
    \draw[-](u8)--(u14);
    \draw[-](u13)--(u14);
    \draw[-](u12)--(u15);
    \draw[-](u12)--(u16);
    \draw[-](u15)--(u16);
    \draw[-](u14)--(u17);
    
    \node[](X1)[right=4.0 of u1] {$X_1$};
    \node[](X2)[below left=0.5 and 0.5 of X1] {$X_2$};
    \node[](X3)[below=0.5 of X1] {$X_3$};
    \node[](X4)[below right=0.5 and 0.5 of X1] {$X_4$};
    \node[](X5)[below  =0.5 of X2] {$X_5$};
    \node[](X6)[below left=0.5 and 0.05 of X4] {$X_6$};
    \node[](X7)[below right  =0.5 and 0.05 of X4] {$X_7$};
    \node[](X8)[below =0.5 of X6] {$X_8$};
    \node[](X9)[below =0.5 of X7] {$X_9$};
    \draw[-](X1)--(X2);
    \draw[-](X1)--(X3);
    \draw[-](X1)--(X4);
    \draw[-](X1)--(X2);
    \draw[-](X2)--(X5);
    \draw[-](X4)--(X6);
    \draw[-](X4)--(X7);
    \draw[-](X6)--(X8);
    \draw[-](X7)--(X9);
    \node[ below left=0.2 and 1.0 of X1](T){$T =$};
    
    \node[](X4')[right=3.5 of X1] {$X_4$};
    \node[](X6')[below left=0.5 and 0.05 of X4'] {$X_6$};
    \node[](X7')[below right  =0.5 and 0.05 of X4'] {$X_7$};
    \node[](X8')[below =0.5 of X6'] {$X_8$};
    \node[](X9')[below =0.5 of X7'] {$X_9$};
    \node[below left= 0.2 and 0.5 of X4'](T4){$T_4=$};
    \draw[-](X4')--(X6');
    \draw[-](X4')--(X7');
    \draw[-](X6')--(X8');
    \draw[-](X7')--(X9');
    
    \node[](X61)[right=2.5 of X4'] {$X_6$};
    \node[below left = 0.1 and 0.1 of X61](T6){$T_6 =$};
    \node[](X81)[below =0.5 of X61] {$X_8$};
    \draw[-](X61)--(X81);
    
    \node[](X71)[below =1.5 of X61] {$X_7$};
    \node[below left = 0.1 and 0.1 of X71](T7){$T_7 =$};
    \node[](X91)[below =0.5 of X71] {$X_9$};
    \draw[-](X71)--(X91);
    
    \node[](X4'')[below left =3.5 and 6.5 of T4] {$X_4$};
    \node[](X6'')[below =0.5 of X4''] {$X_6$};
    \node[](X8'')[below =0.5 of X6''] {$X_8$};
    \node[below left= 0.2 and 0.25 of X4''](T41){$T_4^1 =$};
    \draw[-](X4'')--(X6'');
    \draw[-](X6'')--(X8'');
    
    \node[](X40)[below left = 0.0 and 1.5 of X4''] {$X_4$};
    \node[left= 0.25 of X40](T40){$T_4^0 =$};

    \node[](X4a)[right=2.0 of X4''] {$X_4$};
    \node[](X6a)[below left=0.5 and 0.05 of X4a] {$X_6$};
    \node[](X7a)[below right  =0.5 and 0.05 of X4a] {$X_7$};
    \node[](X8a)[below =0.5 of X6a] {$X_8$};
    \node[](X9a)[below =0.5 of X7a] {$X_9$};
    \node[below left= 0.2 and 0.5 of X4a](T4a){$T_4^2 =$};
    \draw[-](X4a)--(X6a);
    \draw[-](X4a)--(X7a);
    \draw[-](X6a)--(X8a);
    \draw[-](X7a)--(X9a);

    \node[](u31)[right=3.0 of X4a] {$3$};
    \node[](u71)[below left=0.5 and 0.25 of u31] {$7$};
    \node[](u81)[below right=0.5 and 0.25 of u31] {$8$};
    \node[](u111)[below left=0.5 and 0.25 of u71] {$11$};
    \node[](u121)[below right=0.5 and 0.25 of u71] {$12$};
    \node[](u131)[below =0.5 of u81] {$13$};
    \node[](u141)[below right=0.5 and 0.25 of u81] {$14$};
    \node[](u151)[below left=0.5 and 0.25 of u121] {$15$};
    \node[](u161)[below right=0.5 and 0.25 of u121] {$16$};
    \node[](u171)[below = 0.5 of u141] {$17$};
    \node[ left=0.5 of u31](G){$G_4 = G_4^2=$};
    \draw[-](u31)--(u71);
    \draw[-](u31)--(u81);
    \draw[-](u71)--(u81);
    \draw[-](u71)--(u111);
    \draw[-](u71)--(u121);
    \draw[-](u111)--(u121);
    \draw[-](u81)--(u131);
    \draw[-](u81)--(u141);
    \draw[-](u131)--(u141);
    \draw[-](u121)--(u151);
    \draw[-](u121)--(u161);
    \draw[-](u151)--(u161);
    \draw[-](u141)--(u171);
    
    \node[](u30)[right=2.5 of u31] {$3$};
    \node[](u70)[below left=0.5 and 0.25 of u30] {$7$};
    \node[](u80)[below right=0.5 and 0.25 of u30] {$8$};
    \node[ left=0.5 of u30](G){$G_4^0 = $};
    \draw[-](u30)--(u70);
    \draw[-](u30)--(u80);
    \draw[-](u70)--(u80);
    
    \node[](u32)[right=5.5 of u31] {$3$};
    \node[](u72)[below left=0.5 and 0.25 of u32] {$7$};
    \node[](u82)[below right=0.5 and 0.25 of u32] {$8$};
    \node[](u112)[below left=0.5 and 0.25 of u72] {$11$};
    \node[](u122)[below right=0.5 and 0.25 of u72] {$12$};
    \node[](u152)[below left=0.5 and 0.25 of u122] {$15$};
    \node[](u162)[below right=0.5 and 0.25 of u122] {$16$};
    \node[ left=0.5 of u32](G){$G_4^1 =$};
    \draw[-](u32)--(u72);
    \draw[-](u32)--(u82);
    \draw[-](u72)--(u82);
    \draw[-](u72)--(u112);
    \draw[-](u72)--(u122);
    \draw[-](u112)--(u122);
    \draw[-](u122)--(u152);
    \draw[-](u122)--(u162);
    \draw[-](u152)--(u162);

    \end{tikzpicture}
        \caption{Example: $G$ is an undirected graph, and $T$ is a tree decomposition of $G$ where:
$X_1 = \{1, 2, 3\}$,
$X_2 = \{2, 4, 5\}$,
$X_3 = \{2, 3, 6\}$,
$X_4 = \{3, 7, 8\}$,
$X_5 = \{5, 9, 10\}$,
$X_6 = \{7, 11, 12\}$,
$X_7 = \{8, 13, 14\}$,
$X_8 = \{12, 15, 16\}$, and
$X_9 = \{14, 17\}$.
$T_4$, $T_6$, and $T_7$ are induced subtrees of $T$ rooted at $X_4$, $X_6$, and $X_7$, respectively. We assume $X_6$ and $X_7$ are the first and second children of $X_4$. $T_4^0$ is an induced subgraph of $T$ containing only the node $X_4$. $T_4^1$ is an induced subtree of $T$ containing node $X_4$ and the nodes of $T_6$, and $T_4^2$ is an induced subtree of $T_4$ containing node $X_4$ and the nodes of $T_6$ and $T_7$. $G_4$ is the induced subgraph of $G$ represented by $T_4$. $G_4^0$ is the induced subgraph of $G$ represented by $T_4^0$. $G_4^1$ is the induced subgraph of $G$ represented by $T_4^1$. $G_4^2$ is the induced subgraph of $G$ represented by $T_4^2$. Since $T_4 = T_4^2$, it follows that $G_4 = G_4^2$.
}
        \label{fig:graph-example}
\end{figure*}
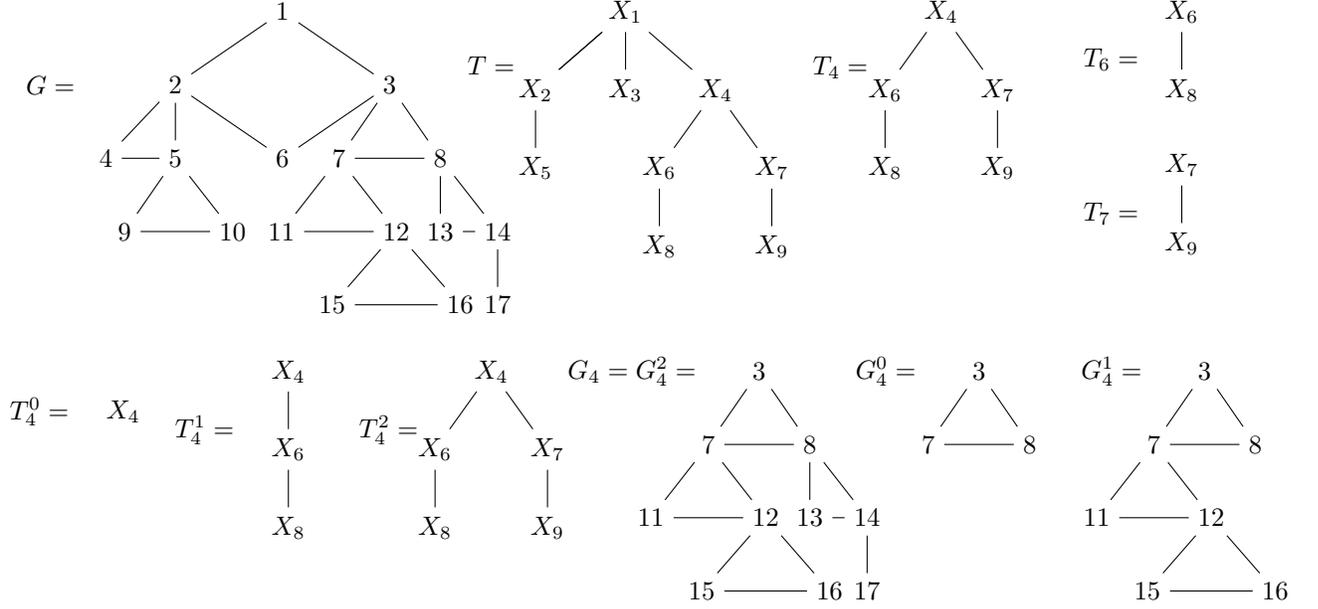

Let $X_{i_1}, X_{i_2}, \ldots, X_{i_l}$ be the children of $X_i$ in the tree $T_i$. For $1 \leq j \leq l$, \cref{obs:relation-between-MECs-of-G-i-j-G-i-j-1-G-ij} establishes a relationship between $G_i^j$, $G_i^{j-1}$, and $G_{i_j}$.
\begin{observation}
\label{obs:relation-between-MECs-of-G-i-j-G-i-j-1-G-ij}
Let $G$ be an undirected graph, and $(X = \{X_1, X_2,\ldots, X_s\}, T)$ be a tree decomposition of $G$. For $1\leq i \leq s$, let $X_{i_1}, X_{i_2}, \ldots, X_{i_l}$ be the children of $X_i$ in the tree $T_i$. Then, for all $j>0$:
\begin{enumerate}
    \item
    \label{item-1-of-obs:relation-between-MECs-of-G-i-j-G-i-j-1-G-ij}
    $G_i^{j-1}$ and $G_{i_j}$ are induced subgraphs of $G_i^j$.
    \item
    \label{item-2-of-obs:relation-between-MECs-of-G-i-j-G-i-j-1-G-ij}
    $G_i^j = G_i^{j-1} \cup G_{i_j}$.
    \item
    \label{item-3-of-obs:relation-between-MECs-of-G-i-j-G-i-j-1-G-ij}
    $X_i \subseteq V_{G_i^{j-1}}$.
    \item
    \label{item-4-of-obs:relation-between-MECs-of-G-i-j-G-i-j-1-G-ij}
    $X_{i_j} \subseteq V_{G_{i_j}}$.
    \item
    \label{item-5-of-obs:relation-between-MECs-of-G-i-j-G-i-j-1-G-ij}
    $I_j = X_i \cap X_{i_j} $ is a vertex separator of $G_i^j$ that separates $V_{G_i^{j-1}}\setminus{I_j}$ and $V_{G_{i_j}}\setminus{I_j}$.
    \item
    \label{item-6-of-obs:relation-between-MECs-of-G-i-j-G-i-j-1-G-ij}
    $V_{G_i^{j-1}} \cap V_{G_{i_j}} = X_i \cap X_{i_j} = I_j$
\end{enumerate}
\end{observation}
\begin{proof} From the construction of $G_i^j, G_i^{j-1}$ and $G_{i_j}$,
\cref{item-1-of-obs:relation-between-MECs-of-G-i-j-G-i-j-1-G-ij,item-2-of-obs:relation-between-MECs-of-G-i-j-G-i-j-1-G-ij,item-3-of-obs:relation-between-MECs-of-G-i-j-G-i-j-1-G-ij,item-4-of-obs:relation-between-MECs-of-G-i-j-G-i-j-1-G-ij} are true. Since $T_i^j$ is a tree decomposition of $G_i^j$, $I_j = X_i \cap X_{i_j}$ serves as a vertex separator of $G_i^j$ that separates $V_{G_i^{j-1}} \setminus I_j$ and $V_{G_{i_j}} \setminus I_j$, validating \cref{item-5-of-obs:relation-between-MECs-of-G-i-j-G-i-j-1-G-ij}. The validation of \cref{item-6-of-obs:relation-between-MECs-of-G-i-j-G-i-j-1-G-ij} comes from the fact that $X_i \subseteq V_{G_i^{j-1}}$, $X_{i_j} \subseteq V_{G_{i_j}}$, and $X_i \cap X_{i_j}$ is a vertex separator of $G_i^j$ that separates $V_{G_i^{j-1}} \setminus I_j$ and $V_{G_{i_j}}\setminus{I_j}$.
\end{proof}

Remember from the above discussion that we want to count MECs of $G_1^p$, where $p$ is the number of children of the root node $X_1$ in the tree decomposition of $G$. \Cref{obs:relation-between-MECs-of-G-i-j-G-i-j-1-G-ij} shows that if we cut the edge $X_i-X_{i_j}$ in the tree $T_i^j$ then we get two subtrees $T_i^{j-1}$ and $T_{i_j}$ of $T_i$. This further gives us two induced subgraphs $G_i^{j-1}$ and $G_{i_j}$ of $G_i^j$ that have the properties \cref{item-1-of-obs:relation-between-MECs-of-G-i-j-G-i-j-1-G-ij,item-2-of-obs:relation-between-MECs-of-G-i-j-G-i-j-1-G-ij,item-3-of-obs:relation-between-MECs-of-G-i-j-G-i-j-1-G-ij,item-4-of-obs:relation-between-MECs-of-G-i-j-G-i-j-1-G-ij,item-5-of-obs:relation-between-MECs-of-G-i-j-G-i-j-1-G-ij,item-6-of-obs:relation-between-MECs-of-G-i-j-G-i-j-1-G-ij} of \cref{obs:relation-between-MECs-of-G-i-j-G-i-j-1-G-ij}.  
We establish a relation between the MECs of $G_i^j$, $G_i^{j-1}$, and $G_{i_j}$that if we have the knowledge about (a) $|\setofMECs{G_i^{j-1}, O_1, P_{11}, P_{12}}|$ for each shadow $(O_1, P_{11}, P_{12})$ of $G_i^{j-1}[X_i\cup N(X_i, G_i^{j-1})]$, and (b) $|\setofMECs{G_{i_j}, O_2, P_{21}, P_{22}}|$ for each shadow $(O_2, P_{21}, P_{22})$ of $G_{i_j}[X_{i_j}\cup N(X_{i_j}, G_{i_j})]$ then we can compute $|\setofMECs{G_i^j, O, P_1, P_2}|$ for each shadow $(O, P_1, P_2)$ of $G_i^j[X_i\cup N(X_i, G_i^j)]$.
This gives us a recursive method to compute $|\setofMECs{G_1^p, O, P_1, P_2}|$ for each shadow $(O, P_1, P_2)$ of $G_1^p$ on $X_1\cup N(X_1, G_1^p)$.
Then, with the assistance of \cref{lem:partition-of-MECs-of-H}, we can effectively count the MECs of $G = G_1^p$.

We first show the relation between $\setofMECs{G_i^j}$, $\setofMECs{G_i^{j-1}}$ and $\setofMECs{G_{i_j}}$, using which we construct a recursive algorithm to count the MECs of $G$. To understand this relation, we first go through a similar simple scenario. Let us consider an undirected graph $H$. Let $H_1$ and $H_2$ be two induced subgraphs of $H$ such that $H = H_1 \cup H_2$, and $I = V_{H_1} \cap V_{H_2}$ is a vertex separator of $H$ that separates $V_{H_1}\setminus I$ and $V_{H_2}\setminus I$. Let $S_1, S_2$ be the  subsets of $V_{H_1}$ and $V_{H_2}$, respectively, such that both $S_1$ and $S_2$ contains nodes of $I$, i.e., $S_1\cap S_2 = I$. Here, $H$ resembles $G_i^j$, $H_1$ resembles $G_i^{j-1}$, $H_2$ resembles $G_{i_j}$, $S_1$ resembles $X_i$, $S_2$ resembles $X_{i_j}$, and $I$ resembles $I_j$. We show some relation between MECs of $H, H_1$, and $H_2$.  

For given shadows $(O_1, P_{11}, P_{12})$ of $H_1[S_1\cup N(S_1, H_1)]$, $(O_2, P_{21}, P_{22})$ of $H_2[S_2\cup N(S_2, H_2)]$ and $(O,P_1, P_2)$ of $H[S_1\cup S_2 \cup N(S_1\cup S_2, H)]$, and MECs $M \in \setofMECs{H, O, P_1, P_2}$, $M_1\in \setofMECs{H_1, O_1, P_{11}, P_{12}}$, and $M_2\in \setofMECs{H_2, O_2, P_{21}, P_{22}}$, we first show a necessary condition for the shadow  $(O, P_1, P_2)$ such that $\mathcal{P}(M, V_{H_1}, V_{H_2}) = (M_1, M_2)$.

\subsection{Necessary condition for a shadow of an MEC}
\label{subsection:necessary-condition-of-shadow}
\begin{lemma}
\label{obs1:O-structure-for-existence-of-MEC}
Let $H$ be an undirected graph, and let $H_1$ and $H_2$ be two induced subgraphs of $H$ such that $H = H_1 \cup H_2$, and $I = V_{H_1} \cap V_{H_2}$ is a vertex separator of $H$ that separates $V_{H_1} \setminus I$ and $V_{H_2} \setminus I$.
Let $S_1$ and $S_2$ be subsets of $V_{H_1}$ and $V_{H_2}$, respectively, such that $S_1 \cap S_2 = I$.
Let $Y = S_1 \cup S_2 \cup N(S_1 \cup S_2, H)$, $Y_1 = S_1 \cup N(S_1, H_1)$, and $Y_2 = S_2 \cup N(S_2, H_2)$.
Let $M$, $M_1$, and $M_2$ be the MECs of $H$, $H_1$, and $H_2$, respectively. Let $(O, P_1, P_2) = \shadowofMEC{M, Y}$, $(O_1, P_{11}, P_{12}) = \shadowofMEC{M_1, Y_1}$, and $(O_2, P_{21}, P_{22}) = \shadowofMEC{M_2, Y_2}$.
If $\mathcal{P}(M, V_{H_1}, V_{H_2}) = (M_1, M_2)$, then:
\begin{enumerate}
    \item \label{item-1-of-obs1:O-structure-for-existence-of-MEC}
    For $a \in \{1, 2\}$, if $u \rightarrow v \in O_a$, then $u \rightarrow v \in O$.
    
    \item \label{item-2-of-obs1:O-structure-for-existence-of-MEC}
    For $a \in \{1, 2\}$, $\mathcal{V}(O_a) = \mathcal{V}(O[V_{O_a}])$ (see \cref{def:v-structure} for $\mathcal{V}(G)$).
    
    \item \label{item-3-of-obs1:O-structure-for-existence-of-MEC}
    For $a \in \{1, 2\}$, if $u-v \in O_a$, then $u \rightarrow v \in O$ if, and only if, either of the following occurs:
    \begin{enumerate}
        \item \label{subitem-1-of-item-3-of-obs1:O-structure-for-existence-of-MEC}
        $u \rightarrow v$ is strongly protected in $O$.
        
        \item \label{subitem-2-of-item-3-of-obs1:O-structure-for-existence-of-MEC}
        There exists $x-y \in O_a$ such that $x \rightarrow y \in O$, and
        $P_{a1}((x,y),(u,v)) = 1$ (i.e., there exists a \tfp{} from $(x,y)$ to
        $(u,v)$ in $M_a$ with $(x, y) \neq (u, v)$).
        
        \item \label{subitem-3-of-item-3-of-obs1:O-structure-for-existence-of-MEC}
        There exists $x-y \in O_a$ such that $x \rightarrow y \in O$, and
        $P_{a2}((x, y), v) = P_{a2}((v, u), x) = 1$ (i.e., $y \neq v$, $u \neq x$, and there exist \tfps{}
         in $M_a$ from $(x,y)$ to $v$, and from $(v,u)$ to $x$).
    \end{enumerate}
    
    \item \label{item-5-of-obs1:O-structure-for-existence-of-MEC}
    For $((u,v),(x,y)) \in E_O \times E_O$, $P_1((u,v), (x,y)) = 1$ if, and only if, $(u, v) \neq (x, y)$ and at least one of the following occurs:
    \begin{enumerate}
        \item \label{subitem-1-of-item-5-of-obs1:O-structure-for-existence-of-MEC}
        There exists a \tfp{} from $(u,v)$ to $(x,y)$ in $O$.
        
        \item \label{subitem-2-of-item-5-of-obs1:O-structure-for-existence-of-MEC}
        For $a \in \{1, 2\}$, $P_{a1}((u,v),(x,y)) = 1$.
        
        \item \label{subitem-3-of-item-5-of-obs1:O-structure-for-existence-of-MEC}
        There exists $(z_1, z_2) \in E_O$ such that $P_1((u,v), (z_1, z_2)) = P_1((z_1, z_2), (x,y)) = 1$.
    \end{enumerate}
    
    \item \label{item-6-of-obs1:O-structure-for-existence-of-MEC}
    For $((u,v),w) \in E_O \times V_O$, $P_2((u,v),w) = 1$ if, and only if, $v \neq w$ and at least one of the following occurs:
    \begin{enumerate}
        \item \label{subitem-1-of-item-6-of-obs1:O-structure-for-existence-of-MEC}
        There exists a \tfp{} from $(u,v)$ to $w$ in $O$.
        
        \item \label{subitem-2-of-item-6-of-obs1:O-structure-for-existence-of-MEC}
        For $a \in \{1, 2\}$, $P_{a2}((u, v), w) = 1$.
        
        \item \label{subitem-3-of-item-6-of-obs1:O-structure-for-existence-of-MEC}
        There exists $(z_1, z_2) \in E_O$ such that
        $P_1((u,v), (z_1, z_2)) = P_2((z_1, z_2), w) = 1$.
    \end{enumerate}
\end{enumerate}
\end{lemma}

\begin{proof}[Proof of \cref{obs1:O-structure-for-existence-of-MEC}]
  Let $M_1 \in \setofMECs{H_1, O_1, P_{11}, P_{12}}$,
$M_2 \in \setofMECs{H_2, O_2, P_{21}, P_{22}}$, and
$M \in \setofMECs{H, O, P_1, P_2}$ be such that
$\mathcal{P}(M, V_{M_1}, V_{M_2}) = (M_1, M_2)$.

\begin{proof}[Proof of \cref{item-1-of-obs1:O-structure-for-existence-of-MEC}:]
Suppose for $a \in \{1,2\}$, $u \rightarrow v \in O_a$. 
From the construction, $u, v \in V_{O_a} = S_a \cup N(S_a, H_a) \subseteq S_1 \cup S_2 \cup N(S_1 \cup S_2, H)$.
From \cref{def:subsets-of-MEC-based-on-O-P1-and-P2}, $O_a$ is an induced subgraph of $M_a$ on $S_a \cup N(S_a, H_a)$. This implies that $u \rightarrow v \in M_a$. Since $M_a$ is a projection of $M$, from \cref{lem:directed-edge-is-same-in-projected-MEC}, $u \rightarrow v \in M$. 
From \cref{def:subsets-of-MEC-based-on-O-P1-and-P2}, $O = M[S_1 \cup S_2 \cup N(S_1 \cup S_2, H)]$.
This implies that $u \rightarrow v \in O$.
\end{proof}

\begin{proof}[Proof of \cref{item-2-of-obs1:O-structure-for-existence-of-MEC}]
From our assumption, for $a\in \{1,2\}$, $M_a \in \setofMECs{H_a, O_a, P_{a1}, P_{a2}}$. This implies $(O_a, P_{a1}, P_{a2})$ is the shadow of $M_a$ on $V_{O_a} = S_a \cup N(S_a, H_a)$. Then, from \cref{item-1-of-def:shadow} of \cref{def:shadow}, $O_a = M_a[S_a \cup N(S_a, H_a)]$. This implies that $\mathcal{V}(O_a) = \mathcal{V}(M_a[S_a \cup N(S_a, H_a)])$. Since $M_a$ is a projection of $M$ onto $V_{H_a}$, and $S_a \cup N(S_a, H_a) \subseteq V_{H_a}$, \cref{corr:projection-of-subgraph} implies that $\mathcal{V}(M_a[S_a \cup N(S_a, H_a)]) = \mathcal{V}(M[S_a \cup N(S_a, H_a)])$. Furthermore, considering that $O$ is an induced subgraph of $M$ over $S_1 \cup S_2 \cup N(S_1 \cup S_2, H)$, and given that $S_a \cup N(S_a, H_a) \subseteq S_1 \cup S_2 \cup N(S_1 \cup S_2, H)$, we can conclude that $\mathcal{V}(M[S_a \cup N(S_a, H_a)]) = \mathcal{V}(O[S_a \cup N(S_a, H_a)])$. Thus, we have successfully established the relationship $\mathcal{V}(O_a) = \mathcal{V}(O[S_a \cup N(S_a, H_a)])$.
This further implies that $\mathcal{V}(O_a) = \mathcal{V}(O[V_{O_a}])$, as from the construction, $V_{O_a} = S_a \cup N(S_a, H_a)$.

\end{proof}

\begin{proof}[Proof of \cref{item-3-of-obs1:O-structure-for-existence-of-MEC}]
We present a more robust result than \cref{item-3-of-obs1:O-structure-for-existence-of-MEC} from \cref{obs1:O-structure-for-existence-of-MEC}.

\begin{claim}
\label{copy-of-item-3-of-obs1:O-structure-for-existence-of-MEC}
For $a\in \{1,2\}$, for an edge $u-v \in M_a$, $u\rightarrow v \in M$ if and only if either of the following occurs:
    \begin{enumerate}
        \item 
        \label{subitem1:copy-of-item-3-of-obs1:O-structure-for-existence-of-MEC}
        $u\rightarrow v$ is strongly protected in $O$.
        \item
        \label{subitem2:copy-of-item-3-of-obs1:O-structure-for-existence-of-MEC}
        There exists an undirected edge $x-y\in O_a$ such that $x\rightarrow y \in O$, and there exists a \tfp{} $Q = (u_1=x, u_2 =y, \ldots, u_{l-1} =u, u_l =v)$ from $(x,y)$ to $(u,v)$ in $M_a$ where $l\geq 3$.
        \item 
        \label{subitem3:copy-of-item-3-of-obs1:O-structure-for-existence-of-MEC}
        There exists an undirected edge $x-y \in O_a$ such that $x\rightarrow y \in O$, and there exist \tfps{} $Q_1 = (u_1=x, u_2 = y, \ldots, u_l = v)$ from $(x,y)$ to $v$ in $M_a$, and $Q_2 = (v_1 =v, v_2 =u, \ldots, v_m =x)$ from $(v,u)$ to $x$ in $M_a$ where $l,m\geq 3$. 
    \end{enumerate}
\end{claim}

\begin{proof} 
\Cref{obs:directed-edge-is-in-O,obs:directed-edge-and-undirected-path-implies-directed-edge,obs:directed-edge-with-cycle-implies-directed-edge} prove the $\rightarrow$ direction of \cref{copy-of-item-3-of-obs1:O-structure-for-existence-of-MEC}.

\begin{observation}
\label{obs:directed-edge-is-in-O}
If $u\rightarrow v$ is strongly protected in $O$, then $u\rightarrow v \in M$.
\end{observation}

\begin{proof}
The fact that $u\rightarrow v$ is strongly protected in $O$ implies that $u\rightarrow v \in O$.
From our assumption, $M\in \setofMECs{H, O, P_1, P_2}$.
From \cref{def:subsets-of-MEC-based-on-O-P1-and-P2}, this implies $(O, P_1, P_2)$ is the shadow of $M$ on $Y$. Therefore, from \cref{item-1-of-def:shadow} of \cref{def:shadow}, $O$ is an induced subgraph of $M$ on $Y$.
Thus, $u\rightarrow v \in O$ implies that $u\rightarrow v \in M$.
\end{proof}

\begin{observation}
\label{obs:directed-edge-and-undirected-path-implies-directed-edge}
For $u-v \in M_a$, if there exists $x-y \in O_a$ such that there is a \tfp{} $Q = (u_1 = x, u_2= y, \ldots, u_{l-1} = u, u_l = v)$ from $(x,y)$ to $(u,v)$ in $M_a$, and $x\rightarrow y \in O$, then for all $1\leq i < l$, $u_i \rightarrow u_{i+1} \in M$, more specifically, $u\rightarrow v \in M$.
\end{observation}
\begin{proof}
Suppose $u-v,x-y  \in M_a$, $x\rightarrow y \in O$, and $Q = (u_1 = x, u_2= y, \ldots, u_{l-1} = u, u_l = v)$ is a \tfp{} from $(x,y)$ to $(u,v)$ in $M_a$. 
Since $O$ is an induced subgraph of $M$, and $x\rightarrow y \in O$, it follows that $x\rightarrow y \in M$. Then, according to \cref{obs:every-edge-of-triangle-free-path-is-directed}, for all $1\leq i < l$, $u_i \rightarrow u_{i+1} \in M$, and more specifically, $u\rightarrow v \in M$.
\end{proof}

\begin{observation}
\label{obs:directed-edge-with-cycle-implies-directed-edge}
For $u-v \in M_a$, if there exists an undirected edge $x-y \in O_a$ such that $x \rightarrow y \in O$, and in $M_a$, there exist \tfps{} $Q_1 = (u_1 = x, u_2 = y, \ldots, u_l = v)$ from $(x,y)$ to $v$ in $M_a$, and $Q_2 = (v_1 = v, v_2 = u, \ldots, v_m = x)$ from $(v,u)$ to $x$ in $M_a$, where $l, m \geq 3$, then $u\rightarrow v \in M$.
\end{observation}
\begin{proof}
Suppose $u-v, x-y \in O_a$, $x \rightarrow y \in O$, $Q_1 = (u_1 = x, u_2 = y, \ldots, u_l = v)$ is a \tfp{} from $(x,y)$ to $v$ in $M_a$, and $Q_2 = (v_1 = v, v_2 = u, \ldots, v_m = x)$ is a \tfp{} from $(v,u)$ to $x$ in $M_a$, with $l, m \geq 3$.

Since $O$ is an induced subgraph of $M$ and $x \rightarrow y \in O$, it follows that $x \rightarrow y \in M$. Then, according to \cref{obs:cond-for-ud-path-in-M_a-to-be-ud-path-in-M}, $Q_1$ is a \tfp{} from $(x,y)$ to $v$ in $M$. Since $x\rightarrow y \in M$, $Q_1$ is a directed path in $M$.

$Q_2$ cannot be a path in $M$, otherwise, we would obtain a directed cycle in $M$ by concatenating $Q_1$ and $Q_2$, which contradicts \cref{item-1-theorem-nec-suf-cond-for-MEC} of \cref{thm:nes-and-suf-cond-for-chordal-graph-to-be-an-MEC}.
But then, from \cref{obs:cond-for-ud-path-in-M_a-to-be-ud-path-in-M}, if $u\rightarrow v \notin M$ then $Q_2$ is a \tfp{} in $M$, a contradiction. This implies $u\rightarrow v \in M$. This completes the proof of \cref{obs:directed-edge-with-cycle-implies-directed-edge}.
\end{proof}

We now prove the $\leftarrow$ direction of \cref{copy-of-item-3-of-obs1:O-structure-for-existence-of-MEC}. W.l.o.g., let $a = 1$. Since $M$ is a chain graph,  $M[V_{M_1}]$ is also a chain graph. 
Then from \cref{prop:every-chain-graph-has-a-topological-ordering},  there must exist an ordering $\tau$ of the vertices of $M_1$ such that for $u,v\in V_{M_1}$, if $u\rightarrow v \in M$ then $\tau(u) < \tau(v)$.
Pick a least rank node $v$ in $\tau$ such that for any edge $x-y \in O_1$, if $x\rightarrow y \in M$, and $\tau(y) < \tau(v)$ then $x\rightarrow y$ obeys \cref{copy-of-item-3-of-obs1:O-structure-for-existence-of-MEC} (i.e., $x\rightarrow y$ either follows \cref{subitem1:copy-of-item-3-of-obs1:O-structure-for-existence-of-MEC} of \cref{copy-of-item-3-of-obs1:O-structure-for-existence-of-MEC}, or it follows \cref{subitem2:copy-of-item-3-of-obs1:O-structure-for-existence-of-MEC} of \cref{copy-of-item-3-of-obs1:O-structure-for-existence-of-MEC}, or it follows \cref{subitem3:copy-of-item-3-of-obs1:O-structure-for-existence-of-MEC} of \cref{copy-of-item-3-of-obs1:O-structure-for-existence-of-MEC}). After picking such a node $v$, pick highest rank node $u$ such that $u-v \in M_1$, $u\rightarrow v \in M$, and for any node $z\in V_{M_1}$ such that $\tau(u) < \tau(z) < \tau(v)$, if $z-v \in M_1$, and $z\rightarrow v \in M$ then $z\rightarrow v$ obeys \cref{copy-of-item-3-of-obs1:O-structure-for-existence-of-MEC}. We show that $u\rightarrow v$ also obeys \cref{copy-of-item-3-of-obs1:O-structure-for-existence-of-MEC}. Our selection of $u-v \in M_1$ implies the following claim:
\begin{claim}
\label{claim:predecessor-of-u-v-obeys-the-claim}
For any edge $u'-v' \in M_1$ such that $u'\rightarrow v' \in M$, if $\tau(v') < \tau(v)$, or $v'=v$ and $\tau(u)< \tau(u')$ then $u'\rightarrow v'$ obeys \cref{copy-of-item-3-of-obs1:O-structure-for-existence-of-MEC},
i.e., either
\begin{enumerate}
    \item $u'\rightarrow v'$ is strongly protected in $O$, or
    \item there exists $x-y \in O_1$ such that $x\rightarrow y \in O$ and there exists a \tfp{} from $(x,y)$ to $(u',v')$ in $M_1$ of length greater than one (i.e., the path containing more than two nodes), or
    \item  there exists $x-y \in O_1$ such that $x\rightarrow y\in O$ and there exist \tfps{} in $M_1$ from $(x,y)$ to $v'$, and from $(v',u')$ to $x$ such that each \tfp{} is of length greater than one.
\end{enumerate}
\end{claim}

Since $M$ is an MEC, and $u\rightarrow v \in M$, from \cref{item-4-theorem-nec-suf-cond-for-MEC} of \cref{thm:nes-and-suf-cond-for-chordal-graph-to-be-an-MEC}, $u\rightarrow v$ must be strongly protected in $M$. This implies that $u\rightarrow v$ is a part of one of the induced subgraphs in \cref{fig:strongly-protected-edge}. We show that $u\rightarrow v$ obeys \cref{copy-of-item-3-of-obs1:O-structure-for-existence-of-MEC} in each such possible induced subgraph of $M$.
We now go through all the possibilities.

\textbf{Case 1:} Suppose $u \rightarrow v$ is strongly protected in $M$ because it is part of an induced subgraph $w \rightarrow u \rightarrow v \in M$ (similar to \Cref{fig:strongly-protected-edge}.a). There are two possibilities: either $w \in V_{H_1}$ or $w \notin V_{H_1}$.

If $w \notin V_{H_1}$, then $w \in V_{H_2} \setminus I$. This implies that $u \in I$, as $u \in V_{H_1}$ is a neighbor of $w \in V_{H_2} \setminus I$ in $H$, and $I$ is a vertex separator of $H$ that separates $V_{H_1} \setminus I$ and $V_{H_2} \setminus I$. This, in turn, implies that $u, v, w \in I \cup N(I, H) \subseteq V_O$, and the induced subgraph is also a part of $O$. Thus, $u \rightarrow v$ is strongly protected in $O$, as shown in \cref{fig:strongly-protected-edge}.a. Let's now consider the other possibility.

Suppose $w \in V_{H_1}$. Since $M_1$ is a projection of $M$ onto $V_{H_1}$, we can deduce, from \cref{corr:directed-edge-of-main-graph-is-either-dir-or-ud-in-proj-graph} and \cref{item-3-theorem-nec-suf-cond-for-MEC} of \cref{thm:nes-and-suf-cond-for-chordal-graph-to-be-an-MEC}, that $w - u \in M_1$. Also, as $w \rightarrow u \rightarrow v \in M$ and by the construction of $\tau$, we have $\tau(w) < \tau(u) < \tau(v)$. According to \cref{claim:predecessor-of-u-v-obeys-the-claim}, $w \rightarrow u$ satisfies \cref{copy-of-item-3-of-obs1:O-structure-for-existence-of-MEC}, specifically, either (1) \cref{subitem1:copy-of-item-3-of-obs1:O-structure-for-existence-of-MEC}, (2) \cref{subitem2:copy-of-item-3-of-obs1:O-structure-for-existence-of-MEC}, or (3) \cref{subitem3:copy-of-item-3-of-obs1:O-structure-for-existence-of-MEC}. We will explore each possibility and demonstrate that in each case, $u \rightarrow v$ obeys \cref{copy-of-item-3-of-obs1:O-structure-for-existence-of-MEC}.

\begin{enumerate}
    \item Suppose $w \rightarrow u$ satisfies \cref{subitem1:copy-of-item-3-of-obs1:O-structure-for-existence-of-MEC}. This implies that $w \rightarrow u$ is strongly protected in $O$, and thus, $w \rightarrow u \in O$. In this scenario, $u \rightarrow v$ obeys \cref{subitem2:copy-of-item-3-of-obs1:O-structure-for-existence-of-MEC}, since there exists a triangle-free path $P = (w, u, v)$ from $w - u$ to $u - v$ in $M_1$ (with length greater than one), and $w \rightarrow u \in O$.

    \item Suppose $w \rightarrow u$ satisfies \cref{subitem2:copy-of-item-3-of-obs1:O-structure-for-existence-of-MEC}. In this case, there exists an undirected edge $x - y \in O_1$ such that $x \rightarrow y \in O$, and a triangle-free path $Q = (u_1 = x, u_2 = y, \ldots, u_{l-1} = w, u_l = u)$ from $(x, y)$ to $(w, u)$ in $M_1$, with $l \geq 3$. We first demonstrate that $Q$ does not contain $v$. According to \cref{obs:undirected-tfp-with-ud-end-implies-source-is-ud}, $Q$ forms a chordless path in $M_1$. If $Q$ contains $v$, then there must be some $i < l-1$ where $u_i = v$ (since $u_l = u$ and $u_{l-1} = w$). This would lead to a contradiction as $Q$ would no longer be a chordless path. Thus, $Q$ does not contain $v$.

    As $Q' = (w, u, v)$ is a triangle-free path from $(w, u)$ to $(u, v)$ in $M_1$, we can use \cref{obs:concatenate-triangle-free-paths} to conclude that $Q'' = (u_1 = x, u_2 = y, \ldots, u_{l-1} = w, u_l = u, v)$ is a triangle-free path from $(x, y)$ to $(u, v)$ in $M_1$ (by concatenating $Q$ and $Q'$). Consequently, $u \rightarrow v$ satisfies \cref{subitem2:copy-of-item-3-of-obs1:O-structure-for-existence-of-MEC} as $Q''$ is a triangle-free path in $M_1$ (of length greater than one, given that $l \geq 3$), and $x \rightarrow y \in O$.

    \item Suppose $w \rightarrow u$ satisfies \cref{subitem3:copy-of-item-3-of-obs1:O-structure-for-existence-of-MEC}. This implies that there exists an undirected edge $x - y \in M_1$, where $x \rightarrow y \in O$, and there are triangle-free paths $Q_1 = (u_1 = x, u_2 = y, \ldots, u_l = u)$ from $(x, y)$ to $u$ in $M_1$, and $Q_2 = (v_1 = u, v_2 = w, \ldots, v_m = x)$ from $(u, w)$ to $x$ in $M_1$, with both $l$ and $m$ are greater than two. No edge in $Q_1$ and $Q_2$ is directed in $M_1$. A directed edge in $Q_1$ or $Q_2$ implies a directed cycle in $M_1$ (which we get by concatenating $Q_1$ and $Q_2)$), contradicting \cref{item-1-theorem-nec-suf-cond-for-MEC} of \cref{thm:nes-and-suf-cond-for-chordal-graph-to-be-an-MEC}.

Therefore, both $v$ and the nodes in $Q_1$ and $Q_2$ belong to the same undirected connected component $\mathcal{C}$ of $M_1$, and both $Q_1$ and $Q_2$ are undirected triangle-free paths in $\mathcal{C}$. By \cref{obv:tfps-are-chordless-in-chordal-graphs}, both $Q_1$ and $Q_2$ are chordless paths in $\mathcal{C}$.

Now, let's consider $Q_3 = (u_1 = x, u_2 = y, \ldots, u_l = u, u_{l+1} = v)$. If $Q_3$ is a chordless path in $\mathcal{C}$, it is also a chordless path in $M_1$, as $\mathcal{C}$ is an induced subgraph of $M_1$. In this case, $u \rightarrow v$ satisfies \cref{subitem2:copy-of-item-3-of-obs1:O-structure-for-existence-of-MEC}, as $Q_3$ is a triangle-free path in $M_1$.

However, if $Q_3$ is not a chordless path in $\mathcal{C}$, this implies that there must exist an index $i \leq l-1$ where an edge exists between $u_i$ and $v$ in $\mathcal{C}$. According to \cref{obs:edge-in--ucc-is-undirected}, $u_i - v$ is an undirected edge in $\mathcal{C}$. Since $Q_1$ is a triangle-free path in $M_1$ and $x \rightarrow y \in M$, from \cref{obs:every-edge-of-triangle-free-path-is-directed}, $Q_1$ forms a fully directed path in $M$. Moreover, as $Q_1$ is fully directed and $u\rightarrow v \in M$, $Q_3$ is also fully directed in $M$. Since $Q_3$ passes through $u_i$, there exists a directed path from $u_i$ to $v$ in $M$. This implies that $u_i \rightarrow v \in M$. Otherwise, $M$ would have a directed cycle $(u_i, u_{i+1}, \ldots, u_l = u, v, u_i)$, which contradicts \cref{item-1-theorem-nec-suf-cond-for-MEC} of \cref{thm:nes-and-suf-cond-for-chordal-graph-to-be-an-MEC}, as $M$ is an MEC.

If $i < l-1$, then $u_i \rightarrow v \leftarrow u_l$ forms a v-structure in $M$ (since there is no edge between $u_i$ and $u_l$ in $M$, as $P_1$ is a chordless path). However, this leads to a contradiction because it would imply that $\mathcal{V}(M[V_{M_1}]) \neq \mathcal{V}(M_1)$ (as $u_i \rightarrow v \leftarrow u_l \in M$ while $u_i - v - u_l \in M_1$). But, from \cref{def:projection}, since $M_1$ is a projection of $M$, both $M_1$ and $M[V_{M_1}]$ must have the same set of v-structures. Thus, the only possible value of $i$ is $l-1$. In this case, $Q_4 = (u_1 = x, u_2 = y, \ldots, u_{l-1}, v)$ is a chordless path in $M_1$. Additionally, since $Q_5 = (v, u, w)$ is a triangle-free path in $M_1$ from $(v, u)$ to $(u, w)$, and $Q_2$ is a triangle-free path in $M_1$ from $(u, w)$ to $x$, we can use \cref{obs:concatenate-triangle-free-paths} to show that $Q_6 = (v, u, w, v_3, v_4, \ldots, v_m = x)$ is a triangle-free path from $(v, u)$ to $(w, x)$ in $M_1$. Both $Q_4$ and $Q_6$ have lengths greater than one. Hence, $u \rightarrow v$ satisfies \cref{subitem3:copy-of-item-3-of-obs1:O-structure-for-existence-of-MEC} of \cref{copy-of-item-3-of-obs1:O-structure-for-existence-of-MEC} as $Q_4$ is a triangle-free path from $(x, y)$ to $v$ in $M_1$, $Q_6$ is a triangle-free path from $(v, u)$ to $x$ in $M_1$, and $x \rightarrow y \in O$.
\end{enumerate}

This implies that in all the cases of Case 1, $u\rightarrow v$ obeys \cref{item-3-of-obs1:O-structure-for-existence-of-MEC}. 

\textbf{Case 2:} Consider the scenario where $u \rightarrow v$ is strongly protected in $M$ because it is part of an induced subgraph $u \rightarrow v \leftarrow w \in M$ (similar to \Cref{fig:strongly-protected-edge}.b). There are two possibilities: either $w \notin V_{H_1}$ or $w \in V_{H_1}$.

If $w \notin V_{H_1}$, as demonstrated in Case 1, the induced subgraph is also an induced subgraph of $O$. Therefore, $u \rightarrow v$ is \spe{} in $O$, and it obeys \cref{subitem1:copy-of-item-3-of-obs1:O-structure-for-existence-of-MEC} of \cref{copy-of-item-3-of-obs1:O-structure-for-existence-of-MEC}. Let's now consider the other possibility.

Suppose $w \in V_{H_1}$. Then, $v \rightarrow u \leftarrow w$ is a v-structure in $M$ such that $u,v,w \in V_{M_1}$.
Since  $M_1$ is a projection of $M$ onto $V_{H_1}$, from \cref{def:projection},  the set of v-structures of $M_1$ and the set of v-structures of $M[V_{M_1}]$ must be the same. This implies $u \rightarrow v \leftarrow w \in M_1$.  This leads to a contradiction since our assumption is that $u - v \in M_1$. Hence, it implies that $w \notin V_{H_1}$.

\textbf{Case 3:} Suppose $u \rightarrow v$ is strongly protected in $M$ because it is part of an induced subgraph $u \rightarrow w \rightarrow v \leftarrow u \in M$ (same as \Cref{fig:strongly-protected-edge}.c). There are two possibilities: either $w \notin V_{H_1}$ or $w \in V_{H_1}$.

If $w \notin V_{H_1}$, as demonstrated in Case 1, the induced subgraph is also an induced subgraph of $O$. Therefore, $u \rightarrow v$ is \spe{} in $O$, and it adheres to \cref{subitem1:copy-of-item-3-of-obs1:O-structure-for-existence-of-MEC} of \cref{copy-of-item-3-of-obs1:O-structure-for-existence-of-MEC}. Let's now consider the other possibility.

Suppose $w \in V_{H_1}$. Since $u \rightarrow w$ and $w \rightarrow v \in M$, and since $M_1$ is a projection of $M$ onto $V_{H_1}$, we can deduce from \cref{corr:directed-edge-of-main-graph-is-either-dir-or-ud-in-proj-graph} that either $u \rightarrow w \in M_1$ or $u - w \in M_1$, and either $w \rightarrow v \in M_1$ or $w - v \in M_1$.

Given that $u - v \in M_1$, if $u \rightarrow w \in M_1$ or $w \rightarrow v \in M_1$, it would lead to a directed cycle $(u,w,v,u)$ in $M_1$, which contradicts \cref{item-1-theorem-nec-suf-cond-for-MEC} of \cref{thm:nes-and-suf-cond-for-chordal-graph-to-be-an-MEC}. Thus, we conclude that $u - w$ and $w - v \in M_1$. As $u \rightarrow w$ and $w \rightarrow v \in M$, based on the ordering in $\tau$, it holds that $\tau(u) < \tau(w) < \tau(v)$. Referring to \cref{claim:predecessor-of-u-v-obeys-the-claim}, we see that both $u \rightarrow w$ and $w \rightarrow v$ adhere to \cref{copy-of-item-3-of-obs1:O-structure-for-existence-of-MEC}. This implies that either $u \rightarrow w$ obeys \cref{subitem1:copy-of-item-3-of-obs1:O-structure-for-existence-of-MEC}, or it follows \cref{subitem2:copy-of-item-3-of-obs1:O-structure-for-existence-of-MEC}, or it satisfies \cref{subitem3:copy-of-item-3-of-obs1:O-structure-for-existence-of-MEC}. Let's examine each of these possibilities.

\begin{enumerate}
    \item Suppose $u\rightarrow w$ obeys \cref{subitem1:copy-of-item-3-of-obs1:O-structure-for-existence-of-MEC} of \cref{copy-of-item-3-of-obs1:O-structure-for-existence-of-MEC}, i.e., $u\rightarrow w$ is strongly protected in $O$. This implies $u,w \in V_O$ and $u\rightarrow w\in O$. As shown earlier, $w\rightarrow v$ also obeys \cref{copy-of-item-3-of-obs1:O-structure-for-existence-of-MEC}.

This implies either $w\rightarrow v$ obeys \cref{subitem1:copy-of-item-3-of-obs1:O-structure-for-existence-of-MEC} of \cref{copy-of-item-3-of-obs1:O-structure-for-existence-of-MEC}, or it obeys \cref{subitem2:copy-of-item-3-of-obs1:O-structure-for-existence-of-MEC}, or it obeys \cref{subitem3:copy-of-item-3-of-obs1:O-structure-for-existence-of-MEC} of \cref{copy-of-item-3-of-obs1:O-structure-for-existence-of-MEC}. We will consider each possibility.

\begin{enumerate}
    \item Suppose $w\rightarrow v$ obeys \cref{subitem1:copy-of-item-3-of-obs1:O-structure-for-existence-of-MEC} of \cref{copy-of-item-3-of-obs1:O-structure-for-existence-of-MEC}, i.e., $w\rightarrow v$ is strongly protected in $O$. This implies $w,v \in V_O$ and $w\rightarrow v \in O$. Since $u,v,w \in V_O$, $u\rightarrow v \in M$, and $O$ is an induced subgraph of $M$, we have $u\rightarrow v \in O$. This shows that $u\rightarrow v$ is also strongly protected in $O$, as it is part of the induced subgraph $u\rightarrow w\rightarrow v \leftarrow u \in O$. This implies $u\rightarrow v$ obeys \cref{subitem1:copy-of-item-3-of-obs1:O-structure-for-existence-of-MEC} of \cref{copy-of-item-3-of-obs1:O-structure-for-existence-of-MEC}.
    
    \item Suppose $w\rightarrow v$ obeys \cref{subitem2:copy-of-item-3-of-obs1:O-structure-for-existence-of-MEC} of \cref{copy-of-item-3-of-obs1:O-structure-for-existence-of-MEC}, i.e., there exists an edge $x-y \in O_1$ such that there is a \tfp{} $Q = (u_1 = x, u_2 = y, \ldots, u_{l-1} = w, u_l = v)$ from $(x,y)$ to $(w,v)$ in $M_1$, with $l\geq 3$ and $x\rightarrow y \in O$.
    
    According to \cref{obs:undirected-tfp-in-M-is-a-cp}, $Q$ is an undirected \cp{} in $M_1$. We first demonstrate that $Q$ does not contain $u$. Suppose $Q$ contains $u$, then for some $i\leq l-2$, $u_{i} = u$ (since $u_l = v$ and $u_{l-1} = w$). However, this contradicts the fact that $Q$ is a \cp{}, as $u-v \in M_1$. Hence, $Q$ does not contain $u$. Since $u-v \in M_1$, nodes of $Q$ and $u$ are in the same \uccc{} of $M_1$.
    
    We now show that $u_{l-2}-u \in M_1$. From \cref{obs:directed-edge-and-undirected-path-implies-directed-edge}, $u_{l-2}\rightarrow w \in M$, since $Q$ is a \tfp{} in $M_1$, and $x\rightarrow y\in M$ (because $O$ is an induced subgraph of $M$ and $x\rightarrow y \in O$). This implies there must be an edge between $u_{l-2}$ and $u$, otherwise, $u_{l-2}\rightarrow w\leftarrow u$ would form a v-structure in $M$. Since $M_1$ is a projection of $M$ on $V_{H_1}$ and $u_{l-2}, w, u \in V_{H_1}$, if $u_{l-2}\rightarrow w\leftarrow u$ is a v-structure in $M$, then $u_{l-2}\rightarrow w\leftarrow u \in M_1$. But this contradicts the construction where $w-u \in M_1$. Therefore, there must be an edge between $u_{l-2}$ and $u$. Since $u_{l-2}$ and $u$ are in the same \ucc{} of $M_1$, from \cref{prop:edge-between-2-nodes-of-same-ucc-is-ud}, $u_{l-2}-u \in M_1$.
    
    We now consider the least $i$ such that $u_i - u \in M_1$. If $i>1$, then $Q' = (u_1 = x, u_2 = y, \ldots, u_i, u,v)$ forms a \cp{} from $(x,y)$ to $(u,v)$, as $v$ cannot be a neighbor of any $u_i$ due to $Q$ being a \cp{} and $i\leq l-2$. This implies that $u\rightarrow v$ obeys \cref{subitem2:copy-of-item-3-of-obs1:O-structure-for-existence-of-MEC} of \cref{copy-of-item-3-of-obs1:O-structure-for-existence-of-MEC}. If $i=1$, then $Q_1 = (v,u, x)$ forms a \tfp{} from $(v,u)$ to $x$ in $M_1$. Therefore, $(u,v)$ obeys \cref{subitem3:copy-of-item-3-of-obs1:O-structure-for-existence-of-MEC} of \cref{copy-of-item-3-of-obs1:O-structure-for-existence-of-MEC}, as $Q$ is a \tfp{} from $(x,y)$ to $v$ in $M_1$ (of length greater than one), $Q_1$ is a \tfp{} from $(v,u)$ to $x$ in $M_1$ (of length greater than one), and $x\rightarrow y \in O$. This concludes the analysis of this case, showing that in each possibility $u\rightarrow v$ obeys \cref{copy-of-item-3-of-obs1:O-structure-for-existence-of-MEC}.
    
    \item Suppose $w\rightarrow v$ obeys \cref{subitem3:copy-of-item-3-of-obs1:O-structure-for-existence-of-MEC} of \cref{copy-of-item-3-of-obs1:O-structure-for-existence-of-MEC}, i.e., there exists an edge $x-y \in O_1$ such that $x\rightarrow y \in O$, and there exist \tfps{} $Q_1 = (u_1=x, u_2 = y, \ldots, u_l= v)$ from $(x,y)$ to $v$ in $M_1$, and $Q_2= (v_1 =v, v_2=w, \ldots, v_m=x)$ from $(v,w)$ to $x$ in $M_1$ such that $l,m\geq 3$.
    
    $Q_1$ and $Q_2$ must both be undirected paths in $M_1$, otherwise combining $Q_1$ and $Q_2$ would create a directed cycle in $M_1$. This implies that all nodes in $Q_1$, all nodes in $Q_2$, and $u$  belong to the same \uccc{} $\mathcal{C}$ of $M_1$. This further implies that $Q_1$ and $Q_2$ are the \tfps{} of a chordal graph $\mathcal{C}$. Then, from \cref{obv:tfps-are-chordless-in-chordal-graphs}, $Q_1$ and $Q_2$ are \cps{} of $M_1$.

    There are two possibilities: either $u$ is in $Q_1$, or $u$ is not in $Q_1$. If $u$ is in $Q_1$, then $u_{l-1} = u$, as $Q_1$ is a \cp{} in $\mathcal{C}$. Then, $Q_1$ is a \tfp{} of $M_1$ from $(x,y)$ to $(u,v)$. That means $(u,v)$ obeys \cref{subitem2:copy-of-item-3-of-obs1:O-structure-for-existence-of-MEC} of \cref{copy-of-item-3-of-obs1:O-structure-for-existence-of-MEC}, as $Q_1$ is a \tfp{} from $(x,y)$ to $(u,v)$ in $M_1$ and $x\rightarrow y \in O$. 

    Now we consider the other possibility. Suppose $u$ is not in $Q_1$. There are two possibilities: either there is an edge between $u_{l-1}$ and $u$ in $M$, or there is no edge between them in $M$. If there is no edge between them in $M$, then $u_{l-1}\rightarrow v\leftarrow u$ forms a v-structure in $M$. Since $M_1$ is a projection of $M$ and $u_{l-1},v,u \in M_1$, the v-structure $u_{l-1}\rightarrow v\leftarrow u \in M$ implies $u_{l-1}\rightarrow v\leftarrow u \in M_1$. This implies $u\rightarrow v \in M_1$. However, this is a contradiction, as $u-v \in M_1$. This implies there is an edge between $u_{l-1}$ and $u$ in $M$. From \cref{corr:directed-edge-of-main-graph-is-either-dir-or-ud-in-proj-graph,corr:undirected-edge-in-an-MEC-implies-undirected-edge-in-projected-MEC}, there is an edge between $u_{l-1}$ and $u$ in $M_1$ as well. Since $u_{l-1}$ and $u$ are in the same \ucc{} of $M_1$, therefore, from \cref{prop:edge-between-2-nodes-of-same-ucc-is-ud}, $u_{l-1}-u \in M_1$.

    We now consider a path $Q_3 \defeq (u_1=x, u_2 = y, \ldots, u_{l-1}, u)$, obtained by replacing $u_l$ with $u$ in the \cp{} $Q_1$. There are two possibilities: either $Q_3$ is not a \cp{} in $\mathcal{C}$, 
or $Q_3$ is a \cp{} in $\mathcal{C}$. We show that in both the possibility $u\rightarrow v$ obeys \cref{copy-of-item-3-of-obs1:O-structure-for-existence-of-MEC}. 

    \begin{enumerate}
        \item Suppose $Q_3$ is not a \cp{} in $\mathcal{C}$. Then there must exist a $k$ such that $k<l-1$ and $u_k-u \in \mathcal{C}$, because $Q_1$ is a \cp{} and $Q_3$ is not a \cp{} in $\mathcal{C}$. Pick the least $k$ such that $u_k-u \in \mathcal{C}$. If $k>1$, then this implies that $Q_3' \defeq (u_1 = x, u_2 = y, \ldots, u_k, u,v)$ is a \cp{} from $(x,y)$ to $(u,v)$ in $\mathcal{C}$ (since $Q_1$ is a \cp{} and there cannot be an edge between $v$ and any $u_i$s in $Q_3$). This further implies that $u\rightarrow v$ obeys \cref{subitem2:copy-of-item-3-of-obs1:O-structure-for-existence-of-MEC} of \cref{copy-of-item-3-of-obs1:O-structure-for-existence-of-MEC}, as $x\rightarrow y \in O$ and $Q_3'$ is a \tfp{} from $(x,y)$ to $(u,v)$ in $M_1$ (since $\mathcal{C}$ is an \ucc{} of $M_1$ and every \cp{} is a \tfp{}). If $k=1$, then $Q_3'' \defeq (v,u,x)$ is a \tfp{} from $(v,u)$ to $x$ in $M_1$ (since $Q_1$ is a \cp{} and there cannot be an edge between $x$ and $v$). This implies that $u\rightarrow v$ obeys \cref{subitem3:copy-of-item-3-of-obs1:O-structure-for-existence-of-MEC} of \cref{copy-of-item-3-of-obs1:O-structure-for-existence-of-MEC}, due to the existence of \tfps{} $Q_1$, $Q_3''$, and $x\rightarrow y \in O$. This shows that in each possibility of this subcase, $u\rightarrow v$ obeys \cref{copy-of-item-3-of-obs1:O-structure-for-existence-of-MEC}.

        \item Suppose $Q_3$ is a \cp{} in $\mathcal{C}$. We claim that $Q_4 \defeq (u, v_2 = w, v_3, \ldots, v_m = x)$ (obtained by replacing the first vertex $v$ of $Q_2$ by $u$) is not a \cp{} in $\mathcal{C}$. 

        \begin{proof}[Proof of the claim]
            Suppose that $Q_4$ is a \cp{} in $\mathcal{C}$. 
            Since $Q_3$ is a \cp{} in $\mathcal{C}$ and $\mathcal{C}$ is an \ucc{} of $M_1$, from \cref{obs:cp-is-tfp}, $Q_3$ is a \tfp{} in $M_1$.
            Since $x\rightarrow y \in O$, and $Q_3$ is a \tfp{} in $M_1$,
            from \cref{obs:directed-edge-and-undirected-path-implies-directed-edge}, $Q_3$ is a directed path in $M$. Also, since $u\rightarrow w \in M$ and $Q_4$ is a \cp{} in $M_1$, from \cref{obs:every-edge-of-triangle-free-path-is-directed}, $Q_4$ is a directed path in $M$. This creates a directed cycle in $M$ by concatenating $Q_3$ and $Q_4$. However, this is a contradiction, as stated in \cref{item-2-theorem-nec-suf-cond-for-MEC} of \cref{thm:nes-and-suf-cond-for-chordal-graph-to-be-an-MEC}, $M$ cannot have a directed cycle. This proves that $Q_4$ is not a \cp{} in $M_1$.
        \end{proof}

        Since $Q_2$ is a \cp{} and $Q_4$ is not a \cp{}, there must exist a $k>2$ such that $u-v_k\in M_1$. Pick the highest $k$, i.e., for all $j>k$, $u-v_j \notin M_1$. Then $Q_5 \defeq (u,v_k, v_{k+1}, \ldots, v_m=x)$ is a \cp{} in $\mathcal{C}$. Since $Q_2$ is a \cp{} in $M_1$, there cannot be an edge between $v$ and $v_i$ for $i\geq k >2$. This further implies that $Q_5' \defeq (v,u,v_k, v_{k+1}, \ldots, v_m=x)$ (obtained by adding $v$ to $Q_5$) is a \cp{} from $(v,u)$ to $x$ in $M_1$. This implies that $u\rightarrow v$ obeys \cref{subitem3:copy-of-item-3-of-obs1:O-structure-for-existence-of-MEC} of \cref{copy-of-item-3-of-obs1:O-structure-for-existence-of-MEC}, since $x\rightarrow y \in O$, $Q_1$ is a \tfp{} from $(x,y)$ to $v$ in $M_1$, and $Q_5'$ is a \tfp{} from $(v,u)$ to $x$ in $M_1$ (every \cp{} is a \tfp{}).
    \end{enumerate}

    \end{enumerate}

This concludes the analysis of all possibilities in this case, showing that in each case $u\rightarrow v$ obeys \cref{copy-of-item-3-of-obs1:O-structure-for-existence-of-MEC}.

\item Suppose $u\rightarrow w$ obeys \cref{subitem2:copy-of-item-3-of-obs1:O-structure-for-existence-of-MEC} of \cref{copy-of-item-3-of-obs1:O-structure-for-existence-of-MEC}, i.e., there exists an edge $x-y \in O_1$ such that there exists a \tfp{} $Q \defeq (u_1 =x, u_2 = y, \ldots, u_{l-1} = u, u_l = w)$ from $(x,y)$ to $(u,w)$ in $M_1$, and $x\rightarrow y \in O$. From \cref{obs:undirected-tfp-in-M-is-a-cp}, $Q$ is a \cp{} in $M_1$. We first show that $Q$ does not contain $v$. 

Suppose $Q$ contains $v$. Then $u_{l-2} = v$, since $Q$ is a \cp{}. Then from \cref{obs:directed-edge-and-undirected-path-implies-directed-edge}, $v\rightarrow u \in M$ (since $Q$ is a \cp{} from $(x,y)$ to $(u,w)$, $x\rightarrow y \in O$, $u_{l-2}=v$, and $u_{l-1} =u$). However, this is a contradiction, as $u\rightarrow v \in M$. This implies $v$ is not in $Q$.

$Q$ not containing $v$ implies that $Q_1 \defeq (u_1 =x, u_2 = y, \ldots, u_{l-1} = u, v)$, obtained by replacing $w$ with $v$ in $Q$, is a path in $M_1$. We show that $Q_1$ is a \cp{} in $M_1$.
This implies that $u\rightarrow v$ obeys \cref{subitem2:copy-of-item-3-of-obs1:O-structure-for-existence-of-MEC} of \cref{copy-of-item-3-of-obs1:O-structure-for-existence-of-MEC}, as $x\rightarrow y \in O$, and $Q_1$ is a \tfp{} from $(x,y)$ to $(u,v)$ in $M_1$ (every \cp{} is a \tfp{}) of length greater than one (since $l\geq 3)$. 
Thus, the only thing we need to complete this subcase is to show $Q_1$ is a \cp{} in $M_1$.

Suppose $Q_1$ is not a \cp{} in $M_1$.  
$Q' \defeq (u_1 =x, u_2 = y, \ldots, u_{l-1} = u)$, a subpath of the \cp{} $Q$, is a \cp{} from $x-y$ to $u$.
Since $Q'$ is a \cp{}, and $Q_1$ is not a \cp{}, there must exist an $i<l-1$ such that $u_i-v \in M_1$. As $x\rightarrow y \in O$ and $Q$ is a \tfp{} from $(x,y)$ to $(u,w)$, from \cref{obs:directed-edge-and-undirected-path-implies-directed-edge}, each edge of $Q$ is directed in $M$. Since $u\rightarrow v \in M$, $Q_1' \defeq (u_i, u_{i+1}, \ldots, u_{l-1}=u, v)$ is a directed path in $M$. This implies that $u_i\rightarrow v \in M$, otherwise, we get a directed cycle $C = (u_i, u_{i+1},\ldots, u_{l-1}, v,u_i)$ in $M$, which contradicts \cref{item-1-theorem-nec-suf-cond-for-MEC} of \cref{thm:nes-and-suf-cond-for-chordal-graph-to-be-an-MEC}. Since $Q$ is a \cp{}, there is no edge between $w$ and $u_i$ (since $i<l-1$). However, then $w\rightarrow v \leftarrow u_i$ forms a v-structure in $M$. Since $M_1$ is a projection of $M$ on $V_{H_1}$, and $w,v,u_i \in V_{H_1}$, from \cref{def:projection},  $w\rightarrow v \leftarrow u_i \in M_1$. But, this is a contradiction, as $w-v \in M_1$. This implies that $Q_1$ is a \cp{} in $M_1$. 

As argued above this concludes the proof for the case when $u\rightarrow w$ obeys \cref{subitem2:copy-of-item-3-of-obs1:O-structure-for-existence-of-MEC} of \cref{copy-of-item-3-of-obs1:O-structure-for-existence-of-MEC}. We now move to the final possibility when $u\rightarrow w$ obeys \cref{subitem3:copy-of-item-3-of-obs1:O-structure-for-existence-of-MEC} of \cref{copy-of-item-3-of-obs1:O-structure-for-existence-of-MEC}.

\item Suppose $u\rightarrow w$ obeys \cref{subitem3:copy-of-item-3-of-obs1:O-structure-for-existence-of-MEC} of \cref{copy-of-item-3-of-obs1:O-structure-for-existence-of-MEC}, i.e., there exists an edge $x-y \in O_1$ such that in $M_1$, there exist \tfps{} $Q_1 \defeq (u_1 = x, u_2 = y, \ldots, u_l = w)$ from $(x,y)$ to $w$, and $Q_2 \defeq (v_1=w, v_2 =u, \ldots, v_m =x)$ from $(w,u)$ to $x$. $Q_1$ and $Q_2$ both must be undirected paths in $M_1$, otherwise, combining $Q_1$ and $Q_2$ gives a directed cycle in $M_1$. This implies that the nodes in $Q_1$, nodes in $Q_2$, and $v$, all belong to the same \uccc{} $\mathcal{C}$ of $M_1$. This further implies that $Q_1$ and $Q_2$ are the \tfps{} of a chordal graph $\mathcal{C}$. Then, from \cref{obv:tfps-are-chordless-in-chordal-graphs}, $Q_1$ and $Q_2$ are \cps{} of $\mathcal{C}$. We now show that $v$ is neither in $Q_1$ nor in $Q_2$.

Suppose $Q_1$ contains $v$, i.e., for some $1\leq i< l$, $u_i = v$. From \cref{obs:directed-edge-and-undirected-path-implies-directed-edge}, each edge of $Q_1$ is directed in $M$, as $Q_1$ is a \tfp{} from $(x,y)$ to $(u_{l-1},w)$ in $M_1$, and $x\rightarrow y \in O$. This implies that $(u_i =v, u_{i+1}, \ldots, u_l = w, v)$ forms a directed cycle in $M$, as from the construction, $w\rightarrow v \in M$. However, this contradicts \cref{item-1-theorem-nec-suf-cond-for-MEC} of \cref{thm:nes-and-suf-cond-for-chordal-graph-to-be-an-MEC}. Therefore, $Q_1$ does not contain $v$.

    Suppose $Q_2$ contains $v$. Since $Q_2$ is a \cp{} in $M_1$, and $w-v \in M_1$, if $Q_2$ contains $v$, then $v_2=v$. But, $v_2 = u$, which is a contradiction. Therefore, $Q_2$ does not contain $v$.

We now show that either $Q_3 \defeq (u_1 = x, u_2 =y, \ldots, u_l = w, v)$ (obtained by adding $v$ to $Q_1$) or $Q_4 \defeq (u_1 = x, u_2 =y, \ldots, u_{l-1}, v)$ (obtained by replacing $w$ with $v$ in $Q_1$) is a \cp{} in $\mathcal{C}$, and $Q_5 = (v,v_2 = u, v_3, \ldots, v_m = x)$ (obtained by replacing the node $w$ with $v$ in $Q_2$) is a \cp{} in $\mathcal{C}$. This will be our main ingredient to prove that $u\rightarrow v$ obeys \cref{copy-of-item-3-of-obs1:O-structure-for-existence-of-MEC}.

Suppose neither $Q_3$ nor $Q_4$ is a \cp{} in $\mathcal{C}$. 
Since $Q_3$ is not a \cp{} in $\mathcal{C}$ and $Q_1$ is a \cp{} in $\mathcal{C}$, for some $i\leq l-1$, $u_i-v \in \mathcal{C}$. Pick the least $i$. If $i\neq l-1$ then we get an undirected \cc{} $(u_i, u_{i+1}, \ldots, u_{l-1}, u_l, v, u_i)$ in the chordal graph $\mathcal{C}$, a contradiction. This implies $u_{l-1} - v \in \mathcal{C}$. Similarly, since $Q_4$ is not a \cp{} in $\mathcal{C}$ and $Q_1$ is a \cp{} in $\mathcal{C}$, $u_{l-2} - v \in \mathcal{C}$. 
Since $Q_1$ is a \tfp{} in $M_1$ and $x\rightarrow y \in O$, from \cref{obs:directed-edge-and-undirected-path-implies-directed-edge}, $u_{l-2}\rightarrow u_{l-1}, u_{l-1}\rightarrow w \in M$. Then, $u_{l-2}\rightarrow v \in M$, otherwise, we get a directed cycle $C=(u_{l-2},u_{l-1},w,v,u_{l-2})$ in $M$ (remember $w\rightarrow v\in M$), which is a contradiction from \cref{item-1-theorem-nec-suf-cond-for-MEC} of \cref{thm:nes-and-suf-cond-for-chordal-graph-to-be-an-MEC}. 
But, then, $u_{l-2}\rightarrow v\leftarrow w$ is a v-structure in $M$ (since $Q_1$ is a \cp{}, there is no edge between $u_{l-2}$ and $w$). Since $M_1$ is a projection of $M$ on $V_{H_1}$, and $u_{l-2},v, w \in V_{H_1}$, from \cref{def:projection}, $u_{l-2}\rightarrow v\leftarrow w \in M_1$. But $w-v$ is an undirected edge in $M_1$, a contradiction. Thus, our assumption that neither $Q_3$ nor $Q_4$ is a \cp{} in $\mathcal{C}$ is wrong. This implies either $Q_3$ or $Q_4$ is a \cp{} in $\mathcal{C}$. We now show that $Q_5$ is a \cp{} in $\mathcal{C}$.

Since $Q_2$ is a \tfp{} in $\mathcal{C}$, from \cref{obs:undirected-tfp-in-M-is-a-cp}, $Q_2$ is a \cp{} in $\mathcal{C}$. This implies that the subpath $Q_5' \defeq (v_2 = u, v_3, \ldots, v_m = x)$ of $Q_2$ is a \cp{} in $\mathcal{C}$. Suppose $Q_5$ is not a \cp{} in $\mathcal{C}$. Then, there exists a $v_i$ such that $i>2$, and $v-v_i \in \mathcal{C}$. Pick the highest $i$ such that $v_i - v \in \mathcal{C}$. Then, $Q_6 \defeq (v, v_i, v_{i+1}, \ldots, v_m = x)$ is a \cp{} from $(v, v_i)$ to $x$ in $\mathcal{C}$.
Then, $v_i\rightarrow v \in M$. Otherwise, from \cref{obs:cond-for-ud-path-in-M_a-to-be-ud-path-in-M}, $Q_6$ is a \tfp{} in $M$, and the concatenation of $Q_1$, $w\rightarrow v$, and $Q_6$ yield a directed cycle in $M$, contradicting \cref{item-1-theorem-nec-suf-cond-for-MEC} of \cref{thm:nes-and-suf-cond-for-chordal-graph-to-be-an-MEC}.
Since $Q_2$ is a \cp{} in $M_1$, there cannot be an edge between $w$ and $v_i$ (as $i>2$). This implies that $w\rightarrow v \leftarrow v_i$ is a v-structure in $M$. Since $M_1$ is a projection of $M$ on $V_{H_1}$, and $w,v,v_i \in V_{H_1}$, from \cref{def:projection}, the v-structure $w\rightarrow v \leftarrow v_i \in M_1$. This gives a contradiction, as $w-v \in M_1$. This implies that $Q_5$ is a \cp{} in $M_1$.

Thus, in this case, there exists a \tfp{} from $(x,y)$ to $v$, in the form of either $Q_3$ or $Q_4$ (from \cref{obs:cp-is-tfp}, every \cp{} is a \tfp{}) such that $x\rightarrow y \in O$, and there exists a \tfp{} from $(v,u)$ to $x$ (in the form of $Q_5$). Since $l,m\geq 3$, the length of all the paths $Q_3, Q_4,$ and $Q_5$ is greater than one. This implies that $u\rightarrow v$ obeys \cref{subitem3:copy-of-item-3-of-obs1:O-structure-for-existence-of-MEC} of \cref{copy-of-item-3-of-obs1:O-structure-for-existence-of-MEC}.
    
\end{enumerate}

\textbf{Case 4:} Suppose $u\rightarrow v$ is strongly protected in $M$ because it is part of an induced subgraph $u\rightarrow v \leftarrow w-u-w'\rightarrow v$ (same as \Cref{fig:strongly-protected-edge}.d). 
  Both $w$ and $w'$ cannot be in $V_{H_1}$, otherwise the v-structure $w\rightarrow v \leftarrow w' \in M_1$, as $M_1$ is a projection of $M$ on $V_{H_1}$, and $w,v,w' \in V_{H_1}$. But, then, $u-v \notin M_1$, otherwise, it creates a directed cycle $(w,v,u,w)$, as from \cref{corr:undirected-edge-in-an-MEC-implies-undirected-edge-in-projected-MEC}, $u-w\in M_1$. This  implies that both $w$ and $w'$ cannot be in $V_{H_1}$. And, if $w$ or $w'$ is not in $V_{H_1}$ then as we have seen in Case 1, the induced subgraph is an induced subgraph of $O$, i.e.,  $u\rightarrow v$ is \spe{} in $O$ (as it is part of an induced subgraph of $O$ as shown in \cref{fig:strongly-protected-edge}.d), and $u\rightarrow v$ obeys \cref{subitem1:copy-of-item-3-of-obs1:O-structure-for-existence-of-MEC} of \cref{copy-of-item-3-of-obs1:O-structure-for-existence-of-MEC}.

We show that in all the possibilities, $u\rightarrow v$ obeys \cref{copy-of-item-3-of-obs1:O-structure-for-existence-of-MEC}. This completes the proof of $\leftarrow$ direction of \cref{copy-of-item-3-of-obs1:O-structure-for-existence-of-MEC}. This further completes the proof of  \cref{copy-of-item-3-of-obs1:O-structure-for-existence-of-MEC}.
\end{proof}
If there exists a \tfp{} $P$ from $(u,v)$ to $(x,y)$ of length more than one then $(u,v)\neq (x,y)$. Similarly, if $Q$ is a \tfp{} from $(u,v)$ to $w$ of length more than one then $v\neq w$. This shows that \cref{copy-of-item-3-of-obs1:O-structure-for-existence-of-MEC} implies  \cref{item-3-of-obs1:O-structure-for-existence-of-MEC} of \cref{obs1:O-structure-for-existence-of-MEC}.
\end{proof}

\begin{proof}[Proof of \cref{item-5-of-obs1:O-structure-for-existence-of-MEC}]
\Cref{obs:trp-in-O-is-a-tfp-in-M,obs:cp-in-Ma-is-a-cp-is-M,obs:two-cp-gives-a-cp} prove $\rightarrow$ of \cref{item-5-of-obs1:O-structure-for-existence-of-MEC}.

\begin{observation}
\label{obs:trp-in-O-is-a-tfp-in-M}
If $(u,v) \neq (x,y)$ and there exists a \tfp{} in $O$ from $(u,v)$ to $(x,y)$, then $P_1((u,v),(x,y)) = 1$.
\end{observation}
\begin{proof}
Since $(O, P_1, P_2)$ is the shadow of $M$, as defined in \cref{item-1-of-def:shadow} of \cref{def:shadow}, $O$ is an induced subgraph of $M$. Moreover, according to \cref{item-2-of-def:shadow} of \cref{def:shadow}, $P_1((u,v),(x,y))$ answers whether there exists a \tfp{} in $M$ from $(u,v)$ to $(x,y)$. If there exists a \tfp{} $P$ in $O$ from $(u,v)$ to $(x,y)$, then $P$ is also a \tfp{} in $M$, and therefore, $P_1((u,v),(x,y)) = 1$.
\end{proof}

\begin{observation}
\label{obs:cp-in-Ma-is-a-cp-is-M}
For $(u,v),(x,y) \in E_O$, if $(u,v) \neq (x,y)$ and for any $a\in \{1,2\}$, there exists a \tfp{} $P$ in $M_a$ from $(u,v)$ to $(x,y)$, i.e., $P_{a1}((u,v),(x,y)) = 1$, then $P$ is also a \tfp{} in $M$, i.e., $P_1((u,v),(x,y)) = 1$.
\end{observation}
\begin{proof}
Since $(O, P_1, P_2)$ is the shadow of $M$, as defined in \cref{item-1-of-def:shadow} of \cref{def:shadow}, $O$ is an induced subgraph of $M$. Moreover, according to \cref{item-2-of-def:shadow} of \cref{def:shadow}, $P_1((u,v),(x,y))$ answers whether there exists a \tfp{} in $M$ from $(u,v)$ to $(x,y)$. For $a\in \{1,2\}$, suppose $(u,v)\neq (x,y)$ and $P$ is a \tfp{} from $(u,v)$ to $(x,y)$ in $M_a$. Since $(u,v) \in E_O$, it follows that $v\rightarrow u \notin O$. This implies $v\rightarrow u \notin M$. Then, from \cref{obs:cond-for-ud-path-in-M_a-to-be-ud-path-in-M}, $P$ is a \tfp{} in $M$, i.e., $P_1((u,v),(x,y)) = 1$.
\end{proof}

\begin{observation}
\label{obs:two-cp-gives-a-cp}
For $u,v,x,y, z_1, z_2 \in V_O$, if $(u,v) \neq (x,y)$, $P_1((u,v),(z_1, z_2)) = 1$, and $P_1((z_1, z_2), (x,y)) = 1$, then $P_1((u,v),(x,y)) = 1$.
\end{observation}
\begin{proof}
Since $(O, P_1, P_2)$ is the shadow of $M$, as mentioned in \cref{item-2-of-def:shadow} of \cref{def:shadow}, $P_1((u,v),(x,y))$ answers whether there exists a \tfp{} in $M$ from $(u,v)$ to $(x,y)$. Suppose $u,v,x,y, z_1, z_2 \in V_O$ such that $(u,v) \neq (x,y)$, and $P_1((u,v),(z_1, z_2)) = P_1((z_1, z_2), (x,y)) = 1$. This implies there exist \tfps{} $Q_1=(u_1 = u, u_2=v, \ldots, u_{l-1} = z_1, u_l = z_2)$ from $(u,v)$ to $(z_1, z_2)$ in $M$ and $Q_2 = (v_1= z_1, v_2 = z_2, \ldots, v_{m-1} = x, v_m = y)$ in $ $M, from $(z_1, z_2)$ to $(x,y)$ in $M$. Then, according to \cref{obs:concatenate-triangle-free-paths}, $P = (u_1 = u, u_2=v, \ldots, u_{l-1} = z_1, u_l = z_2, v_3, v_4, \ldots, v_{m-1} = x, v_m = y)$ is a \tfp{} from $(u,v)$ to $(x,y)$ in $M$. Thus, $P_1((u,v), (x,y)) = 1$.
\end{proof}

We now prove $\leftarrow$ of \cref{item-5-of-obs1:O-structure-for-existence-of-MEC}. 
\begin{observation}
\label{obs:item-5-of-obs1:O-structure-for-existence-of-MEC}
For $((u,v),(x,y)) \in E_O \times E_O$, if $P_1((u,v),(x,y)) = 1$ then $(u,v) \neq (x,y)$ and either of the following occurs:
\begin{enumerate}
        \item
        \label{item-1-of-obs:item-5-of-obs1:O-structure-for-existence-of-MEC}
        There exists a \tfp{} from $(u,v)$ to $(x,y)$ in $O$.
        \item
        \label{item-2-of-obs:item-5-of-obs1:O-structure-for-existence-of-MEC}
        For $a\in \{1,2\}$, $P_{a1}((u,v),(x,y)) = 1$ (i.e.,  there exists a \tfp{} from $(u,v)$ to $(x,y)$ in $M_a$).
        \item 
        \label{item-3-of-obs:item-5-of-obs1:O-structure-for-existence-of-MEC}
        There exists $(z_1, z_2) \in E_O$ such that $P_1((u,v), (z_1, z_2)) = P_1((z_1, z_2), (x,y)) = 1$ (i.e., there exist  \tfps{} in $M$ from $(u,v)$ to $(z_1, z_2)$, and from $(z_1, z_2)$ to $(x,y)$).
    \end{enumerate}
\end{observation}
\begin{proof}
Suppose for some $((u,v),(x,y)) \in E_O \times E_O$,  $P_1((u,v),(x,y)) = 1$, then since $(O, P_1, P_2)$ is the shadow of $M$ on $V_O$, from \cref{def:shadow}, $(u,v) \neq (x,y)$ and there exists a \tfp{} $P= (u_1 = u, u_2 =v, \ldots, u_{l-1} = x, u_l =y)$ from $(u,v)$ to $(x,y)$ in $M$. 
If $P$ is completely in $O$ (i.e., all the nodes of $P$ belong to $V_O$) then $P$ is a \tfp{} in $O = M[V_O]$ (from \cref{def:shadow}, $O$ is an induced subgraph of $M$). 
This obeys \cref{item-1-of-obs:item-5-of-obs1:O-structure-for-existence-of-MEC} of \cref{obs:item-5-of-obs1:O-structure-for-existence-of-MEC}.  If all the nodes of $P$ belongs to $V_{H_a}$ for some $a\in \{1,2\}$  then from \cref{corr:tfp-in-main-graph-implies-tfp-in-projected-graph}, $P$ is a \tfp{} in $M_a$, i.e., $P_{a1}((u,v),(x,y)) =1$. 
This shows if all the nodes of $P$ belongs to $V_{H_a}$ for some $a\in \{1,2\}$  then \cref{item-2-of-obs:item-5-of-obs1:O-structure-for-existence-of-MEC} of \cref{obs:item-5-of-obs1:O-structure-for-existence-of-MEC} is obeyed. 

Suppose $P$ is neither completely in  $O$, nor completely  in $M_1$, nor completely in   $M_2$.
There are two possibilities either $u_3 \in V_O$ or $u_3\notin V_O$. We deal with each possibility.
If $u_3\in V_O$ then we have $(v,u_3) \in E_O$ such that 
(a) $Q_1 = (u_1=u, u_2=v, u_3)$, the subpath of $P$ from $u_1$ to $u_3$,  is a \tfp{} from $(u,v)$ to $(v,u_3)$ in $M$, i.e., $P_1((u,v), (v, u_3)) =1$ (from \cref{def:shadow}, as $(O, P_1, P_2)$ is the shadow of $M$), and 
(b) $Q_2 = (u_2=v, u_3, \ldots, u_{l-1} =x, u_l = y)$, the subpath of $P$ from $u_2$ to $u_l$, is a \tfp{} from $(v,u_3)$ to $(x,y)$ in $M$, i.e., $P_1((v, u_3), (x,y))=1$ (note that $(v,u_3)\neq (x,y)$, otherwise $P$ is completely in $O$, contradicting our assumption). This obeys \cref{item-3-of-obs:item-5-of-obs1:O-structure-for-existence-of-MEC} of \cref{obs:item-5-of-obs1:O-structure-for-existence-of-MEC}.

Suppose $u_3\notin V_O$.
Since $P$ is a \tfp{} in $M$, either $u-v \in M$ or $u\rightarrow v \in M$. 
Since $(u,v) \in E_O$,  either $u-v \in O$ or $u\rightarrow v \in O$. Then, from \cref{claim:u-v-is-either-in-O1-or-in-O2}, either $u,v \in V_{O_1}$ or $u,v \in V_{O_2}$, or both.
\begin{claim}
\label{claim:u-v-is-either-in-O1-or-in-O2}
    If $u-v \in O$ or $u\rightarrow v \in O$ then either $u,v\in V_{O_1}$ or $u,v \in V_{O_2}$.
\end{claim}
\begin{proof}
    From the construction, $V_O = V_{O_1}\cup V_{O_2}$, and $I = V_{O_1}\cap V_{O_2}$ is a vertex separator of $O$. W.l.o.g., suppose $u \in V_{O_1}$. If $v\in V_{O_1}$ then we are done. Suppose, $v\notin V_{O_1}$, i.e., $v\in V_O \setminus V_{O_1} = V_{O_2}\setminus I$. Then, $u\in I$, as $I$ is a vertex separator that separates $V_{O}\setminus V_{O_1}$ and $V_O \setminus V_{O_2}$. But, then, $u,v \in V_{O_2}$, as $I \subseteq V_{O_2}$. 
\end{proof}
W.l.o.g., suppose $u,v \in V_{O_1}$. 
 Then, there must exist a node $z$ in $P$ such that $z \notin V_{M_1}$, otherwise, from \cref{corr:tfp-in-main-graph-implies-tfp-in-projected-graph}, $P$ is a \tfp{} in $M_1$, contradicting our assumption.
Pick the first such $z$ in $P$, i.e., all the nodes before $z$ in $P$ are in $V_{M_1}$. Let $z_2$ be the predecessor of $z$ in $P$, and $z_1$ be the predecessor of $z_2$ in $P$.
Since $z$ is the first node in $P$ which is not in $V_{M_1}$, $z\in V_{M_2}\setminus V_{M_1}$, and $z_2 \in V_{M_1}$. 
Since $I = V_{H_1}\cap V_{H_2}$ is a vertex separator of $H$ that separates $V_{H_1} = V_{M_1}$ and $V_{H_2} = V_{M_2}$, therefore, $z_2\in I$, and $z\in V_{M_2}\setminus I$.
This  implies that $z, z_1,z_2 \in V_O$, as $V_O \supseteq I\cup N(I, O)$.
Since $u_3$ is not in $O$, we can say that $z_1 = u_k$ for some $k> 3$.
Then, $(z_1,z_2) \in E_O$, and 
(a) $Q_1 = (u_1=u, u_2=v, \ldots, u_k=z_1,u_{k+1}=z_2)$, the subpath of $P$ from $u$ to $z_2$, is a \tfp{} from $(u,v)$ to $(z_1,z_2)$ in $M$ such that $((u_1, u_2)\neq (z_1, z_2))$ (since $k> 3$), i.e., $P_1((u_1, u_2), (z_1, z_2)) =1$, and 
(b) $Q_2 = (u_k=z_1, u_{k+1}=z_2, \ldots , u_{l-1}=x, u_l=y)$, the subpath of $P$ from $z_1$ to $y$, is a \tfp{} from $(z_1,z_2)$ to $(x,y)$ in $M$ such that $(z_1, z_2) \neq (x, y)$ (if $(z_1, z_2)=(x,y)$ then $P$ is completely in $M_1$, contradicting our assumption), i.e., $P_1((z_1, z_2), (x,y)) =1$. This obeys \cref{item-3-of-obs:item-5-of-obs1:O-structure-for-existence-of-MEC} of \cref{obs:item-5-of-obs1:O-structure-for-existence-of-MEC}. This  implies that in all the possibilities,  \cref{obs:item-5-of-obs1:O-structure-for-existence-of-MEC} is obeyed. This completes the proof.
\end{proof}

This completes the proof of \cref{item-5-of-obs1:O-structure-for-existence-of-MEC}.
\end{proof}

\begin{proof}[Proof of \cref{item-6-of-obs1:O-structure-for-existence-of-MEC}]
Proof of \cref{item-6-of-obs1:O-structure-for-existence-of-MEC} is similar to the proof of \cref{item-5-of-obs1:O-structure-for-existence-of-MEC}. For completeness, we are giving the proof.

\Cref{obs:trp-in-O-is-a-tfp-in-M-copy,obs:cp-in-Ma-is-a-cp-is-M-copy,obs:two-cp-gives-a-cp-copy} prove $\leftarrow$ of \cref{item-6-of-obs1:O-structure-for-existence-of-MEC}.

\begin{observation}
\label{obs:trp-in-O-is-a-tfp-in-M-copy}
If $v\neq w$, and there exists a \tfp{} in $O$ from $(u,v)$ to $w$, then $P_2((u,v),w) = 1$.
\end{observation}
\begin{proof}
Suppose $v\neq w$.
If there exists a \tfp{} $P$ in $O$ from $(u,v)$ to $w$, then $P$ is also a \tfp{} in $M$. This follows from \cref{item-1-of-def:shadow} of \cref{def:shadow}, which states that $O$ is an induced subgraph of $M$. Consequently, according to \cref{item-3-of-def:shadow} of \cref{def:shadow}, which specifies that $P_2((u,v), w) =1$ if $v\neq w$ and there exists a \tfp{} from $(u,v)$ to $w$, we conclude that $P_2((u,v), w) =1$. 
\end{proof}

\begin{observation}
\label{obs:cp-in-Ma-is-a-cp-is-M-copy}
For $((u,v),w) \in E_O\times V_O$, if $v\neq w$, and for some $a\in \{1,2\}$, there exists a \tfp{} $P$ in $M_a$, from $(u,v)$ to $w$, i.e., $P_{a2}((u,v),w) =1$, then $P$ is also a \tfp{} in $M$, i.e., $P_2((u,v),w) =1$.
\end{observation}
\begin{proof}
Without loss of generality, let us assume $a = 1$.
Suppose $((u,v),w) \in E_O\times V_O$, $v\neq w$, and $P$ is a \tfp{} from $(u,v)$ to $w$ in $M_1$ (i.e., $P_{12}((u,v),w) = 1$).
Since $(u,v) \in E_O$, therefore, $(u,v) \in E_M$. This follows from \cref{item-1-of-def:shadow} of \cref{def:shadow}, which states that $O$ is an induced subgraph of $M$. This further implies that $v\rightarrow u \notin M$. Consequently, from \cref{obs:cond-for-ud-path-in-M_a-to-be-ud-path-in-M}, we deduce that $P$ is also a \tfp{} in $M$. Thus, according to \cref{item-3-of-def:shadow} of \cref{def:shadow}, which states that $P_2((u,v), w) = 1$ if $v\neq w$ and there exists a \tfp{} from $(u,v)$ to $w$, we conclude that $P_2((u,v),w) = 1$.
\end{proof}

\begin{observation}
\label{obs:two-cp-gives-a-cp-copy}
For $u, v, w, z_1, z_2 \in V_O$, if $v \neq w$, $P_1((u, v), (z_1, z_2)) = 1$, and $P_2((z_1, z_2), w) = 1$, then $P_2((u, v), w) = 1$.
\end{observation}

\begin{proof}
Suppose $P_1((u, v), (z_1, z_2)) = P_2((z_1, z_2), w) = 1$. 
From \cref{item-1-of-def:shadow,item-2-of-def:shadow} of \cref{def:shadow}, there exists \tfps{} $Q_1=(u_1 = u, u_2 = v, \ldots, u_{l-1} = z_1, u_l = z_2)$ from $(u, v)$ to $(z_1, z_2)$ in $M$ (as $P_1((u,v), (z_1, z_2)) = 1$), and $Q_2 = (v_1 = z_1, v_2 = z_2, \ldots, v_{m-1}, v_m = w)$ from $(z_1, z_2)$ to $w$ in $M$ (as $P_2((z_1, z_2), w) =  1$). Then, from \cref{obs:concatenate-triangle-free-paths}, $P = (u_1 = u, u_2 = v, \ldots, u_{l-1} = z_1, u_l = z_2, v_3, v_4, \ldots, v_{m-1}, v_m = w)$ is a \tfp{} from $(u, v)$ to $w$ in $M$. Thus, from \cref{item-2-of-def:shadow} of \cref{def:shadow}, $P_2((u, v), w) = 1$.
\end{proof}

We now prove $\leftarrow$ of \cref{item-6-of-obs1:O-structure-for-existence-of-MEC}. 
\begin{observation}
\label{obs:leftarrow-of-item-6-of-obs1:O-structure-for-existence-of-MEC}
For $((u,v),w) \in E_O\times V_O$, if $P_2((u,v),w) = 1$ then $v\neq w$, and either of the following occurs:
\begin{enumerate}
        \item
        \label{item-1-of-obs:leftarrow-of-item-6-of-obs1:O-structure-for-existence-of-MEC}
        There exists a \tfp{} from $(u,v)$ to $w$ in $O$.
        \item
        \label{item-2-of-obs:leftarrow-of-item-6-of-obs1:O-structure-for-existence-of-MEC}
        For some $a\in \{1,2\}$, $P_{a2}((u,v),w) = 1$ (i.e.,  there exists a \tfp{} from $(u,v)$ to $w$).
        \item
        \label{item-3-of-obs:leftarrow-of-item-6-of-obs1:O-structure-for-existence-of-MEC}
        There exists $(z_1, z_2) \in E_O$ such that $P_1((u,v), (z_1, z_2)) = P_2((z_1, z_2), w) = 1$ (i.e., there exist  \tfps{} in $M$ from $(u,v)$ to $(z_1, z_2)$, and from $(z_1, z_2)$ to $w$).
    \end{enumerate}
\end{observation}
\begin{proof}
        From \cref{item-3-of-def:shadow} of \cref{def:shadow}, if $P_2((u, v), w) = 1$, then $v \neq w$, and there exists a \tfp{} $P = (u_1 = u, u_2 = v, \ldots, u_l = w)$ from $(u, v)$ to $w$ in $M$. 

Similar to the proof of \cref{obs:item-5-of-obs1:O-structure-for-existence-of-MEC}, either
\begin{enumerate}
    \item $P$ is a \tfp{} in $O$, or
    \item $P$ is a \tfp{} in $M_1$ or $M_2$, i.e., for some $a \in \{1, 2\}$, $P_{a2}((u, v), w) = 1$, or
    \item there exists $(z_1, z_2) \in E_O$ such that $(u, v) \neq (z_1, z_2)$, $w \neq z_2$, and there exist \tfps{} in $M$ from $(u, v)$ to $(z_1, z_2)$ and from $(z_1, z_2)$ to $w$, i.e., $P_1((u, v), (z_1, z_2)) = P_2((z_1, z_2), w) = 1$.
\end{enumerate}

This completes the proof.

\end{proof}

This completes the proof of $\leftarrow$ of \cref{item-6-of-obs1:O-structure-for-existence-of-MEC}. This further completes the proof of \cref{item-6-of-obs1:O-structure-for-existence-of-MEC}.
\end{proof}
This completes the proof of \cref{obs1:O-structure-for-existence-of-MEC}.
\end{proof}

\subsection{Derived Path Function and Extension of Shadows}
\label{subsection:dependence-of-P1-and-P2-on-O-and-others}
We note that $P_1$ and $P_2$ in \cref{obs1:O-structure-for-existence-of-MEC} are not independent. Knowing $O, P_{11}, P_{12}, P_{21}$, and $P_{22}$, we can compute $P_1$ and $P_2$ using \cref{item-5-of-obs1:O-structure-for-existence-of-MEC,item-6-of-obs1:O-structure-for-existence-of-MEC} of \cref{obs1:O-structure-for-existence-of-MEC} (this is shown in \cref{lem:P1-P2-is-DPF}). This motivates us to define $(P_1, P_2)$ as a \epfs{} of $(O, P_{11}, P_{12}, P_{21}, P_{22})$. \Cref{def:extended-path-function} formally define this. 

\begin{definition}[Derived Path Function]
    \label{def:extended-path-function}
    Let $H, H_1$ and $H_2$ be  undirected graphs such that $H_1$ and $H_2$ are induced subgraphs of $H$, and $I = V_{H_1}\cap V_{H_2}$ is a vertex separator of $H$ that separates $V_{H_1}\setminus I$ and $V_{H_2}\setminus I$. Let $S_1 \subseteq V_{H_1}$ and $S_2\subseteq V_{H_2}$ such that $S_1\cap S_2 = I$. For $a\in \{1,2\}$, let $(O_a, P_{a1}, P_{a2})$ be a shadow of $H_a$ on $S_a\cup N(S_a, H_a)$ and $O\in \setofpartialMECs{H[S_1\cup S_2\cup N(S_1\cup S_2, H)]}$ such that $O$ obeys \cref{item-1-of-obs1:O-structure-for-existence-of-MEC,item-2-of-obs1:O-structure-for-existence-of-MEC,item-3-of-obs1:O-structure-for-existence-of-MEC} of \cref{obs1:O-structure-for-existence-of-MEC}. 
Let $P_1: E_O \times E_O \rightarrow \{0,1\}$, and $P_2: E_O\times V_O \rightarrow \{0,1\}$ be two functions. We call $(P_1, P_2)$ the \epfs{} of $(O, P_{11}, P_{12}, P_{21}, P_{22})$, denoted as $(P_1, P_2) = \EPF{O, P_{11}, P_{12}, P_{21}, P_{22}}$, if it is computed by following the below steps:
    \begin{enumerate}
        \item 
        \label{item-1-of-def:extended-path-function}
        Step 1 (Initialization):
        \begin{enumerate}
            \item $\forall ((u,v), (x,y)) \in E_O\times E_O$, if $(u,v)\neq (x,y)$ and there exists a \tfp{} from $(u,v)$ to $(x,y)$ in $O$ then $P_1((u,v),(x,y)) =1$ else $P_1((u,v),(x,y)) =0$.
            \item 
            $\forall ((u,v), w) \in E_O\times V_O$, if $v\neq w$ and there exists a \tfp{} from $(u,v)$ to $w$ in $O$ then $P_2((u,v),w) =1$ else $P_2((u,v), w) = 0$.            
        \end{enumerate}
        \item
        \label{item-2-of-def:extended-path-function}
        Step 2 (Update 1): 
          \begin{enumerate}
          \item
          \label{subitem-1-of-item-2-of-def:extended-path-function}
          $\forall ((u,v), (x,y)) \in E_O\times E_O$ such that
            $(u, v) \neq (x, y)$ and $P_1((u,v),(x,y)) = 0$, if for some $a\in \{1,2\}$, $P_{a1}((u,v), (x,y)) =1$ then update $P_1$ with $P_1((u,v), (x,y)) =1$.
\item 
          \label{subitem-2-of-item-2-of-def:extended-path-function}
          $\forall ((u,v), w) \in E_O\times V_O$ such that $v \neq w$ and $P_2((u,v),w) =0$, if for any $a\in \{1,2\}$, $P_{a2}((u,v), w) =1$ then update $P_2$ with $P_2((u,v), w) = 1$.
\end{enumerate}
        \item
        \label{item-3-of-def:extended-path-function}
        Step 3 (Update 2: add transitivity): 
        
        Repeat the following until no further change in $(P_1, P_2)$:
        \begin{enumerate}
        \item 
        \label{subitem-1-of-item-3-of-def:extended-path-function}
        $\forall ((u,v), (x,y)) \in E_O\times E_O$ such that
          $(u, v) \neq (x, y), $ and $P_1((u,v), (x,y)) = 0$, if there exists
          $(z_1, z_2)\in E_O$ such that
          $P_1((u,v), (z_1, z_2)) = P_1((z_1, z_2), (x,y)) = 1$ then update
          $P_1$ with $P_1((u,v), (x,y)) = 1$.
        \item 
        \label{subitem-2-of-item-3-of-def:extended-path-function}
        $\forall ((u,v), w) \in E_O\times V_O$, such that $v \neq w$ and
          $P_2((u,v), w) = 0$, if there exists $(z_1, z_2)\in E_O$ such that
          $P_1((u,v), (z_1, z_2)) = P_2((z_1, z_2), w) = 1$ then update $P_2$
          with $P_2((u,v), w) = 1$.
        \end{enumerate}        
    \end{enumerate}
    We emphasize again that $P_2$ is not in general determined by $P_1$; see \cref{rem:p2-not-determined-by-p1} following \cref{def:shadow}.  \end{definition}

\begin{algorithm}
\caption{\text{DPF}$(O ,P_{11}, P_{12}, P_{21}, P_{22})$}
\label{alg:constructDPF}
\SetAlgoLined
\SetKwInOut{KwIn}{Input}
\SetKwInOut{KwOut}{Output}
\SetKwFunction{constructDPF}{construct_DPF}
\KwIn{A graph $O$, and four functions $P_{11}, P_{12}, P_{21}$ and $P_{22}$ such that
$H, H_1$ and $H_2$ are  undirected graphs such that $H_1$ and $H_2$ are induced subgraphs of $H$, and 
$I = V_{H_1}\cap V_{H_2}$ is a vertex separator of $H$ that separates $V_{H_1}\setminus I$ and $V_{H_2}\setminus I$, 
$S_1 \subseteq V_{H_1}$ and $S_2\subseteq V_{H_2}$ such that $S_1\cap S_2 = I$,
for $a\in \{1,2\}$, $(O_a, P_{a1}, P_{a2})$ is a shadow of $H_a$ on $S_a\cup N(S_a, H_a)$ and
$O\in \setofpartialMECs{H[S_1\cup S_2\cup N(S_1\cup S_2, H)]}$ such that $O$ obeys \cref{item-1-of-obs1:O-structure-for-existence-of-MEC,item-2-of-obs1:O-structure-for-existence-of-MEC,item-3-of-obs1:O-structure-for-existence-of-MEC} of \cref{obs1:O-structure-for-existence-of-MEC}.
}
    \KwOut{ $(P_1, P_2)=\EPF{O, P_{11}, P_{12}, P_{21}, P_{22}}$.}
    
    $(P_1, P_2) \leftarrow$ TFP($O$) \label{alg:DPF:P_1-P_2-init}

    \ForEach{$a\in \{1,2\}$\label{alg:DPF:foreach-2-start}}
    {
        \ForEach{$((u,v),(x,y)) \in E_O\times E_O$ such that $(u,v)\neq (x,y)$\label{alg:DPF:foreach-2-a-start}}
        {
                \If{$P_{a1}((u,v),(x,y)) =1$\label{alg:DPF:if-2-start}}
                {
                    $P_1((u,v),(x,y)) = 1$ \label{alg:DPF:P_1-update-1b}
                \label{alg:DPF:if-2-end}}
            \label{alg:DPF:else-1-end} 
}\label{alg:DPF:foreach-2-a-end}

        \ForEach{$((u,v),w) \in E_O\times V_O$ such that $v\neq w$\label{alg:DPF:foreach-2-b-start}}
        {
                \If{$P_{a2}((u,v),w) =1$\label{alg:DPF:if-4-start}}
                {
                    $P_2((u,v),w) = 1$ \label{alg:DPF:P_2-update-1b}
                \label{alg:DPF:if-4-end}}
            \label{alg:DPF:else-2-end}
\label{alg:DPF:foreach-2-b-end}}
    }\label{alg:DPF:foreach-2-end}

    \While{\label{alg:DPF:while-1-start} $\exists ((u,v),(x,y)) \in E_O\times E_O$ such that $(u,v)\neq (x,y)$, $P_1((u,v),(x,y)) =0$, and \\there exists $(z_1,z_2) \in E_O$ such that $P_1((u,v), (z_1, z_2)) = P_1((z_1, z_2), (x,y)) =1$ }
    {
        $P_1((u,v),(x,y)) = 1$ \label{alg:DPF:P_1-update-2}
    }\label{alg:DPF:while-1-end}

    \While{\label{alg:DPF:while-2-start} $\exists ((u,v),w) \in E_O\times V_O$ such that $v\neq w$, $P_1((u,v),w) =0$, and \\there exists $(z_1,z_2) \in E_O$ such that $P_1((u,v), (z_1, z_2)) = P_2((z_1, z_2), w) =1$ }
    {
        $P_2((u,v),w) = 1$ \label{alg:DPF:P_2-update-2}
    }\label{alg:DPF:while-2-end}

  \KwRet $(P_1, P_2)$ \label{alg:DPF:return}
\end{algorithm}
 We construct \Cref{alg:constructDPF} to compute $\EPF{O, P_{11}, P_{12}, P_{21}, P_{22}}$. 
Line~\ref{alg:DPF:P_1-P_2-init} performs the initialization part (Step 1 of \cref{def:extended-path-function}).
Lines~\ref{alg:DPF:foreach-2-start}-\ref{alg:DPF:foreach-2-end} do the first update of $P_1$ and $P_2$ (Step 2 of \cref{def:extended-path-function}). 
Lines~\ref{alg:DPF:while-1-start}-\ref{alg:DPF:while-2-end} do the final update of $P_1$ and $P_2$ (Step 3 of \cref{def:extended-path-function}). 

The following lemma shows that if $(O, P_1, P_2)$ is the shadow of an MEC, and $(O_1, P_{11}, P_{12})$ and $(O_2, P_{21}, P_{22})$ are the shadows of two projections of the MEC then knowing the value of $O, P_{11}, P_{12}, P_{21}$ and $P_{22}$, we can compute $(P_1, P_2)$ using \cref{alg:constructDPF}, as $(P_1, P_2) = \EPF{O, P_{11}, P_{12}, P_{21}, P_{22}}$.

\begin{lemma}
\label{lem:P1-P2-is-DPF}
    Let $H$ be an undirected graph, and $H_1$ and $H_2$ be two induced subgraphs of $H$ such that $H = H_1\cup H_2$, and $I = V_{H_1} \cap V_{H_2}$ is a vertex separator of $H$ that separates $V_{H_1}\setminus I$ and $V_{H_2}\setminus I$. 
Let $S_1, S_2$ be the  subsets of $V_{H_1}$ and $V_{H_2}$, respectively, such that  $S_1\cap S_2 = I$.
Let $Y= S_1\cup S_2\cup N(S_1\cup S_2, H)$, $Y_1 = S_1\cup N(S_1, H_1)$, and $Y_2 = S_2\cup N(S_2, H_2)$.
Let $M, M_1, M_2$ be the MECs of $H, H_1, H_2$, respectively, $(O, P_1,P_2)$ be the \shadow{} of $M$ on $Y$, $(O_1, P_{11}, P_{12})$ be the \shadow{} of $M_1$ on $Y_1$, and $(O_2, P_{21}, P_{22})$ be the \shadow{} of $M_2$ on $Y_2$. If $\mathcal{P}(M, V_{H_1}, V_{H_2}) = (M_1, M_2)$ then $(P_1, P_2) = \EPF{O, P_{11}, P_{12}, P_{21}, P_{22}}$, i.e., $(P_1, P_2)$ is the \epfs{} of $(O, P_{11}, P_{12}, P_{21}, P_{22})$.
\end{lemma}
\begin{proof}
Let $\EPF{O, P_{11}, P_{12}, P_{21}, P_{22}} = (P_1', P_2')$. We show that $(P_1, P_2) = (P_1', P_2')$.  We first show that $P_1 = P_1'$. Since $(O, P_1, P_2)$ is a shadow of $M$, therefore, from \cref{def:projection}, the domain of $P_1$ is $E_O \times E_O$ and the co-domain of $P_1$ is $\{0,1\}$. Also, from line \ref{alg:DPF:P_1-P_2-init} of \cref{alg:constructDPF}, since $(P_1', P_2') = \text{TFP}(O)$, therefore, from \cref{alg:tfp}, the domain of $P_1'$ is $E_O \times E_O$ and the co-domain of $P_1'$ is $\{0,1\}$. This implies both $P_1$ and $P_1'$ have the same domain and co-domain. Thus, to prove $P_1 = P_1'$, it is sufficient to show that (a) $\forall ((u,v), (x,y)) \in E_O\times E_O$, if $P_1'((u,v), (x,y)) =1$ then $P_1((u,v), (x,y)) =1$, and (b) $\forall ((u,v), (x,y)) \in E_O\times E_O$, if $P_1((u,v), (x,y)) =1$ then $P_1'((u,v), (x,y)) =1$.

\begin{claim}
    \label{claim-1-to-prove-P1'-equals-P1}
    For all $((u,v), (x,y)) \in E_O\times E_O$, if $P_1'((u,v), (x,y)) = 1$, then $P_1((u,v), (x,y)) = 1$.
\end{claim}

\begin{proof}
From \cref{def:extended-path-function}, $P_1'$ reached its final version after following steps \ref{item-1-of-def:extended-path-function}, \ref{item-2-of-def:extended-path-function}, and \ref{item-3-of-def:extended-path-function} of \cref{def:extended-path-function}. 

We first show that at the end of step \ref{item-1-of-def:extended-path-function}, for any $(u,v),(x,y)\in E_O$, if $P_1'((u,v), (x,y)) = 1$, then $P_1((u,v), (x,y)) = 1$. From line \ref{alg:DPF:P_1-P_2-init} of \cref{alg:constructDPF}, $(P_1', P_2') = \text{TFP}(O)$. Therefore, from \cref{lem:alg:tfp-is-valid}, if $P_1'((u,v),(x,y)) = 1$, then there exists a \tfp{} from $(u,v)$ to $(x,y)$ in $O$. From \cref{item-1-of-def:shadow} of \cref{def:shadow}, $O$ is an induced subgraph of $M$. This implies there exists a \tfp{} from $(u,v)$ to $(x,y)$ in $M$. Then, from \cref{item-2-of-def:shadow} of \cref{def:shadow}, $P_1((u,v),(x,y)) = 1$. This implies that at the end of the initialization step (i.e., step \ref{item-1-of-def:extended-path-function}) of \cref{def:extended-path-function}, if $P_1'((u,v), (x,y)) = 1$, then $P_1((u,v), (x,y)) = 1$.

We now show that at step \ref{item-2-of-def:extended-path-function} (update 1) of \cref{def:extended-path-function}, for any $((u,v), (x,y)) \in E_O\times E_O$, if $P_1'$ is updated with $P_1'((u,v), (x,y)) = 1$, then $P_1((u,v), (x,y)) = 1$.
Suppose for some $((u,v), (x,y)) \in E_O\times E_O$, $P_1'$ is updated with $P_1'((u,v), (x,y)) = 1$ at step \ref{item-2-of-def:extended-path-function} of \cref{def:extended-path-function}. From \cref{def:extended-path-function}, this implies that $(u,v)\neq (x,y)$ and for some $a\in \{1,2\}$, $P_{a1}((u,v), (x,y)) = 1$.
Since $P_1$ obeys \cref{item-5-of-obs1:O-structure-for-existence-of-MEC} of \cref{obs1:O-structure-for-existence-of-MEC}, from \cref{subitem-2-of-item-5-of-obs1:O-structure-for-existence-of-MEC} of \cref{item-5-of-obs1:O-structure-for-existence-of-MEC} of \cref{obs1:O-structure-for-existence-of-MEC}, $P_1((u,v), (x,y)) = 1$. 
This implies that if for some $((u,v), (x,y)) \in E_O\times E_O$, $P_1'$ is updated with $P_1'((u,v), (x,y)) = 1$ at step \ref{item-2-of-def:extended-path-function} of \cref{def:extended-path-function} , then $P_1((u,v), (x,y)) = 1$. 
This further implies that at the start of step 3, for any $((u,v), (x,y)) \in E_O\times E_O$, if $P_1'((u,v), (x,y)) = 1$, then $P_1((u,v), (x,y)) = 1$.

Suppose while running step 3, if at some iteration, for some $((u,v), (x,y)) \in E_O\times E_O$, we update $P_1'$ with $P_1'((u,v), (x,y)) = 1$. Then $P_1((u,v), (x,y))=1$ provided that before that iteration, for all $((u',v'), (x',y')) \in E_O\times E_O$, if $P_1'((u',v'),(x',y')) = 1$, then $P_1((u',v'), (x',y')) = 1$.
From \cref{def:extended-path-function}, if at step 3 of \cref{def:extended-path-function}, for any $((u,v), (x,y)) \in E_O\times E_O$, $P_1'$ is updated with $P_1'((u,v),(x,y)) = 1$, then $(u,v)\neq (x,y)$ and there exists a $(z_1, z_2) \in E_O$ such that $P_1'((u,v), (z_1, z_2)) = P_1'((z_1, z_2), (x,y)) = 1$. Since at the start of the iteration, $P_1'((u,v), (z_1, z_2)) = P_1'((z_1, z_2), (x,y)) = 1$, therefore, from our assumption, $P_1((u,v), (z_1, z_2)) = P_1((z_1, z_2), (x,y)) = 1$. Since $P_1$ follows \cref{item-5-of-obs1:O-structure-for-existence-of-MEC} of \cref{obs1:O-structure-for-existence-of-MEC}, from \cref{subitem-3-of-item-5-of-obs1:O-structure-for-existence-of-MEC} of \cref{item-5-of-obs1:O-structure-for-existence-of-MEC} of \cref{obs1:O-structure-for-existence-of-MEC}, $P_1((u,v), (x,y)) = 1$. 
This proves that even at the end of step 3, for all $((u,v), (x,y)) \in E_O\times E_O$, if $P_1'((u,v),(x,y)) = 1$, then $P_1((u,v), (x,y)) = 1$.
This proves \cref{claim-1-to-prove-P1'-equals-P1}. 
\end{proof}

\begin{claim}
    \label{claim-21-to-prove-P1'-equals-P1}
    For all $((u,v), (x,y)) \in E_O \times E_O$, if $P_1((u,v), (x,y)) = 1$, then $P_1'((u,v), (x,y)) = 1$.
\end{claim}
\begin{proof}
    Let us construct a directed acyclic graph (DAG) $D$ such that $V_D = E_O$, and for $(u,v), (x,y) \in E_O$, $(u,v) \rightarrow (x,y) \in E_D$ if, and only if, $P_1((u,v), (x,y)) = 1$. We first show that $D$ is a DAG.

    Since $P_1$ obeys \cref{item-5-of-obs1:O-structure-for-existence-of-MEC} of \cref{obs1:O-structure-for-existence-of-MEC}, if $P_1((u,v),(x,y)) = 1$, then $(u,v) \neq (x,y)$. This implies that $D$ does not contain any self-loop. Suppose $D$ is not a DAG. Then, pick the smallest length cycle $((u_1,v_1), (u_2, v_2), \ldots, (u_n, v_n), (u_1, v_1))$ in $D$ for some $n \geq 2$. This implies $P_1((u_1,v_1), (u_2, v_2)) = P_1((u_2,v_2), (u_3, v_3)) = \ldots = P_1((u_{n-1},v_{n-1}), (u_n, v_n)) = P_1((u_n,v_n), (u_1, v_1)) = 1$. Since $P_1$ obeys \cref{item-5-of-obs1:O-structure-for-existence-of-MEC} of \cref{obs1:O-structure-for-existence-of-MEC}, therefore, from \cref{subitem-3-of-item-5-of-obs1:O-structure-for-existence-of-MEC} of \cref{item-5-of-obs1:O-structure-for-existence-of-MEC} of \cref{obs1:O-structure-for-existence-of-MEC}, $P_1((u_1,v_1), (u_2, v_2)) = P_1((u_2,v_2), (u_3, v_3)) = \ldots = P_1((u_{n-1},v_{n-1}), (u_n, v_n)) = 1$ implies $P_1((u_1, v_1), (u_n, v_n)) = 1$. This further implies that in $M$, there exist \tfps{} from $(u_1,v_1)$ to $(u_n, v_n)$, and from $(u_n, v_n)$ to $(u_1, v_1)$. But, this contradicts \cref{lem:nes-tfp-cond-for-an-MEC}. This implies that $D$ does not have any cycles.

    We now show that if $(u,v) \rightarrow (x,y) \in E_D$, then $P_1'((u,v), (x,y)) = 1$, which completes the proof of \cref{claim-21-to-prove-P1'-equals-P1} (remember, from the construction of $D$, $(u,v) \rightarrow (x,y) \in E_D$ if $P_1((u,v), (x,y)) = 1$). Let $\tau$ be an ordering on the vertices of $D$ such that if $(u,v) \rightarrow (x,y) \in E_D$, then $\tau((u,v)) < \tau((x,y))$ (the existence of $\tau$ is because $D$ is a DAG).

    Suppose there exists $(u,v) \rightarrow (x,y) \in E_D$ such that $P_1'((u,v), (x,y)) = 0$. Pick such a $(u,v) \rightarrow (x,y) \in E_D$ such that for all $(u',v') \rightarrow (x',y') \in E_D$, 
    \begin{equation}
    \label{claim:higher-ordering-obeys-P1}
    \text{if $\tau(x',y') < \tau(x,y)$, or $(x',y') = (x,y)$ and $\tau(u,v) < \tau(u',v')$, then $P_1'((u',v'), (x',y')) = 1$.}
    \end{equation}

    From the construction of $D$, since $(u,v) \rightarrow (x,y) \in E_D$, therefore, $P_1((u,v), (x,y)) = 1$. Since $P_1$ obeys \cref{item-5-of-obs1:O-structure-for-existence-of-MEC} of \cref{obs1:O-structure-for-existence-of-MEC}, therefore, either (a) there exists a \tfp{} from $(u,v)$ to $(x,y)$ in $O$, or (b) for some $a\in \{1,2\}$, $P_{a1}((u,v),(x,y)) = 1$, or (c) there exists a $(z_1, z_2) \in E_O$ such that $P_1((u,v), (z_1, z_2)) = P_1((z_1, z_2), (x,y)) = 1$. If $P_1((u,v), (x,y)) = 1$ due to (a) or (b), then from steps \ref{item-1-of-def:extended-path-function} and \ref{item-2-of-def:extended-path-function} of \cref{def:extended-path-function}, $P_1'((u,v), (x,y)) = 1$. If $P_1((u,v), (x,y)) = 1$ due to (c), then there exists a $(z_1, z_2) \in E_O$ such that $P_1((u,v), (z_1, z_2)) = P_1((z_1, z_2), (x,y)) = 1$. This implies $\tau((u,v)) < \tau((z_1, z_2)) < \tau((x,y))$. Then from \cref{claim:higher-ordering-obeys-P1}, $P_1'((u,v), (z_1, z_2)) = 1$ (because $\tau((z_1, z_2))< \tau((x,y))$), and $P_1'((z_1, z_2), (x,y)) = 1$ (because $\tau((u,v))< \tau((z_1, z_2))$). But, then $P_1'((u,v), (x,y)) = 1$, updated at step 3 of \cref{def:extended-path-function}. This implies that for all $(u,v), (x,y) \in E_O$, if $P_1((u,v), (x,y)) = 1$, then $P_1'((u,v), (x,y)) = 1$. This completes the proof of \cref{claim-21-to-prove-P1'-equals-P1}.
\end{proof}

\Cref{claim-21-to-prove-P1'-equals-P1,claim-1-to-prove-P1'-equals-P1} proves $P_1 = P_1'$.

We now prove $P_2 = P_2'$. Since $(O, P_1, P_2)$ is a shadow of $M$, therefore, according to \cref{def:projection}, the domain of $P_2$ is $E_O \times V_O$, and the co-domain of $P_1$ is $\{0,1\}$. Also, from line \ref{alg:DPF:P_1-P_2-init} of \cref{alg:constructDPF}, since $(P_1', P_2') = \text{TFP}(O)$, therefore, based on \cref{alg:tfp}, the domain of $P_2'$ is $E_O \times V_O$, and the co-domain of $P_1'$ is $\{0,1\}$. This implies that both $P_2$ and $P_2'$ have the same domain and co-domain. Thus, to prove $P_1 = P_1'$, it is sufficient to show that (a) for all $((u,v),w) \in E_O \times V_O$, if $P_2((u,v), w) = 1$, then $P_2'((u,v), w) = 1$, and (b) for all $((u,v),w) \in E_O \times V_O$, if $P_2'((u,v), w) = 1$, then $P_2((u,v), w) = 1$. The proof of (a) is similar to \cref{claim-1-to-prove-P1'-equals-P1}, and the proof of (b) is similar to \cref{claim-21-to-prove-P1'-equals-P1}. For completeness, we provide the proof.

\begin{claim}
    \label{claim-1-to-prove-P2'-equals-P2}
    $\forall ((u,v), w) \in E_O\times V_O$, if $P_2'((u,v), w) = 1$ then $P_2((u,v), w) = 1$. 
\end{claim}

\begin{proof}
From \cref{def:extended-path-function}, $P_2'$ reached its final version after following steps \ref{item-1-of-def:extended-path-function}, \ref{item-2-of-def:extended-path-function}, and \ref{item-3-of-def:extended-path-function} of \cref{def:extended-path-function}. 

Firstly, we demonstrate that at the end of step \ref{item-1-of-def:extended-path-function}, for any $((u,v),w) \in E_O \times V_O$, if $P_2'((u,v), w) = 1$, then $P_2((u,v), w) = 1$. From line \ref{alg:DPF:P_1-P_2-init} of \cref{alg:constructDPF}, we have $(P_1', P_2') = \text{TFP}(O)$. Therefore, from \cref{lem:alg:tfp-is-valid}, if $P_2'((u,v),w) = 1$, then $v\neq w$, and there exists a \tfp{} from $(u,v)$ to $w$ in $O$. Since $P_2$ satisfies \cref{item-6-of-obs1:O-structure-for-existence-of-MEC} in \cref{obs1:O-structure-for-existence-of-MEC}, from \cref{subitem-1-of-item-6-of-obs1:O-structure-for-existence-of-MEC}, we have $P_2((u,v), w) = 1$. This implies that at the end of the initialization step (i.e., step \ref{item-1-of-def:extended-path-function}) of \cref{def:extended-path-function}, if $P_2'((u,v), w) = 1$, then $P_2((u,v), w) = 1$.

Next, we demonstrate that if at step \cref{item-2-of-def:extended-path-function} (update 1) of \cref{def:extended-path-function}, for any $((u,v), w) \in E_O\times V_O$, $P_2'$ is updated with $P_2'((u,v), w) = 1$, then $P_2((u,v), w) = 1$.

Suppose for some $((u,v), w) \in E_O\times V_O$, $P_2'$ is updated with $P_2'((u,v), w) = 1$ at step \ref{item-2-of-def:extended-path-function} of \cref{def:extended-path-function}. From \cref{def:extended-path-function}, this implies that $v\neq w$ and for some $a\in \{1,2\}$, $P_{a2}((u,v), w) = 1$. Since $P_2$ satisfies  \cref{obs1:O-structure-for-existence-of-MEC}, from \cref{subitem-2-of-item-6-of-obs1:O-structure-for-existence-of-MEC} of \cref{item-6-of-obs1:O-structure-for-existence-of-MEC} of \cref{obs1:O-structure-for-existence-of-MEC}, we have $P_2((u,v), w) = 1$. This implies that if for some $((u,v), w) \in E_O\times V_O$, $P_2'$ is updated with $P_2'((u,v), w) = 1$ at step \ref{item-2-of-def:extended-path-function}, then $P_2((u,v), w) = 1$. This further implies that at the start of step 3, for any $((u,v), w) \in E_O\times V_O$, if $P_2'((u,v), w) = 1$, then $P_2((u,v), w) = 1$.

Suppose while running step 3, if at some iteration, for some $((u,v), w) \in E_O\times V_O$, we update $P_2'$ with $P_2'((u,v), w) = 1$. Then $P_2((u,v), w) = 1$ provided that before that iteration, for all $((u',v'), w') \in E_O\times V_O$, if $P_2'((u',v'),w') = 1$, then $P_2((u',v'), w') = 1$. From \cref{def:extended-path-function}, if at step 3 of \cref{def:extended-path-function}, for any $((u,v), w) \in E_O\times V_O$, $P_2'$ is updated with $P_2'((u,v),w) = 1$, then $v\neq w$, and there exists a $(z_1, z_2) \in E_O$ such that $P_1'((u,v), (z_1, z_2)) = P_2'((z_1, z_2), w) = 1$. Since we have already shown that $P_1 = P_1'$, therefore, $P_1((u,v),(z_1,z_2)) = 1$. Also, since at the start of the iteration, we have $ P_2'((z_1, z_2), w) = 1$, therefore, from our assumption, $P_2((z_1, z_2), w) = 1$. Since $P_2$ follows  \cref{obs1:O-structure-for-existence-of-MEC}, from \cref{subitem-3-of-item-6-of-obs1:O-structure-for-existence-of-MEC} of \cref{item-6-of-obs1:O-structure-for-existence-of-MEC} of \cref{obs1:O-structure-for-existence-of-MEC}, we have $P_2((u,v), w) = 1$. This proves that even at the end of step 3, for all $((u,v), w) \in E_O\times V_O$, if $P_2'((u,v),w) = 1$ then $P_2((u,v), w) =1$.
This completes the proof of \cref{claim-1-to-prove-P2'-equals-P2}.
\end{proof}

\begin{claim}
    \label{claim-2-to-prove-P2'-equals-P2}
    $\forall ((u,v), w) \in E_O\times V_O$, if $P_2((u,v), w) =1$ then $P_2'((u,v), w) =1$.
\end{claim}
\begin{proof}
    We show that for all $((u,v), w) \in E_O\times V_O$, if $P_2((u,v), w) =1$ then $P_2'((u,v), w) =1$. 
    Let us construct a directed acyclic graph (DAG) $D$ such that $V_D = E_O$, and for $(u,v), (x,y) \in E_O$, $(u,v)\rightarrow (x,y) \in E_D$ if, and only if, $P_1((u,v), (x,y)) =1$. We have shown in the proof of \cref{claim-21-to-prove-P1'-equals-P1} that $D$ is a DAG (read the first two paragraphs of \cref{claim-21-to-prove-P1'-equals-P1} for the proof that $D$ is a DAG). Let $\tau$ be an ordering on the vertices of $D$ such that if $(u,v)\rightarrow (x,y) \in E_D$ then $\tau((u,v)) < \tau((x,y))$ (existence of $\tau$ is because $D$ is a DAG). If possible, suppose there exists a $((u,v), w) \in E_O\times V_O$ such that $P_2((u,v), w) =1$ and $P_2'((u,v), w) =0$. For $w$, pick such a $(u,v)$ which is highest in $\tau$, i.e.,  
    \begin{equation}
    \label{claim:higher-ordering-obeys-P2}
    \text{$\forall (u', v') \in E_O$ such that $\tau((u, v))< \tau((u',v'))$, if $P_2((u',v'), w) =1$ then $P_2'((u',v'), w) =1$.}
    \end{equation}
    Since $P_2$ obeys \cref{item-6-of-obs1:O-structure-for-existence-of-MEC} of \cref{obs1:O-structure-for-existence-of-MEC}, therefore, either (a) there exists a \tfp{} from $(u,v)$ to $w$ in $O$, or (b) for any $a\in \{1,2\}$, $P_{a2}((u,v),w) = 1$, and $v\rightarrow u \notin O$, or (c) there exists a $(z_1, z_2) \in E_O$ such that $P_1((u,v), (z_1, z_2)) = P_2((z_1, z_2), w) =1$. 
    If $P_2((u,v), w) =1$ due to (a) or (b) then from \cref{item-1-of-def:extended-path-function,item-2-of-def:extended-path-function} of \cref{def:extended-path-function}, $P_2'((u,v), w) = 1$. And, if $P_2((u,v), w) =1$ due to (c) then there exists a $(z_1, z_2) \in E_O$ such that $P_1((u,v), (z_1, z_2)) = P_2((z_1, z_2), w) =1$.
    Since, we have shown that $P_1 =P_1'$, therefore, $P_1'((u,v), (z_1, z_2)) = P_1((u,v), (z_1, z_2)) = 1$. 
    Also, since $P_1((u,v), (z_1, z_2)) =1$, therefore, $\tau((u,v))< \tau((z_1, z_2))$. 
    Then from \cref{claim:higher-ordering-obeys-P2}, $P_2'((z_1, z_2), w) =1$.  But, then $P_2'((u,v), w) =1$, updated at step 3 of  \cref{def:extended-path-function}. This implies that for all $((u,v), w) \in E_O\times V_O$, if $P_2((u,v), w) =1$ then $P_2'((u,v), w) =1$. This completes the proof of \cref{claim-2-to-prove-P2'-equals-P2}.
\end{proof}

\Cref{claim-2-to-prove-P2'-equals-P2,claim-1-to-prove-P2'-equals-P2} proves $P_2 = P_2'$. This implies $(P_1, P_2) = (P_1', P_2') = \EPF{O, P_{11}, P_{12}, P_{21}, P_{22}}$. This completes the proof of \cref{lem:P1-P2-is-DPF}.
\end{proof}

We define the following, to summarize the necessary conditions provided by \cref{obs1:O-structure-for-existence-of-MEC}.

\begin{definition}[Valid \epfs{}]
    \label{def:valid-epfs}
$(P_1, P_2)$ is said to be a ``valid \epfs{}'' of $(O, P_{11}, P_{12}, P_{21}, P_{22})$ if:
     \begin{enumerate}
         \item
         \label{item-1-of-def:valid-epfs}
$(P_1, P_2)= \EPF{O, P_{11}, P_{12}, P_{21}, P_{22}}$, and
         \item
         \label{item-2-of-def:valid-epfs}
         For all $(u,v),(x,y) \in E_O$, if $P_1((u,v),(x,y)) =1$ then $P_1((x,y), (u,v)) =0$.
     \end{enumerate}
     
\end{definition}

\begin{definition}[Extension]
 \label{def:extension-of-O1-O2-P11-P12-P21-P22}
 Let $H$ be an undirected graph, and $H_1$ and $H_2$ be two induced subgraphs of
 $H$ such that $H = H_1 \cup H_2$, and $I = V_{H_1} \cap V_{H_2}$ is a vertex
 separator of $H$ that separates $V_{H_1}\setminus{I}$ and
 $V_{H_2}\setminus{I}$. Let $S_1$ and $S_2$ be the subsets of $V_{H_1}$ and
 $V_{H_2}$, respectively, such that $S_1 \cap S_2 = I$. Let
 $A = H[S_1 \cup S_2 \cup N(S_1 \cup S_2, H)]$,
 $B_1 = H_1[S_1 \cup N(S_1, H_1)]$, and $B_2 = H_2[S_2 \cup N(S_2, H_2)]$. Let
 $O_1\in \setofpartialMECs{B_1}$ and $O_2\in \setofpartialMECs{B_2}$ be partial
 MECs.
For each $a \in \{1,2\}$, let
 $P_{a1}: E_{O_a} \times E_{O_a} \rightarrow \{0,1\}$, and
 $P_{a2}: E_{O_a} \times V_{O_a} \rightarrow \{0,1\}$ be two functions.  We
 define $O \in PMEC(A)$ to be an \emph{extension} of $(O_1, P_{11}, P_{12})$ and
 $(O_2, P_{21}, P_{22})$, denoted as
 $O \in \mathcal{E}(O_1, P_{11}, P_{12}, O_2, P_{21}, P_{22})$, if
  \begin{enumerate}
      \item
      \label{item-1-of-def:extension-of-O1-O2-P11-P12-P21-P22}
      For $a\in \{1,2\}$, for $u, v \in S_a\cup N(S_a, H_a)$, if $u\rightarrow v \in O_a$ then $u\rightarrow v \in O$.
      \item
      \label{item-2-of-def:extension-of-O1-O2-P11-P12-P21-P22}
      For $a\in \{1,2\}$, $\mathcal{V}(O_a) = \mathcal{V}(O, S_a\cup N(S_a, H_a))$.
      \item
      \label{item-3-of-def:extension-of-O1-O2-P11-P12-P21-P22}
      For $a\in \{1,2\}$, for $u-v \in O_a$, $u\rightarrow v \in O$ if, and only
      if, at least one of the following occurs:
      \begin{enumerate}
          \item
          \label{subitem-1-of-item-3-of-def:extension-of-O1-O2-P11-P12-P21-P22}
          $u\rightarrow v$ is strongly protected in $O$.
          \item
          \label{subitem-2-of-item-3-of-def:extension-of-O1-O2-P11-P12-P21-P22}
          $\exists x-y \in O_a$ such that  $x\rightarrow y \in O$, and $P_{a1}((x,y), (u, v)) = 1$.
          \item
          \label{subitem-3-of-item-3-of-def:extension-of-O1-O2-P11-P12-P21-P22}
          $\exists x-y \in O_a$ such that $x\rightarrow y \in O$,  $P_{a2}(x, y, v) =1$, and $P_{a2}(v, u, x) =1$.
      \end{enumerate}

    \item
    \label{item-4-of-def:extension-of-O1-O2-P11-P12-P21-P22}
    $\EPF{O, P_{11}, P_{12}, P_{21}, P_{22}}$ is valid (see \cref{def:valid-epfs} for valid \epfs{}). 
\end{enumerate}
With a slight abuse of notation, we  say $(O, P_1, P_2)$ as an extension of $(O_1, P_{11}, P_{12}, O_2, P_{21}, P_{22})$, if
\begin{enumerate}
    \item $O$ obeys \cref{item-1-of-def:extension-of-O1-O2-P11-P12-P21-P22,item-2-of-def:extension-of-O1-O2-P11-P12-P21-P22,item-3-of-def:extension-of-O1-O2-P11-P12-P21-P22,item-4-of-def:extension-of-O1-O2-P11-P12-P21-P22} of \cref{def:extension-of-O1-O2-P11-P12-P21-P22}, and
    \item  $(P_1, P_2) = \EPF{O, P_{11}, P_{12}, P_{21}, P_{22}}$.
\end{enumerate}
 \end{definition}

Since $P_1$ in \cref{obs1:O-structure-for-existence-of-MEC} obeys \cref{item-2-of-def:valid-epfs} of \cref{def:valid-epfs} (from \cref{lem:nes-tfp-cond-for-an-MEC}), therefore,
 \cref{lem:P1-P2-is-DPF,def:extension-of-O1-O2-P11-P12-P21-P22,def:extended-path-function,def:valid-epfs} summarises \cref{obs1:O-structure-for-existence-of-MEC} as follows:
 \begin{lemma}
\label{obs:nes-conditions-of-the-shadow}
     Let $H$ be an undirected graph, and $H_1$ and $H_2$ be two induced subgraphs of $H$ such that $H = H_1\cup H_2$, and $I = V_{H_1} \cap V_{H_2}$ is a vertex separator of $H$ that separates $V_{H_1}\setminus I$ and $V_{H_2}\setminus I$. 
Let $S_1, S_2$ be the  subsets of $V_{H_1}$ and $V_{H_2}$, respectively, such that  $S_1\cap S_2 = I$.
Let $Y= S_1\cup S_2\cup N(S_1\cup S_2, H)$, $Y_1 = S_1\cup N(S_1, H_1)$, and $Y_2 = S_2\cup N(S_2, H_2)$.
Let $M, M_1, M_2$ be MECs of $H, H_1, H_2$, respectively. Let $(O, P_1,P_2)$ be the \shadow{} of $M$ on $Y$, $(O_1, P_{11}, P_{12})$ be the \shadow{} of $M_1$ on $Y_1$, and $(O_2, P_{21}, P_{22})$ be the \shadow{} of $M_2$ on $Y_2$. 
If $\mathcal{P}(M, V_{H_1}, V_{H_2}) = (M_1, M_2)$ then
\begin{enumerate}
    \item
    \label{item-1-of-obs:nes-conditions-of-the-shadow}
    $O \in \mathcal{E}(O_1, P_{11}, P_{12}, O_2, P_{21}, P_{22})$, and
    \item 
    \label{item-2-of-obs:nes-conditions-of-the-shadow}
    $(P_1, P_2) = \EPF{O, P_{11}, P_{12}, O_2, P_{21}, P_{22}}$.
\end{enumerate}
 \end{lemma}

 We show in \cref{subsection:sufficient-condition} that the necessary conditions in \cref{obs:nes-conditions-of-the-shadow} are also sufficient. We then use this result along with \cref{obs:nes-conditions-of-the-shadow} to construct an algorithm that counts MECs of $H$. To prove this, we need a modified version of the  LBFS algorithm of \citet{rose1976algorithmic}. In the next subsection,  we introduce a modified version of the LBFS algorithm. We will also see some properties of the new LBFS algorithm. All these results will be used in \cref{subsection:sufficient-condition} to show that the necessary conditions provided in \cref{obs:nes-conditions-of-the-shadow} are sufficient.

  \subsection{Lexicographical breadth-first search (LBFS) algorithm}
\label{subsection:LBFS}
In this section, we present a modified version of the LBFS algorithm. Let $H$ be an undirected graph, and let $H_1$ and $H_2$ be two induced subgraphs of $H$ such that $H = H_1 \cup H_2$, and $I = V_{H_1} \cap V_{H_2}$ is a vertex separator of $H$. For each $a \in \{1,2\}$, let $S_a$ be a subset of $V_{H_a}$ such that $I \subseteq S_a$, and let $M_a$ be an MEC of $H_a$. The shadow of $M_a$ on $S_a \cup N(S_a, H_a)$ is denoted as $(O_a, P_{a1}, P_{a2})$. Consider $O \in \mathcal{E}(O_1, P_{11}, P_{12}, O_2, P_{21}, P_{22})$ (i.e., $O$ is an extension of $(O_1, P_{11}, P_{12}, O_2, P_{21}, P_{22})$; refer to \Cref{def:extension-of-O1-O2-P11-P12-P21-P22} for the definition of extension).

We introduce a modified LBFS algorithm, \Cref{alg:LBFSwithO}. This algorithm takes as input (a) $O$ and (b) $\mathcal{C}$,  an \ucc{} of $M_a$ for $a\in \{1,2\}$. \Cref{alg:LBFSwithO} returns an LBFS ordering of $\mathcal{C}$ that obeys $O$, i.e., if $\tau$ is the returned LBFS ordering of $\mathcal{C}$, then for any $u, v \in V_{\mathcal{C}}$ such that if $u \rightarrow v \in O$ then $\tau(u) < \tau(v)$.

\begin{algorithm}
  \caption{LBFS($\mathcal{C},O$): A modified version of the Lexicographic
BFS (\cite{rose1976algorithmic}). If the algorithm is executed with $O$ as a null graph the algorithm performs a normal LBFS.}
\SetAlgoLined
\SetKwInOut{KwIn}{Input}
\SetKwInOut{KwOut}{Output}
\SetKwFunction{AMO-Union}{AMO-Union}
\KwIn{(a) An undirected connected chordal graph $\mathcal{C}$, and (b) a chain graph $O$ such that\\
there exist three undirected graphs $H, H_1$ and $H_2$ such that\\
(i) $V_{H_1}\cap V_{H_2} = I$ is a vertex separator of $H$,\\
(ii) there exist $S_1 \in V_{H_1}$ and $S_2 \in V_{H_2}$ such that $I = V_{S_1}\cap V_{S_2}$,\\
(iii) $M_1$ and $M_2$ are MECs of $H_1$ and $H_2$, respectively,\\
(iv) for each $a\in \{1,2\}$, $(O_a, P_{a1}, P_{a2})$ is the shadow of $M_a$ on $S_a\cup N(S_a, H_a)$,\\
(v) $O\in \mathcal{E}(O_1,P_{11}, P_{12}, O_2, P_{21}, P_{22})$, and \\
(vi) $\mathcal{C}$ is an \ucc{} of $M_1$ or $M_2$. 
}
    \KwOut{an LBFS ordering $\tau$ of $\mathcal{C}$ that is consistent with $O$
     }

    $\mathcal{S}\leftarrow$ sequence of sets initialized with $(V_{\mathcal{C}})$\label{alg:S-init} \label{algLBFS:initialization-of-S}
    
    $\tau \leftarrow$ empty list \label{algLBFS:initialization-of-tau}
    
    \While{$\mathcal{S}$ is non-empty\label{algLBFS:while-loop-start}}{
    $X\leftarrow$ first non-empty set of $\mathcal{S}$\label{algLBFS:set-X}
    
    $v\leftarrow$ arbitrary canonical source node of  $U_M(\mathcal{C},O)[X]$  \label{algLBFS:pick-vertex-v}
    
    Add vertex $v$ to the end of $\tau$.\label{algLBFS:add-vertex-to-tau}

    Replace the set $X$ in the sequence $\mathcal{S}$ by the set $X \setminus \{v\}$. \label{algLBFS:remove-picked-node}
    
    $N(v)\leftarrow \{ x|x\notin \tau \textup{ and } v-x \in E_{\mathcal{C}}\}$\label{algLBFS:neighbour-v}
    
    Denote the current $\mathcal{S}$ by $(S_1,\ldots ,S_k)$ \label{algLBFS:expansion-of-S}
    
    Replace each $S_i$ by $S_i \cap N(v, \mathcal{C}), S_i \setminus  N(v, \mathcal{C})$.\label{algLBFS:refine} \label{algLBFS:replacement-of-S}
    
    Remove all empty sets from $\mathcal{S}$. \label{algLBFS:removal-of-empty-sets}
  }\label{algLBFS:while-end}

    \KwRet $\tau$ \label{algLBFS:return}
    \label{alg:LBFSwithO}
\end{algorithm}

\Cref{alg:LBFSwithO} is a variation of the LBFS algorithm introduced by \cite{rose1976algorithmic}. Instead of selecting a vertex arbitrarily, \Cref{alg:LBFSwithO} chooses a canonical source node (\Cref{def:canonical-source-node}) at line~\ref{algLBFS:pick-vertex-v} of \Cref{alg:LBFSwithO}. The crucial point to verify is the existence of a canonical source node in $U_M(\mathcal{C}, O)[X]$, where $U_M(\mathcal{C}, O)$ represents the Markov union of $\mathcal{C}$ and $O$ (see \Cref{def:Markov-union-of-graphs} for the definition of the Markov union), and $X \subseteq V_{\mathcal{C}}$.

\begin{claim}
    \label{claim:LBFS-algorithm-always-find-cannonical-node}
    There always exists a canonical source vertex in $U_M(\mathcal{C}, O)[X]$, where $X\subseteq V_{\mathcal{C}}$.
\end{claim}
\begin{proof}
\newcommand{\ncc}{certificate}
Suppose $U_M(\mathcal{C}, O)[X]$ does not contain any canonical source vertices. \Cref{def:certificate} implies that if  a node $u$ is not a \canonical{} in
$U_M(\mathcal{C}, O)[X]$ then  there exists a node $v$ that certifies that $u$ is not a \canonical{}. And the certification is an  edge $v\rightarrow v' \in U_M(\mathcal{C}, O)[X]$ and a \tfp{} from $(v,v')$ to $u$ in $U_M(\mathcal{C}, O)[X]$).
If $U_M(\mathcal{C}, O)[X]$ does not have contain any canonical source nodes, then, since $X$ has a finite number of nodes,  there exists a sequence $(u_1, u_2, \ldots, u_l, u_{l+1}=u_1)$ such that $u_1, u_2, \ldots, u_l$ are distinct, and for each $1\leq i\leq l$, $u_i$ is a certificate of $u_{i+1}$.  
Now, we will prove the following claim, which is an important component in demonstrating that we will not obtain any such sequence. This, in turn, implies that there always exists a canonical source node in $U_M(\mathcal{C}, O)[X]$.

\begin{claim}
 \label{claim:Q-is-a-tfp}
 In $U_M(\mathcal{C}, O)[X]$, suppose $u$, $v$, and $w$ are three distinct vertices such that $u$ is a certificate of $v$, and $v$ is a certificate of $w$. Then, either $u \rightarrow v \in O$, or $u$ is a certificate of $w$.
\end{claim}

\begin{proof}
 Suppose in $U_M(\mathcal{C}, O)[X]$, $u$ is a certificate of $v$, and $v$ is a certificate of $w$. Then, from \Cref{def:certificate}, since $u$ is a certificate of $v$, there must exist a node $u' \in X$ such that $u \rightarrow u' \in U_M(\mathcal{C}, O)[X]$, and there exists a \tfp{} $Q_1 = (x_1 = u, x_2 = u',\ldots, x_l = v)$ from $(u, u')$ to $v$ in $U_M(\mathcal{C}, O)[X]$. Similarly, since $v$ is a certificate of $w$, there exists a node $v' \in X$ such that $v \rightarrow v' \in U_M(\mathcal{C}, O)[X]$, and there exists a \tfp{} $Q_2 = (y_1 = v, y_2 = v',\ldots, y_m = w)$ from $(v, v')$ to $w$ in $U_M(\mathcal{C}, O)[X]$.

 We first show that for any \tfp{} $P$ in $U_M(\mathcal{C}, O)[X]$, $P$ is a \tfp{} in $\mathcal{C}$. From \Cref{def:Markov-union-of-graphs}, if $u \rightarrow v \in U_M(\mathcal{C}, O)$, then either $u \rightarrow v \in \mathcal{C}$ or $u \rightarrow v \in O$. Since $\mathcal{C}$ is an undirected graph (recall that $\mathcal{C}$ is an \ucc{} of $M_a$), therefore, $u \rightarrow v \in O$. From the construction, for $Y = V_{\mathcal{C}} \cap V_{O}$, $\skel{\mathcal{C}[Y]} = \skel{O[Y]}$. Since $X \subseteq V_{\mathcal{C}}$ ($X\in \mathcal{S}$, and in \Cref{alg:LBFSwithO}, the elements of $\mathcal{S}$ are the subsets of $V_{\mathcal{C}}$), therefore, if $u \rightarrow v \in U_M(\mathcal{C}, O)[X]$, then $u \rightarrow v \in O$ and $u-v \in \mathcal{C}$. This further implies that each \tfp{} in $U_M(\mathcal{C}, O)[X]$ is a \tfp{} in $\mathcal{C}$.

 The above discussion implies that $Q_1$ and $Q_2$ are \tfps{} in $\mathcal{C}$ and $u \rightarrow u', v \rightarrow v' \in O$.

 If $l = 2$, i.e., the length of $Q_1$ is one, then $u' = v$, and $u \rightarrow v \in O$, and we are done.

 Suppose $l > 2$. We show that there exists a \tfp{} from $(u, u')$ to $w$ in $U_M(\mathcal{C}, O)[X]$. As shown above, $Q_1$ and $Q_2$ are the \tfps{} in $\mathcal{C}$. Since $\mathcal{C}$ is an \ucc{} of an MEC $M_a$, therefore, from \cref{item-2-theorem-nec-suf-cond-for-MEC} of \Cref{thm:nes-and-suf-cond-for-chordal-graph-to-be-an-MEC}, $\mathcal{C}$ is chordal. Then from \Cref{obs:tfp-in-M-is-a-cp}, $Q_1$ and $Q_2$ are \cps{} of $\mathcal{C}$. Pick the least $i < l$ such that for some $1 \leq j \leq m$, $x_i - y_j \in \mathcal{C}$. After this, pick the highest $j$ such that $x_i - y_j \in \mathcal{C}$. The existence of such an edge is due to $x_{l-1} - y_1 \in \mathcal{C}$ (in $Q_1$, $x_{l-1} - y_1$ is an edge as $x_l = y_1$).

 If $i \neq 1$, then $Q = (x_1 = u, x_2 = u', \ldots , x_i, y_j, \ldots, y_m = w)$ is a \tfp{} from $(u,  u')$ to $w$ in $\mathcal{C}$. We now show that $Q$ is also a \tfp{} in $U_M(\mathcal{C}, O)[X]$.

 Since $Q_1$ and $Q_2$ are \tfps{} in $U_M(\mathcal{C}, O)[X]$, to show that $Q$ is a \tfp{}, the only thing we have to show is $y_j \rightarrow x_i \notin O$. If both $x_i$ and $y_j$ are not nodes of $V_O$, then there cannot be $y_j \rightarrow x_i \in O$. Suppose $y_j, x_i \in V_O$. Since $Q$ is a \tfp{} in $\mathcal{C}$, the subpath $Q' = (x_1 = u, x_2 = u', \ldots , x_i, y_j)$ of $Q$ is also a \tfp{} in $\mathcal{C}$. Since $\mathcal{C}$ is an \ucc{} of $M_a$, a \tfp{} in $\mathcal{C}$ is a \tfp{} in $M_a$.  This implies $P_{a1}((u, u'), (x_i, y_j)) = 1$. Since $O$ is an extension of $(O_1, P_{11}, P_{12}, O_2, P_{21}, P_{22})$, from \cref{subitem-2-of-item-3-of-def:extension-of-O1-O2-P11-P12-P21-P22} of \cref{item-3-of-def:extension-of-O1-O2-P11-P12-P21-P22} of \Cref{def:extension-of-O1-O2-P11-P12-P21-P22}, $x_i \rightarrow y_j \in O$ (because $P_{a1}((u, u'), (x_i, y_j)) = 1$ and $u\rightarrow u' \in O$). This shows we get the required \tfp{} in $U_M(\mathcal{C}, O)[X]$, $Q$, from $(u, u')$ to $w$.

 Thus, the only possibility that remains is $i = 1$. We show that this case cannot occur. Suppose, if possible, that $i = 1$. Recall that $l > 2$. Pick the least $j$ such that $x_1 - y_j \in \mathcal{C}$. Since $Q_1$ is a \cp{} in $\mathcal{C}$, and $l > 2$, $j$ must be greater than 1, since $y_1 = x_l$. Then $Q = (y_1 = v, y_2 = v', \ldots, y_j, x_1 = u)$ is a \cp{} in $\mathcal{C}$, from $(v, v')$ to $u$. As $\mathcal{C}$ is an \ucc{} of $M_a$, and $Q_1$ and $Q$ are \tfps{} in $\mathcal{C}$ from $(u, u')$ to $v$ and from $(v, v')$ to $u$, respectively, therefore, $P_{a2}((u, u'), v) =  P_{a2}((v, v'), u) = 1$. Since $O$ is an extension of $(O_1, P_{11}, P_{12}, O_2, P_{21}, P_{22})$, therefore, from \cref{subitem-3-of-item-3-of-def:extension-of-O1-O2-P11-P12-P21-P22} of \cref{item-3-of-def:extension-of-O1-O2-P11-P12-P21-P22}, $v' \rightarrow v \in O$ (as $u \rightarrow u' \in O$ and $P_{a2}((u, u'), v) =  P_{a2}((v, v'), u) = 1$). But, this is a contradiction, as from our assumption, $v \rightarrow v' \in O$. This shows that $i =1$ does not occur. This completes the proof of \Cref{claim:Q-is-a-tfp}.
\end{proof}

We now prove the following claim to complete the proof of \Cref{claim:LBFS-algorithm-always-find-cannonical-node}.
 
\begin{claim}
\label{claim:cyclic-sequence-does-not-exist}
    There cannot exist a sequence $(u_1, u_2, \ldots, u_l, u_{l+1}=u_1)$ such that for each $1\leq i\leq l$, $u_i$ is a certificate of $u_{i+1}$. 
\end{claim}

We will prove \Cref{claim:cyclic-sequence-does-not-exist} using induction on $l$. From \Cref{obs:self-certification-not-possible}, we know that self-certification is not possible, implying that $l\geq 2$.

\paragraph{Base Case: $l=2$.} 
Assume, for the sake of contradiction, that \Cref{claim:cyclic-sequence-does-not-exist} is false for $l=2$. In other words, there exist nodes $u_1$ and $u_2$ such that $u_1$ is a certificate of $u_2$, and $u_2$ is a certificate of $u_1$. This implies the existence of edges $u_{1}-v_{1}$ and $u_{2}-v_{2}$ in $\mathcal{C}$ such that $u_{1}\rightarrow v_{1}$ and $u_{2}\rightarrow v_{2}$ in $O$. Additionally, there exist \tfp{}s $Q_1$ from $(u_{1}, v_{1})$ to $u_{2}$ and $Q_2$ from $(u_{2}, v_{2})$ to $u_{1}$ in $\mathcal{C}$. Since $\mathcal{C}$ is an \ucc{} of the MEC $M_a$ (for some $a\in \{1,2\}$), $\mathcal{C}$ is an undirected chordal graph (cf. \cref{item-2-theorem-nec-suf-cond-for-MEC} of \Cref{thm:nes-and-suf-cond-for-chordal-graph-to-be-an-MEC}). Thus, $Q_1$ and $Q_2$ are \tfp{}s $M_a$.

This implies that $v_{1}\neq u_{2}$ and $v_2 \neq u_1$. If, for instance, $v_{1} = u_{2}$, then by applying \Cref{cor-adjacent-node-in-triangle-free-path} to the undirected chordal graph $\mathcal{C}$, the adjacency of $v_{1} = u_{2}$ and $u_{1}$ and the \tfp{} $Q_2$ from $(u_{2}, v_{2})$ to $u_{1}$ would imply $v_{2} = u_{1}$ as well. Similarly, if $v_2 = u_1$, it would imply $v_{1} = u_{2}$. Thus, we must have both $v_{1} = u_{2}$ and $v_2 = u_1$, which leads to the existence of both $u_{1}\rightarrow u_{2}$ and $u_{2}\rightarrow u_{1}$ in $O$, contradicting our initial assumption. Therefore, it follows that $v_{1}\neq u_{2}$ and $v_2 \neq u_1$.

Given that $O$ is an extension of $(O_1, P_{11}, P_{12}, O_2, P_{21}, P_{22})$, we can apply \cref{subitem-3-of-item-3-of-def:extension-of-O1-O2-P11-P12-P21-P22} of \cref{item-3-of-def:extension-of-O1-O2-P11-P12-P21-P22} of \Cref{def:extension-of-O1-O2-P11-P12-P21-P22}. This reveals that $u_{1}\rightarrow v_{1} \in O$ together with the \tfp{}s $Q_1$ from $(u_{1}, v_{1})$ to $u_{2}$ (with $v_1 \neq u_2$) and $Q_2$ from $(u_{2}, v_{2})$ to $u_{1}$ (with $v_2 \neq u_1$) imply $v_{2}\rightarrow u_{2} \in O$. This, however, contradicts the assumption that $u_{2}\rightarrow v_{2}\in O$. Thus, we have shown that \Cref{claim:cyclic-sequence-does-not-exist} holds for $l=2$.

\paragraph{Induction: $l \geq 3$.} 
Assume that \Cref{claim:cyclic-sequence-does-not-exist} is true for $l=k-1$, where $k\geq 3$. We will now use \Cref{claim:Q-is-a-tfp} to prove \Cref{claim:cyclic-sequence-does-not-exist} for $l=k$.

Assume, for the sake of contradiction, that \Cref{claim:cyclic-sequence-does-not-exist} is not true for a sequence $(u_1, u_2, \ldots, u_k, u_{k+1} = u_1)$. According to \Cref{claim:Q-is-a-tfp}, for each $1\leq i \leq l-1$, either $u_i$ is a certificate for $u_{i+2}$ or $u_i\rightarrow u_{i+1} \in O$, and for $i = l$, either $u_l$ is a certificate for $u_2$ or $u_{l}\rightarrow u_{1} \in O$. If, for any $1\leq i \leq k-1$, $u_i$ is a certificate for $u_{i+2}$, it would create a sequence $(u_1, u_2, \ldots, u_i, u_{i+2}, \ldots, u_k, u_{k+1}=u_1)$ of $k-1$ distinct nodes that violates our induction hypothesis. Similarly, if $u_k$ is a certificate for $u_2$, it would create another sequence $(u_2, u_3, \ldots, u_k, u_2)$ of $k-1$ distinct nodes that violates our induction hypothesis. Thus, the only possibility is that for each $1\leq i\leq k$, $u_{i}\rightarrow u_{i+1} \in O$. However, this leads to the presence of a directed cycle in $O$, which is a contradiction because $O$ is a chain graph (as per \cref{item-1-of-def:partial-MEC} of \Cref{def:partial-MEC}). Therefore, we have proven \Cref{claim:cyclic-sequence-does-not-exist} for all $l$ by induction.

This completes the proof of \Cref{claim:LBFS-algorithm-always-find-cannonical-node}.
\end{proof}

The following observation and its corollaries show that the output ordering is consistent with $O$.

\begin{observation}
\label{obs:direction-of-O-in-alg-LBFs}
Let $\mathcal{C}$ be an \ucc{} of $M_a$ for $a\in \{1,2\}$.
Suppose \Cref{alg:LBFSwithO} is called for  $\mathcal{C}$ and $O$. And, at the beginning of an iteration of the while loop (lines \ref{algLBFS:while-loop-start}-\ref{algLBFS:while-end} of \Cref{alg:LBFSwithO}), let $\tau = (a_1, a_2, \ldots, a_r)$ and $\mathcal{S} = (S_1, S_2, \ldots, S_l)$. Then,  for any \tfp{}  $P = (u_1, u_2, \ldots, u_m)$ of $\mathcal{C}$ such that $u_1\rightarrow u_2 \in O$, the following statements hold for all $1\leq i < j \leq m$,
\begin{enumerate}
    \item
    \label{item1-of-obs-direction-of-O-in-alg-LBFS}
    If $u_j \in \tau$, then $u_i \in \tau$ and $u_i$  arrived in $\tau$ before $u_j$.
    \item
    \label{item2-of-obs-direction-of-O-in-alg-LBFS}
    If $u_j \in S_x$ for some $1\leq x \leq l$,  then either $u_i \in \tau$, or $u_i \in S_y$ for $y\leq x$.
    \item
    \label{item3-of-obs-direction-of-O-in-alg-LBFS}
    If $u_1 \in \tau$, then at most one vertex of $P$ is in $S_1$.
\end{enumerate}
\end{observation}
Before proving \cref{obs:direction-of-O-in-alg-LBFs}, we go through the following claim, which is a necessary ingredient for the proof of \cref{obs:direction-of-O-in-alg-LBFs}.

\begin{claim}
\label{claim:P-is-a-tfp-in-Markov-union-of-graph}
    If $P = (u_1, u_2, \ldots, u_l)$ is a \tfp{} in $\mathcal{C}$ and $u_1\rightarrow u_2 \in O$ then $P$ is a \tfp{} in $U_M(\mathcal{C}, O)$. 
\end{claim}

\begin{proof}
Suppose $P$ is not a \tfp{} in $U_M(\mathcal{C}, O)$ then there must exist an $i$ such that $u_i\leftarrow u_{i+1} \in U_M(\mathcal{C}, O)$. Since $\mathcal{C}$ is an undirected graph, if there exists a directed edge in $U_M(\mathcal{C}, O)$ then the directed edge must belong to $O$. This implies $u_i\leftarrow u_{i+1} \in O$. From our assumption, $u_1\rightarrow u_2 \in O$. This implies $i> 1$. From our assumption, $P$ is a \tfp{} in $\mathcal{C}$, an \ucc{} of $M_a$ for some $a\in \{1,2\}$. Also, from the construction, $V_{M_a} \cap V_O = V_{O_a}$. This implies $u_1, u_2, u_i, u_{i+1} \in V_{O_a}$. Since $O_a$ is an induced subgraph of $M_a$, $(u_1, u_2), (u_i, u_{i+1}) \in E_{O_a}$. This implies $P_{a1}((u_1, u_2), (u_i, u_{i+1})) = 1$, as a subpath of $P$ from $(u_1, u_2)$ to $(u_i, u_{i+1})$ is also a \tfp{} in $M_a$, and $(u_1, u_2) \neq (u_i, u_{i+1})$ because $i > 1$. Since $O \in \mathcal{E}(O_1, P_{11}, P_{12}, O_2, P_{21}, P_{22})$,  from \cref{subitem-2-of-item-3-of-def:extension-of-O1-O2-P11-P12-P21-P22}, $u_i\rightarrow u_{i+1} \in O$, a contradiction. This implies $P$ is a \tfp{} in $U_M(\mathcal{C}, O)$.
\end{proof}

\begin{proof}[Proof of \cref{obs:direction-of-O-in-alg-LBFs}]
We will establish the proof through induction. Initially, at the start of the first iteration of the while loop in \Cref{alg:LBFSwithO}, we have $\mathcal{S}=(V_{\mathcal{C}})$, and $\tau$ is an empty list. Consequently, \Cref{obs:direction-of-O-in-alg-LBFs} is satisfied at this stage.

Now, let's assume that at the beginning of some iteration of the while loop, $\tau = (a_1, a_2, \ldots, a_r)$ and $\mathcal{S} = (S_1, S_2, \ldots, S_l)$, and that \Cref{obs:direction-of-O-in-alg-LBFs} is upheld. We aim to demonstrate that, even at the end of this iteration of the while loop, \Cref{obs:direction-of-O-in-alg-LBFs} continues to hold.

During the run of the iteration, we do the following: At line~\ref{algLBFS:pick-vertex-v}, we select a canonical source node $v$ from $U_M(\mathcal{C}, O)[S_1]$. We remove vertex $v$ from $S_1$ and add it to $\tau$ (lines \ref{algLBFS:add-vertex-to-tau}-\ref{algLBFS:remove-picked-node}). Subsequently, we update each set $S_i$ as follows: $S_i$ is replaced by $S_i\cap N(v, \mathcal{C})$ and $S_i \setminus N(v, \mathcal{C})$ (line \ref{algLBFS:replacement-of-S}). This results in the modified $\mathcal{S} = (S_{11}, S_{12}, \ldots, S_{l1}, S_{l2})$, where $S_{11} = (S_1\setminus v)\cap N(v, \mathcal{C})$, $S_{12} = (S_1\setminus v)\setminus N(v, \mathcal{C})$, and for $i>1$, $S_{i1} = S_i\cap N(v, \mathcal{C})$, and $S_{i2} = S_i \setminus N(v, \mathcal{C})$. And, in the end, we remove the empty sets in $\mathcal{S}$ (line \ref{algLBFS:removal-of-empty-sets}).

We now demonstrate that if a \tfp{} $P = (u_1, u_2, \ldots, u_m)$ satisfies $u_1\rightarrow u_2 \in O$ and adheres to \Cref{obs:direction-of-O-in-alg-LBFs} at the start of the iteration, it will continue to satisfy \cref{item1-of-obs-direction-of-O-in-alg-LBFS,item2-of-obs-direction-of-O-in-alg-LBFS,item3-of-obs-direction-of-O-in-alg-LBFS} of \Cref{obs:direction-of-O-in-alg-LBFs} even at the end of this iteration.
There are two possibilities: either $P$ contains the picked node $v$ (the canonical source node picked at step-\ref{algLBFS:pick-vertex-v}), or it does not.
We go through both possibilities one by one.
\begin{enumerate}
    \item Suppose that $P$ contains the picked node $v$. Let $u_a = v$, for some $1 \leq a \leq m$.  There are two cases: either $a = 1$ (i.e., $v$ is the first node of $P$), or $a > 1$.

\begin{enumerate}
    \item Suppose $a = 1$. Since $u_a = v$ has been picked from $S_1$, this implies that at the start of the iteration, $u_1$ is in $S_1$. Since, at the start of the iteration, $P$ obeys \Cref{obs:direction-of-O-in-alg-LBFs},  \cref{item1-of-obs-direction-of-O-in-alg-LBFS} of \Cref{obs:direction-of-O-in-alg-LBFs} implies that at the start of the iteration, none of the nodes of $P$ are in $\tau$.

At the end of the iteration, $u_1$ is in $\tau$, and all other nodes of $P$ remain in some set of $\mathcal{S}$. We will now demonstrate that for all possible pairs of nodes $u_i$ and $u_j$, where $1\leq i<j \leq m$, at the end of the iteration, $u_i$ and $u_j$ satisfy \cref{item1-of-obs-direction-of-O-in-alg-LBFS,item2-of-obs-direction-of-O-in-alg-LBFS} of \Cref{obs:direction-of-O-in-alg-LBFs}. For $1 = i < j \leq m$, $u_i$ and $u_j$ satisfy \cref{item1-of-obs-direction-of-O-in-alg-LBFS,item2-of-obs-direction-of-O-in-alg-LBFS} of \Cref{obs:direction-of-O-in-alg-LBFs}, as at the end of the iteration, $u_i$ is in $\tau$ and $u_j$ is in some set of $\mathcal{S}$.

Suppose $1 < i < j \leq m$. Then, at the start of the iteration, $u_i$ and $u_j$ are in some set $S_x$ and $S_y$ of $\mathcal{S}$, respectively.  Since at the start of the iteration, $u_i$ and $u_j$ satisfy \cref{item2-of-obs-direction-of-O-in-alg-LBFS} of \Cref{obs:direction-of-O-in-alg-LBFs}, we have $x \leq y$.

If $x < y$, then at the end of the iteration, the set in which $u_i$ lies (either $S_{x1}$ or $S_{x2}$) comes before the set in which $u_j$ lies (either $S_{y1}$ or $S_{y2}$). Thus, if $x<y$, then at the end of the iteration, $u_i$ and $u_j$ satisfy \cref{item1-of-obs-direction-of-O-in-alg-LBFS,item2-of-obs-direction-of-O-in-alg-LBFS} of \Cref{obs:direction-of-O-in-alg-LBFs}.

Suppose $x = y$, i.e., at the start of the iteration, both $u_i$ and $u_j$ lie in the same set $S_x$ of $\mathcal{S}$. Since $P$ is a \tfp{} of a chordal graph $\mathcal{C}$, from \cref{obv:tfps-are-chordless-in-chordal-graphs}, $P$ is a chordal graph. This implies that the only neighbor of $u_1 = v$ in $P$ is $u_2$. Therefore, if $i = 2$, then $u_i$ is a neighbor of $v$ and at the end of the iteration, it will move to set $S_{x1}$; otherwise, $u_i$ is not a neighbor of $v$ and at the end of the iteration, it will move to $S_{x2}$. Since $1< i < j$, $j$ cannot be equal to 2. This implies $u_j$ cannot be a neighbor of $v$ and at the end of the iteration, it will move to $S_{x2}$. This shows that at the end of the iteration, the set in which $u_i$ lies does not come after the set in which $u_j$ lies in $\mathcal{S}$. This shows that at the end of the iteration, $u_i$ and $u_j$ satisfy \cref{item1-of-obs-direction-of-O-in-alg-LBFS,item2-of-obs-direction-of-O-in-alg-LBFS} of \Cref{obs:direction-of-O-in-alg-LBFs}.

We now show that at the end of the iteration, \cref{item3-of-obs-direction-of-O-in-alg-LBFS} of \Cref{obs:direction-of-O-in-alg-LBFs} is also satisfied.

As shown above, at the end of the iteration, $u_1$ is in $\tau$, and all the remaining nodes of $P$ are in some set of $\mathcal{S}$. We have shown that at the end of the iteration, $P$ obeys \cref{item1-of-obs-direction-of-O-in-alg-LBFS,item2-of-obs-direction-of-O-in-alg-LBFS} of \Cref{obs:direction-of-O-in-alg-LBFs}. If at the end of the iteration, there does not exist any node of $P$ that is in the first set of $\mathcal{S}$, then \cref{item3-of-obs-direction-of-O-in-alg-LBFS} of \Cref{obs:direction-of-O-in-alg-LBFs} is satisfied. Suppose at the end of the iteration, there exists a node of $P$ that lies in the first set of $\mathcal{S}$. Pick the highest $i$ such that at the end of the iteration, $u_i$ is in the first set of $\mathcal{S}$. If $i = 2$, then at the end of the iteration, only one node of $P$ is in the first set of $\mathcal{S}$. This implies at the end of the iteration, \cref{item3-of-obs-direction-of-O-in-alg-LBFS} of \Cref{obs:direction-of-O-in-alg-LBFs} is satisfied. Suppose $i > 2$, then, from \cref{item2-of-obs-direction-of-O-in-alg-LBFS} of \Cref{obs:direction-of-O-in-alg-LBFs}, at the end of the iteration, $u_2, u_3, \ldots, u_i$ all are in the first set of $\mathcal{S}$. This further implies even at the start of the iteration, $u_2, u_3, \ldots, u_i$ are in the same set of $\mathcal{S}$, say $S_x$. But, we know that among these nodes only $u_2$ is a neighbor of $u_1
= v$. This implies at the end of the iteration, $u_2$ must be in $S_{x1}$, and all the remaining nodes $u_3, \ldots, u_i$ must be in $S_{x2}$. This contradicts our assumption that at the end of the iteration, $u_2, u_3, \ldots, u_i$ are in the same set of $\mathcal{S}$. This implies $i > 2$ does not occur. This implies even at the end of the iteration, \cref{item3-of-obs-direction-of-O-in-alg-LBFS} of \Cref{obs:direction-of-O-in-alg-LBFs} is obeyed.
   This implies that in all the possibilities of this case, \cref{item1-of-obs-direction-of-O-in-alg-LBFS,item2-of-obs-direction-of-O-in-alg-LBFS,item3-of-obs-direction-of-O-in-alg-LBFS} of \Cref{obs:direction-of-O-in-alg-LBFs} are satisfied. We will now move on to the second case.

    \item Suppose $a > 1$. Since $u_a = v$ has been picked from $S_1$,  this implies that at the start of the iteration, $u_a$ is in $S_1$. From \cref{item2-of-obs-direction-of-O-in-alg-LBFS} of \cref{obs:direction-of-O-in-alg-LBFs}, it follows that at the start of the iteration,  $u_1$ must be in either $S_1$ or in $\tau$. If $u_1$ is in $S_1$ at the start of the iteration, then since $u_a$ is also in $S_1$, \cref{item2-of-obs-direction-of-O-in-alg-LBFS} of \cref{obs:direction-of-O-in-alg-LBFs} implies that at the start of the iteration, the nodes between $u_1$ and $u_a$ in $P$ (i.e., $u_1, u_2, u_3, \ldots, u_{a-1}, u_a$) are all in $S_1$.
    This implies $Q = (u_1, u_2, \ldots, u_a)$ is a \tfp{} in $\mathcal{C}[S_1]$. But, then, from \cref{claim:P-is-a-tfp-in-Markov-union-of-graph}, $u_a$ is not a \canonical{} in $U_M(\mathcal{C}, O)[S_1]$, due to the existence of $Q$ and $u_1\rightarrow u_2 \in O$. But, this is a contradiction, as we have picked $u_a$ at line \ref{algLBFS:pick-vertex-v} of \cref{alg:LBFSwithO} because it is a \canonical{} in $U_M(\mathcal{C}, O)[S_1]$. Therefore, $u_1$ cannot be in $S_1$, but must be in $\tau$. Then, \cref{item3-of-obs-direction-of-O-in-alg-LBFS} of \Cref{obs:direction-of-O-in-alg-LBFs} implies that at the start of the iteration, the only node of $P$ in $S_1$ is $u_a = v$. 

From \cref{item2-of-obs-direction-of-O-in-alg-LBFS} of \Cref{obs:direction-of-O-in-alg-LBFs}, it follows that at the start of the iteration, $u_1, u_2, \ldots, u_{a-1}$ are in $\tau$. And, at the end of the iteration, $u_a$ also joins $\tau$. 

We will now show that, for all possible pairs of nodes $u_i$ and $u_j$, where $1\leq i<j \leq m$, at the end of the iteration, $u_i$ and $u_j$ obey \cref{item1-of-obs-direction-of-O-in-alg-LBFS,item2-of-obs-direction-of-O-in-alg-LBFS} of \Cref{obs:direction-of-O-in-alg-LBFs}.

If at the start of the iteration, both $u_i$ and $u_j$ are in $\tau$, then even at the end of the iteration, they remain in $\tau$ and obey \cref{item1-of-obs-direction-of-O-in-alg-LBFS,item2-of-obs-direction-of-O-in-alg-LBFS}. If at the start of the iteration, $u_i$ is in $\tau$ and $u_j$ is in some set of $S$, then at the end of the iteration, either (a) both $u_i$ and $u_j$ are in $\tau$ (this case occurs when $u_j = u_a = v$) and $u_i$ joins $\tau$ before $u_j$, or (b) $u_i$ is in $\tau$, and $u_j$ is in some set of $S$ (this case occurs when $j \neq a$). Both possibilities obey \cref{item1-of-obs-direction-of-O-in-alg-LBFS,item2-of-obs-direction-of-O-in-alg-LBFS}. 

Suppose at the start of the iteration, both $u_i$ and $u_j$ are in some sets $S_x$ and $S_y$ of $S$, respectively. Since at the start of the iteration, $u_i$ and $u_j$ obey \Cref{obs:direction-of-O-in-alg-LBFs}, $x\leq y$. 

If $i = a$, i.e., $u_i = v$, then at the end of the iteration, the pair $u_i$ and $u_j$ obey \cref{item1-of-obs-direction-of-O-in-alg-LBFS,item2-of-obs-direction-of-O-in-alg-LBFS}, as at the end of the iteration, $u_i$ joins $\tau$ and $u_j$ remains in some set of $\mathcal{S}$. 

Suppose $u_i\neq v$. If $x < y$, then even at the end of the iteration, the set in which $u_i$ lies (either $S_{x1}$ or $S_{x2}$) comes before the set in which $u_j$ lies (either $S_{y1}$ or $S_{y2}$). Suppose $x=y$. Then, at the start of the iteration, both $u_i$ and $u_j$ are in the same set $S_x$.

Similar to the previous case, $P$ is a \cp{}. Therefore, there can be at most two neighbors of $v = u_a$ in $P$: one is $u_{a-1}$, and another is $u_{a+1}$ (if $a< m$). As discussed earlier, at the start of the iteration,  $u_{a-1}$ is in $\tau$. This, in turn, implies that either none of $u_i$ and $u_j$ are neighbors of $u_a = v$, or only $u_i$ is a neighbor of $v$ (this scenario arises when $i = a+1$). 

If neither is a neighbor of $v$, then at the end of the iteration, both are in $S_{x2}$. If $u_i$ is a neighbor of $v$, then at the end of the iteration, $u_i \in S_{x1}$, and $u_j \in S_{x2}$. This implies the set in which $u_i$ lies does not come after the set in which $u_j$ lies in $\mathcal{S}$. This shows that, in all the possibilities of this case, \cref{item1-of-obs-direction-of-O-in-alg-LBFS,item2-of-obs-direction-of-O-in-alg-LBFS} of \Cref{obs:direction-of-O-in-alg-LBFs} are obeyed.

Similar to the previous case, we can show that at the end of the iteration, the only possible node of $P$ which can be in the first set of $P$ is $u_{a+1}$. This validates \cref{item3-of-obs-direction-of-O-in-alg-LBFS} of \Cref{obs:direction-of-O-in-alg-LBFs}. Thus, we show that in this case, \cref{item1-of-obs-direction-of-O-in-alg-LBFS,item2-of-obs-direction-of-O-in-alg-LBFS,item3-of-obs-direction-of-O-in-alg-LBFS} of \Cref{obs:direction-of-O-in-alg-LBFs} are obeyed. We now move to the next possibility.

\end{enumerate}

\item Suppose $P$ does not contain $v$. This implies that no node of $P$ moves to $\tau$ in this iteration. We will now demonstrate that, at the end of the iteration, for all pairs of nodes $u_i$ and $u_j$ in $P$ such that $1 \leq i < j \leq m$, \Cref{obs:direction-of-O-in-alg-LBFs} is satisfied. To prove \cref{item1-of-obs-direction-of-O-in-alg-LBFS} of \cref{obs:direction-of-O-in-alg-LBFs}, we assume that at the end of the iteration, $u_j$ is in $\tau$. To prove \cref{item2-of-obs-direction-of-O-in-alg-LBFS} of \cref{obs:direction-of-O-in-alg-LBFs}, we assume that at the end of the iteration, $u_j$ is in some set of $\mathcal{S}$. Finally, we prove \cref{item3-of-obs-direction-of-O-in-alg-LBFS} of \cref{obs:direction-of-O-in-alg-LBFs}.

Suppose at the end of the iteration, $u_j$ is in $\tau$. This implies that even at the start of the iteration, $u_j \in \tau$, as no node of $P$ moved to $\tau$ in this iteration. Since at the start of the iteration, $P$ adheres to \Cref{obs:direction-of-O-in-alg-LBFs}, according to \cref{item1-of-obs-direction-of-O-in-alg-LBFS} of \Cref{obs:direction-of-O-in-alg-LBFs}, $u_i$ has joined $\tau$ before $u_j$. Hence, at the end of the iteration, \cref{item1-of-obs-direction-of-O-in-alg-LBFS} of \Cref{obs:direction-of-O-in-alg-LBFs} is upheld.

Suppose at the end of the iteration, $u_j$ is in some set of $\mathcal{S}$. Then, even at the start of the iteration, $u_j$ must be in some set, say $S_y$, of $\mathcal{S}$. Since at the start of the iteration, $P$ obeys \Cref{obs:direction-of-O-in-alg-LBFs}, from \cref{item2-of-obs-direction-of-O-in-alg-LBFS} of \Cref{obs:direction-of-O-in-alg-LBFs}, at the start of the iteration, either $u_i$ is in $\tau$ or $u_i$ is in some set, say $S_x$, of $\mathcal{S}$ such that $x\leq y$.

If at the start of the iteration, $u_i$ is in $\tau$ then even at the end of the iteration, $u_i$ is in $\tau$, and $u_i$ and $u_j$ obey \cref{item2-of-obs-direction-of-O-in-alg-LBFS} of \Cref{obs:direction-of-O-in-alg-LBFs}. Now, suppose at the start of the iteration, $u_i$ is in some set $S_x$ of $\mathcal{S}$.

If $x < y$, then at the end of the iteration, $u_i$ is in a set of $\mathcal{S}$ (either $S_{x1}$ or $S_{x2}$) that comes before the set in which $u_j$ belongs (either $S_{y1}$ or $S_{y2}$), as both $S_{x1}$ and $S_{x2}$ come before $S_{y1}$ and $S_{y2}$. This implies that \Cref{obs:direction-of-O-in-alg-LBFs} is obeyed in this case.

If $x = y$, then either both $u_i$ and $u_j$ belong to the same set $S_{x1}$ (if both are neighbors of $v$) or both $u_i$ and $u_j$ belong to the same set $S_{x2}$ (if neither of them is a neighbor of $v$) or $u_i \in S_{x1}$ and $u_j \in S_{x2}$ (if $u_i$ is a neighbor of $v$, and $u_j$ is not a neighbor of $v$) or $u_j \in S_{x1}$ and $u_i \in S_{x2}$ (if $u_j$ is a neighbor of $v$, and $u_i$ is not a neighbor of $v$). The first three possibilities obey \Cref{obs:direction-of-O-in-alg-LBFs}, as in the first and second cases, $u_i$ and $u_j$ both are in the same set of $\mathcal{S}$, and in the third case, $u_i$ is in $S_{x1}$, which comes before $S_{x2}$ in which $u_j$ belongs. We now show that the last possibility cannot occur.

Suppose $u_j$ is a neighbor of $v$, and $u_i$ is not a neighbor of $v$. Pick the least $k$ such that $u_k$ comes after $u_i$ in $P$, and $u_k$ is a neighbor of $v$. Such a $u_k$ exists due to the presence of $u_j$.

Then, $Q = (u_1, u_2, \ldots, u_k, v)$ is a \tfp{} in $\mathcal{C}$ such that $u_1\rightarrow u_2 \in O$. 

Now, from our induction hypothesis, at the start of the iteration, all the \tfps{} of $\mathcal{C}$ obey \cref{obs:direction-of-O-in-alg-LBFs}. This implies $Q$ also obeys \cref{obs:direction-of-O-in-alg-LBFs} at the start of the iteration.

$v$ has been picked in this iteration implying that at the start of the iteration, $v\in S_1$. From \cref{item2-of-obs-direction-of-O-in-alg-LBFS} of \cref{obs:direction-of-O-in-alg-LBFs}, at the start of the iteration, all other nodes of $Q$ are either in $\tau$ or in $S_1$. Since $u_i$ is not in $\tau$ at the start of the iteration, this implies that at the start of the iteration, $u_i$ is in $S_1$ along with $v$. Then, \cref{item3-of-obs-direction-of-O-in-alg-LBFS} of \cref{obs:direction-of-O-in-alg-LBFs} implies that at the start of the iteration, all the nodes of $Q$ are in $S_1$.

But then, $v$ cannot be a canonical source node in $U_M(\mathcal{C}, O)[S_1]$ due to the existence of a \tfp{} $Q$ in $U_M(\mathcal{C}, O)[S_1]$ and $u_1\rightarrow u_2 \in O$. The fact that $Q$ is a \tfp{} follows from \cref{claim:P-is-a-tfp-in-Markov-union-of-graph}.

This implies this possibility cannot occur. This shows that in this case, \cref{item1-of-obs-direction-of-O-in-alg-LBFS,item2-of-obs-direction-of-O-in-alg-LBFS} of \Cref{obs:direction-of-O-in-alg-LBFs} are obeyed.

We now show that at the end of the iteration, $P$ also obeys \cref{item3-of-obs-direction-of-O-in-alg-LBFS} of \Cref{obs:direction-of-O-in-alg-LBFs}. Since in this iteration, the picked node $v$ is not a node of $P$, this implies that the set of nodes of $P$ in $\tau$ at the start of the iteration is the same as the nodes of $P$ in $\tau$ at the end of the iteration. If none of the nodes of $P$ are in $\tau$, then \cref{item2-of-obs-direction-of-O-in-alg-LBFS} is satisfied. 

Suppose a node of $P$ is in $\tau$. Since at the end of the iteration, $P$ obeys \cref{item1-of-obs-direction-of-O-in-alg-LBFS,item2-of-obs-direction-of-O-in-alg-LBFS} of \Cref{obs:direction-of-O-in-alg-LBFs}, there exists some $i \geq 1$ such that at the end of the iteration, $u_1, u_2, \ldots, u_i$ are in $\tau$ and each of the remaining nodes of $P$ is in some set of $\mathcal{S}$. If at the end of the iteration, no node of $P$ is in the first set of $\mathcal{S}$, then \cref{item3-of-obs-direction-of-O-in-alg-LBFS} is obeyed. 

Suppose at the end of the iteration, there exists a node of $P$ in the first set of $\mathcal{S}$. Pick the highest $j$ such that at the end of the iteration, $u_j$ is in the first set of $\mathcal{S}$. Since at the end of the iteration, $P$ obeys \cref{item1-of-obs-direction-of-O-in-alg-LBFS,item2-of-obs-direction-of-O-in-alg-LBFS} of \Cref{obs:direction-of-O-in-alg-LBFs}, at the end of the iteration, there must be $u_{i+1}, u_{i+2}, \ldots, u_j$ all are in the first set of $\mathcal{S}$. If $j = i+1$, then \cref{item3-of-obs-direction-of-O-in-alg-LBFS} is obeyed, as in that case there exists only a single node of $P$ in the first set of $\mathcal{S}$. We now show that $j = i+1$. 

Suppose $j > i+1$. Then, at the start of the iteration when $u_i$ has been picked and has been moved to $\tau$,  $u_{i+1}, u_{i+2}, \ldots, u_j$ must be in the same set, say $S_x$, of $\mathcal{S}$ (because from \cref{alg:LBFSwithO}, once two nodes of $\mathcal{C}$ gets separated and move to two different set of $\mathcal{S}$, they never unite).  Since $P$ is a \tfp{} of a chordal graph $\mathcal{C}$, from \cref{obv:tfps-are-chordless-in-chordal-graphs}, $P$ is a \cp{}. This implies that only $u_{i+1}$ is a neighbor of $u_i$ and rest of the nodes $u_{i+2}, u_{i+3} \ldots, u_j$ are not the neighbors of $u_i$ in $P$. This further implies that at the end of the iteration in which $u_i$ has been picked at step~\ref{algLBFS:pick-vertex-v} of \cref{alg:LBFSwithO}, $u_{i+1}$ moves to the set $S_{x1}$ of $\mathcal{S}$ and $u_j$ moves to $S_{x2}$ of $\mathcal{S}$. This further implies that even at the end of the current iteration, $u_{i+1}$ and $u_j$ are in different sets of $\mathcal{S}$, a contradiction. This implies $j > i+1$ cannot occur. This shows that at the end of the iteration,  \cref{item3-of-obs-direction-of-O-in-alg-LBFS} of \cref{obs:direction-of-O-in-alg-LBFs} is obeyed.
\end{enumerate}

This demonstrates that in all scenarios, at the end of the iteration, \cref{item1-of-obs-direction-of-O-in-alg-LBFS,item2-of-obs-direction-of-O-in-alg-LBFS,item3-of-obs-direction-of-O-in-alg-LBFS} of \Cref{obs:direction-of-O-in-alg-LBFs} are true. Thus, \Cref{obs:direction-of-O-in-alg-LBFs} is validated.
\end{proof}

\Cref{obs:direction-of-O-in-alg-LBFs} implies the following corollaries:

\begin{corollary}
\label{corr:directed-chordless-path-and-LBFS-ordering}
For input $\mathcal{C}$ and $O$, let $\tau$ be an LBFS ordering returned by \Cref{alg:LBFSwithO}. If $(u_1, u_2, \ldots, u_l)$ is a chordless path in $\mathcal{C}$ such that $u_1\rightarrow u_2 \in O$, then $\tau(u_1) < \tau(u_2) < \ldots < \tau(u_l)$.
\end{corollary}

\begin{corollary}
\label{corr:directed-edges-of-O-follow-tau}
For input $\mathcal{C}$ and $O$, let $\tau$ be an LBFS ordering returned by \Cref{alg:LBFSwithO}. For $a-b \in \mathcal{C}$, if $a\rightarrow b \in O$, then $\tau(a) < \tau(b)$. 
\end{corollary}

As discussed earlier, \Cref{alg:LBFSwithO} is a modified LBFS ordering of the LBFS ordering given by Rose et al. (\cite{rose1976algorithmic}). Instead of picking a node arbitrarily at line \ref{algLBFS:pick-vertex-v} of \Cref{alg:LBFSwithO}, we choose an arbitrary canonical source node. Since there always exists a canonical source node (as shown in \Cref{claim:LBFS-algorithm-always-find-cannonical-node}), \Cref{alg:LBFSwithO} outputs an LBFS ordering. \Cref{corr:directed-edges-of-O-follow-tau} demonstrates that the ordering returned by \Cref{alg:LBFSwithO} is consistent with $O$. \Cref{lemma:alg:LBFS-with-O-returns-LBFS-ordering-and-consistent-with-O} summarizes this.

\begin{lemma}
    \label{lemma:alg:LBFS-with-O-returns-LBFS-ordering-and-consistent-with-O}
    Let $H$ be an undirected graph, $H_1$ and $H_2$ be two induced subgraphs of $H$ such that $H=H_1\cup H_2$, and $I = V_{H_1}\cap V_{H_2}$ is a vertex separator of $H$. For each $i\in \{1,2\}$, let $S_i$ be a subset of $V_{H_i}$ such that $I\subseteq S_i$, $M_i$ be an MEC of $H_i$, and the shadow of $M_i$ on $S_i\cup N(S_i, H_i)$ be $(O_i, P_{i1}, P_{i1})$. Let $O \in \mathcal{E}(O_1, P_{11}, P_{12}, P_{12}, O_2, P_{21}, P_{22})$. For $i\in \{1,2\}$, let $\mathcal{C}$ be a universal critical component of $M_i$. For inputs  $\mathcal{C}$ and $O$, \Cref{alg:LBFSwithO} returns an LBFS ordering $\tau$ that satisfies $O$, i.e., for $u\rightarrow v\in O$ such that $u,v \in V_{\mathcal{C}}$, $\tau(u) < \tau(v)$.
\end{lemma}
 \subsection{\texorpdfstring{Sufficient condition for a shadow of an MEC of $H$}{Sufficient condition for a shadow of an MEC of H}}
\label{subsection:sufficient-condition}
In this subsection, we show that \cref{item-1-of-obs:nes-conditions-of-the-shadow,item-2-of-obs:nes-conditions-of-the-shadow} of \cref{obs:nes-conditions-of-the-shadow} constitutes the required sufficient condition.
\begin{lemma}
\label{obs2:O-structure-for-existence-of-MEC}  
Let $H$ be an undirected graph, and $H_1$ and $H_2$ be the induced subgraph of $H$ such that $H = H_1\cup H_2$, and $I = V_{H_1} \cap V_{H_2}$ is a vertex separator of $H$ that separates $V_{H_1}\setminus I$ and $V_{H_2}\setminus I$. Let $S_1, S_2$ be the  subsets of $V_{H_1}$ and $V_{H_2}$, respectively, such that  $S_1\cap S_2 = I$. 
Let $M_1$ be an MEC of $H_1$, and $M_2$ be an MEC of $H_2$.
Let $(O_1, P_{11}, P_{12})$ be the shadow of $M_1$ on $S_1\cup N(S_1, H_1)$, and $(O_2, P_{21}, P_{22})$ be the shadow of $M_2$ on $S_2\cup N(S_2, H_2)$. 
Let $A = H[S_1\cup S_2\cup N(S_1\cup S_2, H)]$ and $O\in \setofpartialMECs{A}$.
If $O$ is an extension of $(O_1, P_{11}, P_{12}, O_2, P_{21}, P_{22})$ then there exists a unique MEC $M$ such that $M[V_O] = O$, and $\mathcal{P}(M, V_{H_1}, V_{H_2}) = (M_1, M_2)$. If $\EPF{O, P_{11}, P_{12}, P_{21}, P_{22}} = (P_1,P_2)$ then the shadow of $M$ on $V_O$ is $(O, P_1, P_2)$.

\end{lemma}
\begin{algorithm}
\caption{Construct MEC}
\label{alg:constructMEC}
\SetAlgoLined
\SetKwInOut{KwIn}{Input}
\SetKwInOut{KwOut}{Output}
\SetKwFunction{constructMEC}{construct_MEC}
\KwIn{Three graphs $M_1, M_2$, and $O$ such that\\ 
$H$ is an undirected graph and \\
$H_1$ and $H_2$ are two induced subgraphs of $H$ such that\\
$H = H_1\cup H_2$, and $I = V_{H_1}\cap V_{H_2}$ is a vertex separator of $H$.\\
$S_1$ and $S_2$ are subsets of $V_{H_1}$ and $V_{H_2}$, respectively, such that $S_1\cap S_2 =I$.\\
$M_1$ and $M_2$ are MECs of $H_1$ and $H_2$, respectively.\\
$\forall a\in \{1,2\}$, the shadow of $M_a$ on $S_a\cup N(S_a, H_a)$ is $(O_a, P_{a1}, P_{a2})$, and\\
$O\in \mathcal{E}(O_1, P_{11},P_{12}, O_2,P_{21}, P_{22})$.
}
    \KwOut{ $M$: an MEC of $H_1\cup H_2$ such that  \\ $M \in \setofMECs{H, O, P_1, P_2}$, $M[V_O] = O$, and $\mathcal{P}(M, V_{H_1}, V_{H_2}) = (M_1, M_2)$.}
  $M \leftarrow U_M(M_1, M_2, O)$ \label{algconstructMEC:M-initialization}
  
  $I \leftarrow V_{M_1}\cap V_{M_2}$
  
  \ForEach{$i \in \{1,2\}$  \label{algconstructMEC:foreachMi}}
    {
        \ForEach {undirected connected component $\mathcal{C}$ of $M_i$ \label{algconstructMEC:foreachUCC}}
        {

                $\tau \leftarrow$ LBFS($\mathcal{C}$,$O$). \label{algconstructMEC:contructLBFS}
                
                Denote a vertex $v$  of $\mathcal{C}$ as  $u_i$ if $\tau(v) = i$.
                
                $a \leftarrow 2$
                
                \While{$a \leq |V_{\mathcal{C}}|$ \label{algconstructMEC:while-start}}
                {
                    $b \leftarrow a-1$
                    
                    \While{$b\geq 1$ \label{algconstructMEC:inner-while}}
                    {
                        \If{ $\exists c$ such that $c<b$, and $u_c-u_b-u_a$ is an induced subgraph of $\mathcal{C}$, and $u_c\rightarrow u_b - u_a \in M$, or $\exists c$ such that $b < c < a$, and $u_b-u_c-u_a-u_b \in \mathcal{C}$, and $u_b\rightarrow u_c \rightarrow u_a -u_b \in M$ \label{algconstructMEC:constraint-for-orientation}}
                        {
                            Replace the edge $u_b-u_a$ in  $M$ with $u_b\rightarrow u_a$. \label{algconstructMEC:add-orientation}
                        }\label{algconstructMEC:if-end}
                        
                        $b = b -1$
                    }\label{algconstructMEC:inner-while-end}
                    
                    $a = a+1$
                }\label{algconstructMEC:while-end}

        }\label{algconstructMEC:foreachUCC-end}
      }\label{algconstructMEC:foreachMi-end}
\KwRet $M$
\end{algorithm}

 \begin{proof}
Suppose $M_1 \in \setofMECs{H_1, O_1, P_{11}, P_{12}}$, $M_2 \in \setofMECs{H_2, O_2, P_{21}, P_{22}}$, $O \in \mathcal{E}(O_1, P_{11}, P_{12}, O_2, P_{21}, P_{22})$, and $(P_1, P_2) = \EPF{O, P_{11}, P_{12}, P_{21}, P_{22}}$. We show that there exists a unique MEC $M$ such that $M[V_O] = O$,  $\mathcal{P}(M, V_{H_1}, V_{H_2}) = (M_1, M_2)$, and the shadow of $M$ on $V_O$ is $(O, P_1, P_2)$.
We first construct such an MEC $M$ using \cref{alg:constructMEC}. \Cref{alg:constructMEC} follows the following steps to construct the MEC:
\begin{enumerate}
    \item 
    \label{item:initialize-M}
    Initialize $M = U_M(M_1, M_2, O)$ (line-\ref{algconstructMEC:M-initialization} of \cref{alg:constructMEC}), a Markov union (\cref{def:Markov-union-of-graphs}) of $M_1, M_2$ and $O$.
\item
    \label{item:reframe-ucc-of-Mi}
    For each $i\in \{1,2\}$ (lines \ref{algconstructMEC:foreachMi}-\ref{algconstructMEC:foreachMi-end} of \cref{alg:constructMEC}), for each undirected connected component $\mathcal{C}$ of $M_i$ (lines \ref{algconstructMEC:foreachUCC}-\ref{algconstructMEC:foreachUCC-end} of \cref{alg:constructMEC}), 
we do the following:
        \begin{enumerate}
            \item
            \label{item-a-of-item:reframe-ucc-of-Mi}
            Construct an LBFS ordering $\tau$ of $\mathcal{C}$ using \cref{alg:LBFSwithO} for input $\mathcal{C}$ and $O$ (line \ref{algconstructMEC:contructLBFS} of \cref{alg:constructMEC}).
            
            \item
            \label{item-b-of-item:reframe-ucc-of-Mi}
            For $u-v \in M$ such that $\tau(u) < \tau(v)$, replace $u-v$ with $u\rightarrow v$ (line \ref{algconstructMEC:add-orientation} of \cref{alg:constructMEC}) if either
                \begin{enumerate}
                    \item there exists a vertex $x \in \mathcal{C}$ such that $\tau(x) < \tau(u)$, and  there exists $x-u -v$, an induced subgraph of $\mathcal{C}$, such that $x\rightarrow u \in M$ (line \ref{algconstructMEC:constraint-for-orientation}), or
                    \item there exists a vertex $x \in \mathcal{C}$ such that $\tau(u) < \tau(x) < \tau(v)$, and $u\rightarrow x, x\rightarrow v \in M$ (line \ref{algconstructMEC:constraint-for-orientation}).
                \end{enumerate}
        \end{enumerate}
\end{enumerate}

We begin by demonstrating that \cref{alg:constructMEC} correctly executes steps \ref{item:initialize-M} and \ref{item:reframe-ucc-of-Mi}. The implementation of step \ref{item:initialize-M} can be found at line~\ref{algconstructMEC:M-initialization} of \cref{alg:constructMEC}. Next, we verify that for each $i \in \{1,2\}$ and for every \ucc{} $\mathcal{C}$ of $M_i$, steps \ref{item-a-of-item:reframe-ucc-of-Mi} and \ref{item-b-of-item:reframe-ucc-of-Mi} are executed correctly.

Step \ref{item-a-of-item:reframe-ucc-of-Mi} is implemented at line \ref{algconstructMEC:contructLBFS} of \cref{alg:constructMEC}, where we construct an LBFS ordering $\tau$ for $\mathcal{C}$. To demonstrate that step \ref{item-b-of-item:reframe-ucc-of-Mi} is also executed correctly, we utilize \Cref{corr:an-edge-of-an-ucc-gets-directed-in-M-either-at-initiazation-or-in-the-iteration-when-ucc-is-considered} to provide assurance in the accuracy of \cref{alg:constructMEC}.

Let's suppose that step \ref{item-b-of-item:reframe-ucc-of-Mi} is not accurately executed by \cref{alg:constructMEC}. This would imply the existence of an $x$ such that either (a) $x-u-v$ forms an induced subgraph of $\mathcal{C}$, with $\tau(x) < \tau(u) < \tau(v)$, and $x\rightarrow u-v \in M$, or (b) $u-x-v-u$ forms an induced subgraph of $\mathcal{C}$, with $\tau(u) < \tau(x) < \tau(v)$, and $u\rightarrow x$, $x\rightarrow v - u \in M$.

Let's address both cases individually:

Case (a): Assuming there exists an $x$ such that $x-u-v$ forms an induced subgraph of $\mathcal{C}$, $\tau(x) < \tau(u) < \tau(v)$, and $x\rightarrow u-v \in M$. According to \cref{corr:an-edge-of-an-ucc-gets-directed-in-M-either-at-initiazation-or-in-the-iteration-when-ucc-is-considered}, $x\rightarrow u$ was added to $M$ either during initialization or when $a = \tau(u)$ and $b = \tau(x)$. This implies that when we have $a = \tau(v)$ and $b = \tau(u)$, $x\rightarrow u$ is in $M$. However, during the execution of lines \ref{algconstructMEC:constraint-for-orientation}--\ref{algconstructMEC:if-end} (when we have $a = \tau(v)$ and $b = \tau(u)$), $u\rightarrow v$ is added to $M$, leading to a contradiction. Therefore, this case cannot occur.

Case (b): Assuming there exists an $x$ such that $u-x-v-u$ forms an induced subgraph of $\mathcal{C}$, with $\tau(u) < \tau(x) < \tau(v)$, and $u\rightarrow x$, $x\rightarrow v - u \in M$. According to \cref{corr:an-edge-of-an-ucc-gets-directed-in-M-either-at-initiazation-or-in-the-iteration-when-ucc-is-considered}, $u\rightarrow x$ was added to $M$ either during initialization or when $a = \tau(x)$ and $b = \tau(u)$, and $x\rightarrow v$ was added to $M$ either during initialization or when $a = \tau(v)$ and $b = \tau(x)$. This implies that when we have $a = \tau(v)$ and $b = \tau(u)$, both $u\rightarrow x$ and $x\rightarrow u$ are in $M$. However, during the execution of lines \ref{algconstructMEC:constraint-for-orientation}--\ref{algconstructMEC:if-end} (when we have $a = \tau(v)$ and $b = \tau(u)$), $u\rightarrow v$ is added to $M$, again leading to a contradiction. Therefore, this case also cannot occur.

In conclusion, we have demonstrated that \ref{item-b-of-item:reframe-ucc-of-Mi} is executed correctly by \cref{alg:constructMEC}. This further affirms that \cref{alg:constructMEC} accurately implements steps \ref{item:initialize-M} and \ref{item:reframe-ucc-of-Mi}.

\Cref{lem:O-is-an-induced-subgraph-of-M,lem:M-is-an-MEC,lem:M1-and-M2-are-projections-of-M,lem:uniqueness-of-M,lem:derived-path-function-is-a-part-of-shadow-of-M} validate $M$. 

\begin{lemma}
    \label{lem:O-is-an-induced-subgraph-of-M}
$M[V_O] = O$.
\end{lemma}

\begin{lemma}
    \label{lem:M-is-an-MEC}
    $M$ is an MEC of $H$.
\end{lemma}

\begin{lemma}
    \label{lem:M1-and-M2-are-projections-of-M}
    $\mathcal{P}(M, V_{M_1}, V_{M_2}) = (M_1, M_2)$.
\end{lemma}

\begin{lemma}
    \label{lem:uniqueness-of-M}
    If there exists an MEC $M'$ of $H$ such that $M'[V_O] = O$, and $\mathcal{P}(M', V_{M_1}, V_{M_2}) = (M_1, M_2)$ then $M' = M$.
\end{lemma}

\begin{lemma}
    \label{lem:derived-path-function-is-a-part-of-shadow-of-M}
    If $\EPF{O, P_{11}, P_{12}, P_{21}, P_{22}} = (P_1, P_2)$ then $M \in \setofMECs{H, O, P_1, P_2}$.
\end{lemma}
\Cref{lem:M-is-an-MEC} shows that the graph returned by \cref{alg:constructMEC} is an MEC. \Cref{lem:O-is-an-induced-subgraph-of-M} shows that $O$ is an induced subgraph of $M$, i.e., $M[V_O] = O$, and \cref{lem:M1-and-M2-are-projections-of-M} shows that $M_1$ and $M_2$ are projections of $M$. \Cref{lem:uniqueness-of-M} shows that there exists a unique MEC with $O$ as its induced subgraph, and $M_1$ and $M_2$ as its projections. Additionally, \cref{lem:derived-path-function-is-a-part-of-shadow-of-M} shows that if $(P_1, P_2) = \EPF{O, P_{11}, P_{12}, P_{21}, P_{22}}$, then the shadow of $M$ on $V_O$ is $(O, P_1, P_2)$. This also implies that given $O, P_{11}, P_{12}, P_{21}$, and $P_{22}$, we can compute the shadow of $M$ on $V_O$.

Proofs of \cref{lem:O-is-an-induced-subgraph-of-M,lem:M-is-an-MEC,lem:M1-and-M2-are-projections-of-M,lem:uniqueness-of-M,lem:derived-path-function-is-a-part-of-shadow-of-M} are provided below. Given these lemmas, \cref{obs2:O-structure-for-existence-of-MEC} follows immediately.

\end{proof}

We now prove \cref{lem:O-is-an-induced-subgraph-of-M}.
In preparation for proving
\cref{lem:O-is-an-induced-subgraph-of-M}, we first prove the following observation.
The proof of
\cref{lem:O-is-an-induced-subgraph-of-M} appears on
page~\pageref{proof-of-obs:edges-of-subMECs-does-not-change-its-orientation}.

\begin{observation}
\label{obs:when-u-v-becomes-directed}
For $i\in \{1,2\}$,  for any \ucc{} $\mathcal{C}$ of $M_i$, at line-\ref{algconstructMEC:contructLBFS} of \cref{alg:constructMEC}, let $\tau$ be the LBFS ordering returned by \cref{alg:LBFSwithO}. Then for $u-v \in \mathcal{C}$,   at any iteration of the while loop (lines \ref{algconstructMEC:while-start}-\ref{algconstructMEC:while-end}) of \cref{alg:constructMEC}, if $u\rightarrow v \in M$ then either of the following occurs:
\begin{enumerate}
    \item
    \label{item-obs:when-u-v-becomes-directed:u-v-is-in-O}
    $u\rightarrow v \in O$ \item
    \label{item-obs:when-u-v-becomes-directed:u-v-is-part-of-a-directed-cp}
    $\exists a\rightarrow b \in O$ such that there exists a \cp{}
    from $(a,b)$ to $(u,v)$
    in $\mathcal{C}$ such that $(a, b) \neq (u, v)$. \item 
    \label{item-obs:when-u-v-becomes-directed:u-v-is-part-of-a-directed-cycle}
    $\exists a\rightarrow b \in O$ such that there exist \cps{} $P$ and $Q$
    such that $P$ is a \cp{}
    from $(a,b)$ to $v$ in $\mathcal{C}$, and $Q$ is a \cp{} from $(v,u)$ to $a$ in $\mathcal{C}$, such that
    $b \neq v$ and $u \neq a$.  \end{enumerate}
\end{observation}
\begin{proof}
We start with the following observation:
\begin{observation}
    \label{obs:undirected-in-M_1-and-directed-in-M-implies-directed-in-O}
    For $i\in \{1,2\}$, if $u-v\in M_i$ and $u\rightarrow v \in U_M(M_1, M_2, O)$ then $u\rightarrow v \in O$.
\end{observation}
\begin{proof}
    From \cref{def:Markov-union-of-graphs},  if $u-v\in M_1$ and $u\rightarrow v \in U_M(M_1, M_2,O)$ then either $u\rightarrow v \in M_2$ or $u\rightarrow v \in O$.
    $u\rightarrow v \in M_2$ implies that $u,v \in V_{M_1} \cap V_{M_2} = I$.
    Since $(O_2, P_{21}, P_{22})$ is the shadow of $M_2$ on $V_{O_2}$, therefore, from \cref{item-1-of-def:shadow} of \cref{def:shadow}, $O_2 = M_2[V_{O_2}]$. This implies $u\rightarrow v \in O_2$ as from the construction,  $I \subseteq V_{O_2}$.
    Since $O\in \mathcal{E}(O_1, P_{11}, P_{12}, O_2, P_{21}, P_{22})$, therefore, from \cref{item-1-of-def:extension-of-O1-O2-P11-P12-P21-P22} of \cref{def:extension-of-O1-O2-P11-P12-P21-P22}, we have $u\rightarrow v \in O$.
    Thus, we can say that if $u-v\in M_1$ and $u\rightarrow v \in U_M(M_1, M_2,O)$ then $u\rightarrow v \in O$.
    Similarly, if $u-v\in M_2$ and $u\rightarrow v \in U_M(M_1, M_2,O)$ then $u\rightarrow v \in O$. This completes the proof of \cref{obs:undirected-in-M_1-and-directed-in-M-implies-directed-in-O}.
  \end{proof}

  Let us number the \uccs{} of $M_1$ and $M_2$ 
in the sequence in which they are considered on
  line~\ref{algconstructMEC:foreachUCC} as
  $\mathcal{C}_1, \mathcal{C}_{2}, \dots$.  Suppose that
  \cref{obs:when-u-v-becomes-directed} is true for all $C_{\beta}$ with
  $\beta < \alpha$. We now prove by induction on the iteration of the while loop
  (lines \ref{algconstructMEC:while-start}-\ref{algconstructMEC:while-end}) of
  \cref{alg:constructMEC} corresponding to the connected component
  $\mathcal{C}_{\alpha}$ that the \cref{obs:when-u-v-becomes-directed} is true
  for $\mathcal{C} = \mathcal{C}_{\alpha}$ as well.
\paragraph{Base case:} We first show that at the start of the first iteration
  of the while loop corresponding to $\mathcal{C} = \mathcal{C}_{\alpha}$, it
  must be the case that for any edge $u-v \in \mathcal{C}$, if
  $u\rightarrow v \in M$ then $u \rightarrow v \in O$.  This will establish
  that \cref{obs:when-u-v-becomes-directed} holds at the start of the first
  iteration of the while loop corresponding to the \uccc{}
  $\mathcal{C} = \mathcal{C}_{\alpha}$.

  Suppose, if possible, at the start of the first iteration of the while loop,  there exists
  $u - v \in \mathcal{C}$ such that $u\rightarrow v \in M$. There are two possibilities: either $u\rightarrow v$ is added at the initialization step,  i.e., at line \ref{algconstructMEC:M-initialization}  of \cref{alg:constructMEC}, or $u \rightarrow v$ is added to $M$ during the processing of
  some $\mathcal{C}_{\gamma}$ with $\gamma < \alpha$ at lines \ref{algconstructMEC:foreachMi}-\ref{algconstructMEC:foreachMi-end} of \cref{alg:constructMEC}. 
At the initialization step, $M = U_M(M_1, M_2, O)$. $\mathcal{C}$ is an \ucc{} of $M_i$ for some $i \in \{1,2\}$. This implies $u-v \in M_i$. Then from \cref{obs:undirected-in-M_1-and-directed-in-M-implies-directed-in-O}, if $u\rightarrow v \in U_M(M_1, M_2, O)$ then $u\rightarrow v \in O$. This validates \cref{obs:when-u-v-becomes-directed}. We now go through the other possibility that $u \rightarrow v$ has been added to $M$ during the processing of some $\mathcal{C}_{\gamma}$ with $\gamma < \alpha$.

  W.l.o.g., let us assume that $C_{\alpha}$ is an \ucc{} of $M_1$. Then $C_{\gamma}$ cannot be an \ucc{} of $M_1$ (two distinct \uccs{} of a graph do not share any vertex). This implies $C_{\gamma}$ is an \ucc{} of $M_2$. This further implies $u-v$ is in both $M_1$ and $M_2$. That means $u,v \in V_{M_1}\cap V_{M_2} = I \subseteq V_O$. 
  Since  $u-v \in C_{\gamma}$, an \ucc{} of $M_2$, therefore,  $u-v \in M_2$. Also, since $u,v \in I \subseteq V_{O_2}$, and $O_2$ is an induced subgraph of $M_2$ (from \cref{item-1-of-def:shadow} of \cref{def:shadow} as $(O_2, P_{21}, P_{22})$ is a shadow of $M_2$), therefore, $u-v \in O_2$.
  Since we assumed that
  \cref{obs:when-u-v-becomes-directed} is true for such a $C_{\gamma}$, 
  one of  the three conditions (\cref{item-obs:when-u-v-becomes-directed:u-v-is-in-O,item-obs:when-u-v-becomes-directed:u-v-is-part-of-a-directed-cp,item-obs:when-u-v-becomes-directed:u-v-is-part-of-a-directed-cycle})
  of \cref{obs:when-u-v-becomes-directed} must be true for
  $u - v$, as $u-v$ is considered as an edge in $C_{\gamma}$.  
  
  If
  \cref{item-obs:when-u-v-becomes-directed:u-v-is-in-O} of \cref{obs:when-u-v-becomes-directed} holds then we get
  $u \rightarrow v \in O$.  
  If one of
  \cref{item-obs:when-u-v-becomes-directed:u-v-is-part-of-a-directed-cp} or
  \cref{item-obs:when-u-v-becomes-directed:u-v-is-part-of-a-directed-cycle} of \cref{obs:when-u-v-becomes-directed}
  holds, then in the \uccc{} $C_{\gamma}$ of $M_2$, there exists an edge
  $x-y \in \mathcal{C}_{\gamma}$ such that $x\rightarrow y \in O$, and either (a) $(x,y) \neq (u,v)$ and
  there exists a \cp{} from $(x,y)$ to $(u,v)$ in $\mathcal{C_{\gamma}}$ (due to \cref{item-obs:when-u-v-becomes-directed:u-v-is-part-of-a-directed-cp} of \cref{obs:when-u-v-becomes-directed}),
  or (b) $y\neq v$, $u\neq x$, and there
  exist \cps{} $Q_1$ and $Q_2$ in $C_{\gamma}$ such that $Q_1$ is a \cp{} from
  $(x,y)$ to $v$, and $Q_2$ is a \cp{} from $(v,u)$ to $x$ (due to  \cref{item-obs:when-u-v-becomes-directed:u-v-is-part-of-a-directed-cycle} of \cref{obs:when-u-v-becomes-directed}). 
  Since $C_{\gamma}$ is an \ucc{} of $M_2$, therefore, a \cp{} in $C_{\gamma}$ is a \cp{} of $M_2$. And, from \cref{obs:cp-is-tfp}, every \cp{} is a \tfp{}. This implies \cps{} in  $C_{\gamma}$ are \tfps{} of $M_2$. 
  Also, since $(x,y) \in C_{\gamma}$, an \ucc{} of $M_2$, and $x\rightarrow y \in O$, this implies $x,y \in V_{M_2} \cap V_O = V_{O_2}$. Since $(O_2, P_{21}, P_{22})$ is the shadow of $M_2$ on $V_{O_2}$, from \cref{def:shadow}, $O_2$ is an induced subgraph of $M_2$. This implies $(x,y)\in E_{O_2}$, and if there exists a \cp{} in $C_{\gamma}$, from $(x,y)$ to $(u,v)$, then $P_{21}((x,y),(u,v)) =1$, and if there exist \cps{} in $C_{\gamma}$, from $(x,y)$ to $v$ and from $(v,u)$ to $x$ then $P_{22}((x,y),v) = P_{22}((v,u),x) = 1$. 
  Since $O \in \mathcal{E}(O_1, P_{11}, P_{12}, O_2, P_{21}, P_{22})$, therefore, in both the
  possibilities, from \cref{item-3-of-def:extension-of-O1-O2-P11-P12-P21-P22} of
  \cref{def:extension-of-O1-O2-P11-P12-P21-P22}, we again have
  $u\rightarrow v \in O$. This shows that for any $u-v \in C_{\alpha}$,  when $C_{\alpha}$ is processed at lines \ref{algconstructMEC:foreachUCC}-\ref{algconstructMEC:foreachUCC-end} of \cref{alg:constructMEC}, at the start of the first iteration of the while loop at lines \ref{algconstructMEC:while-start}-\ref{algconstructMEC:while-end}, if $u\rightarrow v \in M$  then $u\rightarrow v \in O$.  This establishes the base case of the induction.

  \paragraph{Induction step:} Suppose now that at the start of an iteration of
  the while loop (at lines \ref{algconstructMEC:while-start}-\ref{algconstructMEC:while-end} of \cref{alg:constructMEC}),
  \cref{obs:when-u-v-becomes-directed} is valid, that means, at the start of the iteration, for any edge $u'-v' \in \mathcal{C}$, if $u'\rightarrow v' \in M$ then $u'\rightarrow v'$ follows \cref{obs:when-u-v-becomes-directed}. We show that even at the end of the iteration, \cref{obs:when-u-v-becomes-directed} is valid, i.e., for any $u-v \in \mathcal{C}$, during the iteration, if $u\rightarrow v$ is added to $M$ then the new
  directed edge $u\rightarrow v$ obeys
  \cref{obs:when-u-v-becomes-directed}. During the iteration, let there exists an edge
  $u-v \in \mathcal{C}$ such that $u\rightarrow v$ is added to $M$ in the while loop while processing $\mathcal{C} = \mathcal{C_{\alpha}}$.  We show that $u\rightarrow v$ obeys \cref{obs:when-u-v-becomes-directed}. 
  Let $\tau$ be the LBFS ordering of $\mathcal{C}$ returned by \cref{alg:LBFSwithO} at line \ref{algconstructMEC:contructLBFS} of \cref{alg:constructMEC} for the input $\mathcal{C}$ and $O$.
  From  lines \ref{algconstructMEC:constraint-for-orientation}-\ref{algconstructMEC:if-end} of  \cref{alg:constructMEC}, if $u\rightarrow v$ is added to $M$ then at the start of the iteration, either of the
  following has occurred: (a) there exists an induced subgraph
  $x-u-v \in \mathcal{C}$ such that $\tau(x)<\tau(u)<\tau(v)$, and
  $x\rightarrow u -v \in M$, or (b) there exists an induced subgraph
  $u-x-v-u \in \mathcal{C}$ such that $\tau(u) < \tau(x) < \tau(v)$,
  $u\rightarrow x, x\rightarrow v \in M$ and $u-v \in M$. We go through each
  case and show that in each case $u\rightarrow v$ obeys
  \cref{obs:when-u-v-becomes-directed}.

  \begin{enumerate}
  \label{proof-of-obs:when-u-v-becomes-directed-1.a}
      \item Suppose that during the iteration, $u-v \in M$ is replaced
  with $u\rightarrow v$ because there exists an induced subgraph
  $x-u-v \in \mathcal{C}$ such that $\tau(x)<\tau(u)<\tau(v)$, and at the start
  of the iteration, $x\rightarrow u-v \in M$.  By the induction hypothesis, we
  have that $x\rightarrow u \in M$ because one of the following holds:
(i) $x\rightarrow u \in O$, or 
  (ii) $\exists a\rightarrow b \in O$ such that there exists a \cp{}
  $(p_1 = a, p_2 = b, \ldots, p_{l-1} = x, p_{l} = u)$ from $a-b$ to $x-u$ in
  $\mathcal{C}$ (with $(a, b) \neq (x, u)$), or
(iii) $\exists a\rightarrow b \in O$ such that there exists a \cp{} from $a-b$
  to $u$, and a \cp{} from $u-x$ to $a$ in $\mathcal{C}$ (with $b \neq u$ and
  $x \neq a$). We start with the first possibility.

  \begin{enumerate}
      \item Suppose that $x\rightarrow u \in M$ because 
$x\rightarrow u \in O$. Then the new
directed edge $u\rightarrow v$ obeys
\cref{item-obs:when-u-v-becomes-directed:u-v-is-part-of-a-directed-cp} of
\cref{obs:when-u-v-becomes-directed}, due to the existence of the \cp{}
$(x, u, v)$ in $\mathcal{C}$ and $x\rightarrow u \in O$. We now go through the second possibility.

\item 
Suppose that $x\rightarrow u \in M$ because $\exists a\rightarrow b \in O$ such that there exists a \cp{} $Q = (p_1 =a, p_2 =b, \ldots, p_{l-1} =x, p_l =u)$ from $a-b$ to $x-u$ in $\mathcal{C}$ (with $(a, b) \neq (x, u)$).
By the hypothesis of this case, $Q' = (x,u,v)$ is a \cp{} in $\mathcal{C}$ (remember in this case, $x-u-v$ is an induced subgraph of $\mathcal{C}$). Then, from \cref{lem:concatenation-of-chordal-paths}, $P =  (p_1 =a, p_2 =b, \ldots, p_{l-1} =x, p_l =u, v)$, concatenation of $Q$ and $Q'$, is a \cp{} in $\mathcal{C}$.

Thus, we have a \cp{} $P$ in $\mathcal{C}$ from $(a,b)$ to $(u,v)$ with $(a, b) \neq (u, v)$ ($Q$ being a path precludes $a = u$) and $a\rightarrow b \in O$. This further implies that the new directed edge $u\rightarrow v$ obeys \cref{item-obs:when-u-v-becomes-directed:u-v-is-part-of-a-directed-cp} of \cref{obs:when-u-v-becomes-directed}.  We now move to the third possibility.

\item 
Suppose that $x\rightarrow u \in M$ because $\exists a\rightarrow b \in O$ such that there exists a \cp{} $Q_1 = (p_1 =a, p_2 =b, \ldots, p_{l} =u)$ from $(a,b)$ to $u$ in $\mathcal{C}$, and a \cp{} $Q_2 = (q_1=u, q_2=x, \ldots, q_m=a)$ from $(u,x)$ to $a$ in $\mathcal{C}$ (with $b \neq u$ and $x \neq a$).

Since $Q_1$ is a \cp{} and $a\rightarrow b \in O$, according to \cref{corr:directed-chordless-path-and-LBFS-ordering}, we have $\tau(a) = \tau(p_1) < \tau(p_{2}) < \tau(p_{3}) < \dots < \tau(p_l) = \tau(u)$.

Now, we claim that
\begin{equation}
\label{claim:v-pi-not-in-C}
\text{For $1 \leq i\leq l-2$, $v-p_i \notin \mathcal{C}$.}
\end{equation}

To see this, suppose, if possible, that for some $i\leq l-2$, $v-p_i \in \mathcal{C}$. As argued above, $\tau(p_i) < \tau(u)$, and by the hypothesis for this case, we have $\tau(u) < \tau(v)$. Then, from \cref{obs:LBGS-gives-PEO}, we must have $p_i- u \in \mathcal{C}$. Since $i \leq l - 2$, this constitutes a chord of $Q_1$ and contradicts the fact that $Q_1$ is a \cp{}. This establishes \cref{claim:v-pi-not-in-C}.

Due to \cref{claim:v-pi-not-in-C}, only the following two options are now allowed: either (i) for each $1 \leq i\leq l-1$, $v-p_i \notin \mathcal{C}$, or (ii) $v-p_{l-1} \in \mathcal{C}$ and for each $1 \leq i\leq l-2$, $v-p_i \notin \mathcal{C}$.

If the first option holds, then $(p_1=a, p_2 = b, \ldots, p_{l-1}, p_l=u, v)$ is a \cp{} (because $Q_1$ is a \cp{}) from $(a, b)$ to $(u, v)$ with $(a, b) \neq (u, v)$ ($a = u$ is precluded by $Q_1$ being a path), and this implies that $u\rightarrow v$ obeys \cref{item-obs:when-u-v-becomes-directed:u-v-is-part-of-a-directed-cp} of \cref{obs:when-u-v-becomes-directed}.

In case the second option holds, $Q_3 = (p_1=a, p_2 = b, \ldots, p_{l-1}, v)$ is a \cp{} from $(a,b)$ to $v$. Since $\tau(p_2) = \tau(b) < \tau(u)$ (as argued above) and $\tau(u) < \tau(v)$ (by the hypothesis for Case 1), we also have $b \neq v$. 

Now, we claim that
\begin{equation}
  \label{claim:Q3-is-a-cp}
  Q_4 = (v, q_1 = u, q_2 =x, \ldots, q_{m} = a) \text{ is a \cp{} in } \mathcal{C} \text{ from $(v, u)$ to $a$}.
\end{equation}

By the hypothesis of this case, $Q_2' = (v, u, x)$ is a \cp{} in $\mathcal{C}$ (remember $x-u-v$ is an induced subgraph of $\mathcal{C}$ in this case). From \cref{lem:concatenation-of-chordal-paths}, $Q_4 = (v, q_1 = u, q_2 = x, q_3, \ldots, q_m = a)$, concatenation of $Q_2'$ and $Q_2$, is a \cp{} in $\mathcal{C}$. Also, $Q_2$ being a path precludes $u = a$. 

We thus obtain a \cp{} $Q_4$ from $(v,u)$ to $a$ (with $u \neq a$). Along with the \cp{} $Q_3$ from $(a, b)$ to $v$ (with $b \neq v$ and $a\rightarrow b \in O$), this implies that $u\rightarrow v$ obeys \cref{item-obs:when-u-v-becomes-directed:u-v-is-part-of-a-directed-cycle} of \cref{obs:when-u-v-becomes-directed}.

  \end{enumerate}
  
We now move to the second case.

\item 
Suppose that during the iteration, $u-v \in M$ is replaced by $u\rightarrow v$ because there exists an induced subgraph $u-x-v-u \in \mathcal{C}$ such that $\tau(u) < \tau(x) < \tau(v)$, $u\rightarrow x, x\rightarrow v \in M$ and $u-v \in M$.  Again, by the induction hypothesis, we have that $u\rightarrow x \in M$ because one of the following holds: (i) $u\rightarrow x \in O$, or (ii) $\exists a\rightarrow b \in O$ such that there exists a \cp{} $(p_1=a, p_2 =b, \ldots, p_{l-1}=u, p_l=x)$ from $(a,b)$ to $(u,x)$ in $\mathcal{C}$ (with $(a, b) \neq (u, x)$), or (iii) $\exists a\rightarrow b \in O$ such that there exists a \cp{} from $(a,b)$ to $v$, and a \cp{} from $(v,x)$ to $a$ in $\mathcal{C}$ (with $b \neq v$ and $x \neq a$). We again go through each possibility.

\begin{enumerate}
    \item Suppose $u\rightarrow x \in M$ because
$u\rightarrow x \in O$. Again, by the induction hypothesis, we have that $x\rightarrow v \in M$
because one of the following holds: (i) $x\rightarrow v \in O$, or
(ii) $\exists c\rightarrow d \in O$ such
that there exists a \cp{} $(q_1 = c, q_2 = d, \ldots, q_{m-1} = x, q_{m} = v)$
from $(c,d)$ to $(x,v)$ in $\mathcal{C}$ (with $(c, d) \neq (x, v)$), or
(iii) $\exists c\rightarrow d \in O$ such that there exists a \cp{} from $(c,d)$
to $v$, and a \cp{} from $(v,x)$ to $c$ in $\mathcal{C}$ (with $d \neq v$ and
$x \neq c$).  We go through each sub-case.

\begin{enumerate}
    \item Let's assume that $x\rightarrow v \in M$ due to the fact that $x\rightarrow v \in O$. The presence of both $u\rightarrow x$ and $x\rightarrow v$ in $O$ implies that $u$, $v$, and $x$ are all elements of $V_O$. 

At the start of the iteration, we have $u-v \in M$. This implies even after the
initialization step (line-\ref{algconstructMEC:M-initialization}),  we have $u-v \in M$. From the construction of $M$, $\skel{M[V_O]} = \skel{O}$. 
That means $u-v \in \skel{O}$. Since $O$ is an extension of $(O_1, P_{11}, P_{12}, O_2, P_{21}, P_{22})$, by \cref{def:extension-of-O1-O2-P11-P12-P21-P22}, $O$ is a partial MEC. Then, from \cref{item-1-of-def:partial-MEC} of \cref{def:partial-MEC}, $O$ is a chain graph. This implies $u\rightarrow v \in O$ (or else $O$ would contain a directed cycle $(u,w,v,u)$). 

But, then, from the initialization of $M$ (remember $M$ is initialized with $U_M(M_1, M_2, O)$), $u\rightarrow v \in M$, a contradiction. This contradiction implies that this particular sub-case cannot occur. We will now proceed to examine the remaining sub-cases.

\item Let's assume that $x\rightarrow v \in M$ because there exists $c\rightarrow d \in O$ such that there is a chordless path $Q = (p_1 = c, p_2 = d, \ldots, p_{l-2}, p_{l-1} = x, p_{l} = v)$ from $(c,d)$ to $(x,v)$ in $\mathcal{C}$ (with $(c, d) \neq (x, v)$). It's worth noting that, due to the hypothesis of this case, $u$ and $v$ are adjacent, while $c$ and $v$ are not (since $Q$ is a chordless path). This implies $c \neq u$.

According to \cref{corr:directed-chordless-path-and-LBFS-ordering}, we have $\tau(c) = \tau(p_1) < \tau(p_2) < \ldots < \tau(p_{l-1}) = \tau(x) < \tau(p_l) = \tau(v)$. Since $\tau(p_{l-2}) < \tau(x)$, and $\tau(u) < \tau(x)$ (as per the hypothesis of this case), \cref{obs:LBGS-gives-PEO} implies that $u-p_{l-2} \in \mathcal{C}$ (since $x$ is adjacent to both $u$ and $p_{l-2}$ in $\mathcal{C}$).

Let $1 \leq i \leq l - 2$ be the smallest number such that $u-p_i \in \mathcal{C}$. If $i>1$, then, as $Q$ is a chordless path, $Q' = (p_1 = c, p_2 = d, \ldots, p_i, u, v)$ is a chordless path in $\mathcal{C}$ from $(c,d)$ to $(u,v)$ with $(c, d) \neq (u, v)$ (since $c \neq u$). Since $c\rightarrow d \in O$, this implies that $u\rightarrow v$ satisfies \cref{item-obs:when-u-v-becomes-directed:u-v-is-part-of-a-directed-cp} of \cref{obs:when-u-v-becomes-directed}.

If $i = 1$, then $Q_1 = (v,u,c)$ is a chordless path from $(v,u)$ to $c$ ($c$ is not adjacent to $v$ in $\mathcal{C}$ because $Q$ is a \cp{}). Since $Q= (p_1 = c, p_2 = d, \ldots, p_{l-1} = x, p_{l} = v)$ is a chordless path from $(c,d)$ to $v$ (with $d \neq v$) and $Q_1$ is a chordless path from $(v,u)$ to $c$ (with $c \neq u$), $u\rightarrow v$ satisfies \cref{item-obs:when-u-v-becomes-directed:u-v-is-part-of-a-directed-cycle} of \cref{obs:when-u-v-becomes-directed}.

Therefore, in this subcase, $u\rightarrow v$ complies with \cref{obs:when-u-v-becomes-directed}. We will now consider the remaining sub-case.

  \item 
    Suppose $x \rightarrow v \in M$ because $\exists c \rightarrow d \in O$ such that there exists a \cp{} $Q_1= (q_1 = c, q_2 =d, \ldots, q_m = v)$ from $(c,d)$ to $v$ in $\mathcal{C}$ (in which $d \neq v$), and a \cp{} $Q_2 = (p_1 = v, p_2 = x, \ldots, p_l=c)$ from $(v,x)$ to $c$ in $\mathcal{C}$ (in which $c \neq x$). Note that $c$ is not adjacent to $v$ (since $Q_1$ is a \cp{}), and $u$ is adjacent to $v$ (by the hypothesis of this case). We thus also have $u \neq c$.

We now claim that there exists a \cp{} $Q$ from $(v,u)$ to $c$ in $\mathcal{C}$. This further implies that $u \rightarrow v$ obeys \cref{item-obs:when-u-v-becomes-directed:u-v-is-part-of-a-directed-cycle} of \cref{obs:when-u-v-becomes-directed} due to the existence of \cps{} $Q_1$ from $(c,d)$ to $v$ (with $d \neq v$) and $Q$ from $(v,u)$ to $c$ (with $u \neq c$) such that $c \rightarrow d \in O$.

\begin{proof}[Proof of the claim]
    There are two possibilities: either $\tau(p_3) < \tau(p_2)$, or $\tau(p_3) > \tau(p_2)$. We go through both possibilities and show that in each possibility there exists a \cp{} from $(v,u)$ to $c$.

    Suppose $\tau(p_3) < \tau(p_2 = x)$. Since we also have $\tau(u) < \tau(x)$ (by the hypothesis of this case), and since $x$ is adjacent to both $u$ and $p_3$, \cref{obs:LBGS-gives-PEO} implies that $u-p_3 \in \mathcal{C}$. Pick the highest $i$ such that $u-p_i \in \mathcal{C}$. Since $u-p_3 \in \mathcal{C}$, $i$ must be greater than 2. Then, $Q = (v, u, p_i, \ldots, p_l = c)$ is a \cp{} from $(v,u)$ to $c$ (since $v$ is adjacent to $u$ and $Q_1$ is a \cp{}). We now move to the other possibility.

    Suppose $\tau(p_3) > \tau(p_2 =x)$. Since $Q_2$ is a \cp{} in $\mathcal{C}$, \cref{obs:chordless-path-and-LBFS-ordering} implies that $\tau(p_2 = x) < \tau(p_3) < \ldots < \tau(p_l = c)$, i.e., $\tau(x) < \tau(c)$. From the hypothesis of this case, $\tau(u) < \tau(x)$. This implies that $\tau(u) < \tau(c)$. We now claim that $u$ must be adjacent to every vertex in the \cp{} $Q_1 = (q_1 = c, q_2 = d, \ldots, q_m = v)$, and in particular to $c$. Since $Q_1$ is a \cp{} and $c \rightarrow d \in O$, therefore, from \cref{corr:directed-chordless-path-and-LBFS-ordering}, $\tau(q_1) = \tau(c) < \tau(q_2) = \tau(d) < \ldots < \tau(q_{m-1}) < \tau(q_m) = \tau(v)$. Since $\tau(u) < \tau(c)$, we have $\tau(u) < \tau(q_i) < \tau(q_{i+1})$ for every $1 \leq i \leq m-1$. Suppose if possible that there exists an $i$ such that $q_i-u \notin \mathcal{C}$, then pick the highest such $i$. Since $u$ is adjacent to $v = q_m$ in $\mathcal{C}$ (by the hypothesis of this case), we must have $i \leq m-1$. Further, by the choice of $i$, we must have that $u$ is adjacent to $q_{i+1}$. But, then, since $\tau(q_i) < \tau(q_{i+1})$ and $\tau(u) < \tau(q_{i+1})$, \cref{obs:LBGS-gives-PEO} implies that $q_i-u \in \mathcal{C}$ (since $q_i$ and $u$ both are adjacent to $q_{i+1}$), which is a contradiction to the choice of $i$. Thus, $u$ is adjacent to every vertex in $Q_1$, and in particular to $c$. This implies that $Q = (c, u, v)$ is a \cp{} in $\mathcal{C}$ (as argued above $v$ and $c$ are not adjacent and $u \neq c$).

    Thus, we show that in both possibilities there exists a \cp{} $Q$ from $(v,u)$ to $c$.
\end{proof}
\end{enumerate}
We now move to the next possibility.
\item Let's assume that $u\rightarrow x \in M$ because there exists an edge $a-b \in \mathcal{C}$ such that $a\rightarrow b \in O$, and there exists a \cp{} $Q = (p_1 =a, p_2 =b, \ldots, p_{l-1} = u, p_l = x)$ from $(a,b)$ to $(u,x)$ in $\mathcal{C}$ (with $(a, b) \neq (u, x)$). Note that we trivially also have $a \neq u$, since $Q$ is a path.

Now, by \cref{corr:directed-chordless-path-and-LBFS-ordering} applied to $Q$, we also have $\tau(p_1) = \tau(a) < \tau(p_2) = \tau(b) < \ldots < \tau(p_{l-1}) < \tau(p_l)$.

We now claim that for each $1\leq i \leq l-2$, $v-p_i \notin \mathcal{C}$. To see this, suppose, if possible, that $v$ is adjacent to $p_i$ for some $1 \leq i \leq l - 2$. Then, since $\tau(p_i) < \tau(x)$, $\tau(x) < \tau(v)$ (by the hypothesis for Case 2), and since $v$ is adjacent to both $x$ and $p_i$, \cref{obs:LBGS-gives-PEO} would imply that $x-p_i \in \mathcal{C}$. But this would be a chord in the \cp{} $Q$, thereby leading to a contradiction. We conclude that for each $1\leq i \leq l-2$, $v-p_i \notin \mathcal{C}$. This implies that $(p_1 =a, p_2 =b, \ldots, p_{l-1} = u, v)$ is a \cp{} from $(a, b)$ to $(u, v)$ with $(a, b) \neq (u, v)$ (since $a \neq u$). 

This further implies that in this case, $u\rightarrow v$ adheres to \cref{item-obs:when-u-v-becomes-directed:u-v-is-part-of-a-directed-cp} of \cref{obs:when-u-v-becomes-directed}.

\item Suppose that $u\rightarrow x \in M$ because there exists an
edge $a-b \in \mathcal{C}$ such that $a\rightarrow b \in O$, and there exists a
\cp{} $Q_1 = (p_1 =a, p_2 =b, \ldots, p_l = x)$ from $(a,b)$ to $x$ in $\mathcal{C}$
(with $b \neq x)$, and a \cp{}
$Q_2 = (q_1 =x, q_2 =u, \ldots, q_m = a)$ from $(x,u)$ to $a$ in
$\mathcal{C}$ (with $u \neq a)$.  By \cref{corr:directed-chordless-path-and-LBFS-ordering} applied to
$Q_1$, we have $\tau(p_1) = \tau(a) < \tau(p_2) < \ldots < \tau(p_{l-1}) < \tau(p_l) = \tau(x)$.
 Since
$\tau(q_m) = \tau(a) < \tau(x) = \tau(q_1)$,
\cref{obs:chordless-path-and-LBFS-ordering} applied to the LBFS ordering $\tau$
and the \cp{} $Q_2$ of the undirected chordal graph $\mathcal{C}$ then also
implies that we must have $\tau(q_i) < \tau(x)$ for every $2 \leq i \leq m$. By the hypothesis of Case 2, $\tau(v) < \tau(x)$. Therefore, for all $i$, $\tau(q_i) < \tau(v)$. We now claim that 
\begin{equation}
\label{claim:v-qi-not-in-C}
\text{For $3 \leq i\leq m$, $v-q_i \notin \mathcal{C}$.}
\end{equation}

To see this, suppose, if possible, that for some $3 \leq i\leq m$, $v-q_i \in \mathcal{C}$. As argued above, we have $\tau(q_i) < \tau(v)$ and $\tau(q_1) < \tau(v)$. Therefore, from \cref{obs:LBGS-gives-PEO}, $q_1-q_i \in \mathcal{C}$ (as both $q_1$ and $q_i$ are adjacent to $v$). Since $3 \leq i\leq m$, this constitutes a chord of $Q_2$ and thus contradicts that $Q_2$ is a \cp{}. This establishes \cref{claim:v-qi-not-in-C}. 
Due to \cref{claim:v-qi-not-in-C}, we have that $Q_3 = (v, q_2 = u, q_3, \ldots, q_m =a)$ (which we get by replacing $q_1$ with $v$ in $Q_2$) is a \cp{} in $\mathcal{C}$.

We now show the existence of a \cp{} from $(a,b)$ to $v$ in $\mathcal{C}$. Since $Q_1$ is a \cp{}, and $a\rightarrow b \in O$, from \cref{corr:directed-chordless-path-and-LBFS-ordering}, $\tau(p_1 = a) < \tau(p_2 =b) <\ldots < \tau(p_l=x)$. 
Suppose for some $i\leq l-2$, $v-p_i \in \mathcal{C}$. Then, from \cref{obs:LBGS-gives-PEO}, $p_i-p_l \in \mathcal{C}$, as $\tau(p_i) < \tau(p_l = x) < \tau(v)$ (from our hypothesis of this case, $\tau(x) < \tau(v)$), which further  implies that $Q_1$ is not a \cp{}, which is a contradiction. This  implies that for all $i\leq l-2$, $v-p_i \notin \mathcal{C}$.
Thus, we get that either
$Q_4 \coloneqq (p_1 =a, p_2 =b, \ldots, p_{l-1},v)$ is a \cp{} from $(a,b)$ to $v$
(if $p_{l-1}-v \in \mathcal{C}$), or else
$Q_4 = (p_1 =a, p_2 =b, \ldots, p_{l},v)$ is a \cp{} from $(a,b)$ to $v$ (if
$p_{l-1}-v \notin \mathcal{C}$).  Note that in either case, we also have
$b \neq v$ since $\tau(b) = \tau(p_2) < \tau(x)$ and also $\tau(x) < \tau(v)$
(by the hypothesis for Case 2).

Existence of the \cps{} $Q_3$ and $Q_4$ implies that $u\rightarrow v$ obeys
\cref{item-obs:when-u-v-becomes-directed:u-v-is-part-of-a-directed-cycle} of
\cref{obs:when-u-v-becomes-directed}.
\end{enumerate}  
This concludes the second case.
  \end{enumerate}

We have thus shown that in all cases $u\rightarrow v$ obeys
\cref{obs:when-u-v-becomes-directed}. This completes the induction, and hence
also the proof of \cref{obs:when-u-v-becomes-directed}.
\end{proof}

\Cref{obs:when-u-v-becomes-directed} implies the following corollary.
\begin{corollary}
    \label{corr:an-edge-of-an-ucc-gets-directed-in-M-either-at-initiazation-or-in-the-iteration-when-ucc-is-considered}
    For $i\in \{1,2\}$, let $\mathcal{C}$ be an \ucc{} of $M_i$. 
    Let $\tau$ be the LBFS ordering returned at line \ref{algconstructMEC:contructLBFS} of \cref{alg:constructMEC} in the iteration of the second foreach loop (lines \ref{algconstructMEC:foreachUCC}-\ref{algconstructMEC:foreachUCC-end} of \cref{alg:constructMEC}) when  $\mathcal{C}$ is considered. 
    For an undirected edge $u-v$ in  $\mathcal{C}$, if $u\rightarrow v \in M$ then either $u\rightarrow v$ is added to $M$ in the initialization step (line \ref{algconstructMEC:M-initialization} of \cref{alg:constructMEC}) or in the second foreach loop (lines \ref{algconstructMEC:foreachUCC}-\ref{algconstructMEC:foreachUCC-end} of \cref{alg:constructMEC}) when  $\mathcal{C}$ is considered, in the iteration of the first while loop (lines \ref{algconstructMEC:while-start}--\ref{algconstructMEC:while-end} of \cref{alg:constructMEC}) when $a = \tau(v)$ and in the iteration of the second while loop (lines \ref{algconstructMEC:inner-while}--\ref{algconstructMEC:inner-while-end} of \cref{alg:constructMEC}) when $b = \tau(u)$.
\end{corollary}
\begin{proof}
Without loss of generality, let's assume $i = 1$, i.e., $\mathcal{C}$ is an \ucc{} of $M_1$. 
From the construction of $M$, it's evident that if $u\rightarrow v  \in M$, then it's either added at the initialization step (line \ref{algconstructMEC:M-initialization} of \cref{alg:constructMEC}), or $u-v$ is in some \ucc{} $\mathcal{C}'$ of $M_1$ or $M_2$, and it gets directed at line \ref{algconstructMEC:add-orientation} when $\mathcal{C}'$ is considered in the second foreach loop (lines \ref{algconstructMEC:foreachUCC}-\ref{algconstructMEC:foreachUCC-end} of \cref{alg:constructMEC}). 

If $u\rightarrow v$ is added at the initialization step, then we are done. Suppose $u\rightarrow v$ is not added at the initialization step; we'll now show that $\mathcal{C} = \mathcal{C}'$. 

Suppose $\mathcal{C} \neq \mathcal{C}'$. Then, $\mathcal{C}'$ must be an \ucc{} of $M_2$, as two different \ucc{} of a graph cannot share a node, and both $\mathcal{C}$ and $\mathcal{C}'$ contain $u-v$. This implies $u, v  \in V_{M_1} \cap V_{M_2} = I \subseteq V_O$. 

Since $u-v \in \mathcal{C}'$ and $u\rightarrow v  \in M$, $u\rightarrow v$ must obey \cref{obs:when-u-v-becomes-directed}. We now show that this implies $u\rightarrow v \in O$.

If \cref{item-obs:when-u-v-becomes-directed:u-v-is-in-O} of \cref{obs:when-u-v-becomes-directed} holds, then we have $u \rightarrow v \in O$. If one of \cref{item-obs:when-u-v-becomes-directed:u-v-is-part-of-a-directed-cp} or \cref{item-obs:when-u-v-becomes-directed:u-v-is-part-of-a-directed-cycle} of \cref{obs:when-u-v-becomes-directed} holds, then in the \uccc{} $\mathcal{C}'$ of $M_2$, there exists an edge $x-y \in \mathcal{C}'$ such that $x\rightarrow y \in O$, and either (a) $(x,y) \neq (u,v)$ and there exists a \cp{} from $(x,y)$ to $(u,v)$ in $\mathcal{C}'$ (due to \cref{item-obs:when-u-v-becomes-directed:u-v-is-part-of-a-directed-cp} of \cref{obs:when-u-v-becomes-directed}), or (b) $y\neq v$, $u\neq x$, and there exist \cps{} $Q_1$ and $Q_2$ in $\mathcal{C}'$ such that $Q_1$ is a \cp{} from $(x,y)$ to $v$, and $Q_2$ is a \cp{} from $(v,u)$ to $x$ (due to  \cref{item-obs:when-u-v-becomes-directed:u-v-is-part-of-a-directed-cycle} of \cref{obs:when-u-v-becomes-directed}). 

Since $\mathcal{C}'$ is an \uccc{} of $M_2$ (from \cref{item-2-theorem-nec-suf-cond-for-MEC} of \cref{thm:nes-and-suf-cond-for-chordal-graph-to-be-an-MEC}), a \cp{} in $\mathcal{C}'$ is a \cp{} of $M_2$. And, from \cref{obs:cp-is-tfp}, every \cp{} in $\mathcal{C}'$ is a \tfp{}. This implies \cps{} in $\mathcal{C}'$ are \tfps{} of $M_2$. Also, since $(x,y) \in \mathcal{C}'$, an \ucc{} of $M_2$, and $x\rightarrow y \in O$, this implies $x,y \in V_{M_2} \cap V_O = V_{O_2}$. 

Since $(O_2, P_{21}, P_{22})$ is the shadow of $M_2$ on $V_{O_2}$, from \cref{def:shadow}, $O_2$ is an induced subgraph of $M_2$. This implies $(x,y)\in E_{O_2}$, and if there exists a \cp{} in $\mathcal{C}'$, from $(x,y)$ to $(u,v)$, then $P_{21}((x,y),(u,v)) =1$, and if there exist \cps{} in $\mathcal{C}'$, from $(x,y)$ to $v$ and from $(v,u)$ to $x$ then $P_{22}((x,y),v) = P_{22}((v,u),x) = 1$. 

Since $O \in \mathcal{E}(O_1, P_{11}, P_{12}, O_2, P_{21}, P_{22})$, therefore, in both possibilities, from \cref{item-3-of-def:extension-of-O1-O2-P11-P12-P21-P22} of \cref{def:extension-of-O1-O2-P11-P12-P21-P22}, we again have $u\rightarrow v \in O$. 

But, if $u\rightarrow v \in O$, then $u\rightarrow v$ has been added to $M$ at the initialization step, contradicting our assumption. This implies $\mathcal{C} = \mathcal{C}'$. This further implies $u\rightarrow v$ is added to $M$ at the iteration of the second foreach loop (lines \ref{algconstructMEC:foreachMi}--\ref{algconstructMEC:foreachMi-end}) when $\mathcal{C}$ is considered. It is obvious from the algorithm that in that iteration, $u\rightarrow v$ is added to $M$ when $a = \tau(v)$ and $b = \tau(u)$. This completes the proof of \cref{corr:an-edge-of-an-ucc-gets-directed-in-M-either-at-initiazation-or-in-the-iteration-when-ucc-is-considered}.
\end{proof}

The following corollary is implied by \cref{corr:directed-chordless-path-and-LBFS-ordering,corr:directed-edges-of-O-follow-tau} when applied to \cref{obs:when-u-v-becomes-directed}.

\begin{corollary}
\label{corr-of-obs:when-u-v-becomes-directed}
For $i\in \{1,2\}$, and any \ucc{} $\mathcal{C}$ of $M_i$ at line-\ref{algconstructMEC:contructLBFS} of \cref{alg:constructMEC}, let $\tau$ be the LBFS ordering returned by \cref{alg:LBFSwithO}. Then, for $u-v \in \mathcal{C}$, at any iteration of the while loop (lines \ref{algconstructMEC:while-start}-\ref{algconstructMEC:while-end}) of \cref{alg:constructMEC}, if $u\rightarrow v \in M$, then either of the following occurs:

\begin{enumerate}
    \item \label{item-corr-of-obs:when-u-v-becomes-directed:u-v-is-in-O}
    $u \rightarrow v$ is strongly protected in $O$, and $\tau(u) < \tau(v)$, or
    
    \item \label{item-corr-of-obs:when-u-v-becomes-directed:u-v-is-part-of-a-directed-cp}
    $\exists a\rightarrow b \in O$ such that there exists a \cp{} $Q = (p_1 = a, p_2 = b, \ldots, p_{l-1} = u, p_l = v)$ from $(a,b)$ to $(u,v)$ in $\mathcal{C}$ such that $(a, b) \neq (u, v)$ and $\tau(p_1) < \tau(p_2) < \ldots < \tau(p_{l-1}) < \tau(p_l)$. In particular, $\tau(u) < \tau(v)$.
    
    \item \label{item-corr-of-obs:when-u-v-becomes-directed:u-v-is-part-of-a-directed-cycle}
    $\exists a\rightarrow b \in O$ such that there exists a \cps{} $Q_1$ and $Q_2$ such that $Q_1 = (p_1 = a, p_2 = b, \ldots, p_{l-1}, p_l = v)$ is a \cp{} from $(a,b)$ to $v$ in $\mathcal{C}$ and $Q_2$ is a \cp{} from $(v,u)$ to $a$ in $\mathcal{C}$ such that $b \neq v$ and $u \neq a$, and $\tau(p_1 = a) < \tau(p_2 = b) < \ldots < \tau(p_{l-1}) < \tau(p_l = v)$, and $\tau(z) < \tau(v)$ for every node $z \neq v$ in $Q_2$ (in particular, $\tau(u) < \tau(v)$).
\end{enumerate}
\end{corollary}

\begin{proof}
We prove the relative ordering of nodes (mentioned in \cref{item-corr-of-obs:when-u-v-becomes-directed:u-v-is-in-O,item-corr-of-obs:when-u-v-becomes-directed:u-v-is-part-of-a-directed-cp,item-corr-of-obs:when-u-v-becomes-directed:u-v-is-part-of-a-directed-cycle} of \cref{corr-of-obs:when-u-v-becomes-directed}) at the end of the proof. We first prove the remaining parts of the corollary.
Let for $i\in \{1,2\}$, $\mathcal{C}$ be an \ucc{} of $M_i$, and $u-v\in \mathcal{C}$.
From \cref{obs:when-u-v-becomes-directed}, if $u\rightarrow v \in M$ then either (a) $u \rightarrow v \in O$ (\cref{item-obs:when-u-v-becomes-directed:u-v-is-in-O} of \cref{obs:when-u-v-becomes-directed}), or (b) there exists an edge $a-b \in \mathcal{C}$ such that there exists a \cp{} from $(a,b)$ to (u,v) in $\mathcal{C}$ (\cref{item-obs:when-u-v-becomes-directed:u-v-is-part-of-a-directed-cp} of \cref{obs:when-u-v-becomes-directed}), or (c) there exists an edge $a-b \in \mathcal{C}$ such that there exist \cps{} in $\mathcal{C}$ from $(a,b)$ to $v$ and from $(v,u)$ to $a$ (\cref{item-obs:when-u-v-becomes-directed:u-v-is-part-of-a-directed-cycle} of \cref{obs:when-u-v-becomes-directed}). 

Suppose $u-v\in \mathcal{C}$, and $u\rightarrow v \in M$ because it obeys \cref{item-obs:when-u-v-becomes-directed:u-v-is-in-O} of \cref{obs:when-u-v-becomes-directed}, i.e.,  $u\rightarrow v \in O$.
This implies $u,v \in V_{M_i} \cap V_O = I \subseteq V_{O_i}$. As $(O_i, P_{i1}, P_{i2})$ is the shadow of $M_i$ on $V_{O_i}$, therefore, from \cref{item-1-of-def:shadow} of \cref{def:shadow}, $O_i$ is an induced subgraph of $M_i$. Since $u,v \in V_{O_i}$ and $u-v \in \mathcal{C}$, an \ucc{} of $M_i$, therefore, $u-v \in O_i$.
Since $O\in \mathcal{E}(O_1, P_{11}, P_{12}, O_2, P_{21}, P_{22})$, therefore, $O$ obeys
\cref{def:extension-of-O1-O2-P11-P12-P21-P22}.
\Cref{item-3-of-def:extension-of-O1-O2-P11-P12-P21-P22}
  of \cref{def:extension-of-O1-O2-P11-P12-P21-P22} implies that if
  $u\rightarrow v \in O$ then either 
  (i) $u\rightarrow v$ is strongly protected in $O$, or 
  (ii) there exists an edge $a-b \in O_i$ such that $P_{i1}((a,b),(u,v)) =1$, or
  (iii) there exists an edge $a-b \in O_i$ such that $P_{i2}((a,b),v) = P_{i2}((v,u),a) = 1$.
  
  If $u\rightarrow v$ is strongly protected in $O$ then it obeys \cref{item-corr-of-obs:when-u-v-becomes-directed:u-v-is-in-O} of \cref{corr-of-obs:when-u-v-becomes-directed}. 
  
  Suppose $u\rightarrow v \in O$ and  there exists an edge $a-b \in O_i$ such that $P_{i1}((a,b),(u,v)) =1$. We first show that $a-b$ and $u-v$ are in the same \ucc{} $\mathcal{C}$ of $M_i$. Since $P_{i1}((a,b),(u,v)) =1$ this implies $(a,b) \neq (u,v)$ and there exists a \tfp{} $Q = (q_1 =a, q_2 =b, \ldots, q_{l-1} = u, q_l = v)$ from $(a,b)$ to $(u,v)$ in $M_i$.
  We claim that all the edges of $P$ are undirected. If this is not true then pick a directed edge $u_j \rightarrow u_{j+1}$ of $P$ such that $j$ is maximum. Since $u-v$ is an undirected edge in $M_i$, therefore, $j < l-1$. But, then $u_j\rightarrow u_{j+1} - u_{j+2}$ is an induced subgraph of $M_i$, contradicting \cref{item-3-theorem-nec-suf-cond-for-MEC} of \cref{thm:nes-and-suf-cond-for-chordal-graph-to-be-an-MEC} (remember $M_i$ is an MEC). This validates our claim that all the edges of $Q$ are undirected. This implies $Q$ is a \tfp{} in the \ucc{} $\mathcal{C}$ of $M_i$. From \cref{item-2-theorem-nec-suf-cond-for-MEC} of \cref{thm:nes-and-suf-cond-for-chordal-graph-to-be-an-MEC}, $\mathcal{C}$ is a chordal graph. Therefore, from \cref{obv:tfps-are-chordless-in-chordal-graphs}, $Q$ is a \cp{} in $\mathcal{C}$. This implies $u\rightarrow v$ obeys \cref{item-corr-of-obs:when-u-v-becomes-directed:u-v-is-part-of-a-directed-cp} of \cref{corr-of-obs:when-u-v-becomes-directed}.

  Suppose $u\rightarrow v \in O$ because  there exists an edge $a-b \in O_i$ such that $P_{i2}((a,b),v) = P_{i2}((v,u),a) = 1$. This implies there exist \tfps{} $Q_1$ and $Q_2$ in $M_i$ such that $Q_1$ is a \tfp{} from $(a,b)$ to $v$, and $Q_2$ is a \tfp{} from $(v,u)$ to $a$. 
All the edges in $Q_1$ and $Q_2$ must be undirected, otherwise, by concatenating $Q_1$ and $Q_2$ we get a directed cycle in $M_i$, contradicting \cref{item-1-theorem-nec-suf-cond-for-MEC} of \cref{thm:nes-and-suf-cond-for-chordal-graph-to-be-an-MEC}, as $M_i$ is an MEC. This implies $Q_1$ and $Q_2$ are \tfps{} in $\mathcal{C}$, an \ucc{} of $M_i$.   
From \cref{item-2-theorem-nec-suf-cond-for-MEC} of \cref{thm:nes-and-suf-cond-for-chordal-graph-to-be-an-MEC}, $\mathcal{C}$ is a chordal graph. Therefore, from \cref{obv:tfps-are-chordless-in-chordal-graphs}, $Q_1$ and $Q_2$ are \cps{} in $\mathcal{C}$. This implies $u\rightarrow v$ obeys \cref{item-corr-of-obs:when-u-v-becomes-directed:u-v-is-part-of-a-directed-cycle} of \cref{corr-of-obs:when-u-v-becomes-directed}.

Suppose $u\rightarrow v \in M$ because it obeys \cref{item-obs:when-u-v-becomes-directed:u-v-is-part-of-a-directed-cp} of \cref{obs:when-u-v-becomes-directed}. Then, $u\rightarrow v$ obeys \cref{item-corr-of-obs:when-u-v-becomes-directed:u-v-is-part-of-a-directed-cp} of \cref{corr-of-obs:when-u-v-becomes-directed}. Similarly, if $u\rightarrow v \in M$ because it obeys \cref{item-obs:when-u-v-becomes-directed:u-v-is-part-of-a-directed-cycle} of \cref{obs:when-u-v-becomes-directed}. Then, $u\rightarrow v$ obeys \cref{item-corr-of-obs:when-u-v-becomes-directed:u-v-is-part-of-a-directed-cycle} of \cref{corr-of-obs:when-u-v-becomes-directed}. 
  
  The ranking of the nodes comes from
  \cref{obs:chordless-path-and-LBFS-ordering,corr:directed-edges-of-O-follow-tau}.
  If $u\rightarrow v \in O$ then from
  \cref{corr:directed-edges-of-O-follow-tau}, $\tau(u) < \tau(v)$.
  In the \cp{} of
  \cref{item-corr-of-obs:when-u-v-becomes-directed:u-v-is-part-of-a-directed-cp}
  of \cref{corr-of-obs:when-u-v-becomes-directed}, $a\rightarrow b \in O$. From
  \cref{corr:directed-edges-of-O-follow-tau}, $\tau(a) < \tau(b)$. Then, from
  \cref{obs:chordless-path-and-LBFS-ordering}, the nodes of the \cp{} follows
  the order: $\tau(p_1=a) < \tau(p_2=b) < \ldots < \tau(p_{l-1}) <
  \tau(p_l)$. Similarly, the nodes in the \cp{} $P_1$ of
  \cref{item-corr-of-obs:when-u-v-becomes-directed:u-v-is-part-of-a-directed-cycle} of
  \cref{corr-of-obs:when-u-v-becomes-directed} follow the order:
  $\tau(p_1 = a) < \tau(p_2 = b) < \ldots < \tau(p_{l-1}) < \tau(p_l = v)$. In
  particular, this also implies $\tau(a) < \tau(v)$. Thus,
  \cref{obs:chordless-path-and-LBFS-ordering} implies that for all the nodes
  $z\neq v$ of the \cp{} $P_2$ of
  \cref{item-corr-of-obs:when-u-v-becomes-directed:u-v-is-part-of-a-directed-cycle} of
  \cref{corr-of-obs:when-u-v-becomes-directed}, we have $\tau(z) < \tau(v)$.  This
  completes the proof of \cref{corr-of-obs:when-u-v-becomes-directed}.
\end{proof}

We are now ready to prove \cref{lem:O-is-an-induced-subgraph-of-M}.

\begin{proof}[\textbf{Proof of \cref{lem:O-is-an-induced-subgraph-of-M}}]
  \label{proof-of-obs:edges-of-subMECs-does-not-change-its-orientation}
  
  We initialize $M$ as $U_M(M_1, M_2, O)$ (line~\ref{algconstructMEC:M-initialization} of \cref{alg:constructMEC}). This implies that after the initialization step, if $u\rightarrow v \in O$ then $u\rightarrow v \in M$. Throughout the algorithm, we only convert an undirected edge into a directed one (lines~\ref{algconstructMEC:constraint-for-orientation}-\ref{algconstructMEC:if-end}). This implies that if $u\rightarrow v \in O$ then $u\rightarrow v$ belongs to the graph $M$ returned by the algorithm. To complete the proof, we have to show that even for any undirected edge $u-v \in O$, $M$ returned by \cref{alg:constructMEC} has $u-v$.
  
  Suppose, if possible, that there is an edge $u-v$ of $O$ that is directed as $u \rightarrow v$ in $M$. If $u-v \in O$ then $u-v \in \skel{O}$. From the construction, $\skel{O} = \skel{O_1} \cup \skel{O_2}$. This implies either $u-v \in \skel{O_1}$ or $u-v \in \skel{O_2}$ or both. 
  
  Without loss of generality, let us assume that $u-v \in \skel{O_1}$. Since $O\in \mathcal{E}(O_1, P_{11}, P_{12}, O_2, P_{21}, P_{22})$ and $u-v \in O$, from \cref{item-1-of-def:extension-of-O1-O2-P11-P12-P21-P22} of \cref{def:extension-of-O1-O2-P11-P12-P21-P22}, $u-v \in O_1$. This implies $u-v \in M_1$ as $O_1$ is an induced subgraph of $M_1$ (from \cref{item-1-of-def:shadow} of \cref{def:shadow}, as $(O_1, P_{11}, P_{12})$ is the shadow of $M_1$ on $V_{O_1}$). 
  
  This further implies there exists an \ucc{} $\mathcal{C}$ of $M_1$ that contains $u-v$. Since from our assumption, $u\rightarrow v \in M$ and $u\rightarrow v \notin O$, therefore, from \cref{obs:when-u-v-becomes-directed}, either (a) there exists $a-b \in \mathcal{C}$ such that $a\rightarrow b \in O$ and there exists a \cp{} from $(a,b)$ to $(u,v)$ in $\mathcal{C}$, or (b) there exists $a-b \in \mathcal{C}$ such that there exists \cps{} in $\mathcal{C}$ from $(a,b)$ to $v$ and from $(v,u)$ to $a$. 
  
  In both of the possibilities, $a\rightarrow b \in O$ and $a-b \in M_1$ (remember $\mathcal{C}$ is an \ucc{} of $M_1$). This implies $a,b\in V_{M_1}\cap V_O = I \subseteq V_{O_1}$. This further implies $a-b \in O_1$ as $O_1$ is an induced subgraph of $M_1$. In the first possibility, since there is a \cp{} from $(a,b)$ to $(u,v)$ in $\mathcal{C}$, therefore, there is a \cp{} from $(a,b)$ to $(u,v)$ in $M_1$. 
  
  From \cref{obs:cp-is-tfp}, this further implies that there is a \tfp{} from $(a,b)$ to $(u,v)$ in $M_1$. That means $P_{11}((a,b),(u,v)) = 1$. But, then from \cref{item-3-of-def:extension-of-O1-O2-P11-P12-P21-P22} of \cref{def:extension-of-O1-O2-P11-P12-P21-P22}, $u\rightarrow v \in O$, a contradiction. Similarly, in the second possibility, we have $a-b \in O_1$ and there exist \tfps{} from $(a,b)$ to $v$ and from $(v,u)$ to $a$ in $M_1$. 
  
  This implies $P_{22}((a,b),v) = P_{22}((v,u),a) = 1$. But, then from \cref{item-3-of-def:extension-of-O1-O2-P11-P12-P21-P22} of \cref{def:extension-of-O1-O2-P11-P12-P21-P22}, $u\rightarrow v \in O$, a contradiction. This implies that for any undirected edge $u-v \in O$, $M$ returned by  \cref{alg:constructMEC} has $u-v$. The above discussion implies that $M[V_O] = O$.
\end{proof}

The following observations can be easily derived from \cref{corr-of-obs:when-u-v-becomes-directed}. These observations will be used later, and if preferred, the reader can skip ahead to \cref{lem:M-is-chain-graph} directly.

\begin{observation}
\label{item-1-of-obs:when-u-v-is-directed-reverse}
For $a\in \{1,2\}$, and for any \ucc{} $\mathcal{C}$ of $M_a$, let $\tau$ be the LBFS ordering obtained from \cref{alg:LBFSwithO} when it is called at line~\ref{algconstructMEC:contructLBFS} of \cref{alg:constructMEC}. Then, for $u-v \in \mathcal{C}$, if $u\rightarrow v \in O$ then $u\rightarrow v \in M$, and $\tau(u) < \tau(v)$.
\end{observation}

\begin{proof}
At the initialization step, line \ref{algconstructMEC:M-initialization}, of \cref{alg:constructMEC}, if $u\rightarrow v \in O$, then $u\rightarrow v \in M$. From \cref{item-obs:when-u-v-becomes-directed:u-v-is-in-O} of \cref{corr-of-obs:when-u-v-becomes-directed}, it follows that $\tau(u) < \tau(v)$.
\end{proof}

\begin{observation}
\label{obs:directed-edge-respects-tau}
For $a \in \{1,2\}$, and for any \ucc{} $\mathcal{C}$ of $M_a$, let $\tau$ be the LBFS ordering returned by \cref{alg:LBFSwithO} when it is called at line~\ref{algconstructMEC:contructLBFS} of \cref{alg:constructMEC}. Then, for $u-v \in \mathcal{C}$, if $u\rightarrow v \in M$, then $\tau(u) < \tau(v)$.
\end{observation}

\begin{proof}
Without loss of generality, let us assume $a = 1$. For $u-v \in \mathcal{C}$, if $u\rightarrow v \in M$, then either it got directed at the initialization step (line \ref{algconstructMEC:M-initialization} of \cref{alg:constructMEC}), or it gets directed at lines~\ref{algconstructMEC:constraint-for-orientation}-\ref{algconstructMEC:if-end} of \cref{alg:constructMEC}.

At the initialization step, $M = U_M(M_1, M_2, O)$. Therefore, if $u\rightarrow v \in M$ and $u\rightarrow v \notin M_1$, then either $u\rightarrow v \in M_2$ or $u\rightarrow v \in O$. If $u\rightarrow v \in M_2$, then $u,v \in V_{M_1} \cap V_{M_2} = I \cup V_{O_2}$. Since $O_2$ is an induced subgraph of $M_2$, therefore $u\rightarrow v \in O_2$.

Since $O$ is an extension of $(O_1, P_{11}, P_{12}, O_2, P_{21}, P_{22})$, from \cref{def:extension-of-O1-O2-P11-P12-P21-P22}, $u\rightarrow v \in O$. This further implies that if $u\rightarrow v \in M$, then either $u\rightarrow v \in O$ or it gets directed during the run of lines~\ref{algconstructMEC:constraint-for-orientation}-\ref{algconstructMEC:if-end} of \cref{alg:constructMEC}.

If $u\rightarrow v \in O$, then from \cref{item-1-of-obs:when-u-v-is-directed-reverse}, $\tau(u) < \tau(v)$. And, if $u\rightarrow v$ is directed during the run of lines~\ref{algconstructMEC:constraint-for-orientation}-\ref{algconstructMEC:if-end} of \cref{alg:constructMEC}, then from \cref{item-corr-of-obs:when-u-v-becomes-directed:u-v-is-part-of-a-directed-cp,item-corr-of-obs:when-u-v-becomes-directed:u-v-is-part-of-a-directed-cycle} of \cref{corr-of-obs:when-u-v-becomes-directed}, $\tau(u) < \tau(v)$.
\end{proof}

\begin{observation}
\label{item-2-of-obs:when-u-v-is-directed-reverse}
For $a \in \{1,2\}$, and for any \ucc{} $\mathcal{C}$ of $M_a$, let $\tau$ be the LBFS ordering obtained from \cref{alg:LBFSwithO}. Then, for $u-v \in \mathcal{C}$, if there exists $a-b \in \mathcal{C}$ such that either $a\rightarrow b \in M$ or $a \rightarrow b \in O$, and there exists a \cp{} $Q = (u_1=a, u_2=b, \ldots, u_{l-1} = u, u_l = v)$ from $(a,b)$ to $(u,v)$ in $\mathcal{C}$ then $\tau(u_1) < \tau(u_2) < \ldots < \tau(u_l)$, and for all $1 \leq i \leq l-1$, $u_i \rightarrow u_{i+1} \in M$. In particular, $u\rightarrow v \in M$.
\end{observation}

\begin{proof}
  Note that when $a\rightarrow b \in O$, from the initialization step (line~\ref{algconstructMEC:M-initialization}) of \cref{alg:LBFSwithO}, we have $a\rightarrow b \in M$. Thus, we can assume that $a \rightarrow b \in M$.
  
  Now, since $Q = (u_1=a, u_2=b, \ldots, u_{l-1} = u, u_l = v)$ is a \cp{} from $(a,b)$ to $(u,v)$ in $\mathcal{C}$, and $a \rightarrow b \in M$, \cref{obs:directed-edge-respects-tau} implies that $\tau(u_1=a) < \tau(u_2 =b)$. Since $Q$ is a \cp{} in the undirected chordal graph $\mathcal{C}$ and $\tau$ is an LBFS ordering of $\mathcal{C}$, \cref{obs:chordless-path-and-LBFS-ordering} implies that we must then have $\tau(u_1) < \tau(u_2) < \ldots < \tau(u_{l-1}) < \tau(u_l)$.
  
  Since $u_1 = a$ and $u_2 = b$, we have $u_1\rightarrow u_2 \in M$. Lines \ref{algconstructMEC:while-start}-\ref{algconstructMEC:while-end} of \cref{alg:constructMEC} show that for $1< i < l$, in the iteration of the while loop at line \ref{algconstructMEC:while-start} when $a = i+1$, before the end of the iteration of the while loop at line \ref{algconstructMEC:inner-while} when $b = i$, we get $u_i \rightarrow u_{i+1} \in M$. Because, if at the initialization step, we do not have $u_i \rightarrow u_{i+1}$ in $M$ then due to the existence of the induced subgraph $u_{i-1} \rightarrow u_i - u_{i+1}$ in $M$, the if condition at line \ref{algconstructMEC:constraint-for-orientation} is satisfied and we get $u_i \rightarrow u_{i+1} \in M$ at line \ref{algconstructMEC:add-orientation}. Thus, we conclude that for all $1 \leq i \leq l-1$, $u_i \rightarrow u_{i+1} \in M$.
\end{proof}

We now prove \cref{lem:M-is-an-MEC}.
\begin{proof}[\textbf{Proof of \cref{lem:M-is-an-MEC}}]
We first show that $M$ is an MEC. Then we show that $\skel{M} = \skel{H}$. Together these two imply that $M$ is an MEC of $H$. 
\Cref{thm:nes-and-suf-cond-for-chordal-graph-to-be-an-MEC} provides necessary and sufficient conditions for an MEC.
\Cref{lem:M-is-chain-graph,lem:UCC-of-M-chordal,lem:a->b-c-does-not-exist,lem:every-directed-edge-of-M-is-strongly-ptotected} validate that $M$ obeys the required conditions mentioned in \cref{thm:nes-and-suf-cond-for-chordal-graph-to-be-an-MEC} to be an MEC. This further implies that $M$ is an MEC.
\begin{lemma}
\label{lem:M-is-chain-graph}
$M$ is a chain graph.
\end{lemma}

\begin{lemma}
\label{lem:UCC-of-M-chordal}
Every undirected connected component of $M$ is chordal.
\end{lemma}

\begin{lemma}
\label{lem:a->b-c-does-not-exist}
$M$ does not have an induced subgraph of the form $a\rightarrow b -c$.
\end{lemma}

\begin{lemma}
\label{lem:every-directed-edge-of-M-is-strongly-ptotected}
Every directed edge $u \rightarrow v \in M$ is strongly protected in $M$.
\end{lemma}

Proofs of \cref{lem:M-is-chain-graph,lem:UCC-of-M-chordal,lem:a->b-c-does-not-exist,lem:every-directed-edge-of-M-is-strongly-ptotected} are given below. 

We now prove that $\skel{M} = H$. From the initialization step (line \ref{algconstructMEC:M-initialization} of \cref{alg:constructMEC}), $M = U_M(M_1, M_2, O)$. The remaining lines of \cref{alg:constructMEC} orient the undirected edges of $M$ and do not change the skeleton of $M$. Therefore, $\skel{M} = \skel{M_1} \cup \skel{M_2} \cup \skel{O}$. 

From the construction, $\skel{M_1} = \skel{H_1}$, $\skel{M_2} = \skel{H_2}$, and $\skel{H} = \skel{H_1} \cup \skel{H_2} \cup \skel{O}$. This implies $\skel{M} = \skel{H}$. Thus, $M$ is an MEC of $H$, as it is an MEC with $H$ as its skeleton.
\end{proof}

We will now prove \cref{lem:M-is-chain-graph,lem:UCC-of-M-chordal,lem:a->b-c-does-not-exist,lem:every-directed-edge-of-M-is-strongly-ptotected}. If preferred, readers may choose to bypass the proofs and proceed directly to \cref{lem:M1-and-M2-are-projections-of-M}, where we demonstrate that $M_1$ and $M_2$ are projections of $M$.

First, we present the proof of \cref{lem:M-is-chain-graph}. The proof of \cref{lem:UCC-of-M-chordal} can be found on page~\pageref{proof-of-lem:UCC-of-M-chordal}, the proof of \cref{lem:a->b-c-does-not-exist} is on page~\pageref{proof-of-lem:a->b-c-does-not-exist}, and the proof of \cref{lem:every-directed-edge-of-M-is-strongly-ptotected} is on page~\pageref{proof-of-lem:every-directed-edge-of-M-is-strongly-ptotected}.

\begin{proof}[\textbf{Proof of \cref{lem:M-is-chain-graph}}]
The following observations, to be proved later, imply that $M$ is a chain graph.

\begin{observation}
\label{obs:induced-graph-of-M-on-M1-is-a-chain-graph}
For $a \in \{1,2\}$, $M[V_{H_a}]$ is a chain graph.
\end{observation}

\begin{observation}
\label{obs:no-chordless-cycle-in-M}
$M$ does not have any chordless cycle of length greater than 3.
\end{observation}

Suppose $M$ is not a chain graph. Then $M$ has a directed cycle. Choose a cycle $C = (u_1, u_2, \ldots, u_l, u_{l+1} = u_1)$ of the smallest length. $l$ must be greater than two.

Suppose $l = 3$. We claim that either all the nodes of the cycle are in $M_1$ (i.e., $u_1, u_2, u_3 \in V_{M_1}$) or all the nodes of the cycle are in $M_2$ (i.e., $u_1, u_2, u_3 \in V_{M_2}$).

Suppose this is not true. Then there exists a node in the cycle that is in $M_1$ but not in $M_2$, and there exists a node in the cycle that is in $M_2$ but not in $M_1$. Without loss of generality, let us assume that $u_1$ is in $M_1$ but not in $M_2$ (i.e., $u_1 \in V_{M_1} \setminus V_{M_2} = V_{M_1} \setminus I$, where $I = V_{H_1} \cap V_{H_2}$ is a vertex separator of $H$ separating $V_{H_1} \setminus I$ and $V_{H_2} \setminus I$). Then $u_2$ and $u_3$, the neighbors of $u_1$ in $C$, must be in $M_1$ (neither $u_2$ nor $u_3$ can be in $V_{M_2} \setminus I$ since there is no edge between $V_{M_1} \setminus I$ and $V_{M_2} \setminus I$, as they are separated by $I$). This validates our claim.

From \cref{obs:induced-graph-of-M-on-M1-is-a-chain-graph}, there cannot be a directed cycle in $M[V_{M_1}]$ or $M[V_{M_2}]$. Therefore, $C$ must have a length greater than three.

Suppose $l > 3$. Since we have chosen a cycle of the smallest length, $C$ must be a chordless cycle. If $C$ had a chord (an edge between two non-adjacent nodes of the cycle), we would obtain a directed cycle of smaller length in $M$, regardless of whether the chord is directed or undirected. But from \cref{obs:no-chordless-cycle-in-M}, $M$ cannot have any chordless cycle. This implies this case cannot occur.

Thus, we have shown that $M$ does not have any directed cycle, i.e., $M$ is a chain graph. This completes the proof of \cref{lem:M-is-chain-graph}.
\end{proof}

\begin{proof}[\textbf{Proof of \cref{obs:induced-graph-of-M-on-M1-is-a-chain-graph}}]
If $M[V_{H_a}]$ is not a chain graph, then there exists a directed cycle $C = (u_1, u_2, \ldots, u_l, u_{l+1} = u_1)$ in $M$ such that each node of the cycle is in $V_{H_a}$. 

From the construction of $M$, for $p,q \in V_{H_a}$, if there is an edge $p\rightarrow q$ in $M$, then either $p\rightarrow q \in M_a$ or $p-q \in M_a$, and if $p-q \in M$, then $p-q \in M_a$. This implies $C$ is also a cycle in $M_a$. This further implies that none of the edges in $C$ are directed in $M_a$, otherwise $M_a$ contains a directed cycle, contradicting \cref{item-1-theorem-nec-suf-cond-for-MEC} of \cref{thm:nes-and-suf-cond-for-chordal-graph-to-be-an-MEC}.

This implies all the edges of the cycle are undirected edges in $M_a$, and all the nodes of the cycle are part of the same \ucc{} $\mathcal{C}$ of $M_a$. This further implies that for any \ucc{} $\mathcal{C}$ of $M_a$, $M[V_{\mathcal{C}}]$ contains a directed cycle. From \cref{obs:no-directed-cycle-exists-in-a-UCC}, this is not possible.

\begin{observation}
\label{obs:no-directed-cycle-exists-in-a-UCC}
For $a\in \{1,2\}$, for any \ucc{} $\mathcal{C}$ of $M_a$, there does not exist a directed cycle in $M[V_{\mathcal{C}}]$.
\end{observation}

Proof of \cref{obs:no-directed-cycle-exists-in-a-UCC} is given below. This completes the proof of \cref{obs:induced-graph-of-M-on-M1-is-a-chain-graph}.
\end{proof}

\begin{proof}[\textbf{Proof of \cref{obs:no-directed-cycle-exists-in-a-UCC}}]
Suppose for some $a\in \{1,2\}$, for some \ucc{} $\mathcal{C}$ of $M_a$, there exist a directed cycle in $M[V_{\mathcal{C}}]$.
Pick a directed cycle $C = (u_1, u_2, \ldots, u_l, u_{l+1} = u_1)$ in $M[V_{\mathcal{C}}]$ of the smallest length. 
We first show that $l$ cannot be greater than three.

Suppose $l >  3$. Then, $C$ must be a chordless cycle, because if $C$ has a chord then whatever the direction of the chord, we get a directed cycle of smaller length. 
From the construction of $M$, $\skel{M[V_{\mathcal{C}}]} = \skel{M_a[V_{\mathcal{C}}]} = \mathcal{C}$. 
This further implies $C$ is a chordless cycle in $\mathcal{C}$. But, since $\mathcal{C}$ is a chordal graph (from \cref{item-2-theorem-nec-suf-cond-for-MEC} of \cref{thm:nes-and-suf-cond-for-chordal-graph-to-be-an-MEC}, as $\mathcal{C}$ is an \ucc{} of MEC $M_a$), it cannot contain any chordless cycle, a contradiction.  Therefore, if there exists any directed cycle in $M[V_{\mathcal{C}}]$ then the shortest length directed cycle must contain only three vertices. 

There are three possibilities: either all the edges of the cycle are directed, or at most two edges of the cycle are directed, or only one edge of the cycle is directed. 

Suppose there exist vertices $x,y,z \in \mathcal{C}$ such that $x-y-z-x \in \mathcal{C}$ and $x\rightarrow y \rightarrow z \rightarrow x \in M$, i.e., all edges of the cycle are directed.  Let $\tau$ be the LBFS ordering of $\mathcal{C}$ we get at
  line~\ref{algconstructMEC:contructLBFS} of \cref{alg:constructMEC}. From
  \cref{obs:directed-edge-respects-tau}, for $u-v \in V_{\mathcal{C}}$, if  $u\rightarrow v \in M$ then $\tau(u) <\tau(v)$. This implies that
  $x\rightarrow y \rightarrow z \rightarrow x$ cannot occur in $M$, otherwise, we have $\tau(x) < \tau(y) < \tau(z) < \tau(x)$, a contradiction. 

Suppose there exist vertices $x,y,z \in \mathcal{C}$ such that $x-y-z-x \in \mathcal{C}$ and   $x\rightarrow y \rightarrow z - x \in M$ (two edges of the cycle are directed). Then from
  \cref{obs:directed-edge-respects-tau}, $\tau(x)< \tau(y) < \tau(z)$. And, while processing the edge $x-z$ on lines
  \ref{algconstructMEC:constraint-for-orientation}-\ref{algconstructMEC:if-end}
  of \cref{alg:constructMEC},  we
  replace $x-z$ with $x\rightarrow z$.  
  This implies that $x\rightarrow y \rightarrow z-x$
  cannot occur in $M$. We thus have the following: if vertices $x, y, z$ in $\mathcal{C}$ form a directed cycle $C$ in $M$, then $C$ has exactly one directed edge.
The only possibility that remains therefore is of a cycle of the form
  $x\rightarrow y-z-x$.  \Cref{obs:triangle-directed-cycle-with-one-directed-edge-does-not-occur-in-M} shows that even this possibility cannot occur.
\begin{observation}
\label{obs:triangle-directed-cycle-with-one-directed-edge-does-not-occur-in-M}
For $a\in \{1,2\}$, let $\mathcal{C}$ be an \ucc{} of $M_a$, and $M$ be the graph returned by \cref{alg:constructMEC} for input $M_1, M_2$ and $O$.
There does not exist vertices $x,y,z \in \mathcal{C}$ such that $x-y-z-x \in \mathcal{C}$ and $x\rightarrow y -z -x \in M$. 
\end{observation}
This completed the proof of \cref{obs:no-directed-cycle-exists-in-a-UCC}.

\end{proof}

\begin{proof}[\textbf{Proof of \cref{obs:triangle-directed-cycle-with-one-directed-edge-does-not-occur-in-M}}]
\label{proof-of-obs:triangle-directed-cycle-with-one-directed-edge-does-not-occur-in-M}
Let $\tau$ be the LBFS ordering returned by \cref{alg:LBFSwithO} at line~\ref{algconstructMEC:contructLBFS} of \cref{alg:constructMEC}. From \cref{obs:directed-edge-respects-tau}, for $p-q\in V_{\mathcal{C}}$, if $p\rightarrow q \in M$, then $\tau(p)< \tau(q)$. We proceed by structuring the analysis based on the relative ordering of $x, y$, and $z$ in $\tau$. \Cref{statement-for-induction} lists all such potential orderings and shows that none of them are allowed, thus completing the proof of \cref{obs:triangle-directed-cycle-with-one-directed-edge-does-not-occur-in-M}.

\begin{claim}
\label{statement-for-induction}
There do not exist $u, v, w \in \mathcal{C}$ such that $\tau(u) < \tau(v) < \tau(w)$, $u-v-w-u \in \mathcal{C}$ and either $u \rightarrow v -w-u \in M$ or $v \rightarrow w -u-v \in M$ or $u \rightarrow w - v - u \in M$.
\end{claim}

We establish the validity of \cref{statement-for-induction} using induction on $\tau(w)$.

\textbf{Base Case} ($\tau(w) = 2$): This case is vacuously true, as any node in $\tau$ must have a rank greater than or equal to 1. Therefore, there are no nodes $u$ and $v$ satisfying $1 \leq \tau(u) < \tau(v) < \tau(w) = 2$.

\textbf{Inductive Hypothesis:} Suppose \cref{statement-for-induction} is true for $\tau(w) = k$, where $k \geq 2$.

\textbf{Inductive Step:} We now demonstrate that \cref{statement-for-induction} holds for $\tau(w) = k+1$.

Suppose, for the sake of argument, that there exist $u, v, w \in \mathcal{C}$ such that $\tau(u) < \tau(v) < \tau(w) = k+1$, and either $u\rightarrow v-w-u \in M$, or $u-v \rightarrow w-u \in M$, or $u-v-w\leftarrow u \in M$.

We select $u$ to be the node with the maximum possible rank, i.e., for any $u'$ with $\tau(u) < \tau(u')$, there does not exist any $v$ satisfying $\tau(u') < \tau(v) < \tau(w) = k+1$ and either $u'\rightarrow v-w-u \in M$, or $u'-v \rightarrow w-u \in M$, or $u'-v-w\leftarrow u \in M$ hold. We then select $v$ to be the node with the minimum possible rank, i.e., for any $v'$ with $\tau(u) < \tau(v') < \tau(v) < \tau(w) = k+1$, neither $u\rightarrow v'-w-u \in M$, nor $u-v' \rightarrow w-u \in M$, nor $u-v'-w\leftarrow u \in M$ hold.

The following claims show that neither $u\rightarrow v-w-u \in M$, nor $u-v \rightarrow w-u \in M$, nor $u-v-w\leftarrow u \in M$ is satisfied:

\begin{claim}
$u\rightarrow v-w-u \notin M$.
\label{claim1-of-statement-for-induction}
\end{claim}

\begin{claim}
$u-v\rightarrow w -u \notin M$.
\label{claim2-of-statement-for-induction}
\end{claim}

\begin{claim}
$u-v- w \leftarrow u \notin M$.
\label{claim3-of-statement-for-induction}
\end{claim}

Proofs of \cref{claim1-of-statement-for-induction,claim2-of-statement-for-induction,claim3-of-statement-for-induction} are given below. These claims validate \cref{statement-for-induction}.
\end{proof}

\begin{proof}[\textbf{Proof of \cref{claim1-of-statement-for-induction}}]
\label{proof-of-claim1-of-statement-for-induction}
Suppose $u\rightarrow v-w-u \in M$. 
From \cref{corr-of-obs:when-u-v-becomes-directed}, either of the following has occurred:
 \begin{enumerate}
     \item $u \rightarrow v$ is strongly protected in $O$.
     
     \item There exists $a\rightarrow b \in O$ such that $(a,b) \neq (u,v)$ and there exists a \cp{} from $(a,b)$ to $(u,v)$ in $\mathcal{C}$.
     
     \item There exists $a\rightarrow b \in O$ such that $b\neq v$, $u\neq a$ and there exist \cps{} in $\mathcal{C}$ from $(a,b)$ to $v$ and from  $(v,u)$ to $a$.
 \end{enumerate} 
 
We go through each possibility.

\begin{enumerate}
    \item 
    Assume that $u\rightarrow v$ is strongly protected in $O$. Consequently, $u\rightarrow v$ is part of one of the induced subgraphs of $O$, as depicted in \cref{fig:strongly-protected-edge}.
\begin{enumerate}
    \item Consider the scenario where $u\rightarrow v$ is strongly protected due to its presence in an induced subgraph $p\rightarrow u \rightarrow v \in O$ (similar to \Cref{fig:strongly-protected-edge}.a). Keep in mind that $u$, $v$, and $w$ are members of $\mathcal{C}$, an \ucc{} of $M_a$ for some $a\in \{1,2\}$. We need to examine whether $p$ belongs to $M_a$ or not.

Suppose $p \notin M_a$. Without loss of generality, let's assume $a = 1$. From the construction, if $p\notin V_{M_1}$, then $u$, a neighbor of $p$ in $M$, belongs to $I = V_{M_1}\cap V_{M_2}$, as $u\in V_{M_1}$ (recall that $u$ belongs to $\mathcal{C}$, an \ucc{} of $M_a$), and $I$ is a vertex separator of $H$ that separates $V_{M_1}\setminus I$ and $V_{M_2}\setminus I$. However, this implies $u,v,w \in I \cup N(I, H) \subseteq V_O$. Consequently, $u\rightarrow v-w-u \in O$, as from \cref{lem:O-is-an-induced-subgraph-of-M}, $M[V_O] = O$. Nevertheless, this contradicts \cref{item-1-of-def:partial-MEC} of \cref{def:partial-MEC}, as $O$ is a partial MEC. Thus, the only remaining possibility is $p\in M_a$. 

Suppose $p\in M_a$. 
Since $u,p \in V_{M_a}$ and $p\rightarrow u \in M$, from the construction of $M$,  either $p\rightarrow u \in M_a$ or $p-u\in M_a$. If $p\rightarrow u \in M_a$, then $p\rightarrow u -v$ is an induced subgraph in $M_a$, a contradiction from \cref{item-2-theorem-nec-suf-cond-for-MEC} of \cref{thm:nes-and-suf-cond-for-chordal-graph-to-be-an-MEC}. This implies $p-u \in M$, which means $p \in \mathcal{C}$.

From \cref{item-1-of-obs:when-u-v-is-directed-reverse}, $\tau(p) < \tau(u)$.   
This further implies $\tau(p) < \tau(u) < \tau(v)<\tau(w)$ (from our hypothesis $\tau(u) < \tau(v) < \tau(w)$). 
Since $p,v \in \mathcal{C}$, and there is no edge between $p$ and $v$ in $M$, from the construction of $M$, there must be no edge between $p$ and $v$ in $\mathcal{C}$. 
This also implies $p-w \notin \mathcal{C}$; otherwise, since $\tau(v)< \tau(w)$, $\tau(p)<\tau(w)$, and $v-w\in \mathcal{C}$, from \cref{obs:LBGS-gives-PEO}, $p-v \in \mathcal{C}$, leading to a contradiction. This implies that $p\rightarrow u-w \in M$. 
However, this is not possible as while running lines \ref{algconstructMEC:constraint-for-orientation}-\ref{algconstructMEC:if-end} of \cref{alg:constructMEC}, we get $u\rightarrow w\in M$.
Hence, this case does not occur. 

Considering the above analysis, it is evident that neither $p\in M_a$ nor $p\notin M_a$ holds true. Consequently, this case cannot occur. We will now proceed to explore other possibilities.

\item Suppose $u\rightarrow v$ is strongly protected because it is part of an induced subgraph $u \rightarrow v \leftarrow p \in O$ (same as \Cref{fig:strongly-protected-edge}.b). Similar to the above case, the only possibility we need to consider is when $p \in \mathcal{C}$. Since $u,v,p \in V_O$ as well as in $\mathcal{C}$, an \ucc{} of $M_a$, therefore, $u,v,w \in O_a =  V_O \cap V_{M_a}$. Since $O \in \mathcal{E}(O_1, P_{11}, P_{12}, O_2, P_{21}, P_{22})$, therefore, from \cref{item-2-of-def:extension-of-O1-O2-P11-P12-P21-P22} of \cref{def:extension-of-O1-O2-P11-P12-P21-P22}, the v-structure $u\rightarrow v\leftarrow p$ also belongs to $O_a$. Since $(O_a, P_{a1}, P_{a2})$ is a shadow of $M_a$ on $V_{O_a}$, therefore, from \cref{item-1-of-def:shadow} of \cref{def:shadow}, $O_a$ is an induced subgraph of $M_a$. This further implies $u\rightarrow v\leftarrow p \in M_a$. But, this is a contradiction as $u-v \in M_a$ (since $\mathcal{C}$ is an \ucc{} of $M_a$, if $u-v \in \mathcal{C}$ then $u-v\in M_a$).  This contradiction implies that this case also does not occur.

\item 
Suppose $u\rightarrow v$ is strongly protected in $O$ because it is part of an induced subgraph $u\rightarrow p\rightarrow v\leftarrow u \in O$ (same as \Cref{fig:strongly-protected-edge}.c). Similar to the above Case 1.a, the only possibility we need to consider is when $p \in \mathcal{C}$. 
Then from \cref{obs:directed-edge-respects-tau}, $\tau(u)<\tau(p)<\tau(v)$. Combining this with our induction hypothesis, we get $\tau(u)<\tau(p)<\tau(v)<\tau(w)$. There are two possibilities: either $p-w\notin \mathcal{C}$, or $p-w \in \mathcal{C}$.

If $p-w \notin \mathcal{C}$ then while running
lines~\ref{algconstructMEC:constraint-for-orientation}-\ref{algconstructMEC:if-end} of
\cref{alg:constructMEC}, we get $v\rightarrow w \in M$ (because $\tau(p)< \tau(v) <\tau(w)$ and $p\rightarrow v-w $ is an induced subgraph in $M$). This means that this case does not
occur as $v-w \in M$.

Suppose then that $p-w\in \mathcal{C}$. Then there are three possibilities:
either $p-w\in M$, or $p\leftarrow w \in M$, or $p\rightarrow w \in M$.  We will show
that none of the three can occur.
\begin{enumerate}
\item If $p-w\in M$, then we would have a triplet $(p,v,w)$ that satisfies the conditions $\tau(u) < \tau(p) < \tau(v) < \tau(w)$, and $p\rightarrow v - w - p \in M$. However, this directly contradicts our choice of $u$ as the vertex with the maximal $\tau(u)$ (recall that we have selected $u$ as a node with maximum rank in $\tau$ such that there exists a $v$ satisfying $\tau(u) < \tau(v) < \tau(w) = k+1$ and $u\rightarrow v -w - u \in M$). Hence, it is concluded that $p-w\notin M$.

\item Since $\tau(p) < \tau(w)$, therefore, according to \cref{obs:directed-edge-respects-tau},
  $w\rightarrow p \notin M$.

\item If $p\rightarrow w \in M$ then from
  lines~\ref{algconstructMEC:constraint-for-orientation}-\ref{algconstructMEC:if-end} of
  \cref{alg:constructMEC}, $u\rightarrow w \in M$ (due to $\tau(u)<\tau(p)<\tau(w)$, and $u\rightarrow p\rightarrow w-u \in M$). However, this leads to a contradiction as
  $u-w\in M$.
\end{enumerate}
This analysis demonstrates that none of the three possibilities can occur, thereby ruling out the possibility of case $p-w \in \mathcal{C}$. With this, we conclude that this case also does not occur.

\item 
Suppose $u\rightarrow v$ is strongly protected in $O$ because it is part of an induced subgraph $u\rightarrow v \leftarrow p-u-p'\rightarrow v$ (same as \Cref{fig:strongly-protected-edge}.d). Similar to above case 1.a, the only possibility we need to consider is when $p,p' \in \mathcal{C}$.
Similar to Case 1.b, the v-structure $p\rightarrow v \leftarrow p'$ in $O$ implies that the v-structure also belongs to $M_a$, i.e., $p\rightarrow v\leftarrow p' \in M_a$.
But, since $\mathcal{C}$ is an \ucc{} of $M_a$, and $p,v \in \mathcal{C}$, therefore, from \cref{prop:edge-between-2-nodes-of-same-ucc-is-ud}, $p\rightarrow v$ cannot be in $M_a$. This leads to a contradiction.
Hence, it can be deduced that this particular case also does not come to pass.
\end{enumerate}

We show that $u\rightarrow v$ is a part of none of the induced subgraphs shown in \cref{fig:strongly-protected-edge}. This implies that this possibility does not occur, i.e., $u\rightarrow v$ is not strongly protected in $O$. We now move to the other possibilities.

\item 
Suppose $u\rightarrow v \in O$ because of the existence of
$x-y \in \mathcal{C}$ such that $x\rightarrow y \in O$, and there exists a \cp{}
$Q = (p_1 =x, p_2 = y, \ldots, p_{l-1} = u, p_l = v)$ from $(x,y)$ to $(u,v)$ in
$\mathcal{C}$. From \cref{item-2-of-obs:when-u-v-is-directed-reverse},
$p_{l-2}\rightarrow p_{l-1} = p_{l-2} \rightarrow u\in M$.  We thus have the
induced sub-graph $p \rightarrow u \rightarrow v$ in $M$ with $p = p_{l-2}$.
This case is now exactly the same as the first case of the first possibility, and hence as proved there,
cannot occur.

\item 
Suppose there exists $x - y \in O$ such that
$x \rightarrow y \in O$, and there exist  \cps{} in $\mathcal{C}$:
$Q_1 = (p_1 = x, p_2 = y, \ldots, p_l = v)$ from $(x, y)$ to $v$ and a \cp{}
$Q_2 = (q_1 = v, q_2 = u, \ldots, q_m = x)$ from $(v, u)$ to $x$ in
$\mathcal{C}$, such that $y \neq v$ and $u \neq x$.

Note that by \cref{item-2-of-obs:when-u-v-is-directed-reverse},
we then have
\begin{equation}
  \label{eq:all-edges-directed}
  \tau(p_i) < \tau(v)
  \text{ and }
  p_i \rightarrow p_{i+1}
  \text{ for every } 1 \leq i \leq l - 1.
\end{equation}

There are two possibilities: either $Q_1' = (p_1 = x, p_2 = y, \ldots, p_l = v, w)$
is a \cp{} in $\mathcal{C}$, or $Q_1'$ is not a \cp{} in $\mathcal{C}$. If
$Q_1'$ were a \cp{}, then from \cref{item-2-of-obs:when-u-v-is-directed-reverse},
$v \rightarrow w \in M$. This is a contradiction, since by hypothesis, we have
$v - w \in M$. Thus, $Q_1'$ cannot be a \cp{}.

We now show that $Q \coloneqq (p_1 = x, p_2 = y, \ldots, p_{l-1}, w)$ (obtained by replacing the node $p_l$ with $w$ in $Q_1$) is a \cp{}.

Since $Q_1$ is a \cp{} in $\mathcal{C}$, and $Q_1'$ is not a \cp{} in
$\mathcal{C}$, there must exist an edge $p_i - w \in \mathcal{C}$ for some
$i < l$. As observed in \cref{eq:all-edges-directed}, for every $i \leq l - 1$, $\tau(u_i) < \tau(v) < \tau(w)$ (the
second inequality comes from the hypothesis of \cref{statement-for-induction}).
Thus, for any $i \leq l - 2$, if $p_i - w \in \mathcal{C}$, then from
\cref{obs:LBGS-gives-PEO}, $p_i - v \in \mathcal{C}$.
But, then $Q_1$ is not a \cp{}, a contradiction.
This implies that the neighbor of $w$ in $Q_1$, responsible for ensuring that
$Q_1'$ is not a \cp{}, must be $p_{l-1}$.
This further implies that $Q$ is a \cp{}.

Note now that since $\tau(p_{l-1}) < \tau(v) < \tau(w)$ and $\tau(u) < \tau(w)$,
\cref{obs:LBGS-gives-PEO} implies that $p_{l-1} - u \in \mathcal{C}$ (since $w$ is
adjacent to both $p_{l-1}$ and $u$). We will now show that
\begin{equation}
  \label{eq:u-u-l-1-not-in-M}
  p_{l-1} \rightarrow w \in M \text{ and } u \rightarrow p_{l-1} \not\in M.
\end{equation}
The first follows from \cref{item-2-of-obs:when-u-v-is-directed-reverse} since
$Q$ as constructed above is a \cp{} in $\mathcal{C}$ from $(x, y)$ to
$(p_{l-1}, w)$ where $x \rightarrow y \in M$.  
Suppose, if possible, that $u \rightarrow p_{l-1} \in M$. Then, by \cref{obs:directed-edge-respects-tau},
and the observation preceding \cref{eq:u-u-l-1-not-in-M},
$\tau(u) < \tau(p_{l-1}) < \tau(w)$. Thus, while running lines \ref{algconstructMEC:constraint-for-orientation}-\ref{algconstructMEC:if-end} of \cref{alg:constructMEC}, we get $u\rightarrow w \in M$ (as $\tau(u) < \tau(p_{l-1}) < \tau(w)$ and $u\rightarrow p_{l-1}, p_{l-1} \rightarrow w \in M$).  This is a contradiction, as
$u - w \in M$. This implies that $u \rightarrow p_{l-1} \notin M$.

We now show that there exists some \(i \geq 2\) for which \(q_{i+1} \rightarrow q_i \in M\). Recall that \(Q_2 = (q_1 = v, q_2 = u, \ldots, q_m = x)\) is a \cp{} from \((v, u)\) to \(x\) in \(\mathcal{C}\), an \ucc{} of \(M_a\), as guaranteed by the hypothesis for this case.

From the construction, \(\skel{M[V_{\mathcal{C}}]} = \mathcal{C}\). Since \(u - p_{l-1} \in \mathcal{C}\), therefore, \(u - p_{l-1} \in \skel{M}\). This implies from \cref{eq:u-u-l-1-not-in-M}, either \(p_{l-1} \rightarrow u \in M\) or \(p_{l-1} - u \in M\). Thus, from \cref{eq:all-edges-directed}, if we do not have \(q_{i+1} \rightarrow q_i \in M\) for any \(i \geq 2\), then \(M\) has a directed cycle contained in the sequence \(C_1 \coloneqq (p_1 = x, p_2, \ldots, p_{l-1}, q_2 = u, q_3, \ldots, q_m = x)\). From \cref{prop:chordal-chain-graph-contains-directed-cycle-of-length-three}, there must thus exist a directed cycle of the form \(C = (a, b, c)\) such that \(a, b, c\) are nodes of the cycle \(C_1\).

From \cref{eq:all-edges-directed}, \(\tau(p_1 = q_m) < \tau(v)\), therefore, from \cref{obs:chordless-path-and-LBFS-ordering}, for path \(Q_2\), for \(i > 1\), \(\tau(q_i) < \tau(v)\). Also, from \cref{eq:all-edges-directed}, for \(1 \leq i < l\), \(\tau(p_i) < \tau(v)\) since \(\tau(v) < \tau(w)\). Therefore, for \(1 \leq i < l\), \(\tau(p_i) < \tau(w)\), and for \(i > 1\), \(\tau(q_i) < \tau(w)\). This implies each node of the cycle \(C_1\) has a rank lower than the rank of \(w\) in \(\tau\). Therefore, the rank of the highest rank node of the cycle \(C\) in \(\tau\) is less than the rank of \(w\).

However, this is a contradiction, as per our induction hypothesis, there cannot be any such directed cycle with the rank of the highest rank node being less than \(k+1\). This validates that for some \(i \geq 2\), \(q_{i+1} \rightarrow q_i \in M\).

Since \(Q_2\) is an undirected \cp{} in \(\mathcal{C}\), therefore, \(\bar{Q_2} = (q_m , q_{m-1}, \ldots, q_2, q_1)\), the reverse of the path \(Q_2\), is also a \cp{}. Consequently, \(Q_3 = (q_{i+1}, q_i, \ldots, q_2, q_1)\), a subpath of \(\bar{Q_2}\), is also a \cp{} in \(\mathcal{C}\). Since \(q_{i+1}\rightarrow q_i \in M\), according to \cref{item-2-of-obs:when-u-v-is-directed-reverse},  \(\tau(q_{i+1})< \tau(q_i)< \ldots, \tau(q_2) < \tau(q_1)\), and for each \(i+1\leq j \leq 2\), \(q_{j} \rightarrow q_{j-1} \in M\).

We thus have the induced sub-graph \(p \rightarrow u \rightarrow v\) in \(M\) with \(p = q_3\) (remember \(q_2 = u\) and \(q_1 = v\)). This case is now exactly the same as the first case of the first possibility, and hence as proved there, this case cannot occur.
\end{enumerate}

 The above discussion shows that there is no possibility for $u\rightarrow v-w-u \in M$.
 \end{proof} 

\begin{proof}[\textbf{Proof of \cref{claim2-of-statement-for-induction}}]
\label{proof-of-claim2-of-statement-for-induction}
  Suppose that $u-v\rightarrow w-u \in M$. From \cref{corr-of-obs:when-u-v-becomes-directed}, either of the following has occurred:

  \begin{enumerate}
      \item $v\rightarrow w$ is strongly protected in $O$.
      \item There exists $a\rightarrow b \in O$ such that $(a,b)\neq (v,w)$, and there exists a \cp{} from $(a,b)$ to $(v,w)$ in $\mathcal{C}$.
      \item There exists $a\rightarrow b \in O$ such that there exist \cps{} in $\mathcal{C}$ from $(a,b)$ to $w$ and from $(w,v)$ to $a$ such that $b\neq w$ and $v\neq a$.
  \end{enumerate}

Let's examine each possibility.

\begin{enumerate}
    \item 
    Suppose $v\rightarrow w$ is strongly protected in $O$.  This implies that $v\rightarrow w$ is part of one of the induced subgraphs of $O$, as illustrated in \cref{fig:strongly-protected-edge}.

\begin{enumerate}
    \item Assume that $v\rightarrow w$ is strongly protected due to its presence in an induced subgraph $p\rightarrow v\rightarrow w \in O$ (similar to \Cref{fig:strongly-protected-edge}.a).
Similar to the first case of the first possibility of the proof of  \cref{claim1-of-statement-for-induction}, the only relevant scenario is when $p$ belongs to $\mathcal{C}$.

From the construction, $\skel{M[V_{\mathcal{C}}]} = \mathcal{C}$. This implies $p-v \in \mathcal{C}$ and $p-w\notin \mathcal{C}$.
Given that $p-v \in \mathcal{C}$ and $p\rightarrow v \in M$, we can infer from \cref{obs:directed-edge-respects-tau} that $\tau(p) < \tau(v)$.
Given that $u-v\in \mathcal{C}$ and $\tau(u) < \tau(v)$, we can deduce from \cref{obs:LBGS-gives-PEO} that $p-u\in \mathcal{C}$. 
Since   $\skel{M[V_{\mathcal{C}}]} = \mathcal{C}$, therefore, either (1) $p-u \in M$ or (2) $p\rightarrow u \in M$ or (3) $u\rightarrow p \in M$.  We will show
that none of the three can occur.

\begin{enumerate}
    \item If $p-u \in M$, then $M$ would possess a directed cycle $(p,v,u,p)$ with the highest rank node $v$ having a rank less than $\tau(w)$, thereby contradicting our induction hypothesis.

    \item If $p\rightarrow u \in M$ then using \cref{obs:directed-edge-respects-tau}, we have $\tau(p) < \tau(u)$. However, since $p-w\notin \mathcal{C}$ and $\tau(p)<\tau(u)<\tau(w)$, applying lines \ref{algconstructMEC:constraint-for-orientation}-\ref{algconstructMEC:if-end} of \cref{alg:constructMEC} would lead to $u\rightarrow w \in M$, which contradicts $u-w\in M$.

    \item And, if $u\rightarrow p \in M$ then using \cref{obs:directed-edge-respects-tau}, we have $\tau(u) < \tau(p)$. Therefore, since $\tau(u) < \tau(p) < \tau(v)$ and $u\rightarrow p \rightarrow v -u \in M$, applying lines \ref{algconstructMEC:constraint-for-orientation}-\ref{algconstructMEC:if-end} of \cref{alg:constructMEC} would lead to $u\rightarrow v \in M$, which contradicts $u-v \in M$.
\end{enumerate}

Hence, based on the analysis above, it becomes evident that none of these scenarios are feasible. Therefore, we conclude that this case cannot occur.

\item 
Assume that $v\rightarrow w$ is strongly protected due to the presence in an induced subgraph $v\rightarrow w \leftarrow p \in O$ (as shown in \Cref{fig:strongly-protected-edge}.b). Similar to the second case of the first possibility of the proof of  \cref{claim1-of-statement-for-induction}, this case does not occur.

\item 
Suppose $v\rightarrow w$ is strongly protected in $O$ because it is part of an induced subgraph $v\rightarrow p \rightarrow w\leftarrow v \in O$ (same as \Cref{fig:strongly-protected-edge}.c). Similar to the third case of the first possibility of the proof of \cref{claim1-of-statement-for-induction}, $p\in \mathcal{C}$. 
From \cref{lem:O-is-an-induced-subgraph-of-M}, $v\rightarrow p \rightarrow w\leftarrow v \in M$. From the construction, $\skel{M[V_{\mathcal{C}}]} = \mathcal{C}$. This implies  $p-w \in \mathcal{C}$ (as $p\rightarrow w \in M$ and $p,w \in \mathcal{C}$).
From \cref{obs:directed-edge-respects-tau}, $\tau(v) < \tau(p) < \tau(w)$. 
From \cref{obs:LBGS-gives-PEO}, this implies $u-p\in \mathcal{C}$ (since $p-w, u-w \in \mathcal{C},\tau(p)< \tau(w)$ and $\tau(u)<\tau(w)$).
Since $u-p \in \mathcal{C}$ and $\skel{M[V_{\mathcal{C}}]} = \mathcal{C}$, therefore, either (1) $u-p \in M$, or (2) $u\rightarrow p \in M$, or (3) $p\rightarrow u \in M$. 
We show that none of these cases occur.
\begin{enumerate}
    \item If $u-p \in M$ then $M$ have a directed cycle $u-v\rightarrow p -u$ such that the rank of the highest rank node of the cycle is less than $\tau(w)$, which contradicts our induction hypothesis. This implies $u-p \notin M$.

    \item Suppose $u\rightarrow p \in M$. Then, since $u\rightarrow p, p\rightarrow w \in M$, while running lines \ref{algconstructMEC:constraint-for-orientation}-\ref{algconstructMEC:if-end} of \cref{alg:constructMEC}, we get $u\rightarrow w \in M$, a contradiction, as $u-w\in M$. This implies $u\rightarrow p \notin M$ does not occur.
    
    \item Since $\tau(u) < \tau(p)$, from \cref{obs:directed-edge-respects-tau}, $p\rightarrow u \notin M$.
\end{enumerate}
The above analysis shows that neither $u-p \in M$ nor $u\rightarrow p \in M$ nor $p\rightarrow u \in M$.  This implies that this case does not occur.

\item 
Suppose $v\rightarrow w$ is strongly protected in $O$ because it is part of an induced subgraph $v\rightarrow w \leftarrow p-v-p'\rightarrow w$ (same as \Cref{fig:strongly-protected-edge}.d). Similar to the fourth case of the first possibility of the proof of   \cref{claim1-of-statement-for-induction}, this case does not occur.
\end{enumerate}

\item 
Suppose $v\rightarrow w \in M$ because of the existence of some $x-y \in \mathcal{C}$ such that $x\rightarrow y\in O$, and there exists a \cp{} $P = (p_1=x, p_2 =y, \ldots, p_{l-1}=v, p_l = w)$ from $(x,y)$ to $(v,w)$ in $\mathcal{C}$. From \cref{item-2-of-obs:when-u-v-is-directed-reverse}, $p_{l-2}\rightarrow p_{l-1} \in M$. 
We thus have the
induced sub-graph $p \rightarrow v \rightarrow w$ in $M$ with $p = p_{l-2}$.
This case is now exactly the same as the first case of the first possibility, and hence as proved there,
cannot occur.

\item 
Suppose there exists $x-y \in O$ such that $x\rightarrow y \in O$, and there exists a \cp{} in $\mathcal{C}$: $Q_1 = (p_1 = x, p_2=y, \ldots, p_l =w)$ from $(x,y)$ to $w$ in $\mathcal{C}$, and  $Q_2 = (q_1 = w, q_2 = v, \ldots, q_m = x)$ from $(w,v)$ to $x$ in $\mathcal{C}$ such that $y\neq w$ and $v\neq x$. According to \cref{item-2-of-obs:when-u-v-is-directed-reverse}, we have $\tau(p_{l-1}) < \tau(w)$, and $p_{l-1} \rightarrow w \in M$.
Considering that $\tau(p_{l-1})  < \tau(w)$, $\tau(v) < \tau(w)$, and $\tau(u) < \tau(w)$, it follows from \cref{obs:LBGS-gives-PEO} that $p_{l-1}-u$ and $p_{l-1} -v$ are in $\mathcal{C}$.

There are two possibilities: either $v\rightarrow p_{l-1} \in M$, or $v\rightarrow p_{l-1}\notin M$. Let's examine both possibilities:

\begin{enumerate}
    \item Suppose $v\rightarrow p_{l-1} \in M$. Since $u-p_{l-1} \in \mathcal{C}$, we have either $u-p_{l-1} \in M$, or $u\rightarrow p_{l-1} \in M$, or $p_{l-1}\rightarrow u \in M$. 
    \begin{enumerate}
        \item Suppose $u-p_{l-1} \in M$. This would result in a directed cycle $(u,v,p_{l-1}, u)$ in $M$ where the node with the highest rank, $u$, has a rank lower than $\tau(w)$, contradicting our induction hypothesis.
        
        \item Suppose $u\rightarrow p_{l-1} \in M$. Since $p_{l-1}\rightarrow w \in M$, running lines \ref{algconstructMEC:constraint-for-orientation}-\ref{algconstructMEC:if-end} of \cref{alg:constructMEC} would yield $u\rightarrow w \in M$. This contradicts the assumption $u-w \in M$. Thus, $u\rightarrow p_{l-1} \notin M$.
        
        \item From the hypothesis of \cref{statement-for-induction}, $\tau(u) < \tau(v)$. Since from our assumption, $v\rightarrow p_{l-1} \in M$, therefore, from \cref{obs:directed-edge-respects-tau}, $\tau(v) < \tau(p_{l-1})$. This implies $\tau(u) < \tau(p_{l-1})$. Then, from \cref{obs:directed-edge-respects-tau}, $p_{l-1} \rightarrow u \notin M$.
    \end{enumerate} 
    The above analysis shows that neither $u-p_{l-1} \in M$ nor  $u\rightarrow p_{l-1} \in M$ nor $p_{l-1} \rightarrow u \in M$. 
    
    \item Suppose $v\rightarrow p_{l-1}\notin M$. Then for some $i\geq 2$, $q_{i+1}\rightarrow q_i$; otherwise, a directed cycle $(p_1 =x, p_2, \ldots, p_{l-1}, q_2=u, q_3, \ldots, q_m=x)$ with a maximum rank lower than $\tau(w)$ would exist in $M$. From \cref{prop:chordal-chain-graph-contains-directed-cycle-of-length-three},  this would lead to a directed cycle in $M$ of length 3 with a maximum rank lower than $\tau(w)$, contradicting our induction hypothesis. Hence, for some $i\geq 2$, $q_{i+1}\rightarrow q_i \in M$. Since $Q_2$ is an undirected \cp{} in $\mathcal{C}$, therefore, $\bar{Q_2} = (q_m , q_{m-1}, \ldots, q_2, q_1)$, the reverse of the path $Q_2$, is also a \cp{}. Consequently, $Q_3 = (q_{i+1}, q_i, \ldots, q_2, q_1)$, a subpath of $\bar{Q_2}$, is also a \cp{} in $\mathcal{C}$. Since $q_{i+1}\rightarrow q_i \in M$, according to \cref{item-2-of-obs:when-u-v-is-directed-reverse},  $\tau(q_{i+1})< \tau(q_i)< \ldots, \tau(q_2) < \tau(q_1)$, and for each $i+1\leq j \leq 2$, $q_{j} \rightarrow q_{j-1} \in M$.
We thus have the
induced sub-graph $p \rightarrow v \rightarrow w$ in $M$ with $p = q_3$ (remember $Q_2$ is a \cp{}, $q_1 = w$ and $q_2 = v$).
This case is now exactly the same as the first case of the first possibility, and hence as proved there,
cannot occur.
\end{enumerate}
\end{enumerate}
The above discussion shows a contradiction in each possibility, thereby demonstrating that $u-v\rightarrow w -u \notin M$. This establishes the truth of \cref{claim2-of-statement-for-induction}.

\end{proof}

\begin{proof}[\textbf{Proof of \cref{claim3-of-statement-for-induction}}]
\label{proof-of-claim3-of-statement-for-induction}
Suppose $u-v-w\leftarrow u \in M$.
From \cref{corr-of-obs:when-u-v-becomes-directed}, either of the following occurred:
 \begin{enumerate}
    \item $u \rightarrow w$ is strongly protected in $O$.
    
  \item There exists
  $a\rightarrow b \in O$ such that $(a,b)\neq (u,w)$ and there exists a \cp{} from $(a,b)$ to $(u,w)$ in $\mathcal{C}$.
  
  \item There exists $a\rightarrow b \in O$ such that there exist \cps{} in $\mathcal{C}$ from $(a,b)$ to $w$
  and from $(w,u)$ to $a$ such that $b\neq w$ and $u\neq a$.  
 \end{enumerate}
We will consider each possibility.
\begin{enumerate}
    \item 
    $u\rightarrow w$ is strongly protected in $O$. Then, $u\rightarrow w$ is part of one of the induced subgraphs of $O$ in \cref{fig:strongly-protected-edge}.

  \begin{enumerate}
      \item 
      Suppose $u\rightarrow w$ is strongly protected because it is part of an induced subgraph $p\rightarrow u\rightarrow w \in O$ (same as \Cref{fig:strongly-protected-edge}.a). Similar to the first case of the first possibility of the proof of  \cref{claim1-of-statement-for-induction}, the only possibility we need to consider is when $p\in \mathcal{C}$. From \cref{obs:directed-edge-respects-tau}, $\tau(p) < \tau(u)$. 
If $p-v \notin \mathcal{C}$, then while running lines \ref{algconstructMEC:constraint-for-orientation}-\ref{algconstructMEC:if-end} of \cref{alg:constructMEC}, we get $u\rightarrow v \in M$, a contradiction, as $p\rightarrow u-v$ is an induced subgraph in $M$ and $\tau(p) < \tau(u) < \tau(v)$ (second inequality comes from the hypothesis of \cref{statement-for-induction}). However, this is a contradiction as $u-v\in M$. This implies $p-v \in \mathcal{C}$. 
There are three  possibilities: either $v-p \in M$, or $v\rightarrow p \in M$, or $p\rightarrow v\in M$. Let's consider each one.
\begin{enumerate}
    \item If $v-p \in M$ then we have a directed cycle $(p,u,v,p)$ in $M$ such that the rank of the highest rank node of the cycle (which is $v$) is less than the rank of $w$, this contradicts the induction hypothesis of \cref{statement-for-induction}. Hence implies $v-p \notin M$. 
    
    \item From \cref{obs:directed-edge-respects-tau}, $v\rightarrow p \notin M$ as $\tau(p) < \tau(v)$.
    
    \item Suppose $p\rightarrow v \in M$. But, then while running lines \ref{algconstructMEC:constraint-for-orientation}-\ref{algconstructMEC:if-end} of \cref{alg:constructMEC}, we get $v\rightarrow w \in M$, a contradiction, as $p\rightarrow v-w$ is an induced subgraph in $M$ and $\tau(p) < \tau(v) < \tau(w)$. 
\end{enumerate}

We have shown that in all the possibilities of this subcase, we get a contradiction. Therefore, this subcase cannot occur.

\item 
Suppose that $u\rightarrow w$ is strongly protected because it is part of an induced subgraph $u\rightarrow w \leftarrow p \in O$ (same as \Cref{fig:strongly-protected-edge}.b). Similar to the second case of the first possibility of the proof of \cref{claim1-of-statement-for-induction}, this case does not occur.

\item 
Suppose $u\rightarrow w$ is strongly protected in $O$ because it is part of an induced subgraph $u\rightarrow p \rightarrow w\leftarrow u \in O$ (same as \Cref{fig:strongly-protected-edge}.c). Similar to the third case of the first possibility of the proof of \cref{claim1-of-statement-for-induction}, the only possibility we need to consider is when $p\in \mathcal{C}$. 
From \cref{obs:directed-edge-respects-tau}, $\tau(u) < \tau(p) < \tau(w)$. 
From \cref{obs:LBGS-gives-PEO}, $v-p\in \mathcal{C}$ (since $p-w, v-w \in \mathcal{C},\tau(p)< \tau(w)$ and $\tau(v) <\tau(w)$). From the construction, $\skel{M[V_{\mathcal{C}}]} = \mathcal{C}$. Therefore, $v-p \in \skel{M}$. This implies either $v-p \in M$ or $v\rightarrow p \in M$ or $p\rightarrow v \in M$. We go through each possibility.
\begin{enumerate}
    \item Suppose $v-p \in M$. Then we get a directed cycle $(u,p,v,u)$ in $M$ such that the highest rank node (which could be either $v$ or $p$) in the cycle is less than the rank of $w$ in $\tau$, this contradicts the induction hypothesis of \cref{statement-for-induction}. This implies $v-p \notin M$.
    
    \item Suppose $v\rightarrow p \in M$. Then from \cref{obs:directed-edge-respects-tau}, $\tau(v) < \tau(p)$. Then while running lines \ref{algconstructMEC:constraint-for-orientation}-\ref{algconstructMEC:if-end} of \cref{alg:constructMEC}, we get $v\rightarrow w \in M$ (as $\tau(v) < \tau(p) < \tau(w)$ and $v\rightarrow p, p\rightarrow w \in M$). But, this is a contradiction as $v-w \in M$. This implies $v\rightarrow p \notin M$.

    \item Suppose $p\rightarrow v \in M$. Then from \cref{obs:directed-edge-respects-tau}, $\tau(p) < \tau(v)$. Then while running lines \ref{algconstructMEC:constraint-for-orientation}-\ref{algconstructMEC:if-end} of \cref{alg:constructMEC}, we get $u\rightarrow v \in M$ (as $\tau(u) < \tau(p) < \tau(v)$ and $u\rightarrow p, p\rightarrow v \in M$). But, this is a contradiction as $u-v \in M$. This implies $p\rightarrow v \notin M$.
\end{enumerate}

The above analysis shows that none of the possibilities can occur. This  implies that this case does not occur.

\item 
Suppose $u\rightarrow w$ is strongly protected in $O$ because it is part of an induced subgraph $u\rightarrow w \leftarrow p-u-p'\rightarrow w$ (same as \Cref{fig:strongly-protected-edge}.d). Similar to the fourth case of the first possibility of the proof of \cref{claim1-of-statement-for-induction}, this case does not occur.
  \end{enumerate}

  \item 
  Suppose $u\rightarrow w \in M$ because of the existence of some $x-y \in \mathcal{C}$ such that $x\rightarrow y\in O$, and there exists a \cp{} $P = (p_1=x, p_2 =y, \ldots, p_{l-1}=u, p_l = w)$ from $(x,y)$ to $(u,w)$ in $\mathcal{C}$. From \cref{item-2-of-obs:when-u-v-is-directed-reverse}, $p_{l-2}\rightarrow p_{l-1} \in M$. 
We thus have the
induced sub-graph $p \rightarrow u \rightarrow w$ in $M$ with $p = p_{l-2}$.
This case is now exactly the same as the first case of the first possibility, and hence as proved there,
cannot occur.

\item 
Suppose there exists $x-y \in O$ such that $x\rightarrow y \in O$, and there exist \cps{} in $\mathcal{C}$: $Q_1 = (p_1 = x, p_2=y, \ldots, p_l =w)$ from $(x,y)$ to $w$ in $\mathcal{C}$, and  $Q_2 = (q_1 = w, q_2 = u, \ldots, q_m = x)$ from $(w,u)$ to $x$ such that $y\neq w$ and $u\neq x$. From \cref{item-2-of-obs:when-u-v-is-directed-reverse}, $\tau(p_{l-1}) < \tau(w)$, and $p_{l-1} \rightarrow w \in M$.
Since $\tau(p_{l-1})  < \tau(w)$, $\tau(v) < \tau(w)$, and $\tau(u) < \tau(w)$ (the last two inequality comes from the hypothesis of \cref{statement-for-induction}), therefore, from \cref{obs:LBGS-gives-PEO}, $u-p_{l-1}, v-p_{l-1} \in \mathcal{C}$.
From the construction, $\skel{M[V_{\mathcal{C}}]} = \mathcal{C}$. This implies $u-p_{l-1}, v-p_{l-1} \in \skel{M}$.
There are two possibilities: either $u\rightarrow p_{l-1} \in M$, or $u\rightarrow p_{l-1}\notin M$ (i.e., either $u-p_{l-1} \in M$ or $p_{l-1} \rightarrow u \in M$). We go through each possibility.

\begin{enumerate}
    \item Suppose $u \rightarrow p_{l-1} \in M$. Since $v - p_{l-1}\in \skel{M}$, either we have $v-p_{l-1} \in M$ or $v\rightarrow p_{l-1} \in M$ or $v\leftarrow p_{l-1} \in M$. 
    \begin{enumerate}
        \item Suppose $v-p_{l-1} \in M$. Then we have a directed cycle $(u, p_{l-1}, v,u)$ in $M$ such that the rank of the highest rank node in the cycle is bigger than $\tau(w)$ in $\tau$. This contradicts the induction hypothesis of \cref{statement-for-induction}. This implies $v-p_{l-1} \notin M$.

        \item Suppose $v\rightarrow p_{l-1} \in M$. Then, since $p_{l-1}\rightarrow w \in M$, therefore, while running lines \ref{algconstructMEC:constraint-for-orientation}-\ref{algconstructMEC:if-end} of \cref{alg:constructMEC},  we get $v\rightarrow w \in M$ (as $v\rightarrow p_{l-1}, p_{l-1}\rightarrow w \in M$). But, from our assumption, $v-w \in M$. This  implies that $v\rightarrow p_{l-1} \notin M$.

        \item Suppose $v\leftarrow p_{l-1} \in M$.  Then, while running lines \ref{algconstructMEC:constraint-for-orientation}-\ref{algconstructMEC:if-end} of \cref{alg:constructMEC},  we get $u\rightarrow v \in M$ (as $u\rightarrow p_{l-1} \in M, p_{l-1}\rightarrow v \in M$). But, from our assumption $u-v \in M$. This  implies that $v\leftarrow p_{l-1} \notin M$.
    \end{enumerate} 
    The above analysis shows that when $u\rightarrow p_{l-1} \in M$, neither $v-p_{l-1} \in M$ nor $v\rightarrow p_{l-1} \in M$ nor $v\leftarrow p_{l-1} \in M$. This implies $u \rightarrow p_{l-1} \in M$ does not occur.
    
    \item Suppose $u\rightarrow p_{l-1}\notin M$. Then for some $i\geq 2$, $q_{i+1}\rightarrow q_i$; otherwise, a directed cycle $(p_1 =x, p_2, \ldots, p_{l-1}, q_2=u, q_3, \ldots, q_m=x)$ with a maximum rank lower than $\tau(w)$ would exist in $M$. From \cref{prop:chordal-chain-graph-contains-directed-cycle-of-length-three}, this would lead to a directed cycle in $M$ of length 3 with a maximum rank lower than $\tau(w)$, contradicting the induction hypothesis of \cref{statement-for-induction}. Hence, for some $i\geq 2$, $q_{i+1}\rightarrow q_i \in M$. Since $Q_2$ is an undirected \cp{}, $\bar{Q_2} = (q_m , q_{m-1}, \ldots, q_2, q_1)$, the reverse of the path $Q_2$, is also a \cp{}. Consequently, $Q_3 = (q_{i+1}, q_i, \ldots, q_2, q_1)$, a subpath of $\bar{Q_2}$, is also a \cp{} in $\mathcal{C}$. Since $q_{i+1}\rightarrow q_i \in M$, according to \cref{item-2-of-obs:when-u-v-is-directed-reverse},  $\tau(q_{i+1})< \tau(q_i)< \ldots, \tau(q_2) < \tau(q_1)$, and for each $i+1\leq j \leq 2$, $q_{j} \rightarrow q_{j-1} \in M$.
We thus have the
induced sub-graph $p \rightarrow u \rightarrow w$ in $M$ with $p = q_3$ (remember $Q_2$ is a \cp{}, $q_1 = w$ and $q_2 = u$).
This case is now exactly the same as the first case of the first possibility, and hence as proved there,
cannot occur.
\end{enumerate}
This indicates a contradiction in all possibilities, thereby demonstrating that $u-v-w \leftarrow u \notin M$. This establishes the truth of \cref{claim3-of-statement-for-induction}.
\end{enumerate} 
\end{proof}

\begin{proof}[\textbf{Proof of \cref{obs:no-chordless-cycle-in-M}}]
    Let $C = (u_1, u_2, \ldots, u_l, u_{l+1} = u_1)$ be a chordless cycle in $M$. If $C$ lies entirely within $M_1$ (i.e., each node of $C$ belongs to $V_{M_1}$), then by \cref{obs:induced-graph-of-M-on-M1-is-a-chain-graph}, $C$ is undirected. However, from the construction of $M$, (a) for any $u, v \in V_{M_1}$, if $u - v \in M$, then $u - v \in M_1$ (due to the initialization step, line \ref{algconstructMEC:M-initialization} of \cref{alg:constructMEC}, a directed edge in $M_1$ is a directed edge in $M$), and (b) if there is no chord in $C$ in $M$ then there is no chord in $C$ in $M_1$ as well (because all the nodes of $C$ are in $M_1$ and $\skel{M[V_{M_1}]} = \skel{M_1}$). This implies that $C$ is an undirected chordless cycle in $M_1$, which contradicts \cref{item-2-theorem-nec-suf-cond-for-MEC} of \cref{thm:nes-and-suf-cond-for-chordal-graph-to-be-an-MEC}. Hence, the nodes in $C$ cannot be exclusively within $M_1$ or $M_2$. Thus, there must exist a node $u_i$ in $C$ that is in $V_{M_1} \setminus I$, a node $u_j$ in $C$ that is in $V_{M_2} \setminus I$, and a node $u_k$ in $C$ that is in $I$, as $I = V_{M_1}\cap V_{M_2}$ is a vertex separator of $H = \skel{M}$ that separates $V_{M_1} \setminus I$ and $V_{M_2} \setminus I$. 
    
    Without loss of generality, let's assume that $u_1 \in I$.
    Therefore, $u_1, u_2, u_l \in V_O$ since $u_2$ and $u_l$ are neighbors of $u_1$, and $I \cup N(I, H) \subseteq V_O$. Since $O$ is an induced subgraph of $M$ (from \cref{lem:O-is-an-induced-subgraph-of-M}), and $C$ is a cycle in $M$, for any edge $(u_i, u_{i+1})$ in $C$, if $u_i, u_{i+1} \in V_O$, then $(u_i, u_{i+1}) \in O$. More specifically, $(u_1, u_2), (u_l, u_1) \in E_O$.

    We enumerate the edges in $C$ present in $O$ in the order they appear starting from $(u_1, u_2)$, ending with $(u_l, u_1)$. This enumeration is represented as $L = ((x_1, y_1), (x_2, y_2), \ldots, (x_m, y_m))$, where $(x_1, y_1) = (u_1, u_2)$ and $(x_m, y_m) = (u_l, u_1)$. Notably, $m \geq 2$ since the list contains the edges $(u_1, u_2)$ and $(u_l, u_1)$. Additionally, let $(x_{m+1}, y_{m+1}) = (x_1, x_2)$, and define $(P_1, P_2) = \EPF{O, P_{11}, P_{12}, P_{21}, P_{22}}$. We present the following claims for the edges in this enumeration:
    
    \begin{enumerate}
        \item For $1 \leq i \leq m$, if $y_i = x_{i+1}$, then $(x_i, y_i, y_{i+1})$ is a \tfp{} in $O$, and $P_1((x_i, y_i), (x_{i+1}, y_{i+1})) = 1$.
        
        \item For $1 \leq i \leq m$, if $y_i \neq x_{i+1}$, then either there exists a \tfp{} from $(x_i, y_i)$ to $(x_{i+1}, y_{i+1})$ in $M_1$, or there exists a \tfp{} from $(x_i, y_i)$ to $(x_{i+1}, y_{i+1})$ in $M_2$, and $P_1((x_i, y_i), (x_{i+1}, y_{i+1})) = 1$.
    \end{enumerate}
    
    \begin{proof}[Proof of the claim]
        For any $1 \leq i \leq m$ where $y_i = x_{i+1}$, due to the chordless nature of $C$ and the presence of $(x_i, y_i)$ and $(x_{i+1}, y_{i+1})$ as edges in $O$, it follows that $(x_i, y_i, y_{i+1})$ is a \tfp{} in $O$. Consequently, considering that $(P_1, P_2) = \EPF{O, P_{11}, P_{12}, P_{21}, P_{22}}$, by \cref{item-1-of-def:extended-path-function} of \cref{def:extended-path-function}, $P_1((x_i, y_i), (x_{i+1}, y_{i+1}))$ is initialized to $1$ due to the presence of a \tfp{} from $(x_i, y_i)$ to $(x_{i+1}, y_{i+1})$ in $O$. This validates the first part of the claim.

        Suppose for some $1 \leq i \leq m$, $y_i \neq x_{i+1}$. Then, $y_i$ cannot be part of $I$. Otherwise, the neighbors of $y_i$ would be within $O$ ($I \cup N(I, H) \subseteq V_O$). That means, for $1\leq j \leq l$, if $y_i = u_j$ then the next node $u_{j+1}$ in $C$ is also in $V_O$. This implies $(u_j, u_{j+1})$ is in $O$, as $C$ is a cycle in $M$, $u_j, u_{j+1} \in V_O$, and $O$ is an induced subgraph of $M$. This implies the next edge after $(x_i, y_i)$ in $L$ is $(u_j, u_{j+1}) = (x_{i+1}, y_{i+1})$. But, then $y_i = x_{i+1}$, a contradiction. This implies that $y_i$ must either be in $V_{M_1} \setminus I$ or in $V_{M_2} \setminus I$. Without loss of generality, assume that $y_i \in V_{M_1} \setminus I$. Consequently, $x_{i+1}$ must not be in $V_{M_2} \setminus I$, or else the path between $y_i$ and $x_{i+1}$ in $C$ would contain an $I$ node, as $I$ serves as a vertex separator between $V_{H_1} \setminus I$ and $V_{H_2} \setminus I$. This implies that within the path from $y_i$ to $x_{i+1}$ in $C$, there exists an edge from $C$ in $O$ due to the containment of $I \cup N(I, H) \subseteq V_O$ and the fact that $C$ is a cycle in $M$, and $O$ is an induced subgraph of $M$. Similarly, $x_{i+1}$ must not be in $I$; otherwise, the edge right before $(x_{i+1}, y_{i+1})$ in $C$ would be in $O$ as well (since $I \cup N(I,H) \subseteq V_O$), but the edge prior to $(x_{i+1}, y_{i+1})$ in $O$ is $(x_i, y_i)$. This implies that $(x_i , y_i)$ is right before $(x_{i+1}, y_{i+1})$ in $C$, leading to $y_i = x_{i+1}$, contradicting the assumption. Hence, $x_{i+1} \in V_{M_1} \setminus I$, which further implies $y_{i+1} \in V_{M_1}$ since a neighbor of a node in $V_{M_1} \setminus I$ within $V_O$ must belong to $V_{M_1}$, and $y_{i+1}$ is the neighbor of $x_{i+1}$ and lies in $V_O$ (due to $(x_{i+1}, y_{i+1}) \in O$). Thus, $x_i$, $y_i$, $x_{i+1}$, and $y_{i+1}$ all belong to $M_1$. Furthermore, the path from $(x_i, y_i)$ to $(x_{i+1}, y_{i+1})$ must exclusively lie within $M_1$, otherwise there would be an edge in $O$ between $(x_i, y_i)$ and $(x_{i+1}, y_{i+1})$ in $C$, a contradiction. With our construction of $M$, for any $u,v \in V_{M_1}$, if $(u,v) \in M$, then $(u,v) \in M_1$. This implies that the \cp{} from $(x_i, y_i)$ to $(x_{i+1}, y_{i+1})$ in $C$ is also a \cp{} in $M_1$, meaning $P_{a1}((x_i, y_i), (x_{i+1}, y_{i+1})) = 1$ (since every \cp{} is a \tfp{}). Also, $y_1 \rightarrow x_i \notin M$, otherwise $C$ wouldn't possess the edge $(x_i, y_i)$. Therefore, by \cref{item-2-of-def:extended-path-function} of \cref{def:extended-path-function}, $P_1((x_i, y_i), (x_{i+1}, y_{i+1})) = 1$.
    \end{proof}

    It follows that for all $1 \leq i \leq m$, $P_1((x_i, y_i), (x_{i+1}, y_{i+1})) = 1$. Now, we claim that $P_1((x_1, y_1), (x_m, y_m)) = 1$. Suppose this isn't true and select the smallest $j \leq m$ for which $P_1((x_1, y_1), (x_j, y_j)) = 0$. Given the previous discussion, since $P_1((x_i, y_i), (x_{i+1}, y_{i+1})) = 1$ for all $1\leq i \leq m$, it implies that $j > 2$. This further leads to $P_1((x_1, y_1), (x_{j-1}, y_{j-1})) = P_1((x_{j-1}, y_{j-1}), (x_j, y_j)) = 1$. But, from \cref{item-3-of-def:extended-path-function} of \cref{def:extended-path-function}, we get $P_1((x_1, y_1), (x_i, y_i)) = 1$, a contradiction. Hence, $P_1((x_1, y_1), (x_m, y_m)) = 1$. Additionally, from our earlier claim, $P_1((x_m, y_m), (x_{m+1}, y_{m+1}) = (x_1, y_1)) =1$.

    According to the assumption of \cref{obs2:O-structure-for-existence-of-MEC}, $O$ is an extension of $(O_1, P_{11}, P_{12}, O_2, P_{21}, P_{22})$. Thus, by \cref{def:extension-of-O1-O2-P11-P12-P21-P22}, $\EPF{O, P_{11}, P_{12}, P_{21}, P_{22}}$ is valid (as per \cref{def:valid-epfs}). However, the fact that $P_1((x_1, y_1), (x_m, y_m)) = P_1((x_m, y_m), (x_1, y_1)) = 1$ demonstrates that $\EPF{O, P_{11}, P_{12}, P_{21}, P_{22}}$ is invalid, leading to a contradiction. Consequently, the existence of a chordless cycle in $M$ is not possible.
\end{proof}

\begin{proof}[\textbf{Proof of \cref{lem:UCC-of-M-chordal}}]
\label{proof-of-lem:UCC-of-M-chordal}
Assume, for the sake of contradiction, that \cref{lem:UCC-of-M-chordal} is false, implying the existence of an undirected connected component denoted as $\mathcal{C}$ in $M$ that is not chordal. Since $\mathcal{C}$ is not chordal, there must exist an undirected chordless cycle $C = (u_1, u_2, \ldots, u_l, u_{l+1} = u_1)$, where $l$ is greater than or equal to 4. This implies that no edges exist between non-adjacent nodes of $C$ within $\mathcal{C}$. We consider two possibilities: (a) there is an edge between non-adjacent nodes of $C$ in $M$, but it is directed; or (b) no edge, either directed or undirected, exists between non-adjacent nodes of $C$ in $M$.

Let's analyze these cases. Firstly, suppose there is a directed edge between two non-adjacent nodes of $C$ in $M$. This would create a directed cycle in $M$, contradicting \cref{lem:M-is-chain-graph} which states that $M$ is a chain graph.

Next, consider the second possibility, where there are no edges, neither directed nor undirected, between non-adjacent nodes in $M$. As a result, $C$ would be a chordless cycle in $M$. However, this contradicts \cref{obs:no-chordless-cycle-in-M}, which establishes that $M$ does not contain any chordless cycles. 

Consequently, neither of these cases is feasible. Hence, we conclude that \cref{lem:UCC-of-M-chordal} holds true, and there are no undirected connected components in $M$ that are not chordal.
\end{proof}

\begin{proof}[\textbf{Proof of \cref{lem:a->b-c-does-not-exist}}]
\label{proof-of-lem:a->b-c-does-not-exist}
We will first demonstrate that there are no induced subgraphs of the form $a\rightarrow b-c$ within $M[V_{M_1}]$ and $M[V_{M_2}]$.

Let's consider the scenario where there is an induced subgraph of the form $a\rightarrow b-c$ within $M[V_{M_i}]$ for some $i\in \{1,2\}$. Based on the construction of $M$, either $a\rightarrow b-c$ or $a-b-c$ is an induced subgraph of $M_i$. According to \cref{item-3-theorem-nec-suf-cond-for-MEC} of \cref{thm:nes-and-suf-cond-for-chordal-graph-to-be-an-MEC}, $a\rightarrow b-c$ cannot be an induced subgraph of $M_i$.

If $a-b-c \in M_i$, then $a$, $b$, and $c$ all belong to an undirected connected component $\mathcal{C}$ of $M_i$, and $(a,b,c)$ forms a chordless path from $(a,b)$ to $(b,c)$ within $\mathcal{C}$. Now, considering \cref{item-2-of-obs:when-u-v-is-directed-reverse}, since $(a,b,c)$ is a chordless path and $a\rightarrow b \in M$, it follows that $b\rightarrow c \in M$. This leads to a contradiction since we assumed that $b-c \in M$. Thus, we can conclude that there cannot exist an induced subgraph of the form $a\rightarrow b-c$ within $M[V_{M_i}]$.

The above analysis implies that if there exists an induced subgraph $a\rightarrow b-c$ in $M$, then all the nodes of the induced subgraph cannot belong exclusively to $V_{M_1}$ or $V_{M_2}$. This implies that at least one node must be in $V_{M_1} \setminus I$ (where $I = V_{H_1} \cap V_{H_2}$), and one node must be in $V_{M_2} \setminus I$. Since $I$ acts as a separator between $V_{H_1} \setminus I$ and $V_{H_2} \setminus I$, the induced subgraph must contain at least one node from $I$. Therefore, we have two possibilities: (i) $a \in V_{M_1} \setminus I$, $b \in I$, and $c \in V_{M_2} \setminus I$, or (ii) $a \in V_{M_2} \setminus I$, $b \in I$, and $c \in V_{M_1} \setminus I$. In both cases, $b$ belongs to $I$.

Consequently, we have $a, b, c \in I \cup N(I, H) \subseteq V_O$. This implies that the induced subgraph is in $O$ as $O$ is an induced subgraph of $M$. However, this contradicts the fact that $O$ is a partial MEC, as stated in \cref{item-2-of-def:partial-MEC} of \cref{def:partial-MEC}. Nonetheless, $O$ is an extension of $(O_1, P_{11}, P_{12}, O_2, P_{21}, P_{22})$, and \cref{def:extension-of-O1-O2-P11-P12-P21-P22} implies that $O$ is a partial MEC. This creates a contradiction. Hence, there cannot be any such induced subgraph in $M$. Thus, \cref{lem:a->b-c-does-not-exist} is validated.
\end{proof}

\begin{proof}[\textbf{Proof of \cref{lem:every-directed-edge-of-M-is-strongly-ptotected}}]
\label{proof-of-lem:every-directed-edge-of-M-is-strongly-ptotected}
Suppose $u\rightarrow v \in M$. Since, from the construction, $\skel{M} = \skel{M_1} \cup \skel{M_2} = H_1 \cup H_2$, there are two possibilities: either $u-v\in E_{H_1}$ or $u-v\in E_{H_2}$. Without loss of generality, we assume $u-v\in E_{H_1}$. From the construction of $M$, either $u\rightarrow v \in M_1$ or $u-v \in M_1$. Because if $v\rightarrow u \in M_1$ then due to the initialization step, line \ref{algconstructMEC:M-initialization} of \cref{alg:constructMEC}, we have $v\rightarrow u 
\in M$, which contradicts our assumption that $u\rightarrow v \in M$. We will show that in both cases, $u\rightarrow v$ is strongly protected in $M$. We will consider each possibility in turn:

\begin{enumerate}
    \item Suppose $u\rightarrow v \in M_1$. Then $u\rightarrow v$ is strongly protected in $M_1$ according to \cref{item-4-theorem-nec-suf-cond-for-MEC} of \cref{thm:nes-and-suf-cond-for-chordal-graph-to-be-an-MEC}. As per the construction of $M$, all the directed edges in $M_1$ are also directed in $M$ (due to line \ref{algconstructMEC:M-initialization} of \cref{alg:constructMEC}). Now, let's consider the induced subgraphs shown in (a), (b), (c), and (d) of \cref{fig:strongly-protected-edge}.

If $u\rightarrow v$ is strongly protected in $M_1$ due to its inclusion in an induced subgraph shown in (a), (b), or (c) of \cref{fig:strongly-protected-edge}, then the fact that all the edges in those induced subgraphs are directed implies that the induced subgraph is also an induced subgraph of $M$ (due to line \ref{algconstructMEC:M-initialization} of \cref{alg:constructMEC}). This further implies $u\rightarrow v$ is also strongly protected in $M$.

Now, let's focus on the case where $u\rightarrow v$ is strongly protected in $M_1$ because it is part of the induced subgraph shown in \cref{fig:strongly-protected-edge}-(d). In this induced subgraph, the edges $w\rightarrow v$, $w'\rightarrow v$, and $u\rightarrow v$ are directed, indicating that $w\rightarrow v$, $w'\rightarrow v$, and $u\rightarrow v$ are also present in $M$ due to the construction of $M$.
To establish the strong protection of $u\rightarrow v$ in $M$, we consider different cases. If $w-u$ and $w'-u$ are present in $M$, then $u\rightarrow v$ is part of the induced subgraph shown in \cref{fig:strongly-protected-edge}-(d), making it strongly protected in $M$. If $u\rightarrow w$ or $u\rightarrow w'$ exists in $M$, then again $u\rightarrow v$ is part of an induced subgraph, either as shown in \cref{fig:strongly-protected-edge}-(c) with $u\rightarrow w\rightarrow v\leftarrow u$, or as shown in \cref{fig:strongly-protected-edge}-(c) with $u\rightarrow w'\rightarrow v\leftarrow u$. In both cases, $u\rightarrow v$ is strongly protected in $M$.

The remaining possibility is that $w\rightarrow u$ or $w'\rightarrow u$ are present in $M$. 
Since $w-u-w'$ is present in $M_1$, it implies that $w$, $u$, and $w'$ belong to the same undirected connected component $\mathcal{C}$ of $M_1$. Now, considering the fact that $w-u-w'$ forms a chordless path in $\mathcal{C}$, according to \cref{item-2-of-obs:when-u-v-is-directed-reverse}, if $w\rightarrow u$ exists in $M$, then $u\rightarrow w'$ must also be present in $M$. Similarly, if $w'\rightarrow u \in M$, then $u\rightarrow w \in M$. And, as shown above, if $u\rightarrow w$ or $u\rightarrow w'$ exists in $M$, then  $u\rightarrow v$ is strongly protected. 

In conclusion, we have shown that in all possible cases, if $u\rightarrow v$ is strongly protected in $M_1$, then it is also strongly protected in $M$. Hence, the claim is verified.

\item Suppose $u-v\in M_1$ and $u\rightarrow v \in M$. 
From \cref{obs:when-u-v-becomes-directed}, $u\rightarrow v \in M$ either because (a) $u\rightarrow v$ is strongly protected in $O$, (b) there exists $x\rightarrow y \in O$ with $(u,v) \neq (x,y)$ and a \cp{} from $(x,y)$ to $(u,v)$ in $M_1$, or (c) there exists $x\rightarrow y \in O$ and \cps{} in $M_1$ from $(x,y)$ to $v$ and from $(v,u)$ to $x$, with $y\neq v$ and $u\neq x$. We will consider each possibility in detail.

\begin{enumerate}
    \item If $u\rightarrow v$ is strongly protected in $O$, then since $O$ is an induced subgraph of $M$ (from \cref{lem:O-is-an-induced-subgraph-of-M}), if $u\rightarrow v$ is strongly protected in $O$, then it is strongly protected in $M$.

    \item Suppose there exists an edge $x\rightarrow y \in O$ with a \cp{} $Q = (u_1=x, u_2=y, \ldots, u_{l-1}=u, u_l=v)$ from $(x,y)$ to $(u,v)$ in $\mathcal{C}$ where $(x,y)\neq (u,v)$. According to \Cref{item-2-of-obs:when-u-v-is-directed-reverse}, $u_{l-2}\rightarrow u_{l-1}\rightarrow u_l \in M$. This implies that $u\rightarrow v$ is strongly protected in $M$ as it is part of an induced subgraph shown in \cref{fig:strongly-protected-edge}-(a).

    \item Suppose there are \cps{} $Q_1= (u_1 =x, u_2=y, \ldots, u_l=u$ in $M_1$ from $(x,y)$ to $u$, and $Q_2 = (v_1=v, v_2=u, \ldots, v_m= x)$ from $(v,u)$ to $x$, where $y\neq u$ and $u\neq x$. All the edges of $Q_1$ and $Q_2$ must be undirected in $M_1$, otherwise, concatenating $Q_1$ and $Q_2$ would create a directed cycle in $M_1$, contradicting \cref{item-1-theorem-nec-suf-cond-for-MEC} of \cref{thm:nes-and-suf-cond-for-chordal-graph-to-be-an-MEC}. This implies that all nodes in $Q_1$ and $Q_2$ belong to an undirected connected component $\mathcal{C}$ of $M_1$.

Let $\tau$ be the LBFS ordering returned by \cref{alg:LBFSwithO} when called at line~\ref{algconstructMEC:contructLBFS} of \cref{alg:constructMEC} for the undirected connected component $\mathcal{C}$ of $M_1$. According to \cref{item-2-of-obs:when-u-v-is-directed-reverse}, $\tau(u_1) < \tau(u_2) < \ldots < \tau(u_l)$, and for all $1\leq i<l-1$, $u_i\rightarrow u_{i+1} \in M$. According to \cref{obs:directed-edge-respects-tau}, $\tau(u) < \tau(v)$. Based on \cref{obs:LBGS-gives-PEO}, $u_{l-1}-u \in \mathcal{C}$.

If $u\rightarrow u_{l-1} \in M$, then $u\rightarrow v$ is strongly protected in $M$ due to the induced subgraph $u\rightarrow u_{l-1}\rightarrow v \leftarrow u$ (as shown in \cref{fig:strongly-protected-edge}-(c)).

Suppose $u\rightarrow u_{l-1} \notin M$. Then there must exist $j \geq 2$ such that $v_{j+1}\rightarrow v_j \in M$, otherwise, $C = (u_1 = x, u_2, \ldots, u_{l-1}, v_2, \ldots, v_m=x)$ is a directed cycle in $M$, contradicting \cref{lem:M-is-chain-graph}. In this case, $u_{j+1} \rightarrow u_j \in M$ for some $j\geq 2$. Then $Q_2' = (u_{j+1}, u_j, \ldots, u_2, u_1)$ is a chordless path and $u_{j+1}\rightarrow u_j \in M$. Then, according to \cref{item-2-of-obs:when-u-v-is-directed-reverse}, $v_3\rightarrow v_2\rightarrow v_1 \in M$. This implies that $u\rightarrow v$ is strongly protected in $M$ as it is part of an induced subgraph $v_3\rightarrow v_2=u \rightarrow v_1 =v$ (as shown in \cref{fig:strongly-protected-edge}-(a)).
\end{enumerate}

Thus, we have shown that in all possible cases, $u\rightarrow v$ is strongly protected in $M$. This verifies the claim.

\end{enumerate}
This completes the proof of \cref{lem:every-directed-edge-of-M-is-strongly-ptotected}.
\end{proof}

\begin{proof}[\textbf{Proof of \cref{lem:M1-and-M2-are-projections-of-M}}]
Without loss of generality, we prove that $\mathcal{P}(M, V_{H_1}) = M_1$. Similarly, we can claim that $\mathcal{P}(M, V_{H_2}) = M_2$, which completes the proof of \cref{lem:M1-and-M2-are-projections-of-M}.

Since $M_1$ is an MEC of $H_1$, it suffices to show that $\mathcal{V}(M_1) = \mathcal{V}(M[V_{H_1}])$ (from \cref{def:projection}).
From the initialization step, all directed edges of $M_1$ are also directed edges of $M$. This implies $\mathcal{V}(M_1) \subseteq \mathcal{V}(M[V_{H_1}])$ as from the construction of $M$, $\skel{M[V_{M_1}]} = \skel{M_1}$. Thus, to establish $\mathcal{V}(M_1) = \mathcal{V}(M[V_{H_1}])$, we only need to demonstrate $\mathcal{V}(M[V_{H_1}]) \subseteq \mathcal{V}(M_1)$.

Suppose $\mathcal{V}(M[V_{H_1}]) \not\subseteq \mathcal{V}(M_1)$. This implies the existence of $u,v,w \in V_{H_1}$ such that $u\rightarrow v \leftarrow w \in M$ and $u\rightarrow v \leftarrow w \notin M_1$. 
From the construction of $M$, if $u\rightarrow v \in M$, then either $u\rightarrow v \in M_1$ or $u-v\in M_1$ (because if $v\rightarrow u \in M_1$ from line \ref{algconstructMEC:M-initialization} of \cref{alg:constructMEC}, $v\rightarrow u \in M$). 
Similarly, if $w\rightarrow v\in M$, then either $w\rightarrow v \in M_1$ or $w-v \in M_1$. Since $u\rightarrow v \leftarrow w \notin M_1$, either $u\rightarrow v-w \in M_1$ or $u-v\leftarrow w \in M_1$ or $u-v-w \in M_1$. From \cref{item-3-theorem-nec-suf-cond-for-MEC} of \cref{thm:nes-and-suf-cond-for-chordal-graph-to-be-an-MEC}, the first two possibilities are not possible. Thus, the only remaining possibility is that $u-v-w$ is an induced subgraph of $M$.
However, this contradicts \cref{item-2-of-obs:when-u-v-is-directed-reverse}, as if $u\rightarrow v \in M$, then $v\rightarrow w \in M$, which contradicts the presence of $w\rightarrow v \in M$.
This confirms that there are no such $u,v,w \in V_{H_1}$ for which $\mathcal{V}(M[V_{H_1}]) \not\subseteq \mathcal{V}(M_1)$. 

Thus, we have established $\mathcal{V}(M[V_{H_1}]) \subseteq \mathcal{V}(M_1)$. Combined with the earlier inclusion that $\mathcal{V}(M_1) \subseteq \mathcal{V}(M[V_{H_1}])$, we conclude that $\mathcal{V}(M_1) = \mathcal{V}(M[V_{H_1}])$. 

Hence, we have successfully shown that $\mathcal{P}(M, V_{H_1}) = M_1$, completing the proof of \cref{lem:M1-and-M2-are-projections-of-M}.
\end{proof}

\begin{proof}[\textbf{Proof of \cref{lem:uniqueness-of-M}}]
    Suppose there exists an MEC $M'$ such that $M'[V_O] = O$, and $\mathcal{P}(M', V_{M_1}, V_{M_2}) = (M_1, M_2)$. We will show that $M' = M$.

We know that each MEC of a graph has a unique set of v-structures. Our aim is to demonstrate that both $M$ and $M'$ contain the same set of v-structures. This implies that $M = M'$.

From \cref{lem:O-is-an-induced-subgraph-of-M}, we have $M[V_O] = O$. Also, from \cref{lem:M1-and-M2-are-projections-of-M}, we know that $\mathcal{P}(M, V_{M_1}, V_{M_2}) = (M_1, M_2)$. Since both $M$ and $M'$ are MECs of $H$, and $\skel{H} = \skel{H_1} \cup \skel{H_2}$, it follows that for any v-structure $u\rightarrow v \leftarrow w$ that belongs to either $M$ or $M'$, one of the following cases holds: 
(a) $u,v,w \in V_{M_1}$,
(b) $u,v,w \in V_{M_2}$,
(c) $v \in I = V_{M_1}\cap V_{M_2}$, $u\in V_{M_1}\setminus I$, and $w\in V_{M_2}\setminus I$,
(d) $u\in V_{M_2}\setminus I$, $v \in I$, and $w\in V_{M_1}\setminus I$.

We now address each of these cases.

For case (a), if $u\rightarrow v \leftarrow w$ is a v-structure in $M$ with $u,v,w \in V_{M_1}$, since $M_1$ is a projection of $M$ on $V_{M_1}$, the v-structure $u\rightarrow v \leftarrow w$ also belongs to $M_1$. Similarly, as $M_1$ is a projection of $M'$ on $V_{M_1}$, the same v-structure belongs to $M'$. A similar argument can be made for v-structures in $M'$.   The same argument holds for v-structures when $u,v,w \in V_{M_2}$ (i.e., for case (b)).

For case (c), if $u\rightarrow v \leftarrow w$ is a v-structure in $M$ with $u\in V_{M_1}\setminus I$, $v \in I$, and $w\in V_{M_2}\setminus I$, this v-structure is in $O$ as $I \cup N(I, H) \subseteq V_O$, and $O$ is an induced subgraph of $M$. Since $O$ is also an induced subgraph of $M'$, the v-structure also belongs to $M'$. A similar argument can be made for v-structures in $M'$ with $u\in V_{M_1}\setminus I$, $v \in I$, and $w\in V_{M_2}\setminus I$. The same argument holds for v-structures when $u\in V_{M_2}\setminus I$, $v \in I$, and $w\in V_{M_1}\setminus I$ (i.e., for case (d)).

The above analysis demonstrates that a v-structure $u\rightarrow v \leftarrow w$ is present in $M$ if, and only if, it is also present in $M'$. This implies that both $M$ and $M'$ have the same set of v-structures, and therefore, $M = M'$. This completes the proof of \cref{lem:uniqueness-of-M}.

\end{proof}

\begin{proof}[\textbf{Proof of \cref{lem:derived-path-function-is-a-part-of-shadow-of-M}}]
    Let $\EPF{O, P_{11}, P_{12}, P_{21}, P_{22}} = (P_1, P_2)$, and $(O', P_1', P_2')$ be the shadow of $M$ on $V_O$. We aim to show $(O, P_1, P_2) = (O', P_1', P_2')$. 
    By definition (\cref{def:shadow}), $O' = M[V_O]$.
    Furthermore, from \cref{lem:O-is-an-induced-subgraph-of-M}, $O = M[V_O]$. Thus, we conclude that $O = O'$.
   
    We are given that the shadow of $M_1$ on $S_1\cup N(S_1, H_1)$ is $(O_1, P_{11}, P_{12})$, and the shadow of $M_2$ on $S_2 \cup N(S_2, H_2)$ is $(O_2, P_{21}, P_{22})$. Also, we are given $O \in \mathcal{E}(O_1, P_{11}, P_{12}, O_2, P_{21}, P_{22})$. Furthermore, from \cref{lem:M1-and-M2-are-projections-of-M}, we have $\mathcal{P}(M, V_{H_1}, V_{H_2}) = (M_1, M_2)$. Therefore, based on \cref{obs:nes-conditions-of-the-shadow}, we conclude that $(P_1', P_2') = (P_1, P_2)$. This shows that $(O, P_1, P_2)$ is the shadow of $M$ on $V_O$, i.e., $M \in \setofMECs{H, O, P_1, P_2}$.
\end{proof}

\subsection{\texorpdfstring{Finding Size of $\setofMECs{H, O, P_1, P_2}$}{Finding Size of MEC(H, O, P\_1, P\_2)}}
\label{subsection:finding-size-of-MEC-H}
We use \cref{def:extension-of-O1-O2-P11-P12-P21-P22} to define a one to one function $f_{O, P_1, P_2}$. We use the function for the computation of $|\setofMECs{H, O, P_1,P_2}|$. 
\begin{definition}
\label{def:one-to-one-function-for-counting-MEC-H-O-P1-P2}
Let $H$ be an undirected graph, and let $H_1$ and $H_2$ be two induced subgraphs of $H$ such that $H = H_1 \cup H_2$, and  $I = V_{H_1} \cap V_{H_2}$ is a vertex separator of $H$ that separates $V_{H_1} \setminus I$ and $V_{H_2} \setminus I$. Define subsets $S_1$ and $S_2$ of $V_{H_1}$ and $V_{H_2}$, respectively, such that $S_1 \cap S_2 = I$. Let $A = H[S_1 \cup S_2 \cup N(S_1 \cup S_2, H)]$, $B_1 = H_1[S_1 \cup N(S_1, H_1)]$, and $B_2 = H_2[S_2 \cup N(S_2, H_2)]$. Consider $(O, P_1, P_2) \in \shadowofudgraph{A}$. Define $T_1 = \setofMECs{H, O, P_1, P_2}$ and $T_2 = \{(M_1, M_2) :$ $\exists (O_1, P_{11}, P_{12}, O_2, P_{21}, P_{22})$ such that $(O_1, P_{11}, P_{12}) \in \shadowofudgraph{B_1}$, $(O_2, P_{21}, P_{22}) \in \shadowofudgraph{B_2}$, $(O, P_1, P_2) \in \mathcal{E}(O_1, P_{11}, P_{12}, O_2, P_{21}, P_{22})$, $M_1 \in \setofMECs{H_1, O_1, P_{11}, P_{12}}$, and $M_2 \in \setofMECs{H_2, O_2, P_{21}, P_{22}}\}$. The function $f_{O, P_1, P_2} : T_1 \rightarrow T_2$ is defined as $f_{O, P_1, P_2}(M) = (M_1, M_2)$ if $\mathcal{P}(M, V_{H_1}, V_{H_2}) = (M_1, M_2)$.
\end{definition}

 \begin{lemma}
 \label{lem:f-is-one-to-one}
 $f_{O, P_1, P_1}$ is bijective.
 \end{lemma}
 \begin{proof}
 We first prove $\rightarrow$ of \cref{lem:f-is-one-to-one}.
 \begin{claim}
 \label{claim:to-part-of-lem:f-is-one-to-one}
 For each $M \in \setofMECs{H, O, P_1, P_2)}$, there exists a unique $(M_1, M_2)$ such that $f_{O, P_1, P_1}(M) = (M_1, M_2)$.
\end{claim}
 \begin{proof}
 Let $M \in \setofMECs{H, O, P_1, P_2)}$. From \cref{def:subsets-of-MEC-based-on-O-P1-and-P2}, $M$ is an MEC of $H$ and its shadow on $V_O$ is $(O, P_1, P_2)$. Let $\mathcal{P}(M, V_{H_1}, V_{H_2}) = (M_1, M_2)$.   From \cref{corr:projection-of-an-MEC-is-unique}, there exists a unique $(M_1, M_2)$ such that $\mathcal{P}(M, V_{H_1}, V_{H_2}) = (M_1, M_2)$. Thus, we only have to show that $(M_1, M_2) \in T_2$. Let $(O_1, P_{11}, P_{12}) \in \shadowofudgraph{B_1}$ and $(O_2, P_{21}, P_{22}) \in \shadowofudgraph{B_2}$. Therefore, from \cref{def:subsets-of-MEC-based-on-O-P1-and-P2}, $M_1 \in \setofMECs{H_1, O_1, P_{11}, P_{12}}$ and $M_2 \in \setofMECs{H_2, O_2, P_{21}, P_{22}}$. From \cref{obs:nes-conditions-of-the-shadow}, $(O, P_1, P_2) \in \mathcal{E}(O_1, P_{11}, P_{12}, O_2, P_{21}, P_{22})$. From the definition of $T_2$ (given at \cref{def:one-to-one-function-for-counting-MEC-H-O-P1-P2}), $(M_1, M_2) \in T_2$. Thus, from \cref{def:one-to-one-function-for-counting-MEC-H-O-P1-P2}, $f_{O, P_1, P_1}(M) = (M_1, M_2)$.
 
\end{proof}
 We now prove $\leftarrow$ of \cref{lem:f-is-one-to-one}.
\begin{claim}
\label{claim:from-part-of-lem:f-is-one-to-one}
Let $(O_1, P_{11}, P_{12})$ be a shadow of $B_1$ and $(O_2, P_{21}, P_{22})$ be a shadow of $B_2$ such that $(O, P_1, P_2)$ is an extension of $(O_1, P_{11}, P_{12}, O_2, P_{21}, P_{22})$. For each $M_1 \in \setofMECs{H_1, O_1, P_{11}, P_{12}}$ and $M_2 \in \setofMECs{H_2, O_2, P_{21}, P_{22}}$, there exists a unique $M \in \setofMECs{H, O, P_1, P_2}$ such that $f_{O, P_1, P_2}(M) = (M_1, M_2)$.
\end{claim}

\begin{proof}
 \Cref{obs2:O-structure-for-existence-of-MEC} proves this.
\end{proof}
\Cref{claim:to-part-of-lem:f-is-one-to-one,claim:from-part-of-lem:f-is-one-to-one} prove \cref{lem:f-is-one-to-one}.
\end{proof}
 
\Cref{lem:f-is-one-to-one}  implies that the following:
 
 \begin{lemma}
\label{lem:equivalence-between-MECs}
 Let $H$ be an undirected graph, and $H_1$ and $H_2$ be two induced subgraphs of $H$ such that $H = H_1 \cup H_2$, and $I = V_{H_1} \cap V_{H_2}$ is a vertex separator of $H$ that separates $V_{H_1}\setminus{I}$ and $V_{H_2}\setminus{I}$. Let $S_1$ and $S_2$ be the subsets of $V_{H_1}$  and $V_{H_2}$, respectively, such that $S_1 \cap S_2 = I$. Let $A = H[S_1 \cup S_2 \cup N(S_1 \cup S_2, H)]$, $B_1 = H_1[S_1 \cup N(S_1, H_1)]$, and $B_2 = H_2[S_2 \cup N(S_2, H_2)]$.  For any $(O, P_1, P_2)\in \shadowofudgraph{A}$,
 \begin{equation}
     |\setofMECs{H, O, P_1, P_2}| = \sum_{ \substack{(O_1, P_{11}, P_{12})\in \shadowofudgraph{B_1} \\ (O_2, P_{21}, P_{22}) \in \shadowofudgraph{B_2}\\(O, P_1,P_2) \in \mathcal{E}(O_1, P_{11}, P_{12}, O_2, P_{21}, P_{22}) }}
     {|\setofMECs{H_1, O_1, P_{11}, P_{12}}| \times |\setofMECs{H_2, O_2, P_{21}, P_{22}}|}
 \end{equation}
 \end{lemma}
\begin{proof}
From \cref{lem:f-is-one-to-one}, $|T_1| = |T_2|$. Since $T_1 = \setofMECs{H, O, P_1, P_2}$, therefore, $|T_1| = | \setofMECs{H, O, P_1, P_2}|$. Since $T_2$ contains tuples $(M_1, M_2)$ such that for some tuple $(O_1, P_{11}, P_{12}, O_2, P_{21}, P_{22})$ such that $(O_1, P_{11}, P_{12}) \in \shadowofudgraph{B_1}$, $(O_2, P_{21}, P_{22}) \in \shadowofudgraph{B_2}$,  and $(O, P_1, P_2) \in \mathcal{E}(O, P_{11}, P_{12}, O_2, P_{21}, P_{22})$, $M_1 \in \setofMECs{H_1, O_1, P_{11}, P_{12}}$ and $M_2 \in \setofMECs{H_2, O_2, P_{21}, P_{22}}$. This implies
\begin{equation}
    T_2 = \bigcup_{\substack{(O_1, P_{11}, P_{12})\in \shadowofudgraph{B_1} \\ (O_2, P_{21}, P_{22}) \in \shadowofudgraph{B_2}\\(O, P_1,P_2) \in \mathcal{E}(O_1, P_{11}, P_{12}, O_2, P_{21}, P_{22}) }}{\setofMECs{H_1, O_1, P_{11}, P_{12}} \times \setofMECs{H_2, O_2, P_{21}, P_{22}}}.
\end{equation}
We show that sets $\setofMECs{H_1, O_1, P_{11}, P_{12}} \times \setofMECs{H_2, O_2, P_{21}, P_{22}}$ are disjoint. Suppose  sets $\setofMECs{H_1, O_1, P_{11}, P_{12}} \times \setofMECs{H_2, O_2, P_{21}, P_{22}}$ are not disjoint. Then there exists $(M_1, M_2)$ such that  $(M_1, M_2) \in  \setofMECs{H_1, O_1, P_{11}, P_{12}} \times \setofMECs{H_2, O_2, P_{21}, P_{22}}$ as well as $(M_1, M_2) \in  \setofMECs{H_1, O_1', P_{11}', P_{12}'} \times \setofMECs{H_2, O_2', P_{21}', P_{22}'}$. 
This implies $M_1 \in \setofMECs{H_1, O_1, P_{11}, P_{12}}$ as well as $M_1 \in \setofMECs{H_1, O_1', P_{11}', P_{12}'}$. But, from \cref{def:shadow}, shadow of $M$ on $V_{B_1}$ is unique, i.e., $(O_1, P_{11}, P_{12}) =  (O_1', P_{11}', P_{12}')$. Similarly, we can say that  $(O_2, P_{21}, P_{22}) =  (O_2', P_{21}', P_{22}')$. This shows $(O_1, P_{11}, P_{12}, O_2, P_{21}, P_{22}) = (O_1', P_{11}', P_{12}', O_2', P_{21}', P_{22}')$. This implies that the sets $\setofMECs{H_1, O_1, P_{11}, P_{12}} \times \setofMECs{H_2, O_2, P_{21}, P_{22}}$ are disjoint. This further implies
\begin{equation}
    |T_2| = \sum_{\substack{(O_1, P_{11}, P_{12})\in \shadowofudgraph{B_1} \\ (O_2, P_{21}, P_{22}) \in \shadowofudgraph{B_2}\\(O, P_1,P_2) \in \mathcal{E}(O_1, P_{11}, P_{12}, O_2, P_{21}, P_{22}) }}{|\setofMECs{H_1, O_1, P_{11}, P_{12}}| \times |\setofMECs{H_2, O_2, P_{21}, P_{22}}|}.
\end{equation}
This proves \cref{lem:equivalence-between-MECs}.
\end{proof}

 As discussed earlier, $H$ resembles $G_i^j$, $H_1$ resembles $G_i^{j-1}$, $H_2$ resembles $G_{i_j}$, and $S_1$ and $S_2$ resemble $X_i$ and $X_{i_j}$, respectively. 
 Let $A = G_i^j[X_i \cup X_{i_j} \cup N(X_i\cup X_{i_j}, G_i^j)]$, $A' = G_i^j[X_i \cup N(X_i, G_i^j)]$   $B_1 = G_i^{j-1}[X_i \cup N(X_i, G_i^{j-1})]$, and $B_2 = G_{i_j}[X_{i_j} \cup N(X_{i_j}, G_{i_j})]$. 
 \Cref{lem:equivalence-between-MECs} implies that if  for all $(O_1, P_{11}, P_{12}) \in \shadowofudgraph{B_1}$ and $(O_2, P_{21}, P_{22}) \in \shadowofudgraph{B_2}$,
 we have the knowledge about $|\setofMECs{G_i^{j-1}, O_1, P_{11}, P_{12}}|$ and $|\setofMECs{G_{i_j}, O_2, P_{21}, P_{22}}|$ then  for each $(O, P_1, P_2) \in \shadowofudgraph{A}$,
 we can compute $|\setofMECs{G_i^j, O, P_1, P_2}|$. But, for the recursion, we need $|\setofMECs{H, O', P_1', P_2'}|$ for each  $(O', P_1', P_2') \in \shadowofudgraph{A'}$. As $(O, P_1, P_2)$ contains more information than $(O', P_1', P_2')$, we define projection of $(O, P_1, P_2)$ to get  $|\setofMECs{H, O', P_1', P_2'}|$ for each required $(O', P_1', P_2')$.

 \begin{definition}[\textbf{\emph{Projection of a shadow}}]
    \label{def:proj-of-partial-MEC-and-functions}
    Let $H$ be an undirected graph. $S$ and $S'$ are two vertex subsets of $H$ such that $S' \subseteq S \subseteq V_H$. Let $(O, P_1, P_2) \in \shadowofudgraph{G[S]}$ and $(O', P_1', P_2') \in \shadowofudgraph{G[S']}$.
    We say $(O', P_1', P_2')$ is a projection of $(O, P_1, P_2)$ on $S'$ denoted as $\mathcal{P}(O, P_1, P_2, S') = (O', P_1', P_2')$ if 
    \begin{enumerate}
        \item
        \label{item-1-of-def:proj-of-partial-MEC-and-functions}
        $O[S'] = O'$,
        \item
        \label{item-2-of-def:proj-of-partial-MEC-and-functions}
        for $((u, v), (x, y)) \in E_{G[S']} \times E_{G[S']}$, $P_1'((u, v), (x, y)) = P_1 ((u, v), (x, y))$, and
        \item
        \label{item-3-of-def:proj-of-partial-MEC-and-functions}
        for $((u, v), w) \in E_{G[S']} \times V_{G[S']}$, $P_2'((u, v), w) = P_2 ((u, v), w)$.
    \end{enumerate}
\end{definition}

We now  compute $|\setofMECs{H, O', P_1', P_2'}|$.
\begin{lemma}
\label{lem:counting-MECs-corresponding-to-projected-partial-MECs}
Let $H$ be an undirected graph,  $S$ and $S'$ are two vertex subsets of $H$ such that $S' \subseteq S \subseteq V_H$, $A = H[S]$, and $A' = H[S']$.  Let $(O', P_1', P_2')\in \shadowofudgraph{A'}$.  Then, 
\begin{equation}
\label{eq-of-lem:counting-MECs-corresponding-to-projected-partial-MECs}
    |\text{MEC}(H, O', P_1', P_2')| = \sum_{\substack{(O, P_1, P_2)\in \shadowofudgraph{A}\\ (O', P_1', P_2') = \mathcal{P}(O, P_1, P_2, S')}}{|\text{MEC}(H, O, P_1, P_2)|}
\end{equation}
\end{lemma}
\begin{proof}
We first show the following:
\begin{equation}
\label{eq:rel-bn-2-set-of-MECs}
    \text{MEC}(H, O', P_1', P_2') = \bigcup_{\substack{(O, P_1, P_2)\in \shadowofudgraph{A}\\ (O', P_1', P_2') = \mathcal{P}(O, P_1, P_2, S')}}{\text{MEC}(H, O, P_1, P_2)}
\end{equation}
Let $g$ be a function with L.H.S. of \cref{eq:rel-bn-2-set-of-MECs} (i.e., $\setofMECs{H, O', P_1',P_2'}$) as its domain, and R.H.S. of \cref{eq:rel-bn-2-set-of-MECs} as its co-domain. $g$ is defined as $g(M) = M$.
\begin{claim}
\label{claim:g-is-bijective}
$g$ is bijective.
\end{claim}
\begin{proof}
We first prove the $\rightarrow$ of \cref{claim:g-is-bijective}.
\begin{claim}
\label{claim:to-g-is-bijective}
For each $M \in \setofMECs{H, O', P_1', P_2'}$, there exists a unique $(O, P_1, P_2) \in \shadowofudgraph{A}$ such that $\mathcal{P}(O,P_1,P_2, S') = (O',P_1',P_2')$, and $M \in \setofMECs{H, O, P_1, P_2}$.
\end{claim}
\begin{proof}
$M \in \setofMECs{H, O', P_1', P_2'}$.
From \cref{def:shadow}, there exists a unique triple $(O, P_1, P_2)$ such that $(O, P_1, P_2) \in \shadowofudgraph{A}$  and $M \in \setofMECs{H, O, P_1, P_2}$. Only thing that we have to check is whether $\mathcal{P}(O,P_1,P_2) = (O',P_1',P_2')$ or not. From \cref{def:subsets-of-MEC-based-on-O-P1-and-P2}, $O = M[S]$, and $O' = M[S']$. Since $S' \subseteq S$,  $O[S'] = O'$. Similarly, since $S' \subseteq S$, $E_{A'}\times E_{A'} \subseteq E_A\times E_A$. This implies for all $((u,v), (x,y)) \in E_{A'}\times E_{A'}$, $P_1'((u,v), (x,y)) = P_1((u,v),(x,y))$ (as both answers whether there exists a \tfp{} from $(u,v)$ to $(x,y)$ in $M$). Similarly, for all $((u,v), w) \in E_{A'}\times V_{A'}$, $P_2'((u,v), w) = P_2((u,v),w)$. This further  implies that $\mathcal{P}(O,P_1,P_2) = (O',P_1',P_2')$ (from \cref{def:proj-of-partial-MEC-and-functions}). This proves the claim.
\end{proof}
We now prove the $\leftarrow$ of \cref{claim:g-is-bijective}.
\begin{claim}
\label{claim:from-g-is-bijective}
Suppose there exists a triplet $(O, P_1, P_2)$ such that $(O, P_1, P_2) \in \shadowofudgraph{A}$ and $\mathcal{P}(O,P_1,P_2, S') = (O',P_1',P_2')$. Then, for each $M \in \setofMECs{H, O, P_1, P_2)}$, $M\in \setofMECs{H, O', P_1', P_2'}$.
\end{claim}
\begin{proof}
Let $M \in \setofMECs{H, O, P_1, P_2)}$.
From \cref{def:subsets-of-MEC-based-on-O-P1-and-P2}, $M[S] = O$. Since $\mathcal{P}(O,P_1,P_2, S') = (O',P_1',P_2')$, $M[S'] = O[S'] = O'$. Since $S' \subseteq S$, $E_{A'}\times E_{A'} \subseteq E_A\times E_A$. This  implies that for all $((u,v), (x,y)) \in E_{A'}\times E_{A'}$, $P_1'((u,v), (x,y)) = P_1((u,v),(x,y))$, as $\mathcal{P}(O,P_1,P_2, S') = (O',P_1',P_2')$.
This implies for each $((u,v), (x,y)) \in E_{A'}\times E_{A'}$, $P_1'$ correctly depicts that whether there is a \tfp{} from $(u,v)$ to $(x,y)$ in $M$.
Similarly, for all $((u,v), (w) \in E_{A'}\times V_{A'}$, $P_1'((u,v), w)$ correctly depicts that whether there is a \tfp{} from $(u,v)$ to $w$ in $M$. 
The above discussion implies that $(O', P_1', P_2')$ is the shadow of $M$ on $S'$ (from \cref{def:proj-of-partial-MEC-and-functions}).
This  implies that $M \in \setofMECs{H, O, P_1, P_2}$.
\end{proof}
\Cref{claim:to-g-is-bijective,claim:from-g-is-bijective} prove \cref{claim:g-is-bijective}.
\end{proof}
\Cref{claim:g-is-bijective} proves \cref{eq:rel-bn-2-set-of-MECs}.
From \cref{def:shadow}, an MEC $M$ of $H$ belongs to a unique $\setofMECs{H, O, P_1, P_2)}$. 
This implies that sets in the union in the R.H.S. of \cref{eq:rel-bn-2-set-of-MECs} are disjoint.
This further  implies \cref{eq-of-lem:counting-MECs-corresponding-to-projected-partial-MECs}. This completes the proof of \cref{lem:counting-MECs-corresponding-to-projected-partial-MECs}.
\end{proof}

Combining \cref{lem:counting-MECs-corresponding-to-projected-partial-MECs,lem:equivalence-between-MECs}, we get the following lemma:
\begin{lemma}
\label{lem:combination-of-two-lemmas}
 Let $H$ be an undirected graph, and $H_1$ and $H_2$ be two induced subgraphs of $H$ such that $H = H_1 \cup H_2$, and $I = V_{H_1} \cap V_{H_2}$ is a vertex separator of $H$ that separates $V_{H_1}\setminus{I}$ and $V_{H_2}\setminus{I}$. Let $S_1$ and $S_2$ be the subsets of $V_{H_1}$  and $V_{H_2}$, respectively, such that $S_1 \cap S_2 = I$. Let $S = S_1 \cup S_2 \cup N(S_1 \cup S_2, H)$, and $S' = S_1\cup N(S_1, H)$,  $A = H[S]$, $A' = H[S']$, $B_1 = H_1[S_1 \cup N(S_1, H_1)]$, and $B_2 = H_2[S_2 \cup N(S_2, H_2)]$. Then, for any partial MEC $(O', P_1', P_2')\in \shadowofudgraph{A'}$,
 \begin{equation}
 |\text{MEC}(H, O', P_1', P_2')| = \sum_{ \substack{(O, P_1, P_2)\in \shadowofudgraph{A}\\(O_1, P_{11}, P_{12})\in \shadowofudgraph{B_1} \\(O_2, P_{21}, P_{22}) \in \shadowofudgraph{B_2}\\ (O, P_1,P_2) \in \mathcal{E}(O_1, P_{11}, P_{12}, O_2, P_{21}, P_{22}) \\ (O', P_1', P_2') = \mathcal{P}(O, P_1, P_2, S')}}
     {|\setofMECs{H_1, O_1, P_{11}, P_{12}}| \times |\setofMECs{H_2, O_2, P_{21}, P_{22}}|}.
 \end{equation}
\end{lemma}  \section{Algorithm for counting MECs of a graph}
\label{sec:algorithm-for-counting-MECs}
Using \cref{lem:combination-of-two-lemmas}, we can compute $|\setofMECs{H, O', P_1', P_2'}|$ for any shadow $(O', P_1', P_2')\in \shadowofudgraph{H[S_1\cup N(S_1, H)]}$. Since $H$ resembles $G_i^j$, the above discussion implies that we can count the MECs of $G_i^j$. This gives us a recursive method to compute the MECs of $G$.

\begin{algorithm}
\caption{Count-MEC($G$)}
\label{alg:counting-MEC}
\SetAlgoLined
\SetKwInOut{KwIn}{Input}
\SetKwInOut{KwOut}{Output}
\SetKwFunction{MEC-Construction}{MEC-Construction}
\KwIn{An undirected connected graph $G$}
    \KwOut{ $|MEC(G)|$}
    
   $(\mathcal{X} = \{X_1, X_2 \ldots, X_l\},T) \leftarrow$ a tree decomposition of $G$.
   \label{alg-counting-MEC:tree-decomposition}
   
   $F = $ Count($G, (\mathcal{X}, T), X_1$);
   \label{alg:counting-MEC:call-count-MEC-2}
   
   $sum \leftarrow 0$
   \label{alg:counting-MEC:initialization-of-sum}
   
   $A\leftarrow G[X_1\cup N(X_1, G)]$
   \label{alg:counting-MEC:A}
   
   \ForEach{$(O, P_1, P_2) \in \shadowofudgraph{A} $\label{alg:counting-MEC:foreach-1-start}}
   {
{
            $sum = sum + F(O, P_1, P_2)$
            \label{alg:counting-MEC:update-sum}
        }
}
   \label{alg:counting-MEC:foreach-1-end}
   \KwRet sum
   \label{alg:counting-MEC:return-sum}
\end{algorithm}

 \begin{algorithm}
\caption{\BaseCount{}($G$)}
\label{alg:counting-MEC-of-general-graph-base-case}
\SetAlgoLined
\SetKwInOut{KwIn}{Input}
\SetKwInOut{KwOut}{Output}
\SetKwFunction{Count_MEC}{Count_MEC}
\KwIn{An undirected connected graph $G$}
    \KwOut{A function $F: \shadowofudgraph{G} \rightarrow \mathbb{Z}$ such that for each $(O, P_1, P_2) \in \shadowofudgraph{G}$,  $F(O, P_1, P_2) = |MEC(G, O, P_1, P_2)|$
    }
    
    $F: \shadowofudgraph{G} \rightarrow \mathbb{Z}$ \label{alg-counting-MEC-of-general-graph-base-case:F-construction}
    
    \ForEach{$(O, P_1, P_2)\in \shadowofudgraph{G}$
    \label{alg-counting-MEC-of-general-graph-base-case:for-each-0-start}}
    {
        $F(O, P_1, P_2) \leftarrow 0$ \label{alg-counting-MEC-of-general-graph-base-case:initialization-of-S}
    }\label{alg-counting-MEC-of-general-graph-base-case:for-each-0-end}

     \ForEach{$O \in \setofMECs{G}$ \label{alg-counting-MEC-of-general-graph-base-case:for-each-1-start}}
    {
$(P_1, P_2) \leftarrow$ TFP($O$) \label{alg-counting-MEC-of-general-graph-base-case:P1-P2-init}
        
        $F(O, P_1, P_2) = 1$ \label{alg-counting-MEC-of-general-graph-base-case:update-F}

    }\label{alg-counting-MEC-of-general-graph-base-case:for-each-1-end}
    
    \KwRet $F$ \label{alg-counting-MEC-of-general-graph-base-case:return-statement}
    
\end{algorithm}

 \begin{algorithm}
\caption{Count($G, (\mathcal{X}, T), R_1$)}
\label{alg:counting-MEC-of-general-graph}
\SetAlgoLined
\SetKwInOut{KwIn}{Input}
\SetKwInOut{KwOut}{Output}
\SetKwFunction{Count_MEC}{Count_MEC}
\KwIn{An undirected graph $G$, a tree decomposition $(\mathcal{X}, T)$ of $G$, and a root node $R_1$ of $T$}
    \KwOut{A function $F: \shadowofudgraph{G} \rightarrow \mathbb{Z}$ such that for each $(O, P_1, P_2) \in \shadowofudgraph{G[R_1\cup N(R_1, G)]}$,  $F(O, P_1, P_2) = |MEC(G, O, P_1, P_2)|$
    }

    \eIf{\label{alg-counting-MEC-of-general-graph-if-1-start}
    $\mathcal{X} = \{R_1\}$}
    {\KwRet \BaseCount{}($G$)\label{alg-counting-MEC-of-general-graph-return-base-case}
    }
    {\label{alg-counting-MEC-of-general-graph-if-1-end}
        $R_2 \leftarrow$ last child of $R_1$ in $T$. \label{alg-counting-MEC-of-general-graph-R'-initialization}
        
        Cut the edge $R_1-R_2$ of  $T$.\label{alg-counting-MEC-of-general-graph-cut-edge-R-R'}
        
        $T_1 \leftarrow$ induced subtree of $T$ containing $R_1$\label{alg-counting-MEC-of-general-graph-T1-initialization}
        
        $T_2 \leftarrow$ induced subtree of $T$ containing $R_2$\label{alg-counting-MEC-of-general-graph-T2-initialization}
        
        $G_1 \leftarrow$ induced subtree of $G$ represented by $T_1$\label{alg-counting-MEC-of-general-graph-G1-initialization}
        
        $G_2 \leftarrow$ induced subtree of $G$ represented by $T_2$ \label{alg-counting-MEC-of-general-graph-G2-initialization}
        
        $\mathcal{X}_1 \leftarrow$ subset of $X$ containing nodes of $T_1$ \label{alg-counting-MEC-of-general-graph-X1-initialization}
        
        $\mathcal{X}_2 \leftarrow$ subset of $X$ containing nodes of $T_2$ \label{alg-counting-MEC-of-general-graph-X2-initialization}

        $F_1 \leftarrow Count(G_1, (\mathcal{X}_1, T_1), R_1)$ \label{alg-counting-MEC-of-general-graph-F1-initialization}
        
        $F_2 \leftarrow Count(G_2, (\mathcal{X}_2, T_2), R_2)$ \label{alg-counting-MEC-of-general-graph-F2-initialization}

        $A' \leftarrow G[R_1\cup N(R_1, G)]$
        \label{alg-counting-MEC-of-general-graph-A'-initialization}

        $A \leftarrow G[R_1\cup R_2\cup N(R_1\cup R_2, G)]$
        \label{alg-counting-MEC-of-general-graph-A-initialization}
        
        $B_1 \leftarrow G_1[R_1\cup N(R_1, G_1)]$
        \label{alg-counting-MEC-of-general-graph-B1-initialization}
        
        $B_2\leftarrow G_2[R_2\cup N(R_2, G_2)]$
        \label{alg-counting-MEC-of-general-graph-B2-initialization}

        $F: \shadowofudgraph{A'} \rightarrow \mathbb{Z}$ 
        \label{alg-counting-MEC-of-general-graph:S-construction}
        
        \ForEach{$(O', P_1', P_2') \in  \shadowofudgraph{A'}$ \label{alg-counting-MEC-of-general-graph:foreach-0-start}}
        {
            $F(O', P_1', P_2') \leftarrow 0$ \label{alg-counting-MEC-of-general-graph:S-initialization}
        }\label{alg-counting-MEC-of-general-graph:foreach-0-end}
        \ForEach{$O \in \setofpartialMECs{A}$ \label{alg-counting-MEC-of-general-graph:foreach-1-start}}
        {
            \ForEach{$(O_1, P_{11}, P_{12}) \in \shadowofudgraph{B_1}$ \label{alg-counting-MEC-of-general-graph:foreach-2-start}}
        {
            \ForEach{$(O_2, P_{21}, P_{22}) \in \shadowofudgraph{B_2}$ \label{alg-counting-MEC-of-general-graph:foreach-3-start}}
        {
            \If{$O \in \mathcal{E}(O_1, P_{11}, P_{12}, O_2, P_{21}, P_{22})$ \label{alg-counting-MEC-of-general-graph:if-start}}
            {
                $(P_1, P_2) \leftarrow \EPF{O, P_{11}, P_{12}, P_{21}, P_{22}}$ \label{alg-counting-MEC-of-general-graph:DPF-init}
                
                $(O', P_1', P_2') \leftarrow \mathcal{P}(O, P_1, P_2, R_1\cup N(R_1, G))$. \label{alg-counting-MEC-of-general-graph:O'-P1'-P2'-initialization}
                
                $F(O', P_1', P_2') = F(O', P_1',P_2') + F_1(O_1, P_{11}, P_{12})\times F_2(O_2, P_{21}, P_{22})$ \label{alg-counting-MEC-of-general-graph:S-update}
            \label{alg-counting-MEC-of-general-graph:if-end}}
        \label{alg-counting-MEC-of-general-graph:foreach-3-end}}
        \label{alg-counting-MEC-of-general-graph:foreach-2-end}}
        }\label{alg-counting-MEC-of-general-graph:foreach-1-end}
        
        \KwRet $F$ \label{alg-counting-MEC-of-general-graph:final-return-statement}
        
    }
 
\end{algorithm} \begin{algorithm}
\caption{Is\textunderscore Extension}
\label{alg:is-extension}
\SetAlgoLined
\SetKwInOut{KwIn}{Input}
\SetKwInOut{KwOut}{Output}
\SetKwFunction{Is_extension}{Is_extension}
\KwIn{An MECs $O$, and two triples $(O_1, P_{11}, P_{12})$ and $(O_2, P_{21}, P_{22})$ such that\\ 
$H$ is an undirected graph, $H_1$ and $H_2$ are induced subgraphs of $H$ such that\\ $I = V_{H_1}\cap V_{H_2}$ is a vertex separator of $H$ that separates $V_{H_1}\setminus I$ and $V_{H_2}\setminus I$,\\
for $a\in \{1,2\}$, $S_a \subseteq V_{H_a}$ such that $S_1\cap S_2 = I$,\\
for $a\in \{1,2\}$, $(O_a, P_{a1}, P_{a2}) \in \shadowofudgraph{H_a[S_a\cup N(S_a, H_a)]}$, and \\
$O \in \setofpartialMECs{H[S_1\cup S_2\cup N(S_1\cup S_2, H)]}$
}
    \KwOut{ 1 : if $O \in \mathcal{E}(O_1, P_{11}, P_{12}, O_2, P_{21}, P_{22})$,\\
     \hspace{4pt}0 : otherwise.
    }

    \ForEach{$a\in \{1,2\}$\label{alg:is-extension:first-for-each-start}}
    {
        \If{$\exists u\rightarrow v \in O_a$ such that $u\rightarrow v\notin O$}
        {
            \KwRet 0.
        }
    }\label{alg:is-extension:first-for-each-end}
    
    \ForEach{$a\in \{1,2\}$\label{alg:is-extension:second-for-each-start}}
    {
        \ForEach{$u,v,w \in V_{O_a}$\label{alg:is-extension:third-for-each-start}}
        {
\If{$u\rightarrow v\leftarrow w$ is a v-structure in $O$ but not a v-structure in $O_a$}
            {
                \KwRet 0.
            }
            
        }\label{alg:is-extension:third-for-each-end}
    }\label{alg:is-extension:second-for-each-end}

    \ForEach{$a\in \{1,2\}$\label{alg:is-extension:forth-for-each-start}}
    {
        \ForEach{$u-v \in O_a$\label{alg:is-extension:fifth-for-each-start}}
        {
            \If{$u\rightarrow v \in O$\label{alg:is-extension:first-if-start}}
            {
                \If{$u\rightarrow v$ is not strongly protected in $O$ \label{alg:is-extension:u-v-is-sp}}
                {
                    \If{$\nexists x-y \in O_a$ such that $x\rightarrow y \in O$ and $P_{a1}((x,y),(u,v)) = 1$ \label{alg:is-extension:edge-is-directed-due-to-path}}
                    {
                        \If{$\nexists x-y \in O_a$ such that $x\rightarrow y \in O$ and $P_{a2}((x,y),v) = P_{a2}((v,u),x) = 1$ \label{alg:is-extension:edge-is-directed-due-to-cycle}}
                        {
                            \KwRet 0.
                        }
                    }
                }
            }\label{alg:is-extension:first-if-end}

            \If{$u-v\in O$\label{alg:is-extension:second-if-start}}
            {
                \If{$\exists x-y \in O_a$ such that $x\rightarrow y \in O$ and $P_{a1}((x,y),(u,v)) = 1$ \label{alg:is-extension:edge-is-ud-even-after-existence-of-path}}
                {
                    \KwRet 0.
                }
                \If{$\exists x-y \in O_a$ such that $x\rightarrow y \in O$ and $P_{a2}((x,y),v) = P_{a2}((v,u),x) = 1$\label{alg:is-extension:edge-is-ud-even-after-existence-of-cycle}}
                {
                    \KwRet 0.
                }
            }\label{alg:is-extension:second-if-end}
            
        }\label{alg:is-extension:fifth-for-each-end}
    }\label{alg:is-extension:forth-for-each-end}

    $(P_1, P_2) \leftarrow \text{DPF}(O, P_{11}, P_{12}, P_{21}, P_{22})$ \label{alg:is-extension:P1-P2-init}

    \ForEach{$((u,v), (x,y)) \in E_O \times E_O$ \label{alg:is-extension:sixth-for-each-start}}
    {
        \If{$(u,v)\neq (x,y)$ and $P_1((u,v),(x,y)) = P_1((x,y),(u,v)) =1$}
        {
            \KwRet 0
        }
    }\label{alg:is-extension:sixth-for-each-end}

    \KwRet 1.
\end{algorithm}
 \begin{algorithm}[ht]
\caption{Projection($O, P_1, P_2, X$)}
\label{alg:projection}
\SetAlgoLined
\SetKwInOut{KwIn}{Input}
\SetKwInOut{KwOut}{Output}
\SetKwFunction{Projection}{Projection}
\KwIn{A partial MEC $O$, two functions $P_1$ and $P_2$, and $X\subseteq V_O$ such that \\
 $P_1: E_O \times E_O \rightarrow \{0,1\}$ and $P_2:E_O\times V_O \rightarrow \{0,1\}$
}
    \KwOut{ $\mathcal{P}(O, P_1, P_2, X)$ \\
} 
    
    $O' \leftarrow O[X]$\label{alg:projection:O'-is-induced-subgraph-of-O}

    $P_1': E_{O'}\times E_{O'} \rightarrow \{0,1\}$ \label{alg:projection:P_1'-is-projection-of-P1}

    $P_2': E_{O'}\times X \rightarrow \{0,1\}$ \label{alg:projection:P_2'-is-projection-of-P2}

    \ForEach{$((u,v), (x,y)) \in E_{O'}\times E_{O'}$ \label{alg:projection:first-for-each-start}}
    {
        $P_1'((u,v), (x,y)) \leftarrow P_1((u,v),(x,y))$
    }\label{alg:projection:first-for-each-end}

    \ForEach{$((u,v), w) \in E_{O'}\times V_{O'}$\label{alg:projection:second-for-each-start}}
    {
        $P_2'((u,v), w) \leftarrow P_2((u,v),w)$
    }\label{alg:projection:second-for-each-end}

   \KwRet $(O', P_1', P_2')$ 
\end{algorithm} Based on the above discussion, we construct \cref{alg:counting-MEC} to count MECs of any undirected connected graph $G$.\Cref{alg:counting-MEC} takes an undirected graph $G$ as its input, and returns $|\setofMECs{G}|$ as the output.
\begin{lemma}
\label{lem:alg:counting-MEC-is-valid}
Let $G$ be an undirected connected graph. Given input $G$, \cref{alg:counting-MEC} returns $|\setofMECs{G}|$.
\end{lemma}
\begin{proof}
\Cref{alg:counting-MEC} constructs a tree decomposition $(\mathcal{X} = \{X_1, X_2, \ldots, X_l\},T)$ of $G$ (at line~\ref{alg-counting-MEC:tree-decomposition} of \cref{alg:counting-MEC}). It then invokes \cref{alg:counting-MEC-of-general-graph} with input $G$, $(\mathcal{X}, T)$, and $X_1$. Let $A = G[X_1\cup N(X_1, G)]$. \Cref{lem:alg:counting-MEC-of-general-graph-is-valid} demonstrates that \cref{alg:counting-MEC-of-general-graph} returns a function $F$ with domain $\shadowofudgraph{A}$ and co-domain $\mathbb{Z}$ such that for each triplet $(O, P_1, P_2) \in \shadowofudgraph{A}$, $F(O, P_1, P_2) = |\setofMECs{G, O, P_1, P_2}|$. Lines \ref{alg:counting-MEC:foreach-1-start}--\ref{alg:counting-MEC:foreach-1-end} of \cref{alg:counting-MEC} implement \cref{lem:partition-of-MECs-of-H} to obtain $|\setofMECs{G}|$. This confirms the validity of \cref{lem:alg:counting-MEC-is-valid}.

\end{proof}

\begin{lemma}
\label{lem:alg:counting-MEC-of-general-graph-is-valid}
Let $G$ be an undirected graph, $(\mathcal{X}, T)$ be a tree decomposition of $G$, $R_1 \in \mathcal{X}$ be the root node of $T$, and $A' = G[R_1 \cup N(R_1, G)]$. For input $G$, $(\mathcal{X}, T)$, and $R_1$, \cref{alg:counting-MEC-of-general-graph} returns a function $F$ with domain $\shadowofudgraph{A'}$ and co-domain $\mathbb{Z}$ such that for any triplet $(O, P_1, P_2) \in \shadowofudgraph{A'}$, $F(O, P_1, P_2) = |\setofMECs{G, O, P_1,P_2}|$.
\end{lemma}
\begin{proof}
If the tree $T$ has only one node, the root node $R_1$, i.e., $\mathcal{X} = \{R_1\}$, then \cref{alg:counting-MEC-of-general-graph} calls \cref{alg:counting-MEC-of-general-graph-base-case}. \Cref{lem:alg:counting-MEC-of-general-graph-base-case-is-valid} shows that \cref{alg:counting-MEC-of-general-graph-base-case} returns the required function $F$. If $T$ has more than one node, we pick the last child $R_2$ of $R_1$ in $T$ (line-\ref{alg-counting-MEC-of-general-graph-R'-initialization} of \cref{alg:counting-MEC-of-general-graph}), and cut the edge $R_1-R_2$ of $T$ (line-\ref{alg-counting-MEC-of-general-graph-cut-edge-R-R'} of \cref{alg:counting-MEC-of-general-graph}). This gives us two induced subgraphs $T_1$ and $T_2$ of $T$ such that $T_1$ contains $R_1$, and $T_2$ contains $R_2$ (lines \ref{alg-counting-MEC-of-general-graph-T1-initialization}-\ref{alg-counting-MEC-of-general-graph-T2-initialization}), and two induced subgraphs $G_1$ and $G_2$ of $G$ such that $T_1$ represents $G_1$, and $T_2$ represents $G_2$. 
The two trees $T_1$ and $T_2$ have sets of nodes $\mathcal{X}_1$ and $\mathcal{X}_2$, respectively (lines \ref{alg-counting-MEC-of-general-graph-X1-initialization}-\ref{alg-counting-MEC-of-general-graph-X2-initialization} of \cref{alg:counting-MEC-of-general-graph}). We define $A = G[R_1 \cup R_2 \cup N(R_1 \cup R_2, G)]$, $A' = G[R_1 \cup N(R_1, G)]$, $B_1 = G_1[R_1 \cup N(R_1, G)]$, and $B_2 = G_2[R_2 \cup N(R_2, G)]$ for simplification of our analysis (lines \ref{alg-counting-MEC-of-general-graph-A'-initialization}-\ref{alg-counting-MEC-of-general-graph-B2-initialization} of \cref{alg:counting-MEC-of-general-graph}). 
We observe that $G = G_1 \cup G_2$, and $R_1 \cap R_2$ is a vertex separator of $G$ (from \cref{prop:tree-decomposition-property}). We find that this scenario resembles what we discussed after \cref{obs:relation-between-MECs-of-G-i-j-G-i-j-1-G-ij}. Here, $G$ resembles $H$, $G_1$ and $G_2$ resemble $H_1$ and $H_2$ respectively, $R_1$ and $R_2$ resemble $S_1$ and $S_2$. 
Thus, to implement \cref{lem:combination-of-two-lemmas}, we need the functions $F_1$ and $F_2$ such that for any triplet  $(O_1, P_{11}, P_{12}) \in \shadowofudgraph{B_1}$, $F_1(O_1, P_{11}, P_{12}) = |\setofMECs{G_1, O_1, P_{11}, P_{12}}|$, and for any triplet  $(O_2, P_{21}, P_{22}) \in \shadowofudgraph{B_2}$, $F_2(O_2, P_{21}, P_{22}) = |\setofMECs{G_2, O_2, P_{21}, P_{22}}|$. Recursively, we compute these two functions $F_1$ and $F_2$ (lines \ref{alg-counting-MEC-of-general-graph-F1-initialization}-\ref{alg-counting-MEC-of-general-graph-F2-initialization} of \cref{alg:counting-MEC-of-general-graph}). After obtaining $F_1$ and $F_2$, we implement \cref{lem:combination-of-two-lemmas} (lines \ref{alg-counting-MEC-of-general-graph:foreach-1-start}-\ref{alg-counting-MEC-of-general-graph:foreach-1-end}) to compute the required function $F$.
\end{proof}

\begin{lemma}
\label{lem:alg:counting-MEC-of-general-graph-base-case-is-valid}
Let $G$ be an undirected graph. For input $G$, \cref{alg:counting-MEC-of-general-graph-base-case} returns a function $F$ such that for any triplet $(O, P_1, P_2)$ where $(O, P_1, P_2) \in \shadowofudgraph{G}$, and $F(O, P_1, P_2) = |\setofMECs{G, O, P_1,P_2}|$.
\end{lemma}
\begin{proof}
We initialize the output function $F$ with $F(O, P_1, P_2) = 0$ for all triplets $(O, P_1, P_2) \in \shadowofudgraph{G}$ (lines \ref{alg-counting-MEC-of-general-graph-base-case:for-each-0-start}--\ref{alg-counting-MEC-of-general-graph-base-case:for-each-0-end}).  
For each MEC $O$ of $G$, we compute TFP($O$) and update $F$ with $F(O, P_1, P_2) = 1$ (lines \ref{alg-counting-MEC-of-general-graph-base-case:P1-P2-init} and \ref{alg-counting-MEC-of-general-graph-base-case:update-F}). According to \cref{lem:alg:tfp-is-valid}, TFP($O$) provides a pair of functions $(P_1, P_2)$ such that
\begin{enumerate}
    \item $P_1:E_{O} \times E_{O} \rightarrow \{0,1\}$, and 
  $P_1((u,v),(x,y)) = 1$ if and only if $(u,v)\neq (x,y)$ and there exists a \tfp{} exists from $(u,v)$ to $(x,y)$ in $O$, and 
  
  \item $P_2:E_{O} \times V_{O} \rightarrow \{0,1\}$, and $P_2((u,v),w) = 1$ if and only if $v\neq w$ and a \tfp{} exists from $(u,v)$ to $w$ in $O$.
\end{enumerate}
 This implies that for any triple $(O, P_1, P_2) \in \shadowofudgraph{G}$, $F(O, P_1, P_2) =1$ if $O$ is an MEC of $G$ and $(O, P_1, P_2)$ is the shadow of $O$ on $V_G$; otherwise, $F(O, P_1, P_2) =0$. 
We now demonstrate that for each $(O, P_1, P_2) \in \shadowofudgraph{G}$, $F(O, P_1, P_2) = |\setofMECs{G, O, P_1, P_2}|$. 
According to \cref{def:subsets-of-MEC-based-on-O-P1-and-P2}, $|\setofMECs{G, O, P_1, P_2}|$ is the count of MECs $M$ such that $(O, P_1, P_2)$ is the shadow of $M$ on $V_O$. 
Since for any partial MEC $O \in \setofpartialMECs{G}$, $V_O = V_G$, if $M \in \setofMECs{G, O, P_1, P_2}$, then $M = O$  and $(P_1, P_2) = TFP(O)$. This implies that if $O$ is not an MEC of $G$, then $|\setofMECs{G, O, P_1, P_2}| = 0$ regardless of $P_1$ and $P_2$, and if $O$ is an MEC of $G$, then if $(P_1, P_2) = TFP(O)$, then $|\setofMECs{G, O, P_1, P_2}| = 1$, else $|\setofMECs{G, O, P_1, P_2}| = 0$. This demonstrates that for each $(O, P_1, P_2) \in L$, $F(O, P_1, P_2) = |\setofMECs{G, O, P_1, P_2}|$. This completes the proof of \cref{lem:alg:counting-MEC-of-general-graph-base-case-is-valid}.
\end{proof}

\begin{lemma}
    \label{lem:alg:projection-isvalid}
    Let $O$ be a partial MEC, $G = \skel{O}$, $X\subseteq V_G$, and $P_1: E_G \times E_G \rightarrow \{0,1\}$ and $P_2: E_G \times V_G \rightarrow \{0,1\}$ be two functions. For the input $O, P_1, P_2$, and $X$, \cref{alg:projection} outputs $\mathcal{P}(O, P_1, P_2, X)$.
\end{lemma}
\begin{proof}
    Let \cref{alg:projection} return $(O', P_1', P_2')$. 
    As per line~\ref{alg:projection:O'-is-induced-subgraph-of-O} in \cref{def:projection}, $O' = O[X]$. This demonstrates that $O$ satisfies \cref{item-1-of-def:proj-of-partial-MEC-and-functions} in \cref{def:proj-of-partial-MEC-and-functions}. Lines~\ref{alg:projection:first-for-each-start}--\ref{alg:projection:first-for-each-end} of \cref{alg:projection} demonstrate that $P_1'$ satisfies \cref{item-2-of-def:proj-of-partial-MEC-and-functions} in \cref{def:proj-of-partial-MEC-and-functions}. Similarly, lines~\ref{alg:projection:second-for-each-start}--\ref{alg:projection:second-for-each-end} of \cref{alg:projection} show that $P_2'$ satisfies \cref{item-3-of-def:proj-of-partial-MEC-and-functions} in \cref{def:proj-of-partial-MEC-and-functions}. Thus, according to \cref{def:proj-of-partial-MEC-and-functions}, $(O', P_1', P_2') = \mathcal{P}(O, P_1, P_2, X)$. 
\end{proof}

\begin{lemma}
    \label{lem:alg:is_extension-is-valid}
    Let $H$ be an undirected graph, $H_1$ and $H_2$ be two induced subgraphs of $H$ such that $I = V_{H_1}\cap V_{H_2}$ is a vertex separator of $H$, separating $V_{H_1}\setminus I$ and $V_{H_2}\setminus I$. For $a\in \{1,2\}$, let $S_a \subseteq V_{H_a}$ such that $S_1\cap S_2 = I$. For each $a\in \{1,2\}$, let $(O_a, P_{a1}, P_{a2}) \in \shadowofudgraph{H_a[S_a\cup N(S_a, H_a)]}$. Additionally, let $O \in \setofpartialMECs{H[S_1\cup S_2\cup N(S_1\cup S_2, H)]}$.
    
    For the input $O, O_1, O_2, P_{11}, P_{12}, P_{21}$, and $P_{22}$, \cref{alg:is-extension} returns 1 if $O \in \mathcal{E}(O_1, P_{11}, P_{12}, O_2, P_{21}, P_{22})$, otherwise it returns 0.
\end{lemma}
\begin{proof}
    By the definition of extension, \cref{def:extension-of-O1-O2-P11-P12-P21-P22}, if $O \in \mathcal{E}(O_1, P_{11}, P_{12}, O_2, P_{21}, P_{22})$, then $O$ satisfies \cref{item-1-of-def:extension-of-O1-O2-P11-P12-P21-P22,item-2-of-def:extension-of-O1-O2-P11-P12-P21-P22,item-3-of-def:extension-of-O1-O2-P11-P12-P21-P22,item-4-of-def:extension-of-O1-O2-P11-P12-P21-P22} of \cref{def:extension-of-O1-O2-P11-P12-P21-P22}.

    Lines~\ref{alg:is-extension:first-for-each-start}--\ref{alg:is-extension:first-for-each-end} handle cases where $O$ does not satisfy \cref{item-1-of-def:extension-of-O1-O2-P11-P12-P21-P22} of \cref{def:extension-of-O1-O2-P11-P12-P21-P22}, causing \cref{alg:is-extension} to return 0.

    If $O$ satisfies lines~\ref{alg:is-extension:first-for-each-start}--\ref{alg:is-extension:first-for-each-end} then each directed edge in $O_1$ and $O_2$ is also a directed edge in $O$. This implies for $a\in \{1,2\}$, $\mathcal{V}(O_a) \subseteq \mathcal{V}(O[V_{O_a}])$, since from the construction, $\skel{O_a} = \skel{O[V_{O_a}]}$. Lines~\ref{alg:is-extension:second-for-each-start}--\ref{alg:is-extension:second-for-each-end} deal with the verification of $\mathcal{V}(O[V_{O_a}]) \subseteq \mathcal{V}(O_a)$.  
      If $O$ does not satisfy $\mathcal{V}(O[V_{O_a}]) \subseteq \mathcal{V}(O_a)$, then \cref{alg:is-extension} returns 0. 
     Therefore, if $O$ passes lines~\ref{alg:is-extension:first-for-each-start}--\ref{alg:is-extension:second-for-each-end} and moves ahead then $O$ satisfies  \cref{item-2-of-def:extension-of-O1-O2-P11-P12-P21-P22} of \cref{def:extension-of-O1-O2-P11-P12-P21-P22}.

    Lines~\ref{alg:is-extension:forth-for-each-start}--\ref{alg:is-extension:forth-for-each-end} cover cases where $O$ does not satisfy \cref{item-3-of-def:extension-of-O1-O2-P11-P12-P21-P22} of \cref{def:extension-of-O1-O2-P11-P12-P21-P22}, causing \cref{alg:is-extension} to return 0. 

    Lines~\ref{alg:is-extension:first-if-start}--\ref{alg:is-extension:first-if-end} and lines~\ref{alg:is-extension:second-if-start}--\ref{alg:is-extension:second-if-end} check the conditions for \cref{item-3-of-def:extension-of-O1-O2-P11-P12-P21-P22} of \cref{def:extension-of-O1-O2-P11-P12-P21-P22}.

    Lines~\ref{alg:is-extension:sixth-for-each-start}--\ref{alg:is-extension:sixth-for-each-end} address cases where $O$ does not satisfy \cref{item-4-of-def:extension-of-O1-O2-P11-P12-P21-P22} of \cref{def:extension-of-O1-O2-P11-P12-P21-P22}, leading \cref{alg:is-extension} to return 0.

    This implies that if \cref{alg:is-extension} returns 1 for the inputs $O, O_1, O_2, P_{11}, P_{12}, P_{21}, P_{22}$, then $O$ satisfies \cref{item-1-of-def:extension-of-O1-O2-P11-P12-P21-P22,item-2-of-def:extension-of-O1-O2-P11-P12-P21-P22,item-3-of-def:extension-of-O1-O2-P11-P12-P21-P22,item-4-of-def:extension-of-O1-O2-P11-P12-P21-P22} of \cref{def:extension-of-O1-O2-P11-P12-P21-P22}. By \cref{def:extension-of-O1-O2-P11-P12-P21-P22}, this further implies that $O \in \mathcal{E}(O_1, P_{11}, P_{12}, O_2, P_{21}, P_{22})$. This concludes the proof of \cref{lem:alg:is_extension-is-valid}.
\end{proof}

 \section{Time Complexity}
\label{sec:time-complexity}

\begin{lemma}
    \label{lem:time-complexity-of-alg:TFP}
    For an input chain graph $G$ with chordal undirected connected components, the time complexity of \cref{alg:tfp} is $O(|E_G|^5)$.
\end{lemma}
\begin{proof}
    The running time of lines \ref{alg:tfp:first-for-each-start}--\ref{alg:tfp:first-for-each-end} in \cref{alg:tfp} is $O(|E_G|^2)$, as the first foreach loop iterates ${|E_G|}^2$ times, and each iteration takes $O(1)$ time.
    Similarly, the running time of lines \ref{alg:tfp:second-for-each-start}--\ref{alg:tfp:second-for-each-end} is $O(|V_G|\cdot|E_G|)$, since the second foreach loop iterates $|V_G|\cdot|E_G|$ times, and each iteration takes $O(1)$ time.
    Furthermore, the running time of lines \ref{alg:tfp:third-for-each-start}--\ref{alg:tfp:third-for-each-end} is $O(|V_G|^3)$, as the third foreach loop iterates ${|V_G|}^3$ times, and each iteration takes $O(1)$ time.

    In each iteration of the while loop at lines \ref{alg:tfp:first-while-loop-start}--\ref{alg:tfp:first-while-loop-end}, the algorithm updates $P_1$ by setting $P_1((u,v), (x,y)) = 1$ for a pair of edges $(u,v)$ and $(x,y)$ where $P_1((u,v),(x,y)) = 0$ at the iteration's beginning. As the number of such pairs is at most ${|E_G|}^2$, the while loop iterates at most ${|E_G|}^2$ times. Each iteration takes $O({|E_G|}^3)$ time to find three edges $(u,v), (x,y), (z_1, z_2)$ satisfying conditions such as $P_1((u,v),(x,y)) = 0$, $(u,v)\neq (x,y)$, and $P_1((u,v),(z_1,z_2)) = P_1((z_1,z_2),(x,y)) = 1$. This implies that the running time of lines \ref{alg:tfp:first-while-loop-start}--\ref{alg:tfp:first-while-loop-end} is $O({|E_G|}^5)$.

    Similarly, in each iteration of the while loop at lines \ref{alg:tfp:second-while-loop-start}--\ref{alg:tfp:second-while-loop-end}, the algorithm updates $P_2$ by setting $P_2((u,v), w) = 1$ for an edge $(u,v)$ and a node $w$, provided that $P_2((u,v),w) = 0$ at the start of the iteration. The number of such edge-node pairs can be at most $|V_G|\cdot |E_G|$, causing the while loop to iterate at most $|V_G|\cdot |E_G|$ times. Each iteration takes $O(|V_G|\cdot{|E_G|}^2)$ time to find two edges $(u,v)$ and $(z_1, z_2)$ and a node $w$ satisfying conditions such as $P_2((u,v),w) = 0$, $v\neq w$, and $P_1((u,v),(z_1,z_2)) = P_2((z_1,z_2),w) = 1$. Thus, the running time complexity of lines \ref{alg:tfp:second-while-loop-start}--\ref{alg:tfp:second-while-loop-end} is $O(|V_G|^2 \cdot {|E_G|}^3)$.

    Combining the above, the total running time complexity of the algorithm is $O(|E_G|^5)$.
\end{proof}

\begin{lemma}
    \label{lem:time-complexity-of-DPF-function}
    For the input elements $O, P_{11}, P_{12}, P_{21}$, and $P_{22}$, the time complexity of \cref{alg:constructDPF} is $O(|E_O|^5)$.
\end{lemma}
\begin{proof}
    Line \ref{alg:DPF:P_1-P_2-init} of \cref{alg:constructDPF} takes $O(|E_O|^5)$ time as it calls \cref{alg:tfp}.
    Lines \ref{alg:DPF:foreach-2-a-start}--\ref{alg:DPF:foreach-2-a-end} run for $O(|E_O|^2)$ time, where the foreach loop at line \ref{alg:DPF:foreach-2-a-start} iterates $O(|E_O|^2)$ times, with each iteration taking $O(1)$ time.
    Similarly, lines \ref{alg:DPF:foreach-2-b-start}--\ref{alg:DPF:foreach-2-b-end} run for $O(|V_O| \cdot |E_O|)$ time. The foreach loop at line \ref{alg:DPF:foreach-2-b-start} iterates for $O(|V_O| \cdot |E_O|)$ time, and each iteration takes $O(1)$ time.
    Thus, the overall running time of lines \cref{alg:DPF:foreach-2-start}--\ref{alg:DPF:foreach-2-end} is $O(|E_O|^2)$.

    Similar to \cref{alg:tfp}, the running time complexity of lines \ref{alg:DPF:while-1-start}--\ref{alg:DPF:while-1-end} and lines \ref{alg:DPF:while-2-start}--\ref{alg:DPF:while-2-end} is $O(|E_O|^5)$ and $O(|V_O|^2 \cdot |E_O|^3)$, respectively. Therefore, the time complexity of \cref{alg:constructDPF} is $O(|E_O|^5)$. 
\end{proof}

\begin{theorem}
\label{thm:time-complexity-of-alg:counting-MEC}
For an input undirected graph $G$, the time complexity of \cref{alg:counting-MEC} is $O(n(2^{O(k^4 \cdot \delta^4)}+ n^2))$, where $\delta$ represents the degree of $G$, and $k$ denotes the treewidth of $G$.
\end{theorem}
\begin{proof}
For a graph of treewidth $k$, we can utilize Korhonen's algorithm (\cite{korhonen2022single}) to construct a tree decomposition $(\mathcal{X}, T)$ of width $k\leq d\leq 2k+1$ at line \ref{alg-counting-MEC:tree-decomposition} of \cref{alg:counting-MEC}. Korhonen provides an algorithm (\cite{korhonen2022single}) that, for a given $n$-vertex graph $G$ of treewidth $k$, outputs a tree decomposition of $G$ of width $d$ such that $k \leq d \leq 2k + 1$ in time $2^{O(k)} \cdot n$,.

At line \ref{alg:counting-MEC:call-count-MEC-2}, we compute the required function $F$ using \cref{alg:counting-MEC-of-general-graph}. From \cref{lem:time-complexity-of-alg:counting-MEC-of-general-graph}, it takes $O(n(2^{O(d^4 \cdot \delta^4)}+ n^2))$ time.
Since the width of the tree decomposition is $d$, the size of $X_1$ is $d+1$. Also, as the degree of the graph is $\delta$, the number of nodes in $X_1\cup N(X_1, G)$ is at most $(1+\delta)(1+d)$. Therefore, the size of $A$ (at line \cref{alg:counting-MEC:A}) is $O(d^2\delta^2)$, and the number of partial MECs of $A$ is $O(3^{d^2\delta^2})$, as for each edge $u-v \in A$, in a partial $O \in \setofpartialMECs{A}$, there are three possibilities: either $u\rightarrow v \in O$ or $v\rightarrow u \in O$ or $u-v\in O$. 

For a partial MEC $O$ of $A$, (a) the number of distinct functions $P_1:E_O \times E_O \rightarrow \{0,1\}$ is at most $2^{{|E_O|}^2} = 2^{O(d^4\delta^4)}$, and (b) the number of distinct functions $P_2:E_O \times V_O \rightarrow {0,1}$ is at most $2^{{|E_O|\cdot |V_O|}} = 2^{O(d^3\delta^3)}$.

This implies the number of shadows $(O, P_1, P_2) \in \shadowofudgraph{A}$ is $O(2^{O(d^4\delta^4)})$.
This implies that the  foreach loop at line \ref{alg:counting-MEC:foreach-1-start} runs for $O(2^{O(d^4\delta^4)})$ iterations.
Each iteration of the foreach loop (line \ref{alg:counting-MEC:update-sum}) takes $O(1)$ time.
Thus, the total running time of lines \ref{alg:counting-MEC:foreach-1-start}--\ref{alg:counting-MEC:foreach-1-end} is $O(2^{O(d^4\delta^4)})$.

Since $k\leq d\leq 2k+1$, the total running time of the algorithm is $O(n(2^{O(k^4 \cdot \delta^4)}+ n^2))$.
\end{proof}

\begin{lemma}
    \label{lem:time-complexity-of-brute-force-count}
    For an input undirected graph $G$, the time complexity of \cref{alg:counting-MEC-of-general-graph-base-case} is $O(2^{O(|E_G|^2)})$.
\end{lemma}
\begin{proof}
    For any edge $u-v\in E_G$, in a partial MEC $O$ of $G$, there are three possibilities: either $u-v\in O$, or $u\rightarrow v \in O$, or $v\rightarrow u \in O$. Therefore, the number of partial MECs of $G$ is at most $3^{|E_G|}$. 
    For any partial MEC $O$ of $G$, the number of different functions $P_1: E_O \times E_O \rightarrow \{0,1\}$ is $2^{|E_O|^2} = O(2^{|E_G|^2})$. 
    Similarly, for any partial MEC $O$ of $G$, the number of different functions $P_2: E_O \times V_O \rightarrow \{0,1\}$ is $2^{|E_O|\cdot |V_O|} = O(2^{|E_G|\cdot |V_G|})$. This implies the number of shadows of $G$ is $O(2^{|E_G|^2})$.
    Therefore, the construction of $F$ and its initialization (lines \ref{alg-counting-MEC-of-general-graph-base-case:F-construction}--\ref{alg-counting-MEC-of-general-graph-base-case:for-each-0-end} of \cref{alg:counting-MEC-of-general-graph-base-case}) takes $O(2^{O(|E_G|^2)})$ time.

    As discussed above, the number of partial MECs with skeleton $G$ is $3^{|E_G|}$. Since each MEC is a partial MEC, therefore the number of MECs with skeleton $G$ is also $3^{|E_G|}$.
    Therefore, the number of iterations in the second foreach loop (lines \ref{alg-counting-MEC-of-general-graph-base-case:for-each-1-start}--\ref{alg-counting-MEC-of-general-graph-base-case:for-each-1-end}) is at most $3^{|E_G|}$. Verifying whether $O$ is an MEC or not, at line \ref{alg-counting-MEC-of-general-graph-base-case:for-each-1-start}, will take $O(|V_G| + |E_G|)$ time, as verifying each condition of \cref{item-1-theorem-nec-suf-cond-for-MEC,item-2-theorem-nec-suf-cond-for-MEC,item-3-theorem-nec-suf-cond-for-MEC,item-4-theorem-nec-suf-cond-for-MEC} will take $O(|V_G| + |E_G|)$ time. 
    From \cref{lem:time-complexity-of-alg:TFP}, the runtime of line \ref{alg-counting-MEC-of-general-graph-base-case:P1-P2-init} is $O(|E_G|^5)$. Therefore, the runtime of the second foreach loop (lines \ref{alg-counting-MEC-of-general-graph-base-case:for-each-1-start}--\ref{alg-counting-MEC-of-general-graph-base-case:for-each-1-end}) is $O(3^{|E_G|}\times |E_G|^5)$. Thus, the time complexity of \cref{alg:counting-MEC-of-general-graph-base-case} is $O(2^{O(|E_G|^2)})$.
\end{proof}

\begin{lemma}
    \label{lem:time-complexity-of-alg:counting-MEC-of-general-graph}
    For the input elements $G, (\mathcal{X}, T)$, and $R_1$, the time complexity of \cref{alg:counting-MEC-of-general-graph} is $O(n(2^{O(d^4 \cdot \delta^4)}+ n^2))$, where $d$ is the width of the tree decomposition $(\mathcal{X}, T)$ of $G$ and $\delta$ is the degree of $G$.
\end{lemma}
\begin{proof}
If $\mathcal{X} = \{R_1\}$, then the tree $T$ has only one node $R_1$, and $V_G = R_1$. If the width of $T$ is $d$, then $|V_G| = d+1$, and $|E_G| = O(d^2)$.
From \cref{lem:time-complexity-of-brute-force-count}, the time complexity to compute \BaseCount{}($G$) is $O(2^{O(d^4)})$.
Therefore, the runtime of lines \ref{alg-counting-MEC-of-general-graph-if-1-start}--\ref{alg-counting-MEC-of-general-graph-if-1-end} is $O(2^{O(d^4)})$.

The time complexity to run lines \ref{alg-counting-MEC-of-general-graph-R'-initialization}--\ref{alg-counting-MEC-of-general-graph-X2-initialization} and lines \ref{alg-counting-MEC-of-general-graph-A'-initialization}--\ref{alg-counting-MEC-of-general-graph-B2-initialization} is $O(|V_G| + |E_G|)$.
If the width of $T$ is $d$, then the size of each $R_1$ and $R_2$ is $d+1$. And, if the degree of $G$ is $\delta$, then the size of each $R_1\cup N(R_1, G_1)$, $R_2\cup N(R_2, G_2)$, $R_1\cup N(R_1, G)$, and $R_1\cup R_2\cup N(R_1\cup R_2, G)$ is $O(d\delta)$. This implies that the number of edges in $A, A', B_1$, and $B_2$ is $O(d^2\delta^2)$. Therefore, as explained in \cref{thm:time-complexity-of-alg:counting-MEC}, the size of each $\shadowofudgraph{A}$, $\shadowofudgraph{A'}$, $\shadowofudgraph{B_1}$, and $\shadowofudgraph{B_2}$ is $O(2^{O(d^4\delta^4)})$.

The initialization of $F$ at lines \ref{alg-counting-MEC-of-general-graph:foreach-0-start}--\ref{alg-counting-MEC-of-general-graph:foreach-0-end} will take $O(2^{O(d^4\delta^4)})$ as the size of $L'$ is $O(2^{O(d^4\delta^4)})$. 

Since the number of different PMECs of $A$ is $O(3^{O(d^2\delta^2)})$, the for loop at line \ref{alg-counting-MEC-of-general-graph:foreach-1-start} runs for $O(3^{O(d^2\delta^2)})$ iterations. Since the sizes of $\shadowofudgraph{B_1}$ and $\shadowofudgraph{B_2}$ are $O(2^{O(d^4\delta^4)})$ (as computed in \cref{lem:time-complexity-of-brute-force-count}), the for loops at lines \ref{alg-counting-MEC-of-general-graph:foreach-2-start} and \ref{alg-counting-MEC-of-general-graph:foreach-3-start} run for $O(2^{O(d^4\delta^4)})$ iterations. 
From \cref{lem:time-complexity-of-alg:is-extension}, the runtime of verifying whether $O$ is an extension of $(O_1, P_{11}, P_{12}, O_2, P_{21}, P_{22})$ or not (at line \ref{alg-counting-MEC-of-general-graph:if-start}) is $O(|E_O|^5) = O(d^{10}\delta^{10})$ (as discussed above the number of edges in $O$ is $O(d^2\delta^2)$). From \cref{lem:time-complexity-of-DPF-function}, the runtime of line \ref{alg-counting-MEC-of-general-graph:DPF-init} is $O(|E_O|^5) = O(d^{10}\delta^{10})$ time. Therefore, the overall runtime of lines \ref{alg-counting-MEC-of-general-graph:foreach-1-start}--\ref{alg-counting-MEC-of-general-graph:foreach-1-end} takes $O(2^{O(d^4\delta^4)} + |V_G| + |E_G|)$ time.

At lines \ref{alg-counting-MEC-of-general-graph-F1-initialization}--\ref{alg-counting-MEC-of-general-graph-F2-initialization}, \cref{alg:counting-MEC-of-general-graph} calls itself. Each time \cref{alg:counting-MEC-of-general-graph} calls itself, it cuts one of the edges of $T$ and breaks it into two disjoint parts $T_1$ and $T_2$. Since the number of edges in the tree decomposition $T$ is $O(|V_G|)$ (\cite{bodlaender1996efficient}), the number of times \cref{alg:counting-MEC-of-general-graph} calls itself is $O(|V_G|)$. 
This implies that the time complexity to run \cref{alg:counting-MEC-of-general-graph} is $O(|V_G|\cdot (2^{O(d^4\delta^4)}+ |V_G| + |E_G|)) = O(n(2^{O(d^4\delta^4)}+ n^2))$. This completes the proof.
\end{proof}

\begin{lemma}
    \label{lem:time-complexity-of-alg:is-extension}
    For the input elements $O, (O_1, P_{11}, P_{12})$, and $(O_2, P_{21}, P_{22})$, the time complexity of \cref{alg:is-extension} is $O({|E_O|}^5)$.
\end{lemma}
\begin{proof}
    Considering the input constraint provided at \cref{alg:is-extension}, the size of $O_1$ and $O_2$ cannot exceed the size of $O$ (since $\skel{O_1}$ and $\skel{O_2}$ are induced subgraphs of $\skel{O}$).
    Therefore, the time complexity to run lines \ref{alg:is-extension:first-for-each-start}--\ref{alg:is-extension:first-for-each-end} is $O(|E_O|)$. 
    And, the time complexity to run lines \ref{alg:is-extension:second-for-each-start}--\ref{alg:is-extension:second-for-each-end} is $O({|V_O|}^3)$.
    The verification of an edge of $O$ to check if it is strongly protected (line \ref{alg:is-extension:u-v-is-sp}) will take cumulatively $O(|V_O| + |E_O|)$ time. 
    The verification of the existence of an edge $x-y \in O_a$ such that $x\rightarrow y \in O$ and $P_{a1}((x,y),(u,v)) =1$ (lines \ref{alg:is-extension:edge-is-directed-due-to-path} and \ref{alg:is-extension:edge-is-ud-even-after-existence-of-path}) takes $O(|E_O|)$ time. 
    The verification of the existence of an edge $x-y \in O_a$ such that $x\rightarrow y \in O$ and $P_{a2}((x,y),v) = P_{a2}((v,u),x) =1$ (lines \ref{alg:is-extension:edge-is-directed-due-to-cycle} and \ref{alg:is-extension:edge-is-ud-even-after-existence-of-cycle}) takes $O(|E_O|)$ time. The foreach loop at lines \ref{alg:is-extension:third-for-each-start}--\ref{alg:is-extension:third-for-each-end} runs for $O(|E_O|)$ iterations. Therefore, the time complexity of running lines \ref{alg:is-extension:third-for-each-start}--\ref{alg:is-extension:third-for-each-end} is $O({|E_O|}^2)$ time. 
    From \cref{lem:time-complexity-of-DPF-function}, the runtime to construct a derived path function $(P_1, P_2)$ at line \ref{alg:is-extension:P1-P2-init} is $O({|E_O|}^5)$ time. 
    The foreach loop present at lines \ref{alg:is-extension:sixth-for-each-start}--\ref{alg:is-extension:sixth-for-each-end} runs for $O({|E_O|}^2)$ iterations. Each iteration of the foreach loop takes $O(1)$ time. Therefore, the running of lines 
    \ref{alg:is-extension:sixth-for-each-start}--\ref{alg:is-extension:sixth-for-each-end} takes $O({|E_O|}^2)$ time. The above discussion implies that the time complexity of \cref{alg:is-extension} is $O({|E_O|}^5)$.
\end{proof}

\begin{lemma}
    \label{lem:time-complexity-of-alg:projection}
    For the input elements $O$, $P_1$, $P_2$, and $X$, the time complexity of \cref{alg:projection} is $O({|E_O|}^2)$.
\end{lemma}
\begin{proof}
    The construction of the induced subgraph $O'$ at line \ref{alg:projection:O'-is-induced-subgraph-of-O} takes $O(|V_O| + |E_O|)$ time.
    The construction of the functions $P_1'$ and $P_2'$ at lines \ref{alg:projection:P_1'-is-projection-of-P1} and \ref{alg:projection:P_2'-is-projection-of-P2}, respectively, takes $O({|E_O|}^2)$ and $O(|V_O|\cdot |E_O|)$ time.
    The number of iterations the for loop at line \ref{alg:projection:first-for-each-start}--\ref{alg:projection:first-for-each-end} runs is $O({|E_{O'}|}^2)$, and the running time of each iteration is $O(1)$.
    Similarly, the number of iterations the for loop at line \ref{alg:projection:second-for-each-start}--\ref{alg:projection:second-for-each-end} runs is $O(|V_{O'}|\cdot |E_{O'}|)$, and the running time of each iteration is $O(1)$.
    Therefore, the running time of lines \ref{alg:projection:first-for-each-start}--\ref{alg:projection:first-for-each-end} and lines \ref{alg:projection:second-for-each-start}--\ref{alg:projection:second-for-each-end} is $O({|E_{O'}|}^2)$ and $O(|V_{O'}|\cdot |E_{O'}|)$ time, respectively. Since $O'$ is an induced subgraph of $O$, it follows that $|E_O| \geq |E_{O'}|$ and $|V_O| \geq |V_{O'}|$. This further implies that the time complexity of \cref{alg:projection} is $O({|E_O|}^2)$.
\end{proof}

\Cref{thm:time-complexity-of-alg:counting-MEC} implies the following:

\begin{theorem}
    \label{thm:main}
        For an undirected connected graph $G$, there exists an algorithm that counts the number of MECs with skeleton $G$  in $O(n(2^{O(k^4\delta^4)}+n^2))$ time, where $k$ and $\delta$ are the treewidth and the degree of $G$, respectively.
\end{theorem}

 \section{Conclusion and Open Problems}
\label{sec:conclusion}
We provide a fixed parameter tractable algorithm with runtime $O(n(2^{O(k^4\delta^4)}+ n^2))$ for the problem of counting the number of MECs with 
a given skeleton, where $n$ is the number of nodes in the input graph and $k$ and $\delta$ are the treewidth and the degree of the input graph. 

The main problem left open by this work is to either provide a fully polynomial time algorithm for this problem, or else to prove that it is computationally hard in general. 
An intermediate open problem towards the goal of understanding the complexity fully is to improve the runtime of the fixed-parameter tractable algorithm itself: perhaps using different parameters.

As stated in the introduction, we consider our results to be a first step towards understanding the computational complexity of this problem.  However, we hope that the  definitions and techniques introduced in this paper will be helpful in further study of the problem.

\end{document}